\documentclass[11pt,a4paper,twoside]{article}
\usepackage[latin1]{inputenc}
\usepackage{cprotect}
\usepackage{latexsym}
\usepackage{a4wide}
\usepackage{mathrsfs}   
\usepackage{amsthm}
\usepackage{amssymb}
\usepackage{amsmath}
\usepackage{enumerate}
\usepackage{amsfonts}
\usepackage[textsize=small]{todonotes}
\usepackage{mathrsfs}   
\usepackage{tikz}       
\usepackage{graphicx}
\usepackage{subcaption}
\allowdisplaybreaks 

\usepackage{hyperref}

\usetikzlibrary{arrows,automata}
\usepackage{tabularx}
\usepackage{multirow}
\usepackage{color}

\usepackage{ifthen} 
 \newboolean{longversion} 
 \setboolean{longversion}{true}

\newcommand{\e}{{\mathrm e}}

\usepackage{fancyhdr}

\newcommand{\Qbold} {{\mathbb Q}}

\newcommand{\Bcal}   {\mathcal{B}}
\newcommand{\Ccal}   {\mathcal{C}}

\newcommand{\Gcal}   {\mathcal{G}}

\newcommand{\Lcal}   {\mathcal{L}}

\newcommand{\Ucal}   {\mathcal{U}}

\newcommand{\Wcal}   {\mathcal{W}}

\newcommand{\nin}{\not\in}

\def\1{{\mathchoice {1\mskip-4mu\mathrm l}      
{1\mskip-4mu\mathrm l}
{1\mskip-4.5mu\mathrm l} {1\mskip-5mu\mathrm l}}}
\newcommand{\indic}[1]{\1_{\{#1\}}}

\newcommand{\indAnimal}{\1^{a}}

\DeclareMathSymbol{\expect}        {\mathord}{AMSb}{"45}
\DeclareMathSymbol{\expec}        {\mathord}{AMSb}{"45}
\DeclareMathSymbol{\prob}        {\mathord}{AMSb}{"50}
\DeclareMathSymbol{\Ibold}        {\mathord}{AMSb}{"49}
\DeclareMathSymbol{\Nbold}        {\mathord}{AMSb}{"4E}
\DeclareMathSymbol{\Rbold}        {\mathord}{AMSb}{"52}
\DeclareMathSymbol{\Zbold}        {\mathord}{AMSb}{"5A}

\newcommand{\w}{\mathbf{w}}

\newcommand{\C}{\mathbb C}
\newcommand{\R}{\Rbold}
\newcommand{\Z}{\Zbold}
\newcommand{\Zd}{\Zbold^d}

\newcommand{\sss}   { \scriptscriptstyle }

\newcommand{\conn}{\longleftrightarrow}

\newcommand{\dbc}{\Longleftrightarrow}

\newcommand{\dbct}[1]     { \stackrel{#1}{\dbc} }

\newcommand{\bb}{\underline{b}}
\newcommand{\tb}{\overline{b}}

\newcommand{\eqarray}   {\begin{eqnarray}}
\newcommand{\enarray}   {\end{eqnarray}}
\newcommand{\lbeq}[1]  {\label{e:#1}}
\newcommand{\refeq}[1] {\eqref{e:#1}}
\newcommand{\eq}{\begin{equation}}
\newcommand{\en}{\end{equation}}
\newcommand{\ben}{\begin{enumerate}}
\newcommand{\een}{\end{enumerate}}
\newcommand{\eqn}[1]{\begin{equation} #1 \end{equation}}
\newcommand{\eqan}[1]{\begin{align} #1 \end{align}}

\newcommand{\nn}{\nonumber}
\newcommand{\nnb}{\nonumber\\}

\renewcommand{\to}{\rightarrow}

\def\Zd{\mathbb{Z}^d}

\def\mA[#1]{ {\bf A}(#1)}
\def\miA[#1]{ {\bf A}^{-1}(#1)}

\def\mD[#1]{{\bf \hat D}(#1)}
\def\mE[#1]{{\bf \hat E}(#1)}
\def\mM[#1]{{\bf \hat M}_1(#1)}
\def\mMa[#1]{{\bf \hat M}_2(#1)}
\def\mE[#1]{{\bf E}_{#1}}
\def\ve[#1]{ {e}_{#1}}

\def\Isupx{\mathcal{J}}

\def\v1{{\vec 1}}
\def\mJ{{\bf J}}
\def\mI{{\bf I}}

\newcommand{\dminanimal}{17}
\newcommand{\dmintree}{16}

\def\gj{g^{\iota}_z}

\def\mPi[#1]{\hat {\bf \Pi}(#1)}
\def\mPiwoz[#1]{\hat {\bf \Pi}(#1)}
\def\mPiM[#1]{\hat {\bf \Pi}_{{\sss M}}(#1)}

\def\vXi[#1]{\vec {\hat \Xi}(#1)}
\def\vXiz[#1]{\vec {\hat \Xi}_z(#1)}
\def\vXiM[#1]{\vec {\hat \Xi}_{\sss M}(#1)}
\def\vRM[#1]{\vec {\hat R}_{\sss M}(#1)}
\def\vPsi[#1]{\vec {\hat \Psi}(#1)}
\def\vPsiT[#1]{\vec {\hat \Psi}^T(#1)}
\def\vPsiz[#1]{\vec {\hat \Psi}_z(#1)}
\def\vPsiMT[#1]{\vec {\hat \Psi}^T_{\sss M}(#1)}
\def\vPsiM[#1]{\vec {\hat \Psi}_{\sss M}(#1)}
\def\hatPhiM[#1]{\hat {\Phi}_{\sss M}(#1)}

\def\lowK{\underline {K}_{\sss \Delta F}}

\def\betaaa{\beta_{\sss \mu}}

\def\lowaf{\underline \beta_{\sss  \alpha,F}}

\def\betaap{\beta_{\sss  |\alpha,\Phi|}}
\def\betaRp{\beta_{\sss  |R,\Phi|}}
\def\lowcp{\underline \beta_{\sss  c,\Phi}}
\def\upcp{\overline \beta_{\sss  c,\Phi}}

\def\aaz{\mu_z}
\def\aabz{\bar {\mu}_z}
\def\aa{\mu}
\def\aab{\bar \mu}

\def\afz{\alpha_{{\sss F},z}}

\def\betadeltaRfzlow{ \underline{\beta}_{\sss \Delta R,F}}

\def\ssss[#1]{{\sss \text{\rm #1}}}
\def\ssc[#1]{{\sss( \text{\rm #1})}}
\def\genC[#1]{\hat {C}_{\mu_z}(#1)}
\def\genG[#1]{\hat {G}_{z}(#1)}
\def\bvtheta[#1]{\vec  {\boldsymbol \theta}_{#1}}

\def\diagRepulsiveLetter{\mathscr}

\def\absT{\mathcal{T}}

\newcommand{\threecolomntable}[3]{
\begin{table}[h!]
\caption{#1}
\label{#2}
{\small\centering
\hspace*{2cm}\begin{tabular}
{|
>{\centering\arraybackslash} m{35mm}|
>{\centering} m{35mm}|
>{\centering\arraybackslash} m{35mm}|
}
  \hline
  Condition & Diagram & Definition \\
  \hline
#3
  \hline
\end{tabular}
}
\end{table}
}

\def\projIndexsetdir[#1]{{\mathsf {IK} }(#1)}
\def\projIndexsetPoints[#1]{{\mathsf {AB} }(#1)}
\def\projIndexsetNumber[#1]{{\mathsf {NM} }(#1)}
\def\projIndexsetPointsTwo[#1]{{\mathsf {X} }(#1)}

\newcommand{\ii}{{\mathrm i}}

\newtheorem{theorem}{Theorem}[section]
\newtheorem{corollary}[theorem]{Corollary}
\newtheorem{lemma}[theorem]{Lemma}
\newtheorem{prop}[theorem]{Proposition}

\newtheorem{definition}[theorem]{Definition}
\numberwithin{equation}{section}
\numberwithin{theorem}{section}
\newtheorem{remark}[theorem]{Remark}

\title{NoBLE for lattice trees and lattice animals:\\
\iflongversion
Extended version
\fi}

\author{
        Robert Fitzner\thanks{Department of Mathematics and
        Computer Science, Eindhoven University of Technology,
        5600 MB Eindhoven, The Netherlands.
        {\tt math@fitzner.nl},{\tt rhofstad@win.tue.nl}}
        \and
        Remco van der Hofstad$^*$
}

\begin{document}
\maketitle
\begin{abstract}
We study lattice trees (LTs) and animals (LAs) on the nearest-neighbor lattice $\Zd$  in high dimensions. We prove that LTs and LAs display mean-field behavior above dimension $\dmintree$ and $\dminanimal$, respectively. Such results have previously been obtain by Hara and Slade in sufficiently high dimensions. The dimension above which their results apply was not yet specified. We rely on the non-backtracking lace expansion (NoBLE) method that we have recently developed. The NoBLE makes use of an alternative lace expansion for LAs and LTs that perturbs around non-backtracking random walk rather than simple random walk, leading to smaller corrections. The NoBLE method then provides a careful computational analysis that improves the dimension above which the result applies. Universality arguments predict that the upper critical dimension, above which our results apply, is equal to $d_c=8$ for both models, as is known for sufficiently spread-out models by the results of Hara and Slade mentioned earlier.

The main ingredients in this paper are  (a) to derive a non-backtracking lace expansion for the LT and LA two-point functions;  (b) to bound the non-backtracking lace-expansion coefficients, thus showing that our general NoBLE methodology can be applied; and (c) to obtain sharp numerical bounds on the coefficients. Our proof is complemented by a computer-assisted numerical analysis that verifies that the necessary bounds used in the NoBLE are satisfied.
\iflongversion
This extended version includes the definition of all bounding diagrams, see Appendix C.
\fi
\end{abstract}

\tableofcontents

\section{Introduction}
\label{sec-intro}
\subsection{Motivation}
\label{sec-motiv}

Lattice trees (LTs) and lattice animals (LAs) are models for branched polymers, where excluded volume creates a self-avoidance interaction between different pieces (vertices or bonds), while the branching nature corresponds to polymers whose building blocks can have covalent bonds to more than two other building blocks. They are the branching equivalents of self-avoiding walks (see \cite{MadSla93} for a detailed introduction, and a brief introduction to LTs and LAs as well). The study of LTs and LAs in low dimensions is quite hard and few results exist. Exceptions are the beautiful relation between  LTs in dimension $d$ and a hard-core lattice gas in dimension $d+2$, as predicted by Parisi and Sourlas \cite{ParSou81} and shown for a continuum model by Brydges and Imbrie \cite{BryImb03} (see also \cite{KenWin09} for a simple and elegant proof). As a result, also LTs and LAs in other dimensions have been investigated.

Like many statistical physics models, LTs and LAs are expected to have a so-called {\em upper critical dimension}, above which their behavior should be similar to a simpler model having less intricate interactions. For LTs and LAs, the upper critical dimension is believed to be $d_c=8$, while this simpler model is branching random walk (BRW). BRW has gained enormous popularity in the probability community, and its critical behavior can be proved rigorously. See for example Perkins \cite{Perk02} for super-process limits of BRWs. The main tool to prove that LTs and LAs show mean-field behavior is the so-called {\em lace expansion}, a perturbation expansion that compares LTs and LAs to BRWs. Many results have so far been proved for LTs, including the identification of some of their critical exponents by Hara and Slade \cite{HarSla92c} (which also applies to LAs), scaling limits of lattice trees of a given number of vertices by Derbez and Slade \cite{DerSla97,DerSla98}, and the structure of lattice trees that are conditioned to have long paths in them by Holmes \cite{Holm08b} (see also \cite{HofHolPer15} for a recent tightness result). Almost all of these results apply to the technically simpler setting of {\em spread-out} LTs and LAs. In this paper, our goal is to study nearest-neighbor LTs and LAs above the upper critical dimension $d_c=8$. Nearest-neighbor models are the simplest models to formulate, which makes them quite popular.

The \emph{lace expansion} was first used by Brydges and Spencer in 1985 \cite{BrySpe85} to prove mean-field behavior for weakly self-avoiding walk. Since then, it has been applied (strictly) self-avoiding walks (SAW), percolation, and lattice trees and animals \cite{HarSla90a,HarSla90b,HarSla92a,Slad87}. It has become one of the key tools (and often the only tool available) to prove mean-field behavior of statistical mechanical models above the upper critical dimension of these models. More recent extensions include oriented percolation \cite {HofSla02, NguYan93,NguYan95}, the contact process \cite{HofSak04,Saka01}, and the Ising model \cite{Saka07}.

The lace expansion is a \emph{perturbative method} in nature, and therefore, applications of the lace expansion typically necessitate a small parameter. This small parameter tends to be the degree of the underlying base graph. This is the reason why it is often convenient to work with so-called \emph{spread-out} models, where long- but finite-range connections over a (large) distance $L$ are possible, as then the degree of the base graph can be made large independently of the dimension. Thus, results can often be proved to hold, for $L$ sufficiently large, all the way up to the critical dimension of the corresponding model. The simplest, and most often studied, version of these models is, however, the nearest-neighbor model. For the nearest-neighbor model, the degree of a vertex is $2d$. Taking $2d$ large in order to obtain a small perturbation parameter then necessitates to take $d$ large in order to prove mean-field results. This leads to suboptimal results in terms of the dimension above which the results hold. A seminal exception is SAW, where Hara and Slade \cite{HarSla92b} have proved that $d\geq 5$ is sufficient for their perturbation analysis to hold, using a computer-assisted method and a detailed analysis of the perturbation terms arising through the lace expansion. For SAW, mean-field results are expected to be false in dimension $d=4$. See the work using the renormalization group to identify the logarithmic corrections to mean-field behavior by Bauerschmidt, Brydges and Slade in \cite{BauBrySla15b, BauBrySla15a} and the references therein.

In this paper, we apply lace-expansion methodology to nearest-neighbor LTs and LAs. Hara and Slade \cite{HarSla90b} derived a lace expansion for LTs and LAs, and used it to prove mean-field behavior for sufficiently spread-out LTs and LAs in the optimal $d>8$, while their results also apply to the nearest-neighbor setting in \emph{sufficiently high} dimensions. Which dimension suffices was not answered by Hara and Slade, and we revisit this question.  We derive a so-called {\em non-backtracking lace expansion} (NoBLE) and build upon the techniques of \cite{HarSla92a} to prove mean-field behavior for LTs in $d\geq \dmintree$ and for LAs for $d\geq \dminanimal$. This extends our recent work deriving general conditions under which such a NoBLE can be applied \cite{FitHof13b}, and its application to percolation for $d>10$  \cite{FitHof13d}.

The main differences of the NoBLE to the classical lace expansion method are that (i) we perturb around {\em non-backtracking random walk,} rather than around simple random walk, so that the lace-expansion coefficients are significantly smaller than in the classical lace expansion as used by Hara and Slade; (ii) our bounds on the lace-expansion coefficients are {\em matrix-based,} so as to profit maximally from the fact that loops present in the lace-expansion coefficients consist of at least four bonds in the NoBLE; and (iii) we use and provide Mathematica notebooks that implement the bounds, and that can be downloaded from the first author's website. We prove that the LT and LA two-point functions satisfy an {\em infra-red bound} that describes its singularity for small Fourier variables. Such an infrared bound immediately implies the finiteness of the so-called {\em square diagram}, and thus implies the existence of certain critical exponents, that take on mean-field values.

Our proof is computer-assisted, and relies on two key ingredients, that were also used to analyze percolation in high-dimensions \cite{FitHof13d}:
\begin{enumerate}
\item[(I)] Rigorous upper bounds on various simple random walk integrals, as first proved by Hara and Slade in \cite{HarSla92a}. This part of the analysis is {\em unchanged} compared to the Hara-Slade proof for SAWs. The crucial reason why we can use these integrals is that the non-backtracking random walk Green's function can be explicitly described in terms of the simple random walk Green's function. Our analysis requires us to compute $140$ such integrals, corresponding to convolutions of random walk Green's functions with itself at various values in $\Z^d$. We further need to compute the number of simple random walks of lengths up to $10$ ending at various points in $\Z^d$, as well as the number of related self-avoiding walks and bond-self-avoiding walks.  These bounds are performed in one Mathematica notebook;
\item[(II)] Two additional Mathematica notebooks. The first implements the computations in our general approach to the non-backtracking lace expansion (NoBLE) in \cite{FitHof13b}, and the second computes the rigorous bounds on the lace-expansion coefficients provided in the present paper, both for LTs as well as for LAs. These notebooks do nothing else than implement the bounds proved here and in \cite{FitHof13b}, and rely on nothing but (many) multiplications, additions, as well as diagonalisations of two five-by-five matrices. These computations {\em could} be performed by hand, but the use of the notebooks tremendously simplifies them.
\end{enumerate}

Let us remark that in dimensions $d\geq 100$, say, we require only the values of $5$ integrals, and no computations concerning the number of random walks ending at various locations. The additional computations are used to obtain a sharper bound on the perturbations. We next introduce the nearest-neighbor LTs and LAs that we investigate, and state our main results. Our methods initially apply to any {\em specific} dimension, but we also derive bounds that are monotone in the dimension and can be applies to all $d\geq 30$ at once. Thus, we can check the necessary dimensions up to $d=29$ one by one, and then analyse $d\geq 30$ in one go.

\subsection{The model}
\label{sec-mod}
A nearest-neighbor \emph{lattice tree} on $\Zd$ is a finite connected set of nearest-neighbor bonds that contains no cycles (closed loops).
A nearest-neighbor \emph{lattice animal} on $\Zd$ is a finite connected set of nearest-neighbor bonds, which may or may not contain cycles.
Although a tree/animal $A$ is defined as a set of bonds, for $x\in\Zd$, we write $x\in A$ to denote that $x$ is an element of a bond of $A$. The number of bonds in $A$ is denoted by $|A|$. We define $t^{\sss(a)}_n(x)$ and $t^{\sss(t)}_n(x)$ to be the number of LAs and LTs, respectively, that consist of exactly $n$ bonds and contain the origin and $x\in\Zd$.
We study LAs and LTs using the \emph{one-point function} $g_z$ and the \emph{two-point function} $\bar G_z$, that are defined by
	\begin{align}
	\lbeq{defLTLAOnePoint}
	&g^{\sss(a)}_z=\bar G^{\sss(a)}_z(0)=\sum_{A\colon A\ni 0}z^{|A|}, &g^{\sss(t)}_z={\bar G}^{\sss(t)}_z(0)=\sum_{T\colon 0\in T}z^{|T|},\\
	\lbeq{defLTLATwoPoint}
	&\bar G^{\sss(a)}_z(x)=\sum_{n=0}^\infty t^{\sss(a)}_n(x) z^n=\sum_{A\colon 0,x\in A}z^{|A|},&\bar G^{\sss(t)}_z(x)=\sum_{n=0}^\infty t^{\sss(t)}_n(x) z^n=\sum_{x\in\Zd}\sum_{T\colon 0,x\in T}z^{|T|},
	\end{align}
where we sum over animals $A$ and trees $T$ respectively, and we choose $z$ such that the above sums make sense. In what follows, we drop the superscripts $(a)$ and $(t)$ when we speak about both models simultaneously. Only when we discuss specific statements for LAs or LTs separately, the superscripts will be shown. We define the \emph{susceptibility} of LAs and LTs by
	\begin{align}
	\lbeq{defLTLAsusceptibility}
	\chi(z)=\sum_{x}\sum_{A\colon 0,x\in A} z^{|A|}
	\end{align}
and denote the radii of convergence of these sums by $z_c=z_c(d)$. As for SAW, $1/z_c$ describes the exponential growth of the number of LTs/LAs as $n$ grows.

We use the notation $\bar G_z$ for the two-point function, and use a normalized two-point function $G_z(x)=\bar G_z(x)/g_z$ for our analysis. We give the reason for this in Section \ref{sec-defineCoeff}. Intuitively, this is due to the fact that it is convenient to have that $G_z(0)=1,$ which is true for the normal two-point functions of SAW and percolation, but not for those of LTs and LAs. Since $\bar G_z(0)=g_z$ by definition, we thus prefer to work with $G_z(x)=\bar G_z(x)/g_z$ instead. Dealing properly with such one-point functions is a major ingredient of our proof.

\paragraph{Fourier transforms.}
Our analysis makes heavy use of Fourier analysis. Unless specified otherwise, $k$ always denotes an arbitrary element from the Fourier dual of the discrete lattice, which is the torus $[-\pi,\pi]^d$. The Fourier transform of a function $f\colon\Z^d\to\C$ is defined by
    \eqn{
    \lbeq{def-FourTrans}
    \hat f(k)=\sum_{x\in\Zd}f(x)\,\e^{\ii k\cdot x}.
    }
For two summable function $f,g\colon \Zd \mapsto \Rbold$, we let $f\star g$ denote their \emph{convolution}, i.e.,
    \eqn{
    \lbeq{definition-convolution}
    (f\star g)(x)= \sum_{x\in\Zd} f(y)g(x-y).
    }
We note that the Fourier transform of $f\star g$ is given by the product of $\hat f$ and $\hat g$. In particular, let $D(x)=\indic{|x|=1}/(2d)$ be the nearest-neighbor random walk transition probability, so that
    \eqn{
    \lbeq{def-Dhat}
    \hat{D}(k)=\frac{1}{2d}\sum_{x\colon |x|=1} \e^{\ii k\cdot x} = \frac{1}{d} \sum_{i=1}^d \cos(k_i).
    }

\paragraph{Critical exponents.}
It is believed that the asymptotic behavior of lattice trees and lattice animals can be described by a small number of {\em critical exponents}.
These critical exponents describe the growth of $t_n$ as $n\to\infty$ and $\bar G_z$ as $z$ approaches its radius of convergence $z_c$.
Further, they describe the typical length-scale of a lattice tree/animal, as characterized by the \emph{average radius of gyration} $R_n$
and the \emph{correlation length of order two} $\xi_2$, that are defined by
	\eqn{
	\lbeq{genTreestatement-R}
	R_{n}^2= \frac {1} {2 \sum_{x\in\Zd} t_n(x)} \sum_{x\in\Zd} |x|^2 t_n(x)
	= \frac {1} {2 \hat{t}_n(0)} \sum_{x\in\Zd} |x|^2 t_n(x),
	}
and
	\eqn{
	\lbeq{genTreestatement-xi}
	\xi_2(z)^2=\frac {\sum_{x}\sum_{A\ni 0,x}|x|^2z^{|A|}}{\sum_{x}\sum_{A\ni 0,x}z^{|A|}}
	=\frac{\sum_{n\geq 0} R_{n}\hat{t}_n(0) z^n}{\sum_{n\geq 0}\hat{t}_n(0) z^n},
	}
where $|x|=(\sum_{i=1}^d x_i^2)^{1/2}$ denotes the Euclidean norm of $x\in \Zd$.

It is believed that there exist $\gamma,\eta,\nu$ such that
	\begin{align}
	\lbeq{genTreestatement1}
	\chi(z)&\sim \frac {1} {(1-z/z_c)^\gamma},
	\qquad
	\text{ and } \qquad \xi_2(z)\sim \frac {1} {(1-z/z_c)^\nu} \text{ as } z\nearrow z_c,\\
	\lbeq{genTreestatement2}
	\bar G_{z_c}(x)&\sim \frac {1} {|x|^{2-d-\eta}},\qquad\text{ and }\qquad \hat {\bar G}_{z_c}(k)\sim |k|^{\eta-2} ,
	\end{align}
as $|x|\to\infty$ and $k\to 0$. These exponents are believed to be {\em universal}, in the sense that they do not depend on the detailed structure of the lattice. In particular, it is believed that the values of $\gamma,\eta$ and $\nu$ are the same in the nearest-neighbor setting, that we consider here, and in the spread-out setting. 
The symbol $\sim$ in \refeq{genTreestatement1}-\refeq{genTreestatement2} can have several meanings,
and we shall always mean that the critical exponent exists in the \emph{bounded-ratio sense}, meaning that there exist $0<c_1<c_2<\infty$ such that, uniformly for $z\in(0,z_c)$,
    \eqn{
    \lbeq{def-bounded-ratios}
    \frac {c_1} {(1-z/z_c)^\gamma}\leq \chi(z) \leq \frac {c_2} {(1-z/z_c)^\gamma}.
    }
As for SAW it is believed that the critical exponents $\gamma, \nu, \eta$ are related by the Fisher relation $\gamma=(2-\eta)\nu$ and it has been proven in \cite{BovFroGla86b} that $\gamma\geq 1/2$ in all dimensions. Further, it is believed that there exists an upper critical dimension $d_c$ such that the critical exponents of LT/LA in $d>d_c$ take their mean-field values, which are $\gamma=1/2, \nu=1/4, \eta=0$. These values correspond to the mean-field model of LTs and LAs, studied in \cite{BorChaHofSla99}.
It is conjectured in \cite{LubIsa79} that the upper critical dimensions of LTs and LAs are $d_c=8$. In \cite{HsuNadGra05}, site LAs and LTs are simulated and the conjectured values of the critical exponents consistent with these simulations (see \cite{HsuNadGra05} or \cite{Jens08} for the precise definition of site LAs and LTs). These computations are only done for site LTs, but as the critical exponents are expected to be universal, the values should be the same as for the bond tree/animal discussed here.
This conjecture is supported by rigorous work in \cite{HarTas87}, where it is shown that if the ``square diagram'' is finite at the critical point, as is believed for $d>8$, then the critical exponent $\gamma$ is at most $1/2$. Hara and Slade \cite{HarSla90b} proved that mean-field behavior holds for LTs and LAs in the \emph{spread-out setting} with $L$ sufficiently large and $d>8$, or for nearest-neighbor LTs and LAs in sufficiently high dimensions. The main aim of this paper is to put an exact value to what `sufficiently large' means. The authors have learnt through private communication with Takashi Hara that Hara and Slade expected the classical lace expansion to only be successful in dimensions {\em much} larger than $d_c=8$.

\subsection{Results}
\label{sec-result}
The main result of this paper is the following \emph{infrared bound:}
\begin{theorem}[Infrared bound]
\label{thm-IRB}
For nearest-neighbor lattice trees and lattice animals in dimensions $d\geq \dmintree$ and $d\geq \dminanimal$, respectively, there exist constants $\bar{A}_1(d)$ and $\bar{A}_2(d)$ such that
    	\eqn{
    	\lbeq{IRB}
    	|\hat{\bar{G}}_z(k)|\leq \frac{\bar{A}_1(d)}{\chi(z)^{-1}+z[1-\hat{D}(k)]}
	\quad\qquad
	\text{and}
	\quad\qquad
	|\hat{\bar{G}}_z(k)|\leq \frac{\bar{A}_2(d)}{1-\hat{D}(k)},
    	}
uniformly for $z\leq z_c(d)$.
\end{theorem}
\medskip
Note that Theorem \ref{thm-IRB} implies that $\eta=0$ in Fourier-space bound \refeq{genTreestatement2}, as $[1-\hat{D}(k)]\approx |k|^2$ for small $k$.
We prove Theorem \ref{thm-IRB} first for $\hat{G}_z(k)$ as this is the central quantity in the NoBLE analysis.
The analysis used for the proof  for $\hat{G}_z(k)$  creates bounds on the amplitudes $A_1(d)$ and $A_2(d)$ replacing $\bar{A}_1(d)$ and $\bar{A}_2(d)$. Since $\bar{G}_z(x)=g_z G_z(x)$, these bounds, together with an upper bound on $g_z$, imply upper bounds on $\hat{\bar{G}}_z(x)$ ,$\bar{A}_1(d)$ and $\bar{A}_2(d)$.

Our methods require a detailed analysis of both the critical value as well as the amplitudes $A_1(d)$ and $A_2(d)$. As a result, we obtain the following numerical bounds:

\begin{theorem}[Bounds on the critical value and amplitude]
\label{thm-bds-crit}
For nearest-neighbor lattice trees the following upper bounds hold:
{\rm \begin{center}
\begin{tabular}{|r||  c|c | c|c |c|c|}
  \hline
  $d$             & 16& 17   & 18 & 19  & 20  & 30 \\
  \hline
    $\max\{\bar{A}_1(d),\bar{A}_2(d)\}\leq$      &3.872 &3.501 & 3.37    &3.284 & 3.222 & 2.973 \\
    $\frac 1 {g_{z_c}} \max\{\bar{A}_1(d),\bar{A}_2(d)\}\leq$      &1.293&1.182 & 1.147   & 1.125& 1.11& 1.052  \\
    $g_{z_c}\leq$      &2.9963 & 2.9619 & 2.9383 & 2.9196 &  2.9043 & 2.8268  \\
  $(2d-1)g_{z_c} z_c(d)\leq$       &1.1023& 1.0897 &1.081 & 1.0741&  1.0685  & 1.04  \\
 \hline
\end{tabular}
\end{center}
}
\noindent
For nearest-neighbor lattice animals the following upper bounds hold:
{\rm \begin{center}
\begin{tabular}{|r|| c|  c|c |c|c|}
  \hline
  $d$              & 17 & 18 & 19  & 20  & 30 \\
  \hline
    $\max\{\bar{A}_1(d),\bar{A}_2(d)\}\leq$       &3.587& 3.41  &3.309 & 3.24 &  2.975\\
    $\frac 1 {g_{z_c}} \max\{\bar{A}_1(d),\bar{A}_2(d)\}\leq$      & 1.21&  1.158   &1.132 &1.114 &1.053  \\
    $g_{z_c}\leq$       & 2.9721&2.9454  &2.925 &  2.9086& 2.8277   \\
  $(2d-1)g_{z_c} z_c(d)\leq$      & 1.0934&1.0836 &1.0761& 1.07 & 1.0403  \\
 \hline
\end{tabular}
\end{center}
}
\end{theorem}
\bigskip

These upper bounds are a by-product of our analysis and we explain in Section \ref{sec-part-d} how they are computed using the Mathematica notebooks available at \cite{FitNoblePage}.
We are able to improve the numerical bounds in Theorem \ref{thm-bds-crit}, with reasonable efforts, but yielding an insignificant improvement. These efforts, however, would not reduce the minimal dimensions above which our results apply. It would, probably, be possible to improve these optimal dimensions slightly, but the effort required for this would be enormous. In particular, we do not expect that our results can be improved to {\em all} $d>d_c=8$ without using substantially new ingredients and insights that are way beyond our proof.

Since
    	\eqan{
    	\square(z_c)&=(\bar{G}_{z_c}\star \bar{G}_{z_c}\star \bar{G}_{z_c}\star  \bar{G}_{z_c})(0)
	=\lim_{z\nearrow z_c} (\bar{G}_{z}\star \bar{G}_{z}\star \bar{G}_{z}\star  \bar{G}_{z})(0)\\
	&=\lim_{z\nearrow z_c} \int_{(-\pi,\pi)^d} \frac {dk} {(2\pi)^d} \hat{\bar{G}}_{z}(k)^4,\nn
    	}
the infrared bound in Theorem \ref{thm-IRB} (which is uniform in $z<z_c$) immediately implies that the square condition holds, and by \cite{HarTas87} thus also that $\gamma=1/2$:

\begin{corollary}[Square condition and critical exponents]
\label{cor-TC-crit-exp}
For nearest-neighbor lattice trees and lattice animals in dimensions $d\geq \dmintree$ and $d\geq \dminanimal$, respectively, the square condition holds, i.e., $\square(z_c)<\infty$. As a result, the susceptibility and radius of gyration critical exponents $\gamma$ and $\nu$ exist in the bounded-ratio sense and takes on the mean-field values $\gamma=1/2$ and $\nu=1/4$.
\end{corollary}


We next move to some further extensions of our results, starting with {\em sharp} asymptotics of the Fourier transform of the critical two-point function $\hat{\bar{G}}_{z_c}(k)$ for $k$ small:

\begin{theorem}[Two-point function $k$-space asymptotics]
\label{thm-k-space}
For nearest-neighbor LTs and LAs in dimensions $d\geq \dmintree$ and $d\geq \dminanimal$, respectively, there exists a constant $\bar{A}(d)$ such that, as $|k|\rightarrow 0$,
    	\eqn{
    	\lbeq{k-space}
    	\hat{\bar{G}}_{z_c}(k)=\frac{\bar{A}(d)}{|k|^2}(1+o(1)).
    	}
\end{theorem}

We next investigate the asymptotics in $x$-space of $\bar{G}_{z_c}(x)$ for $x$ large, using the results of Hara in \cite{Hara08}:

\begin{theorem}[Two-point function $x$-space asymptotics]
\label{thm-x-space}
For nearest-neighbor LTs and LAs in dimensions  $d\geq 27$, 
and with the constant $\bar{A}(d)$ as in Theorem \ref{thm-k-space}, as $|x|\rightarrow \infty$,
    	\eqn{
    	\lbeq{x-space}
    	\bar{G}_{z_c}(x)=\frac{a_d\bar{A}(d)}{|x|^{d-2}}(1+O(|x|^{-2/d})),
	\qquad
	\text{with}
	\qquad
	a_d=\frac{d\Gamma(d/2-1)}{2\pi^{d/2}}.
    	}
\end{theorem}
\medskip

The restriction $d\geq 27$ originates as a condition in \cite{Hara08}, where it is used to start a recursion in \cite[Proof of (1.47), assuming Lemmas 1.8 and 1.9]{Hara08}. It is unclear to us whether this restriction can be avoided. We {\em can}, however, prove the convergence of the classical lace expansions for $d\geq \dmintree$ for LTs and $d\geq \dminanimal$ for LAs. Whether this is enough to obtain the $x$-space asymptotics in Theorem \ref{thm-x-space} is open.

In Section \ref{sec-overview}, we give an overview of the proof of Theorem \ref{thm-IRB}. We next discuss relations to the literature.

\subsection{Relations to the literature}
\label{sec-disc}

\paragraph{Critical exponents.} For spread-out lattice trees and animals, as well as for nearest-neighbor lattice trees in sufficiently high (unspecified) dimensions, the statements that $\gamma=1/2$ and $\nu=1/4$ have been improved to fixed $n$, as in \cite[Theorem 1.1]{HarSla92c}. This means that there exist constants $A_\gamma, A_\nu $such that
	\eqn{
	\sum_{x\in\Zd} t_n(x)=A_{\gamma} n^{\gamma-1}(1+o(1)),
	\qquad
	R_n =A_{\nu} n^{\nu}(1+o(1)).
	}
Here, the critical exponent $\gamma-1$ is sometimes called $\theta$. Such results are proved by using Tauberian Theorems on $\chi(z)=\sum_{n\geq 0} \sum_{x\in\Zd} t_n(x)z^n$ and $\sum_{n\geq 0}\sum_{x\in\Zd} |x|^2 t_n(x) z^n$, where $z\in \mathbb{C}$ with $|z|<z_c$. This requires careful Taylor expansions of these generating functions, with explicit control over the error terms. Possibly our results could be extended in this direction as well, but this would require a more detailed and extensive study of the analytic properties of $\chi(z)$ and $\hat{G}_z(k)$, which we refrain from.

Other interesting critical exponents exist. For example, the {\em correlation length} $\xi(z)$ is defined as
	\eqn{
	\lbeq{xi-def}
	\xi(z)^{-1}=-\lim_{x\rightarrow \infty}\frac{1}{|x|} \log G_z(x).
	}
It is predicted that $\xi(z)$ scales similarly as $\xi_2(z)$, so that, in particular,
	\eqn{
	\xi(z) \sim (z_c-z)^{-1/4},
	}
corresponding to $\nu=1/4$ also for this correlation length.

\paragraph{Asymptotics of $z_c$ and $g_{z_c}$.}
The asymptotic value of $z_c$ and $g_{z_c}=\bar G_{z_c}(0)$, are shown in \cite{MirSla12} to be given by
	\begin{align}
	\lbeq{conj-zc}
	z_c&=\e^{-1}\left[\frac 1 {2d} +\frac {\frac 3 2}{(2d)^2}+\frac {\frac {115}{24}- \indAnimal\frac 1 2 \e^{-1}}{(2d)^3}\right]+o\big((2d)^{-3}\big),\\
	g_{z_c}&=\e\left[\frac 1 {2d} +\frac {\frac 3 2}{(2d)^2}+\frac {\frac {263}{24}- \indAnimal\frac 1 2 \e^{-1}}{(2d)^2}\right]+o\big((2d)^{-2}\big),
	\lbeq{conj-gzc}
	\end{align}
where the factor $\indAnimal$ is $1$ for LA and $0$ for LT. The rigorous bound on the error terms was derived using the lace expansion. These asymptotic expansions have also been studied extensively in the physics literature, where results were obtained up to the order of $1/(2d)^6$, but with non-rigorous estimates on the error term. We refer the reader to \cite{MirSla12} for an overview of such results and a discussion of the $1/d$ expansions of the critical value $z_c$ of SAW, percolation, LT, and LA.

\paragraph{Relations to super-processes.} There has been a considerable effort to identify super-process limits of lattice trees and lattice animals, with a focus on lattice trees. Derbez and Slade \cite{DerSla97, DerSla98} show that the scaling limits of lattice trees of a {\em fixed size} scale to a finite measure called {\em integrated super-Brownian excursion} (ISE). For such a scaling limit, a lattice tree of size $n$ is seen as a random counting measure $X_n$, where $X_n(E)$ is the number of vertices of a random lattice tree that are in $E\subseteq \R^d$. When rescaling space by $n^{1/4}$ (which is related to $\nu=1/4$ in Corollary \ref{cor-TC-crit-exp}), and rescaling mass by $1/n$ (so that the total mass becomes 1, this random measure converges to ISE. Holmes and collaborators, instead consider the super-Brownian motion limits of lattice trees, proving convergence in finite-dimensional distributions  \cite{Holm08b, HofHol13} (see also \cite{HolPer07}) and tightness \cite{HofHolPer15}, in the context of spread-out lattice trees above 8 dimensions. Extensions of such results to nearest-neighbor lattice trees, as well as to lattice animals, are of great interest.

\section{Overview of the proof}
\label{sec-overview}
In this section, we give a brief overview of how we derive our main results. This overview is similar in spirit to that in \cite[Section 2]{FitHof13d}, where nearest-neighbor percolation was investigated. We repeat part of this discussion, as it provides the key ideas in our proof.

\subsection{Philosophy of the proof}
\label{sec-phil-proof}
We reduce the proof of Theorem \ref{thm-IRB} to three key propositions and a computer-assisted proof. These ingredients involve
	\begin{enumerate}[(a)]
  	\item the derivation of the non-backtracking lace expansion (NoBLE) in Proposition \ref{prop-LE};
  	\item the diagrammatic bounds on the NoBLE coefficients in Proposition \ref{prop-bds-LEC};
  	\item the analysis presented in \cite{FitHof13b} to obtain the infrared bound in Theorem \ref{thm-IRB} for all $z\leq z_c$, as stated in Proposition \ref{prop-analysis-is-success}; and
	\item a computer-assisted proof to verify the numerical conditions arising in the analysis in \cite{FitHof13b}.
	\end{enumerate}
\label{parts-proof}
The philosophy behind the proof of these parts is discussed in Sections \ref{sec-part-a}-\ref{sec-part-d}, respectively. In Sections \ref{sec-Expansion} and \ref{secBounds} we prove parts (a) and (b), respectively. In Section \ref{sec-part-c}, we explain how we obtain part (c) using the analysis of \cite{FitHof13b}.
The computer-assisted proof of part (d) is performed in \cite{FitNoblePage} using the results of this paper.
For the analysis in the generalized setting \cite{FitHof13b} we state assumptions, which we verify in  Section \ref{sec-gzAssumption} and in Section \ref{secBoundsSummary}.
 Part (d) is explained in detail in Section \ref{sec-part-d}, where we describe how the necessary computations are performed in several Mathematica notebooks. The mathematics behind the notebooks is explained in \cite{FitHof13b}.

 Let us now explain how our main results are proved. The notebooks contain a routine that verifies whether the numerical assumption on the expansions are satisfied and thereby whether the analysis of \cite{FitHof13b} yields the infrared bounds for a given dimension. The results in \cite{FitHof13b}, combined with a successful verification using the notebooks, thus prove Proposition \ref{prop-analysis-is-success}, which in turn implies Theorem \ref{thm-IRB}. See also Figure \ref{Struct-NoBLE} for a visual description of the proof of Theorem \ref{thm-IRB}.  Theorem \ref{thm-bds-crit} will follow from the numerical bounds derived in the proof. The proofs of Corollary \ref{cor-TC-crit-exp}, Theorem \ref{thm-k-space} and Theorem \ref{thm-x-space} are completed in Section \ref{sec-x-space}, using the results of Hara \cite{Hara08}, as well as an improvement of the analysis there by Hara that we learned about in private communication \cite{Hara08ext}, which yields an improved numerical condition for the convergence of the classical-lace expansion coefficients, allowing us to conclude the $x$-space asymptotics in Theorem \ref{thm-x-space} for $d\geq 27$. We close this section with a discussion of our method and results.

\vskip-1cm

\begin{center}
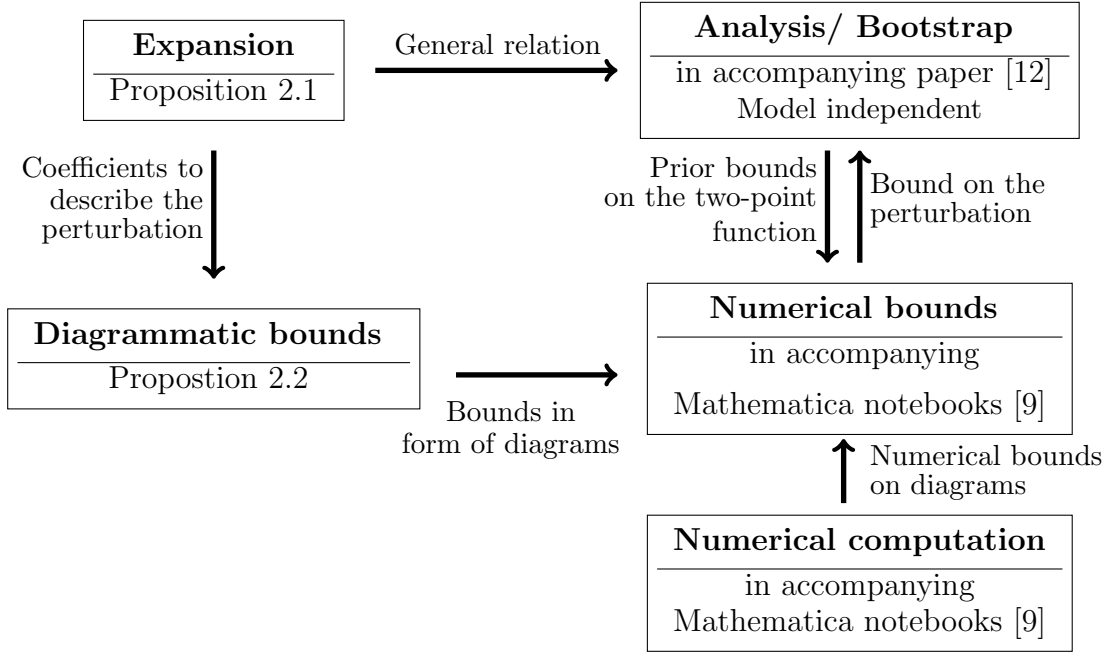
\begin{figure}[ht]
\begin{tikzpicture} [scale=0.85]

  \draw [->,line width=2] (2.5,0) to (6.25,0);
  \draw [->,line width=2] (3.75,-4.75) to (6.25,-4.75);
  \draw [<-,line width=2] (10,-1.25) to (10,-3);
  \draw [->,line width=2] (9.5,-1.25) to (9.5,-3);
  \draw [->,line width=2] (0,-1.25) to (0,-3.25);

  \draw [<-,line width=2] (9.75,-5.75) to (9.75,-6.75);

  \node[above]   at(4.4,0.1)     {General relation};
  \node[below]   at(4.6,-5)   {Bounds in};
  \node[below]   at(4.6,-5.5)   { form of diagrams};
  \node[left]   at(0,-1.5)        {Coefficients to};
  \node[left]   at(0,-2)        {describe the};
  \node[left]   at(0,-2.5)        {perturbation};
  \node[right]   at(10,-1.75)      {Bound on the};
  \node[right]   at(10,-2.25)      {perturbation};

  \node[left]   at(9.5,-1.5)      {Prior bounds};
  \node[left]   at(9.5,-2)       {on the two-point };
  \node[left]   at(9.5,-2.5)      {function};
  \node[right]   at(10,-6)      {Numerical bounds};
  \node[right]   at(10,-6.5)      {on diagrams};

\node [draw] {
\begin{tabular}{c}
{\bf \large Expansion} \\[1mm]
\hline
{\large Proposition \ref{prop-LE}}
\end{tabular}};

\node [draw]  at (10,0){
\begin{tabular}{c}
{\bf \large Analysis/ Bootstrap } \\[1mm]
\hline
\vspace{2mm}
{\large in accompanying paper \cite{FitHof13b}}\\[-2mm]
{\normalsize Model independent}
\end{tabular}};

\node [draw]  at (10,-4.5){
\begin{tabular}{c}
{\bf \large Numerical bounds } \\[1mm]
\hline
\vspace{2mm}
{\large in accompanying}\\
{\large Mathematica notebooks \cite{FitNoblePage}}
\end{tabular}};

\node [draw]  at (10,-8){
\begin{tabular}{c}
{\bf \large Numerical computation} \\[1mm]
\hline
{\large in accompanying}\\
{\large Mathematica notebooks \cite{FitNoblePage}}
\end{tabular}};

\node [draw]  at (0,-4.5){
\begin{tabular}{c}
{\bf \large Diagrammatic bounds } \\[1mm]
\hline
{\large Propostion \ref{prop-bds-LEC} }
\end{tabular}};
\end{tikzpicture}
\caption{Structure of the non-backtracking lace expansion.}
\label{Struct-NoBLE}
\end{figure}
\end{center}

\vskip-1cm

In the following, we explain the philosophy behind the non-backtracking lace expansion (NoBLE), and start by describing simple random walk and non-backtracking walk:
\paragraph{Simple random walk.}
An $n$-step \emph{nearest-neighbor simple random walk} (SRW) on $\Zd$ is an ordered $(n+1)$-tuple $\omega=(\omega_0,\omega_1,\omega_2,\dots, \omega_n)$, with $\omega_i\in\Zd$ and $|\omega_i-\omega_{i+1}|_1=1$, where $|x|_1=\sum_{i=1}^d |x_i|$. Unless stated otherwise, we take $\omega_0=(0,0,\dots,0)$. Then, for $n\geq 1$,
	\eqan{
	p_n(x)&=\#\{\text{$n$-step SRWs with $\omega_n=x$}\} =\sum_{y\in\Zd} 2d D(y)p_{n-1}(x-y)\\
    	&=2d (D \star p_{n-1})(x)=(2d)^{n} D^{\star n}(x),\nn
    	\lbeq{SRWRecScheme}
	}
with $D$ being defined in \refeq{def-Dhat} and $f^{\star n}$ denotes the $n$-fold convolution of a function $f$.
The SRW two-point function is defined by
\begin{eqnarray}
	\lbeq{genSRW}
	C_z(x)&=&\sum_{n=0}^\infty p_n(x)z^n,\qquad\text{and}
	\qquad \hat C_z(k) =\frac {1}{1-2dz\hat D(k)}
\end{eqnarray}
 for $|z|<1/(2d)$, in $x$-space and $k$-space, respectively. The SRW \emph{susceptibility} is given by
	\begin{eqnarray}
	\chi^{\sss\rm SRW}(z)= \hat C_z(0)= \frac {1} {1-2dz},
	\end{eqnarray}
for $|z|<z_c$, with \emph{critical point} $z_c=1/(2d)$.

\paragraph{Notations.}
Before introducing non-backtracking random walk, we introduce some notation that will be used throughout this document. We exclusively use the Greek letters $\iota$ and $\kappa$ for values in $\{-d,-d+1,\dots,-1,1,2,\dots,d\}$ and denote the unit vector in direction $\iota$ by $\ve[\iota]\in\Zd$, e.g. $(\ve[\iota])_i=\text{sign}(\iota) \delta_{|\iota|,i}$. We use ${\mathbb C}^{2d}$-valued and ${\mathbb C}^{2d}\times{\mathbb C}^{2d}$-valued quantities. For a clear distinction between scalar-, vector- and matrix-valued quantities, we always write ${\mathbb C}^{2d}$-valued functions with a vector arrow (e.g.\ $\vec v$) and matrix-valued functions with bold capital letters (e.g.\ ${\bf M}$). We do not use $\{1,2,\dots,2d\}$ as the index set for the elements of a vector or a matrix, but use $\{-d,-d+1,\dots,-1,1,2,\dots,d\}$ instead. We denote the identity matrix by $\mI\in{\mathbb C}^{2d\times 2d}$ and the all-one vector by $ \v1 =(1,1,\dots,1)^T\in{\mathbb C}^{2d}$. Moreover, we define the matrices $\mJ,\mD[k]\in{\mathbb C}^{2d\times 2d}$ by
	\begin{eqnarray}
  	\lbeq{J-D-matrix-def}
	(\mJ)_{\iota,\kappa}=\delta_{\iota,-\kappa}\qquad\text{ and }\qquad
	(\mD[k])_{\iota,\kappa}=\delta_{\iota,\kappa} \e^{\ii k_\iota},
	\end{eqnarray}
where $k\in[-\pi,\pi]^d$ and for negative index $\iota\in\{-d,-d+1,\dots,-1\}$, we write $k_\iota=-k_{|\iota|}$. Now we can introduce non-backtracking random walk and its Green's function.

\paragraph{Non-backtracking walk.}
If an $n$-step SRW $\omega$ satisfies $\omega_i\not=\omega_{i+2}$ for all $i=0,1,2,\dots,n-2$, then we call $\omega$ a \emph{non-backtracking walk} (NBW).  The NBW two-point function $B_z$ is defined, for $|z|<1/(2d-1)$, by
	\begin{align*}
	B_{z}(x)&=\sum_{n=0}^{\infty} \#\{\text{$n$-step NBW with $\omega_n=x$}\}z^n.
	\end{align*}
As derived in \cite[Section 1.2.2]{FitHof13b}, this two-point function satisfies
	\begin{eqnarray}
	\lbeq{NBWGen}
	\hat B_{z}(k)&=& \frac {1} {1 -z\v1^T\left[\mD[k]+z \mJ\right]^{-1}\v1}=\frac {1}{1-2d z\frac {\hat D(k)-z}{1-z^2}}.
	\end{eqnarray}
The critical NBW and SRW two-point functions are thus related by
	\eqn{
	\lbeq{SRWtoNBWlink}
	\hat B_{1/(2d-1)}(k)=\frac {2d-2}{2d-1}\hat C_{1/2d}(k)
	=\frac {2d-2}{2d-1} \cdot \frac 1 {1-\hat D(k)}.
	}
This link allows us to compute values for the NBW two-point function in $x$- and $k$-space, using the SRW two-point function. A detailed analysis of the NBW including a proof that the NBW, when properly rescaled, converges to Brownian motion can be found in \cite{FitHof13a}.

\subsection{Part (a): Non-Backtracking Lace Expansion (NoBLE)}
\label{sec-part-a}
In this section, we explain the shape of the Non-Backtracking Lace Expansion (NoBLE), which is a perturbative expansion of the two-point function. For this, we derive an equation alike \refeq{NBWGen} for the two-point function $G_z(x)$. This expansion is derived in Section \ref{secExp}.

Next to the usual two-point function $G_z$, see \refeq{defLTLATwoPoint},  we use a slight adaptation $G_z^\iota$, with $\iota\in\{\pm1,\pm2,\dots, \pm d\}$ being a distinct direction. We postpone the definition of $G_z^\iota$ to Section \ref{secExp}, see \refeq{defGzIotaanimal}. Intuitively, $G_z^\iota$ considers only LTs and LAs that do not contain $e_\iota$, but for technical reasons, the precise definition is somewhat more involved. Our analysis relies on two expansion identities relating $G_z(x)$ and $G_z^\iota(x)$ that are formalized in the following proposition:

\begin{prop}[Non-backtracking lace expansion]
\label{prop-LE}
For every $x\in \Z^d$, $\iota,\kappa\in \{\pm1,\pm2, \dots, \pm d\}$, the following recursion relations hold:
    \eqan{
    G_z(x)&=\delta_{0,x}+\aaz\sum_{y\in\Zd, \kappa\in\{\pm 1,\dots,\pm d\}} (\delta_{0,y}+\Psi_{z}^\kappa(y))G^{\kappa}(x-y+\ve[\kappa])+\Xi_{z}(x),
    \lbeq{Gx-in-ex-1}\\
    G_z(x)&=G^\iota_z(x)+\aaz G^{-\iota}_z( x-\ve[\iota])
        +\sum_{y\in\Zd, \kappa\in\{\pm 1,\dots,\pm d\}} \Pi_{z}^{\iota,\kappa}(y) G_z^{\kappa}(x-y+\ve[\kappa])+\Xi_{z}^{\iota}(x),
    \lbeq{Gx-in-ex-2}
    }
where
    \eqan{
    \lbeq{PhiPsialt}
    \Pi_{z}^{\iota,\kappa}(x) &= \sum_{N=0}^{\infty} (-1)^N    \Pi^{\ssc[N],\iota,\kappa}(x),
    \qquad
    \Xi_{z}(x)= \sum_{N=0}^{\infty} (-1)^N \Xi^{\ssc[N]}_z(x),\\
   \lbeq{PhiPsialt-2}
    \Psi_{z}^\kappa(x)&=\sum_{N=0}^{\infty} (-1)^N \Psi^{\ssc[N],\kappa}_z(y),
    \qquad
    \Xi_{z}^\iota(x)=\sum_{N=0}^{\infty} (-1)^N \Xi^{\ssc[N],\iota}(x),\\
    \aaz &=  zg^\iota_z,\qquad\qquad\qquad\qquad\qquad \aabz =  zg_z.
   \lbeq{Definition-aaz}
    }
Explicit formulas for the above lace-expansion coefficients are given in Section \ref{sec-defineCoeff}.
\end{prop}
\medskip

Equations \refeq{Gx-in-ex-1} and \refeq{Gx-in-ex-2} are similar in spirit to the equations satisfied by the NBW two-point function, where $G_z(x)$ is replaced by $B_z(x)$ and $G^\iota_z(x)$ by $B^\iota_z(x)$, which is the Green's function of all NBWs for which the first step is unequal to $e_{\iota}$ (recall \cite{FitHof13b}). For this, the similar equation to \refeq{Gx-in-ex-1} is obtained by conditioning on the first step, while \refeq{Gx-in-ex-2} is obtained by splitting depending on whether the first step equals $e_{\iota}$ or not. Equations \refeq{Gx-in-ex-1} and \refeq{Gx-in-ex-2} reduce to the equations for NBW when we set $\Psi^{\kappa}_z=\Xi_{z}=\Pi_{z}^{\iota,\kappa}\equiv 0$, and replace $\aaz$ by $(2d-1)z$ and $\aabz$ by $2dz$. Thus, we can think of  \refeq{Gx-in-ex-1} and \refeq{Gx-in-ex-2} as a perturbations around NBW, and of $\Psi^{\kappa}_z=\Xi_{z}=\Pi_{z}^{\iota,\kappa}$ as the perturbation coefficients, of which we hope to prove that they are `small' in an appropriate sense.

Of course, the precise formulas for the lace-expansion coefficients are crucial for our analysis to succeed. However, at this stage, we refrain from stating their forms explicitly, and refer to Section \ref{sec-defineCoeff} instead where they are derived while performing the NoBLE.

Applying the Fourier transforms to \refeq{Gx-in-ex-1} and \refeq{Gx-in-ex-2} and some simple rearrangements, see \cite[Section 1.3]{FitHof13b}, we derive that
	\begin{align}
    	\lbeq{lace-exp-eq}
    	\hat{G}_z(k)=\frac{1+\hat{\Xi}_{z}(k)-\aaz(\v1+\vPsi[k])^T\big[\mD[k]+\aaz\mJ  +\mPi[k]\big]^{-1} \vXi[k]}
    {1-\aaz(\v1 +\vPsi[k])^T\big[\mD[k]+\aaz \mJ+\mPi[k]\big]^{-1}{\v1}},
	\end{align}
with
    \eqn{
    (\vec {\hat \Psi}(k))_\kappa=\hat \Psi^{\kappa}_z(k),\qquad
    (\vXi[k])_\iota=\hat \Xi^{\iota}_{z}(k), \qquad     (\mPi[k])_{\iota,\kappa}=\hat \Pi^{\iota,\kappa}_{z}(k),
    }
Equation \refeq{lace-exp-eq} is the NoBLE equation, and is the workhorse behind our proof. The goal of the NoBLE is now to show that \refeq{lace-exp-eq} is indeed a perturbation of \refeq{NBWGen}. This amounts to proving that $\hat{\Xi}_{z}(k)$,$\vXi[k]$, $\vPsi[k]$ and $\mPi[k]$ are small, which will only be true in sufficiently high dimensions.

We continue by discussing how to bound the NoBLE coefficients.

\subsection{Part (b): Bounds on the NoBLE}
\label{sec-part-b}
The content of the second key proposition is that the NoBLE coefficients can be bounded by combinations of simple diagrams. Simple diagrams are combinations of two-point functions, alike the following examples:
	\begin{align}
	(2d\aaz)^2 (G_z\star D\star D\star G_z)(e_\iota),   \quad
	\sup_{x\in\Zd\colon \ |x|^2>2}  \sum_{y\in\Zd} |y|^2G_z(y) (G_z\star D)(x-y).
	\end{align}
\begin{prop}[Diagrammatic bound on the NoBLE coefficients]
\label{prop-bds-LEC}
For each $N\geq 0$, the NoBLE coefficients
$ \Pi^{\ssc[N],\iota,\kappa}(y),$ $\Xi^{\ssc[N]}(x),$ $\Psi^{\ssc[N],\kappa}(y)$ and $\Xi^{\ssc[N],\iota}(x)$  can be bounded by a finite combination of sums and products of simple diagrams.
\end{prop}

The explicit form of the bounds in Proposition \ref{prop-bds-LEC} is given in Lemmas \ref{BoundNZero-Xi}-\ref{lemmapercboundXi0minus1}. In Section \ref{secBoundsSummary} we summarize the bounds as required for the analysis in \cite{FitHof13b}. As the complete proof of the bounds on the NoBLE coefficients is highly elaborate, we only sketch the proof.
\iflongversion
We describe the used diagrams used in the bounds in Section \ref{secBuildingBlocks} and give the formal definition in the appendix.
\else
In this paper, we only informally define the diagrams that we use to bound the coefficients.
A complete definition and additional details can be found in the extended version \cite{FitHof13g-ext}.
\fi

\paragraph{Matrix-valued bounds.}
The lace-expansion coefficients arising in the NoBLE describe contributions created by multiple mutually intersecting paths, in the LTs and LAs. We will informally call such path-intersections {\em loops}, and this notion will be made precise below. In the NoBLE, these loops require at least $4$ bonds by design, as direct reversals are excluded by the way that we set up the NoBLE expansion. These lace-expansion coefficients can be bounded in terms of certain Feynman diagrams, of which the lines correspond to various two-point functions, and these lines are required to have various intersections. We will refer to the two-point functions arising in these bounds informally as `lines'.

The diagrams turn out to have the general property that lines can be part of at most two loops. To optimally use the information that loops contain at least 4 bonds, we distinguish five cases for the length and function of paths used by two loops.
Then, we bound the contribution of each loop of the lace-expansion diagram one-by-one, using the information on the lines shared with the previous and preceding loops. We explain this in detail in Section \ref{secBoundsExplained}. This procedure naturally gives rise to a bound on the NoBLE coefficients in terms of matrix products, as formalized in Lemma \ref{LemmaBoundXiLTTwo}. For example, our proof yields that
	\begin{align*}
      \sum_{x\in\Zd}\Xi^{\ssc[N]}_z(x)&\leq \frac {g_z}{\gj} \vec P^T{\bf A}^{N-2}{\vec E}^{\rm{open}},
	\end{align*}
for $N\geq 3$, see \refeq{Bound-Xi2}, for certain vectors $\vec P, {\vec E}^{\rm{open}}$ and 5 by 5 matrices ${\bf A}$. We will give an interpretation of the elements in these vectors and matrices in Section \ref{secBuildingBlocks}. For our analysis, we require a bound on this when summed over $N$. To compute this bound numerically, we perform an eigenvector decomposition of $\vec P$, in terms of the eigenvectors $(\vec v_i)_{i=1}^5$ of ${\bf A}$ with corresponding eigenvalues $(\lambda_i)_{i=1}^5$. In this decomposition, we write $\vec P=\sum_{i=1}^5 \vec v_i,$ so that the eigenvectors used are not normalized\footnote{We do account for the possibility that ${\bf A}$ is not diagonalizable.}.
Then,
	\begin{align*}
	\hat \Xi^{\ssc[N]}_{z}(0)\leq& \sum_{i=1}^5 \vec v_i \lambda^{N-2}_i{\vec E}^{\rm{open}}.
	\end{align*}
The sum of this over $N$ is computed using the geometric sum, see \cite[Section 5.4]{FitHof13b} for more details.

The order of this bound is to a large extent given by the largest eigenvalue of ${\bf A}$. For example, for LTs in $d=16$,
	\begin{align*}
	{\bf A}=\left(
\begin{array}{ccccc}
 0.0325769 & 0.0151778 & 0.0209234 & 0.0151778 & 0.0325769 \\
 0.0196345 & 0.00927948 & 0.0131029 & 0.00927948 & 0.0196345 \\
 0.0271802 & 0.0143934 & 0.0201978 & 0.0160374 & 0.0313802 \\
 0.0196345 & 0.00927948 & 0.0131029 & 0.00927948 & 0.0196345 \\
 0.0325769 & 0.0151778 & 0.0209234 & 0.0151778 & 0.0325769 \\
\end{array}
\right)
	\end{align*}
with largest eigenvalue $\lambda_1=0.103089$. In the classical lace expansion, the corresponding $N$th lace-expansion diagram bound also decays exponentially in $N$, where the base of the exponent is a non-trivial square. We can bound the non-trivial square, as it arises in the classical lace expansion, by $2.53461$.
 Revisiting the classical lace expansion and then using all our improvements, which we explain in Section \ref{sec-BoundIdea-Heuristic}, we can improve the base of the exponent to a value of $0.130586$.
 Using our matrix-valued bounds, reduce the exponent by $30\%$. The gain in the explicite contribution is usually even larger.

\subsection{Part (c): The NoBLE analysis}
\label{sec-part-c}
At the center of our analysis are the so-called {\em bootstrap functions}, which we next discuss.

\paragraph{Bootstrap functions.}
For the bootstrap, we define the following functions:
	\begin{eqnarray}
	\lbeq{defFunc1}
	f_1(z)&:=&\max\left\{(2d-1)zg_z,c_\mu (2d-1)zg_z^\iota\right\},\\
	\lbeq{defFunc2}
	f_2(z)&:=& \sup_{k\in(-\pi,\pi)^d}
	\frac{\hat G_z(k)}{\hat B_{1/(2d-1)}(k)} = \frac{2d-1}{2d-2}\sup_{k\in(-\pi,\pi)^d} [1-\hat D(k)]\ \hat G_z(k),\\
	\lbeq{defFunc3}
	f_3(z)&:=& \max_{\{n,l,S\}\in \mathcal{S}} \frac{\sup_{x\in S}
	\sum_{y}|y|^2G_z(y)(G_z^{\star n}\star D^{\star l})(x-y)}{c_{n,l,S}},
	\end{eqnarray}
where $c_\mu>1$ and $c_{n,l,S}>0$ are some well-chosen constants and $\mathcal{S}$ is some finite set of indices. Let us now discuss the choice of these functions, and explain their use.

The functions $f_1$ and $f_3$ can been seen as combinations of multiple functions.  We group these functions together as they play a similar role and are analyzed in the same way. We do not expect that the values of the bounds on the individual functions that constitute $f_1$ and $f_3$ are comparable. This is the reason that we introduce the constants $c_\mu$ and $c_{n,l,S}$.

The choice of point-sets $S\in \mathcal{S}$ improves the numerical accuracy of our method considerably. For example, we obtain much better estimates in the case when $x=0$ (leading to a closed diagram), than for $x\neq 0$. For $x$ being a neighbor of the origin, we can further use symmetry to improve our bounds significantly. To obtain the infrared bound as stated in Theorem \ref{thm-IRB}, we use the following set $\mathcal{S}$:
	\begin{align}
	\lbeq{calS-def}
  	\mathcal{S}=\big\{
  	&\{1,0,\mathcal{X}\},\{1,1,\mathcal{X}\},\{1,2,\mathcal{X}\},\{1,3,\mathcal{X}\},\{1,6,\{0\}\},\\
  	&\qquad \{2,0,\mathcal{X}\},\{2,1,\mathcal{X}\},\{2,2,\mathcal{X}\},\{2,3,\mathcal{X}\},\{2,6,\{0\}\}\big\},\nn
	\end{align}
with $\mathcal{X}=\{x\in\Zd\colon |x|_2>1\}$. This turns out to be sufficient for our main results.

We apply a \emph{forbidden region or bootstrap argument} that is based on three claims:
	\begin{enumerate}
	\item[(i)] $z\mapsto f_i(z)$ is \emph{continuous} for all $z\in [z_{I}, z_c)$ and $i=1,2,3$ and for some $z_{I}\in(0,z_c)$;
	\item[(ii)] $f_i(z_{I})\leq \gamma_i$ holds for $i=1,2,3$; and
	\item[(iii)] if $f_i(z)\leq \Gamma_i$ holds for $i=1,2,3$, then, in fact, also $f_i(z)\leq \gamma_i$ holds for every $i=1,2,3$, where $\gamma_i<\Gamma_i$ for every $i=1,2,3$.
	\end{enumerate}
Together, these three claims imply that $f_i(z_{I})\leq \gamma_i$ holds for every $i=1,2,3$ and all $z\in [z_{I}, z_c)$. This in turn implies the statement of Theorem \ref{thm-IRB}  for all $z\in [z_{I}, z_c)$.

The continuity in Claim (i) is proven in \cite{FitHof13b} under some assumption that we explain and prove  below. The proofs of the initialization of the bootstrap in Claim (ii) as well as the improvement of the bounds in Claim (iii) use the following relations that are also sketched in Figures \ref{fig-Heuristic-bootstrap-initial} and \ref{fig-Heuristic-bootstrap}:
\begin{enumerate}[(1)]
\item simple diagrams can be bounded by a combination of two-point functions, see \cite[Section 4]{FitHof13b};
\item the NoBLE coefficients can be bounded by a combination of simple diagrams, that, in turn, can be bounded using the bounds on the two-point functions that we obtain from the bootstrap assumption, see Section \ref{secStatingtheBounds};
\item bounds on the NoBLE coefficients imply bounds on the two-point function, see \cite[Section 2]{FitHof13b}, which allows us to improve upon the bounds that we assumed.
\end{enumerate}
Thus, whenever we have numerical bounds on simple diagrams, or on NoBLE coefficients, or on the two-point function, we can also conclude bounds on the other two quantities.

We choose $z_{I}$ such that $G_{z_I}(x)\leq B_{1/(2d-1)}(x)$ holds pointwise in $x$.
Then, using that $G_{z_I}(x)\leq B_{1/(2d-1)}(x)$ and that we can compute $B_{1/(2d-1)}(x)$ numerically, we verify the initialization of the bootstrap in Claim (ii) (i.e., $f_i(z_I)\leq \gamma_i$ for $i=1,2,3$) numerically, see Figure \ref{fig-Heuristic-bootstrap-initial}.

\begin{figure}
 \begin{center}
\begin{tikzpicture}[line width=1pt,auto,scale=0.7]
 \node   at ({3*cos(153)},{3*sin(153)})     {$f_i(z_I)\leq \gamma_i$};
 \node   at ({3*cos(0)},{3*sin(0)})      {Bounds on simple diagrams};
 \node   at ({3*cos(230)},{3*sin(230)+0.2})      {Bounds on coefficients};
 \draw [->,very thick] ({3*cos(215)-0.2},{3*sin(215)})  arc  (215:165:3);

\draw [->,very thick] ({3*cos(350)},{3*sin(350)})  arc  (350:240:3);
 \draw [->,very thick] (4,2.2) to (3,0.5);
 \color{black}
 \node  at(4,2.5)      {$G_{z_I}(x)\leq B_{1/(2d-1)}(x)$};
 \node[left]   at(-3.1,0)   {Conclude a bound};
\end{tikzpicture}
\end{center}
 \caption{Initialization of the bootstrap: proof that $f_i(z_I)\leq \gamma_i$ holds for $i=1,2,3$. Here $\gamma_1,\gamma_2,\gamma_3$ are appropriately and carefully chosen constants.}
 \label{fig-Heuristic-bootstrap-initial}
\end{figure}
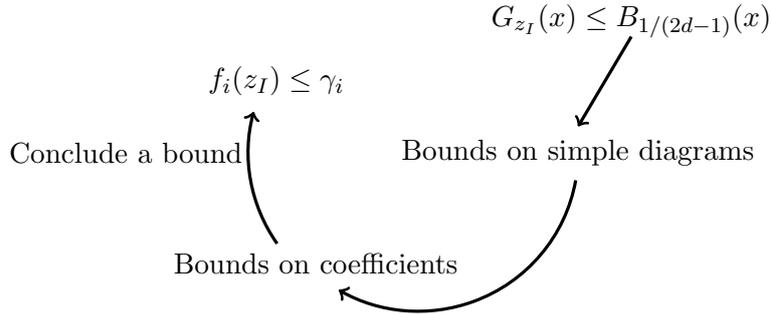

The proof of Claim (iii) is the most elaborate step of our analysis. Its structure is shown in Figure \ref{fig-Heuristic-bootstrap}. We start from the assumption that $f_i(z)\leq \Gamma_i$ holds for every $i=1,2,3$. The function $f_1$ gives a bound on $zg_z$ and $zg_z^\iota$. The function $f_2$ allows us to bound the two-point function in Fourier space by $\hat B_{1/(2d-1)}(k)$, which we can integrate numerically to obtain numerical bounds on simple diagrams. Further, $f_2$ allows us to conclude the infra-red bound in Theorem \ref{thm-IRB}. These, in turn, imply bounds on the NoBLE coefficients, which we use to improve our bounds on the bootstrap functions.

\emph{In the case that} the computed bounds are small enough, we can conclude that $f_i(z)\leq \gamma_i$ holds and thereby that the improvement of the bounds in Claim (iii) holds. Whether we can indeed prove that Claim (iii)  holds depends on the dimension we are in, the quality of the bounds and the analysis used to conclude bounds for the bootstrap function. In high enough dimensions (e.g.\ $d\geq 1000$) the perturbation is rather small so that it is
relatively easy to prove Claim (iii). Proving the claim in lower dimensions is only possible  when the bounds on the lace-expansion coefficients and the analysis are sufficiently sharp and hence sophisticated. It is here that it pays off to apply the NoBLE compared to the classical lace expansion. The third step in the proof of Theorem \ref{thm-IRB} is formalized in the following proposition:

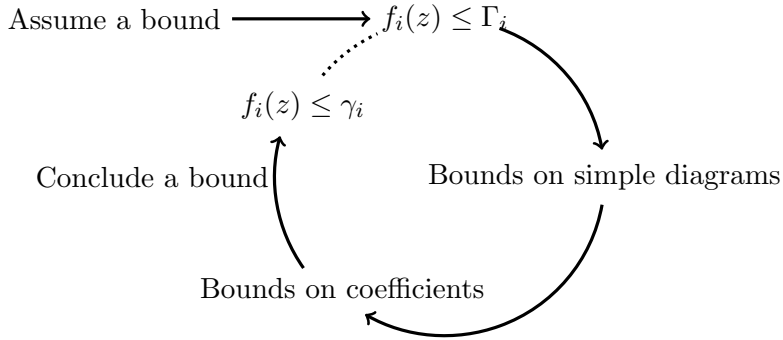
\begin{figure}
 \begin{center}
\begin{tikzpicture}[line width=1pt,auto,scale=0.7]
\node   at ({3*cos(90)},{3*sin(90)})     {$f_i(z)\leq \Gamma_i$};
 \node   at ({3*cos(153)},{3*sin(153)})     {$f_i(z)\leq \gamma_i$};
 \node   at ({3*cos(0)},{3*sin(0)})      {Bounds on simple diagrams};
 \node   at ({3*cos(230)},{3*sin(230)+0.2})      {Bounds on coefficients};
 \draw [->,very thick] ({3*cos(215)-0.2},{3*sin(215)})  arc  (215:165:3);
 \draw [->,very thick] ({3*cos(70)},{3*sin(70)})  arc  (70:10:3);
 \draw [dotted,very thick] ({3*cos(140)},{3*sin(140)})  arc  (140:115:3);
\draw [->,very thick] ({3*cos(350)},{3*sin(350)})  arc  (350:240:3);
 \draw [->,very thick] (-4,3) to ({3*cos(90)-1.4},{3*sin(90)}) ;
  \color{black}
 \node[left]   at(-4,3)      {Assume a bound};
 \node[left]   at(-3.1,0)   {Conclude a bound};
\end{tikzpicture}
\end{center}
 \caption{ Proof of claim (iii): $f_i(z)\leq \Gamma_i$ for $i=1,2,3$ implies that $f_i(z)\leq \gamma_i$ for $i=1,2,3$.}
 \label{fig-Heuristic-bootstrap}
\end{figure}

\begin{prop}[Successful application of NoBLE analysis]
\label{prop-analysis-is-success}
For nearest-neighbor LTs and LAs, in  $d\geq \dmintree$ and $d\geq \dminanimal$, the NoBLE analysis of \cite{FitHof13b} applies and proves the infrared bound in Theorem \ref{thm-IRB}. In particular, there exist constants $\Gamma_1,\Gamma_2,\Gamma_3$ and $\gamma_1,\gamma_2,\gamma_3$ such that $f_i(z_I)\leq \gamma_i$ for $i=1,2,3$ and, for every $z\in [z_I,z_c)$, the bounds $f_i(z)\leq \Gamma_i$ for $i=1,2,3$ imply that $f_i(z)\leq \gamma_i$ for $i=1,2,3$.\footnote{In the percolation paper, we erroneously wrote $p<p_c$ instead, but also there we need that $p\in[p_I,p_c)$.}
\end{prop}
As shown in Figure \ref{fig-Heuristic-bootstrap}, Proposition \ref{prop-analysis-is-success} is proved using the results of Propositions \ref{prop-LE}-\ref{prop-bds-LEC}, the analysis of \cite{FitHof13b} and the computer-assisted proof performed in the Mathematica notebook that can be found on \cite{FitNoblePage}.  To apply the general NoBLE analysis of \cite{FitHof13b}, we need to prove that the assumptions formulated in \cite{FitHof13b} hold.
All but one of these assumptions are proven in Section \ref{sec-gzAssumption}.
The one remaining assumption, involving the bounds on the NoBLE coefficients, is addressed in Section \ref{secBounds}.
\subsection{Part (d): Numerical analysis}
\label{sec-part-d}
In this section, we explain how the numerical computations are performed using Mathematica notebooks that are available from the first author's homepage \cite{FitNoblePage}.

\paragraph{Simple-random walk computations.}
The procedure starts by evaluating the notebook \verb|SRW|.
The file computes the value of the SRW, and thereby the NBW, two-point function for various locations in $\Zd$.
This computation uses pre-computed values of the number SRWs paths and SRW integrals based on numerical integration of certain Bessel functions.
These computations rely on ideas from \cite[Appendix B]{HarSla92a}, and are explained in detail in \cite[Section 5]{FitHof13b}.
The SRW integrals provide {\em rigorous numerical upper bounds} on various convolutions of SRW Green's functions with themselves, evaluated at various $x\in \Z^d$. We rely on 140 of such integrals.

Running these programs takes several hours. For this reason, once computed, the results are saved in two files, \verb|SRWCountData.nb| and \verb|SRWIntegralsData.nb| and are loaded automatically when the notebook is evaluated a second time for the same dimension. 
Alternatively, these two files can also be downloaded directly from the home page of the first author, and put in one's own home directory.\footnote{In Mathematica, the command {\tt \$InitialDirectory} will tell you what this directory is.}

\paragraph{Implementation of the NoBLE analysis for LT and LA.}
After having computed all the simple random walk ingredients, we evaluate the notebook \verb|General|, that implements the bounds of the NoBLE analysis \cite{FitHof13b}. After this, we are ready to perform the NoBLE analysis for the model by evaluating the notebook \verb|LA| and \verb|LA|, respectively. In these notebooks, we implement all the bounds proven in this paper. The computations in \verb|General|, \verb|LT| and \verb|LT| merely implements the bounds proven in this paper and in \cite{FitHof13b}, and rely on many multiplications and additions, as well as the diagonalization of two 5-by-5 matrices. These computations could in principle be done by hand (even though we prefer a computer to do them).

\paragraph{Output of Mathematica Notebooks.} After having evaluated the Mathematica notebooks, we can verify whether the analysis has worked with the chosen constants $\Gamma_1,\Gamma_2, \Gamma_3$.
In Figure \ref{fig-LT-successful} shows the relevant part of the output of the \verb|LT| notebook.
Let us explain this output in more detail. The green dots mean that the bootstrap has been successful for the parameters as chosen. When evaluating the notebook, it is possible that some red dots appear, and this means that these improvements were not successful. The first 3 dots in the first table are the verifications that $f_i(z_I)\leq \gamma_i$ for $i=1,2,3$. The next three dots show that the improvement has been successful for any $z\in(z_I,z_c(d))$. The values for $\Gamma_1, \Gamma_2,\Gamma_3$ are indicated in the second row.
For example, $\Gamma_1=1.10225007$ means that we assumed that $f_1(z)=(2d-1)zg_z\leq1.10225007$.
Using this assumption we concluded that $f_1(z)\leq 1.102250067$ from the improvement step (which relies on the NoBLE analysis), so that we can choose  $\gamma_1=1.02250069$ for our analysis in Proposition \ref{prop-analysis-is-success}. Since this is true for all $z\in(z_I,z_c(d))$, we obtain that $(2d-1)z_cg_{z_c}\leq 1.0225007$. This explains the value in the table of Theorem \ref{thm-bds-crit}.
The stated upper bound on $g_{z_c}$ follow from $g_{z_c}\leq \e\Gamma_1$, which we obtain by combining $f_1(z)<\Gamma_1$ with $z_c>(2d-1)^{-1}\e^{-1}$, which is proven in Section \ref{sec-gzAssumption}.

Similarly $\Gamma_2=1.335307$ implies that
	\begin{align}
  	&\frac{2d-1}{2d-2}\sup_{k\in(-\pi,\pi)^d} [1-\hat D(k)]\frac 1 {g_z}\hat {\bar G}_z(k)\leq \Gamma_2,
	\end{align}
which, in turn, implies that
  	\begin{align}[1-\hat D(k)]\hat {\bar G}_z(k) \leq g_z\frac {2d-2}{2d-1}\Gamma_2
    	\leq g_z\frac {2d-2}{2d-1}\Gamma_2\leq \Gamma_1\Gamma_2\e\leq 3.955 :=\bar{A}_2(d).
	\end{align}
Using these computation, we have computed the bounds stated in Theorem \ref{thm-bds-crit}.
Anyone interested in obtaining improved bounds on $g_{z_c}$, $g_{z_c}z_c$ or $A_2$ for other values of $d$ can play with the notebooks to optimize them. The second and third table in Figure \ref{fig-LT-successful} provides  details on the improvement of $f_3(z)$, which, as indicated in \refeq{defFunc3} and \refeq{calS-def} consists of several contributions, over which the maximum is taken. The assumed bound correspond to the constants $c_{n,l,S}$, with $S\in {\mathcal S}$ in \refeq{calS-def}. The notebooks \verb|LT| and \verb|LA| also includes a routine that optimizes the choices of $\Gamma_i$ and $c_{n,l,S}$. This makes it possible to efficiently find values for which the analysis works (when these exist).

\begin{figure}
\begin{center}
\includegraphics[width=0.75\textwidth]{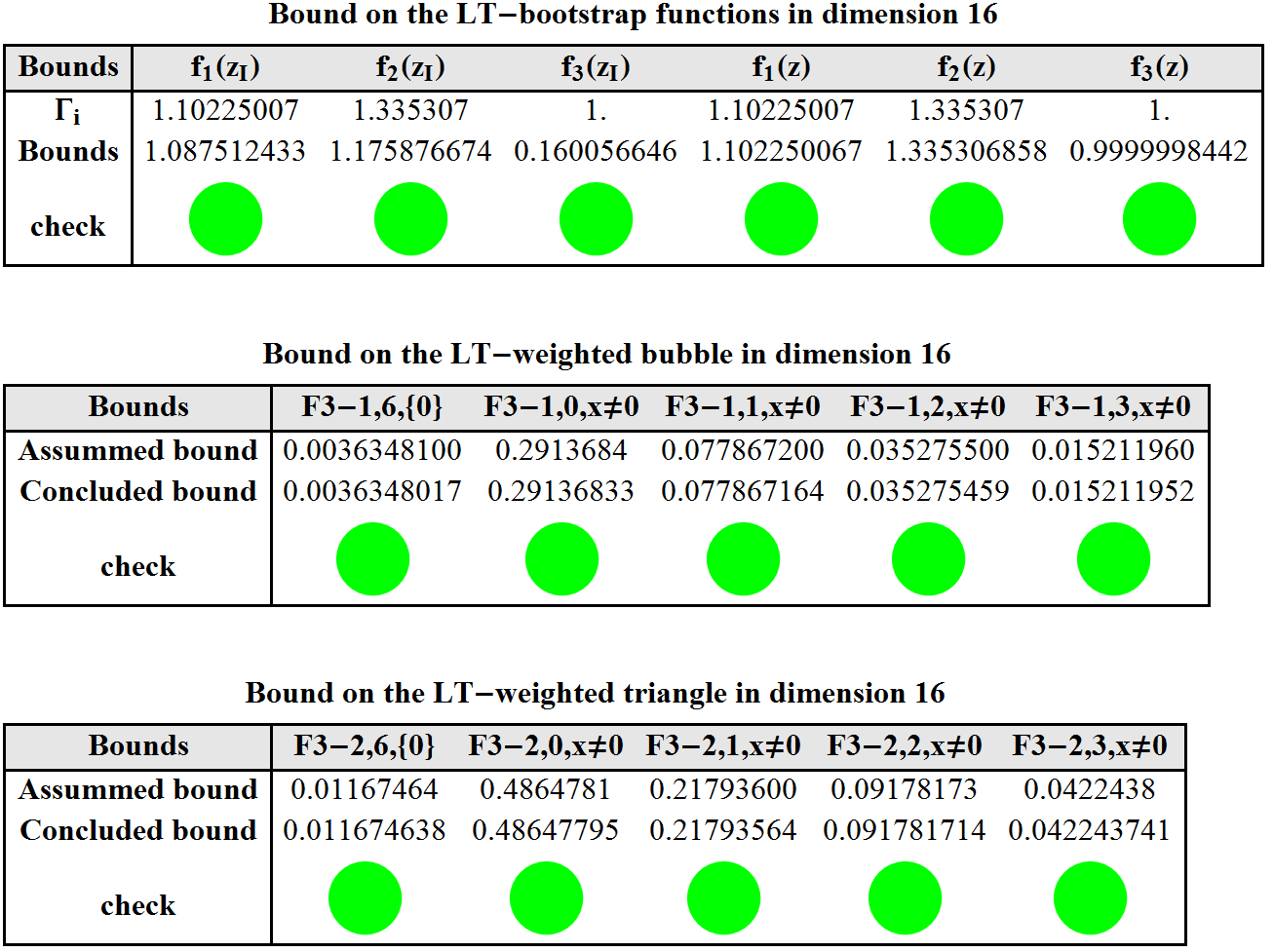}
\end{center}
 \cprotect\caption{\label{fig-LT-successful} Output of the Mathematica notebook {\verb|LT|}, where we can see that the bootstrap argument can be successfully initialized, as well as successfully improved. For completeness, we also show the various bounds of the components of $f_3$.}
\end{figure}

\begin{figure}
\begin{center}
\includegraphics[width=0.75\textwidth]{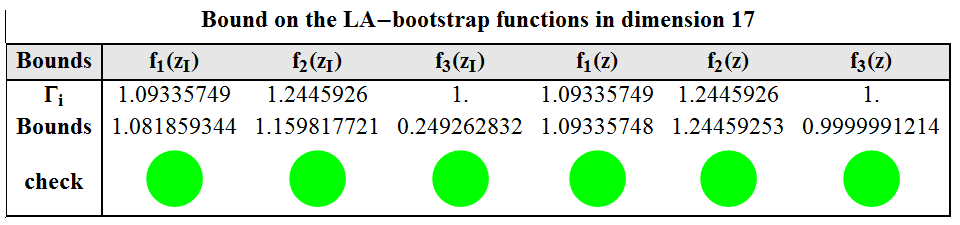}
\end{center}
 \cprotect\caption{\label{fig-LA-successful} Output of the Mathematica notebook {\verb|LA|}, where we can see that the bootstrap argument can be successfully initialized, as well as successfully improved.}
\end{figure}

\subsection{Structure of the NoBLE proof and related results}
\label{sec-summ-proof}
\paragraph{Summary of the proof of the infrared bound in Theorem \ref{thm-IRB}.}
We have already seen how delicately the four parts of the proof described in Section \ref{sec-phil-proof}
 are intertwined. The expansion in part (a) provides a characterisation of $\hat G_z(k)$ as a perturbation of $\hat B_{1/(2d-1)}(k)$ involving the NoBLE coefficients. The analysis in part (c) allows us to compute bounds on $\hat G_z(k)$, when numerical bounds on the coefficients are available.  To obtain such bounds we need to derive diagrammatic bounds, as formulated in part (b), that bound the NoBLE coefficients by simple diagrams.  However, we rely on bounds on $G_z$ to bound such simple diagrams numerically. Thus, we obtain a circular reasoning.

Using the bootstrap argument we can complete the circle, see Figures \ref{fig-Heuristic-bootstrap-initial}-\ref{fig-Heuristic-bootstrap}, to obtain a bound on $\hat G_z(k)$ for all $z\in [z_I, z_c)$. For the bootstrap argument, we need to show that $f_i(z_I)\leq \gamma_i$, as well as the fact that $f_i(z)\leq \Gamma_i$ implies that $f_i(z)\leq \gamma_i$, for appropriately chosen $\gamma_i$ and $\Gamma_i$ for all $z\in(z_I, z_c)$. The verification whether $f_i(z_I)\leq \gamma_i$ holds for $i=1,2,3$. Whether we can conclude from $f_i(z)\leq \Gamma_i$ for $i=1,2,3$ that also $f_i(z)\leq \gamma_i$ holds for $i=1,2,3$ requires a computer-assisted proof as indicated in Section \ref{sec-part-d}. Starting from $G_{z_I}(x)\leq  B_{1/(2d-1)}(x)$, $f_i(z)\leq \Gamma_i$ for $i=1,2,3$ and explicit computations of $B_{1/(2d-1)}(x)$, we obtain numerical bounds on simple diagrams.  These are then used to obtain numerical bounds on the NoBLE coefficients, which we in turn use to verify whether we can actually conclude from $f_i(z)\leq \Gamma_i$ for $i=1,2,3$ that $f_i(z)\leq \gamma_i$ for $i=1,2,3$ holds.

Combining these steps yields the required results for $z\in [z_I, z_c)$. We obtain the statement also for $z=z_c$ by using that $\hat G_z(k)/\hat B_{1/(2d-1)}(k)$ and the NoBLE-coefficients are continuous and uniformly bounded for $z\in [z_I, z_c)$ as well as left-continuous in $x$-space at $z_c$.

\paragraph{Extension to all $d\geq 20$.} Theorems \ref{thm-IRB} and \ref{thm-bds-crit} prove the infrared bound, as well as estimates on amplitudes and critical values, for $d=\dmintree, \ldots, 20$ for lattice trees, and for $d=\dminanimal, \ldots, 20$ for lattice animals. The Mathematica notebooks, and the numerical estimates that these rely upon, are designed to prove our results for any {\em specific} dimension. We successfully execute them for $d$ up to $29$.

To prove that the statements hold for {\em all} higher dimensions $d\geq 30$, without spending eternity checking the condition one dimension at time, we further modify our numerical verifications so that they only use quantities that are monotone decreasing in the dimension. It results in bounds on the NoBLE coefficients and related quantities that hold uniformly in {\em all} $d\geq 30$, implying that the bootstrap succeeds in all these dimensions in {\em one} go. This numerical verification is performed in the Mathematica notebook \verb|LAmonotone| for LA. This analysis also immediately applies to LTs, as all LT quantities are bounded by their LA analogues.

The idea behind these uniform bounds, used in the improvement of the bootstrap bounds, is the following:
\begin{enumerate}[a)]
  \item The SRW quantities that we rely upon, such as the Green's function at various points in $\Z^d$ and contributions of short paths, are monotone decreasing in $d$.
  \item Simple diagram are bounded using the bootstrap assumption $f_i(z)\leq \Gamma_i$ and these SRW quantities.
  Using the same $\Gamma_i$ for all $d\geq 30$, the bounds on simple diagrams are monotone decreasing in $d$ as well.
  \item Our NoBLE coefficients are bounded using simple diagrams, which implies that these bounds are also monotone decreasing in $d$.
  \item Our bounds on the NoBLE coefficients decrease in $d$, so that the bounds on the coefficients appearing in the analysis of \cite{FitHof13b} are monotone decreasing as well. In the end, this guarantees that the improvement of the bootstrap bounds succeeds for all $d\geq 30$ at once.
\end{enumerate}
In Appendix  \ref{sec-monoton-inD}, we explain how we modify our numerical verifications, show that the SRW quantities that we rely upon are monotone decreasing in $d$, and discuss the additional steps needed to make the analysis work for all $d\geq 30$.

\paragraph{Organization of this paper.}
The remainder of this paper is organised as follows. In Section \ref{sec-Expansion}, we perform the NoBLE, and thus prove Proposition \ref{prop-LE}. There, we also give explicit formulas for the NoBLE coefficients. In Section \ref{secBoundsExplained}, we explain how diagrammatic bounds on the NoBLE coefficients can be obtained. These diagrammatic bounds are phrased in terms of various building blocks that are informally defined in Section \ref{secBuildingBlocks}. In Section \ref{secBoundProof}, we explain how such diagrammatic bounds are proved.


\def\ClusterPictureWithPivotalAndLabeling[#1]{
\begin{figure}[ht]
\begin{center}
\begin{tikzpicture}[line width=1pt,auto,scale=#1]

   \node[right]   at(10,1)      {$y$};
   \node[below]   at(0,0)      {$x$};
   \node[below]   at(1,0)      {$w$};
   \node[above]   at(-1,1)      {$S_0$};
  \node  at(5,-2)      {Further, $S_3=S_4=S_6=S_7=S_9=\varnothing$.};
 \fill (10,1) circle (3pt);
  \fill (0,0) circle (3pt);

  \color{red}
  \draw [-] (0,0) to  (1,0);
  \node[below]   at(0.5,0)      {$b_1$};
  \draw [-] (2,1) to  (4,1);
  \node[below]   at(2.5,0)      {$b_2$};
  \node[below]   at(3.5,0)      {$b_3$};
  \draw [-] (4,1) to  (4,2);
  \node[right]   at(4,1.5)      {$b_4$};
  \draw [-] (4,2) to  (5,2);
  \node[below]   at(4.7,2)      {$b_5$};
  \draw [-] (5,2) to  (5,0);
  \node[right]   at(5,1.5)      {$b_6$};
  \node[right]   at(5,0.5)      {$b_7$};
  \draw [-] (5,0) to  (6,0);
  \node[below]   at(5.5,0)      {$b_8$};
  \draw [-] (7,0) to  (7,1);
  \node[right]   at(7,0.5)      {$b_9$};
  \draw [-] (7,1) to  (8,1);
  \node[below]   at(7.7,1)      {$b_{10}$};

 \color{black}
 \path[draw] (3,2) to (0,2) to  (0,0) to  (-1,0) to  (-1,1) to  (-2,1);
  \draw [-] (1,2) to  (1,3);

 \path[draw] (1,0) to  (1,1) to  (2,1) to  (2,0)to  (1,0);
 \draw [-] (2,0) to  (2,-1);

 \path[draw] (1,0) to  (1,1) to  (2,1) to  (2,0)to  (1,0);
 \draw [-] (2,0) to  (2,-1);
- \node[above]   at(1.5,1)      {$S_1$};

 \node[above]   at(3,1)      {$S_2=\varnothing$};
  \draw [-] (5,2) to  (6,2);
 \node[above]   at(5.5,2)      {$S_5$};

 \node[above]   at(6.5,0)      {$S_8$};

  \path[draw]    (6,0)to  (6,-1)to  (7,-1)to  (7,0) to (6,0);
  \draw [-] (7,-1) to  (8,-1);

  \path[draw] (8,1)to  (9,1)to  (9,2)to  (8,2) to  (8,1);
  \path[draw] (9,0) to  (9,1) to  (10,1) to  (10,0) to (9,0) ;
  \node[above]   at(9,2)      {$S_{10}$};

 \end{tikzpicture}
 \caption{\label{fig-Animal-structure} A LA containing $x$ and $y$. All bonds of the backbone $(b_i)_{i=1,\dots,10}$ and all sausages $(S_i)_{i=0,\dots,10}$ are labeled in the picture.}
\end{center}
\end{figure}}

\section{The non-backtracking lace expansion}
\label{sec-Expansion}
\subsection{Derivation of the expansion}
\label{secExp}
In this section we derive the NoBLE for LTs and LAs using an algebraic expansion. The basic idea is similar to the classical lace expansion applied to SAWs derived first in \cite{BrySpe85}, and is an adaptation of the classical lace expansion for LTs and LAs derived in \cite{HarSla92c}. Let us consider a LT that contains $0,x$. As a LT does not contain any loops there exists a unique path connecting $0$ and $x$. We call this path the \emph{backbone}. A LA containing $0,x$ can contain loops. Therefore, a connection between $0$ and $x$ is not necessarily unique. The analogue to the backbone for LAs is formed by the {\em pivotal bonds}, which are defined as follows:

\begin{definition}[Double connections and pivotal elements]
Let $A$ be a LA that contains $x,y\in\Zd$. We say that $x$ and $y$ are \emph{doubly connected} in $A$ if there exist two edge-disjoint paths $p_1,p_2\subset A$ connecting $x$ and $y$. By convention, we say that $x\in A$ is doubly connected to itself. We call a bond $b\in A$ \emph{pivotal} for $x$ and $y$ if the removal of $b$ would disconnect $A$ into two disjoint animals, one containing $x$ and the other $y$.
\end{definition}

\ClusterPictureWithPivotalAndLabeling[1]

Let us consider a LA containing $x$ and $y$ as given in Figure \ref{fig-Animal-structure}. Then, $x$ and $y$ are either doubly connected or there exists at least one pivotal bond $b$. If there are multiple pivotal bonds, then they have a natural order for the connection from $x$ to $y$, because every self-avoiding path from $x$ to $y$ has to use the pivotal bonds $(b_i)_i$ in the same order.

The removal of the bonds $(b_i)_{i\geq 1}$ disconnects the animal into mutually non-intersecting pieces, which we denote in Figure \ref{fig-Animal-structure} by $(S_i)_{i\geq 1}$. We will consider bonds to be directed, and write $b=(x,y)$ for the bond directed from $x$ to $y$. For a directed bond $b=(x,y)$, we write $\bar{b}=y$ and $\underline{b}=x$ for the two ends of the bond. The pieces $(S_i)_{i\geq 1}$ form double connections between the end of one pivotal bond $\bar b_i$ to the beginning of the following pivotal bond $\underline b_{i+1}$.
We call these doubly connected pieces $(S_i)_{i\geq 1}$ {\em sausages}, and the sequence of pivotal bonds $(b_i)_{i\geq 1}$ the {\em backbone} of the animal.
For LTs double connections are not possible, so that  $S_i=\{\bar b_i\}=\{\underline b_{i+1}\}$.

So far, we follow the classical lace expansion as in \cite{HarSla92c}. Now we start to deviate, and we wish to do so in a way that is closer to the NBW.
In order to do this, we define the \emph{rib-walk} and \emph{sausage-walk} to characterize the combination of a backbone and a set of ribs/sausages, that are {\it non-backtracking}.
Using these walks we derive the NoBLE for LTs and LAs, the steps of the classical lace expansion, by expanding a graph-based description of the avoidance constraint of the sausages and of the backbone.
\begin{definition}[Ribs and rib-walks for LTs]\ \\[-6mm]
\label{defLTRibwalks}
\begin{enumerate}[(i)]
\item We call a LT $S$ that contains $x\in\Zd$ a \emph{rib} for $x$ and define that the empty set is also a rib for all $x\in\Zd$ (recall that a LT is a collection of {\em bonds}).

\item For $x,y\in \Zd$ and $n\geq 1$, we call a collection of $n$ oriented nearest-neighbor bonds $(b_i)_{i=1}^n$ and $n+1$ ribs $(S_i)_{i=0}^n$ an \emph{$n$-step rib-walk} from $x$ to $y$, if $S_0$ is a rib for $\underline b_1=x$, $S_n$ is a rib for $\bar b_n=y$, and, for $i=1,\dots,n-1$, the LT $S_i$ is a rib for $\bar b_i=\underline b_{i-1}$. We call $(b_i)_{i=1}^n$ the backbone of the rib-walk.

\item For a rib-walk $\omega=((b_i)_{i=1}^n, (S_i)_{i=0}^n)$, we define $|\omega|$ to be the number of bonds in the backbone. We denote the $i$th bond of the backbone by $b^{\omega}_i$  and the $i$th rib of $\omega$ by $S^{\omega}_i$.

\item We say that any LT containing the origin is a zero step rib-walk to the origin.

\item We call a rib-walk $\omega$ \emph{non-backtracking} if $\bar b^\omega_{i+1}\nin S^\omega_{i},\underline b^\omega_{i}\nin S^\omega_{i}$ and $\bar b^\omega_{i+1}\neq \underline  b^\omega_{i}$ for all $i$.

\item We define $\Wcal^\ssss[T](x)$ as the set of all rib-walks from the origin $0$ to $x$ and
 $\Wcal^{\ssss[T],\iota}(x)$ to be the set of all rib-walks $\omega$ from $0$ to $x$ such that $\ve[\iota]\nin S^\omega_0$ and $\bar b^\omega_1\neq \ve[\iota]$.
\end{enumerate}
\end{definition}
\medskip

We point out that a bond could be part of multiple ribs, so that there is no bijection between rib-walks and LTs containing $0$ and $x$.
This bijection is however possible if we restrict to {\em self-avoiding} rib-walks. The non-backtracking condition is part of this necessary self-avoidance constraint. It rules a specific notion of immediate reversals out. Thus, we can think of a LT as a non-backtracking rib-walk with extra mutual avoidance constraints between the ribs. We continue by defining similar quantities for LAs, to set the stage for an expansion that can treat LTs and LAs at the same time:

\begin{definition}[Sausages and sausage-walks for LAs]\ \\[-6mm]
\label{defLASausagewalks}
\begin{enumerate}[(i)]
\item We call a LA $S$ a \emph{sausage} for $(x,y)\in\Zd\times\Zd$, if $x$ and $y$ are doubly connected in $S$.
Further, we define that the empty set is a sausage for every $(x,x)$ with $x\in\Zd$.

\item For $x,y\in \Zd$ and $n\geq 1$ we call a collection of $n$ oriented nearest-neighbor bonds $(b_i)_{i=1}^n$ and $n+1$ sausages $(S_i)_{i=0}^n$ an \emph{$n$-step sausage-walk} from $x$ to $y$, when $S_0$ is a sausage for $(x,\underline b_1)$, $S_n$ is a sausage for $(\bar b_n,y)$ and $S_i$ is a sausage for $(\bar b_i,\underline b_{i+1})$ for $i=1,\dots,n-1$.

\item For a sausage-walk $\omega=((b_i)_{i=1}^n, (S_i)_{i=0}^n)$ we define $|\omega|$ to be the number of steps of $\omega$. We denote the $i$th bond of the backbone by $b^{\omega}_i$ and  the $i$th sausage of $\omega$ by $S^{\omega}_i$. We call $(b_i^\omega)_{i=1}^n=(b_i)_{i=1}^n$ the {\em backbone} of the sausage-walk.

\item For $x\in\Zd$ we define any sausage for $(0,x)$ to be a zero step sausage-walk from the origin to $x$.

\item We call a sausage-walk \emph{non-backtracking}, if $\bar b^\omega_{i+1}\nin S^\omega_{i},\underline b^\omega_{i}\nin S^\omega_{i}$ and $\bar b^\omega_i\neq \underline  b^\omega_{i-1}$ for all $i$.

\item We define $\Wcal^\ssss[A](x)$ as the set of all sausage-walks from the origin $0$ to $x$ and $\Wcal^{\ssss[A],\iota}(x)$ to be the set of all sausage-walks $\omega$ from $0$ to $x$ such that $\ve[\iota]\nin S^\omega_0$ and $\bar b_1^\omega\neq \ve[\iota]$.
\end{enumerate}
\end{definition}
\medskip

Similarly to Definition \ref{defLTRibwalks}, a key point in Definition \ref{defLASausagewalks} is that a sausage-walk does not necessarily lead to a LA. For a rib-/sausage-walk $\omega$, we define
	\begin{eqnarray}
	\lbeq{defKanimal}
	K[a,b](\omega)=\prod_{s=a}^{b-1}\prod_{s=t+1}^{b} \left(1 +\Ucal_{s,t}(\omega)\right),
	\end{eqnarray}
where $-\Ucal_{s,t}(\omega)$ is the indicator that the ribs/sausages $S_s^\omega$ and $S_t^\omega$ intersect at some point in $\Zd$. Then, $K[0,|\omega|]$ is the indicator that all ribs/sausages of the walk are self-avoiding. Thus, if $K[0,|\omega|](\omega)=1$, then the union of the oriented bonds and all ribs/sausages of $\omega$ is a disjoint union and the resulting object is a LT/LA. The pair interaction in \refeq{defKanimal} thus gives a convenient description of when rib-/sausage-walks lead to a LT/LA, and this representation allows us to expand out this pair interaction in a convenient way.

To capture the contribution of the ribs/sausages we define
	\begin{eqnarray}
	\lbeq{defZanimal}
	Z[a,b](\omega):=\prod_{i=a}^b z^{|S^{\omega}_i|},
	\end{eqnarray}
and remark that $Z[a,b](\omega)=Z[a,c](\omega)Z[c+1,b](\omega)$ for $c\in[a,b)$. We often drop the argument $\omega$ for $\Ucal_{st}, K[a,b]$ and $Z[a,b]$ when this can cause no confusion. Further, we drop the superscript $A$ and $T$ for $\Wcal$ when we consider both models simultaneously. We can now write the two-point function as
	\begin{eqnarray}
	\lbeq{expanimalGen}
	\bar G_{z}(x)&=&\sum_{n=0}^{\infty } t_{n}(x)z^n=\sum_{\omega\in\Wcal(x)}z^{|\omega|}Z[0,|\omega|]K[0,|\omega|].
	\end{eqnarray}
In the NoBLE, we use an adaptation of the two-point function given by
	\begin{eqnarray}
	\lbeq{defGzIotaanimal}
	\bar G^{\kappa}_z(x)&=&\sum_{\omega\in\Wcal^\kappa(x)}z^{|\omega|} Z[0,|\omega|]K[0,|\omega|].
	\end{eqnarray}
We expand $\bar G^{\kappa}$ using the same set of graphs and laces as used by Hara and Slade in \cite{HarSla90b}:

\begin{definition}[Graphs and connected graphs]
Let $a,b\in \Nbold$ with $a<b$. For $s,t\in [a,b]\cap\Nbold$ with $s<t$, the \emph{edge} between $s$ and $t$ is the tuple $(s,t)$. We abbreviate $st$ to denote $(s,t)$. We call a set of edges a \emph{graph}. We call a graph \emph{connected}, if for all $c\in[a,b]$ there exists an edge $st\in\Gamma$ such that $s\leq c\leq t$.
Let $\Bcal[a,b]$ be the set of all graphs on $[a,b]$ and $\Gcal[a,b]$ the set of all connected graphs on $[a,b]$.
\end{definition}

\begin{definition}[Laces and compatible edges]
We call a graph \emph{minimally connected} or a \emph{lace} if the removal of any edge would disconnect the graph and define $\Lcal[a,b]$ to be the set of all minimally connected graphs on $[a,b]$. We define the function ${\rm L}\colon \Gcal[a,b] \mapsto \Lcal[a,b]$ in a constructive manner as follows: For $\Gamma\in\Gcal[a,b]$, we let
	\begin{eqnarray*}
	&s_1=a, &t_1=\max \{t: at\in\Gamma\},  \\
	&t_i=\max \{t\colon \exists s \leq t_{i-1}\text{ such that } st\in\Gamma\},  &s_i=\min \{s\colon  st_i\in\Gamma\}.
	\end{eqnarray*}
This procedure ends after a finite number of steps $N$. We denote the resulting lace $L=\{s_1t_1,s_2t_2,\dots,s_Nt_N\}$ by ${\rm L}(\Gamma)$.
We call an edge $st\not\in L$ {\em compatible} to a lace $L$ if ${\rm L}(L\cup \{st\})=L$. We denote by $\Ccal(L)$ the set of all edges that are compatible with $L$.
\end{definition}
\hspace{-7mm}
We define, for $a>b$,
	\begin{eqnarray}
	\lbeq{defJanimal}
	J[a,b]&=&\sum_{\Gamma\in\Gcal[a,b]}\prod_{st\in \Gamma}\Ucal_{st}=\sum_{L\in\Lcal[a,b]}\prod_{st\in L}\Ucal_{st}\prod_{s't'\in \Ccal(L)}(1+\Ucal_{s't'}).
	\end{eqnarray}
Further, we see that
	\begin{eqnarray}
	K[a,b]&=&\prod_{t=a}^{b-1}\prod_{s=t+1}^{b} \left(1 +\Ucal_{s,t}(\omega)\right)=\sum_{\Gamma\in\Bcal[a,b]}\prod_{st\in \Gamma}\Ucal_{st}
	\end{eqnarray}
and $K[a,a]=J[a,a]=1$. The key observation in the lace expansion is that, for $a<b$, we can write
	\begin{eqnarray}
	\lbeq{Krecrelation}
	K[a,b]&=&\sum_{i=a}^{b-1}J[a,i]K[i+1,b]+J[a,b],
	\end{eqnarray}
see e.g. \cite[Lemma 3.4]{HarSla90b}. We apply \refeq{Krecrelation} to \refeq{expanimalGen} with $a=0$ and $b=|\omega|>0$ to obtain
	\begin{eqnarray}
	\lbeq{Genanimal1}
	\bar G_{z}(x)&=& \sum_{\omega\in\Wcal(x)} z^{|\omega|} Z[0,|\omega|] \sum_{i=0}^{|\omega|-1}J[0,i]K[i+1,|\omega|]\\
	\lbeq{Genanimal2}
	&&+\sum_{\omega\in\Wcal(x)} z^{|\omega|} Z[0,|\omega|]  J[0,|\omega|].
	\end{eqnarray}
Here, the second term also contains the case where $|\omega|=0$.
We define the contribution of \refeq{Genanimal2} to be $\bar \Xi_z(x)$, i.e.,
	\begin{eqnarray}
	\lbeq{defXianimal}
	\bar \Xi_z(x)&=&\sum_{\omega\in\Wcal(x)} z^{|\omega|} Z[0,|\omega|]  J[0,|\omega|].
	\end{eqnarray}
To further rewrite \refeq{Genanimal1}, we cut the non-backtracking rib-/sausage-walk at the $i$th bond of the backbone, $b_i=(y,y-\ve[\kappa])$, into a walk $\omega^1\in\Wcal$ from $0$ to $y$ (which could correspond to a one-point function for $|\omega^1|=0$), and a second walk $\omega^2\in\bigcup_{\kappa}\Wcal^{\kappa}$ from $y-\ve[\kappa]$ to $x$. This leads to
	\begin{align}
	 \sum_{\omega\in\Wcal(x)} z^{|\omega|} Z[0,|\omega|] \sum_{i=0}^{|\omega|-1}J[0,i]K[i+1,|\omega|]
	=& \sum_{y,\kappa}\sum_{\omega^1\in\Wcal(y)} z^{|\omega^1|+1} Z[0,|\omega^1|] J[0,|\omega^1|] \indic{y-\ve[\kappa]\nin S_{|\omega^1|}^{\omega^1}}\nnb
	&\qquad \times\sum_{\omega^2\in\Wcal^{\kappa}(x-y+\ve[\kappa])}z^{|\omega^2|} Z[0,|\omega^2|]K[0,|\omega^2|]\nnb
	=& z\sum_{y,\kappa}\bar \Psi^{\kappa}(y) \bar G_z^{\kappa}(x-y+\ve[\kappa]),
	\lbeq{GenanimalSplit}
	\end{align}
where
	\begin{eqnarray}
	\lbeq{defGzIotaanimal2}
	\bar G_z^{\kappa}(x)&=&\sum_{\omega\in\Wcal^{\kappa}(x)} z^{|\omega|} Z[0,|\omega|]K[0,|\omega|],\\
	\bar \Psi^{\kappa}_z(x)&=&\sum_{\omega\in\Wcal(y)} z^{|\omega|} Z[0,|\omega|] J[0,|\omega|] \indic{x-\ve[\kappa]\nin S_{|\omega|}^{\omega}}.
	\end{eqnarray}
In this way, we have obtained a recurrence relation for the two-point function given by
	\begin{eqnarray}
	\lbeq{animalbasic1}
	\bar G_z(x)&=&\bar \Xi_z(x)+z\sum_{y,\kappa}\bar \Psi^{\kappa}(y) \bar G_z^{\kappa}(x-y+\ve[\kappa]),
	\end{eqnarray}
which is the first step towards \refeq{Gx-in-ex-1}. For \refeq{Gx-in-ex-2}, we instead consider
	\begin{eqnarray}
	\bar G_{z}(x)- \bar G^{\iota}_{z}(x)&=&\sum_{\omega\in\Wcal(x)\setminus \Wcal^{\iota}(x)} z^{|\omega|} Z[0,|\omega|]K[0,|\omega|].
	\end{eqnarray}
As $\omega\in\Wcal(x)\setminus \Wcal^{\iota}(x)$ we know that $\ve[\iota]\in S_0^\omega$ or $\bar b_1=\ve[\iota]$.
For convenience, we define
	\begin{eqnarray}
	\lbeq{def-Eiota-forLTLA}
	\1_{\iota}(\omega)=\indic{|\omega|>0}\indic{\bar b_1=\ve[\iota]}+\indic{\ve[\iota]\in S_0^\omega},
	\end{eqnarray}
and remark that the non-backtracking condition of the rib-/sausage-walk excludes that $\ve[\iota]\in S_0^\omega$ and $\bar b_1=\ve[\iota]$ occur for the same walk. Therefore,
	\begin{align}
	\bar G_{z}(x)- \bar G^{\iota}_{z}(x)=&\sum_{\omega\in\Wcal(x)} z^{|\omega|} Z[0,|\omega|]K[0,|\omega|]\1_{\iota}(\omega)\nnb
	\stackrel{\refeq{Krecrelation}}=&\sum_{\omega\in\Wcal(x)} z^{|\omega|} Z[0,|\omega|]\1_{\iota}(\omega)
	\lbeq{ExpLTLATmp1}
	\left(\sum_{i=0}^{|\omega|-1}J[0,i]K[i+1,|\omega|]+J[0,|\omega|]\right).
	\end{align}
The contribution of $J[0,|\omega|]$ gives rise to
	\begin{align}
	\lbeq{defXiIotaanimal}
	\bar \Xi^{\iota}_z(x)&=\sum_{\omega\in\Wcal(x)} \1_{\iota}(\omega) z^{|\omega|} Z[0,|\omega|]J[0,|\omega|].
	\end{align}
Again, this term incorporates the contribution when $|\omega|=0$.
The dominant contribution of $J[0,i]K[i+1,|\omega|]$ in \refeq{ExpLTLATmp1} is given by $|\omega|\geq 1, i=0$ and $b^{\omega}_0=(0,\ve[\iota])$, for which we see that
	\begin{eqnarray*}
	\lbeq{defGzIotaanimal-rewrite}
	\sum_{\omega\in\Wcal(x)} \indic{b^{\omega}_1=(0,\ve[\iota])}z^{|\omega|} Z_R[0,|\omega|] J[0,0]K[1,|\omega|]=z\gj \bar G^{-\iota}_z(x-\ve[\iota]),
	\end{eqnarray*}
where $\gj=\bar G^{\iota}_z(0)$. We extract this contribution explicitly, and split the remainder at $b_i=(y,y-\ve[\kappa])$, as in \refeq{GenanimalSplit}, which leads to
	\begin{align}
	\lbeq{ExpLTLATmp2}
	\sum_{\omega\in\Wcal(x)}& z^{|\omega|} Z[0,|\omega|]\sum_{i=0}^{|\omega|-1}J[0,i]K[i+1,|\omega|]\1_{\iota}(\omega)(1-\delta_{i,0}\indic{b_1^\omega=(0,\ve[\iota])}) \\
	&=\sum_{y\in\Zd}\sum_\kappa\sum_{\omega^1\in\Wcal(y)}\sum_{\omega^2\in\Wcal^\kappa(x-y+\ve[\kappa])}
	z^{|\omega^1|+|\omega^2|+1} Z[0,|\omega^1|](\omega^1) Z[0,|\omega^2|](\omega^2) \nnb
	&\qquad\times J[0,|\omega^1|](\omega^1)K[0,|\omega^2|](\omega^2) (
	\1_{\iota}(\omega^1)+\indic{|\omega^1|=0} \indic{y\neq 0} \indic{y-\ve[\kappa]=\ve[\iota]})\indic{y-\ve[\kappa]\nin S^{\omega^1}_{|\omega^1|}}.\nn
	\end{align}
We define
	\begin{align}
	\lbeq{defPiIotaanimalIndic}
	\1^{\Pi}_{\iota}(\omega,x,\kappa)&=(\1_{\iota}(\omega)+ \indic{|\omega|=0}\indic{x\neq 0}\indic{x-\ve[\kappa]=\ve[\iota]}),\\
	\Pi^{\iota,\kappa}_z(x)&=\sum_{\omega\in\Wcal(x)}\1^{\Pi}_{\iota}(\omega,x,\kappa)
	\indic{x-\ve[\kappa]\nin S_{|\omega|}^{\omega}} z^{|\omega|+1} Z[0,|\omega|]J[0,|\omega|],
	\lbeq{defPiIotaanimal}
	\end{align}
and see that sums over $\omega^1$ and $\omega^2$ in \refeq{defGzIotaanimal-rewrite} factorize, to conclude that
	\begin{eqnarray}
	\lbeq{animalbasic2}
	\bar G_{z}(x)- \bar G^{\iota}_{z}(x) &=&\bar \Xi^{\iota}_z(x)+z\gj \bar G^{-\iota}_z(x-\ve[\iota])+\sum_{y,\kappa}\Pi_z^{\iota,\kappa}(y)\bar G_{z}^{\kappa}(x-y+\ve[\kappa]).
	\end{eqnarray}
This completes the derivation of the expansion for LT and LA for $\bar G_{z}(x)$. To obtain \refeq{Gx-in-ex-1} and \refeq{Gx-in-ex-2}, we need to divide \refeq{animalbasic1} and \refeq{animalbasic2} by $g_z$, as we will explain in more detail in the next section.
\qed

\subsection{Definitions used in the generalized analysis}
\label{sec-defineCoeff}
In this section we complete the expansion as stated in Proposition \ref{prop-LE} and used in the general analysis of \cite{FitHof13b}.
For the analysis, we use the {\em normalized} two-point functions LT and LA defined as
	\begin{align}
	G_z(x)&=\frac 1 {g_z} \bar G_z(x),\qquad\qquad G^\iota_z(x)=\frac 1 {g_z} \bar G^\iota_z(x).
	\end{align}
This has the advantage that the analysis and bounds on the lace-expansion coefficients are naturally divided into two parts.
The one-point functions $g_z$ and $g_z^\iota$ and their influence, see \refeq{conj-gzc}, are controlled using the bootstrap function $f_1$, see \refeq{defFunc1}.
The spatial dependence of the two-point functions is controlled using the bootstrap function $f_2$, see \refeq{defFunc2}.

To improve the performance of our analysis,
we extract some dominant contributions of the lace-expansion coefficients to be used explicitly within the analysis, so as to improve the numerical accuracy of the method. These explicit terms are defined in Section \ref{secSplitCoeff}. Before that we now complete the proof of Proposition \ref{prop-LE} by identifying the lace-expansion coefficients arising in it.

The NoBLE coefficients can be written into an alternating series of non-negative functions. The series arise in a natural way by the negative signs of the ${\mathcal U}_{st}$ terms in the expansion for LTs and LAs, see \refeq{defJanimal} as well as \cite[Chapter 2]{Fit13}, as we now explain in more detail:

\begin{definition}[Laces for with fixed number of edges] For $n\geq 1$ and $N\geq 1$,
let $\Lcal^{\ssc[N]}[0,n]\subseteq\Lcal[0,n]$ be the set of all laces $L\subseteq\Lcal[0,n]$ that consist of exactly $N$ edges.
\end{definition}
\noindent
Define
	\begin{align}
	\lbeq{defJNanimal}
	J^{\ssc[N]}[a,b](\omega)&= \sum_{L\in\Lcal^\ssc[N][a,b]}\prod_{st\in L}\Ucal_{st}\prod_{s't'\in \Ccal(L)}(1+\Ucal_{s't'}).
	\end{align}
A sausage/rib-walk $\omega$ for which the indicator $J^\ssc[N]$ equals one has $N$ intersecting sausages/ribs. These intersections are characterized by the lace $L$ in \refeq{defJNanimal}.
This gives a convenient interpretation to the NoBLE coefficients that allows for sharp bounds. For $N\geq 0$ and $x\in\Zd$, we define
	\begin{align}
	\lbeq{defXiNanimal}
	\bar \Xi^{\ssc[N]}_z(x) &= (-1)^N \sum_{\omega\in\Wcal(x)} z^{|\omega|} Z[0,|\omega|]  J^\ssc[N][0,|\omega|],\\
	\lbeq{defXiNIotaanimal}
	\bar \Xi^{\ssc[N],\iota}_z(x)&=(-1)^N \sum_{\omega\in\Wcal(x)} z^{|\omega|} Z[0,|\omega|]J^\ssc[N][0,|\omega|]\1_{\iota}(\omega),\\
	\lbeq{defPsiNanimal}
	\bar \Psi^{\ssc[N],\kappa}_z(x)&=(-1)^N\sum_{\omega\in\Wcal(y)} z^{|\omega|} Z[0,|\omega|] J^\ssc[N][0,|\omega|] \indic{x-\ve[\kappa]\nin S^{\omega}_{|\omega|}},\\
	\Pi^{\ssc[N],\iota,\kappa}_z(x)&=(-1)^N \sum_{\omega\in\Wcal(x)} z^{|\omega|+1} Z[0,|\omega|] J^\ssc[N][0,|\omega|]\indic{x-\ve[\kappa]\nin S_{|\omega|}^{\omega}}
		\1^{\Pi}_{\iota}(\omega,x,\kappa),
	\lbeq{defPiNanimal}
	\end{align}
to be the restrictions of the NoBLE coefficients to laces of fixed size, with the convention that $J^\ssc[0][0,|\omega|]=\delta_{0,|\omega|}$.

The dominant contributions of these coefficients are given by
	\begin{align}
	\lbeq{dominate-piece-part1}
	\bar \Xi^{\ssc[0]}_z(0) =&\bar G_z(0)=g_z,\quad \bar \Psi^{\ssc[0],\kappa}_z(0)= \bar G^{\iota}_z(0)=\gj,\quad
	\Xi^{\ssc[0],\iota}_z(0)=\bar G_z(\ve[\iota]).
	\end{align}
We use these terms explicitly in our analysis. We perform the NoBLE analysis using the following coefficients:
	\begin{align}
	\lbeq{defXiNanimalPrime}
	\Xi^{\ssc[N]}_z(x)&=\frac 1 {g_z} \big(\bar \Xi^{\ssc[N]}_z(x)-\delta_{0,x}\delta_{0,N}\bar \Xi^{\ssc[0]}_z(0)\big),\qquad \Xi^{\ssc[N],\iota}_z(x)=\frac 1 {g_z} \bar \Xi^{\ssc[N],\iota}_z(x),\\
	\lbeq{defPsiNanimalPrime}
	\Psi^{\ssc[N],\kappa}_z(0)&=\frac 1 {g^\iota_z} \big(\bar \Psi^{\ssc[N],\kappa}_z(x)-\delta_{0,x}\delta_{0,N}\bar \Psi^{\ssc[0],\kappa}_z(0)\big),\\
	\Xi_z(x)&=\sum_{N=0}^{\infty}(-1)^N \Xi^{\ssc[N]}_z(x),\qquad \Xi^{\iota}_z(x)=\sum_{N=0}^{\infty}(-1)^N \Xi^{\ssc[N],\iota}_z(x),\\
	\Psi^{\kappa}_z(x)&=\sum_{N=0}^{\infty}(-1)^N \Psi^{\ssc[N],\kappa}_z(x),\qquad \Pi^{\iota,\kappa}_z(x)=\sum_{N=0}^{\infty}(-1)^N \Pi^{\ssc[N],\iota,\kappa}_z(x),
	\lbeq{LALTCoefficientNormalization}
	\end{align}
where we explicitly subtract the dominant contributions, of which most of them are also present for the NBW. The situation simplifies considerably for LTs, due to the absence of double connections, so that	
	\eqn{
	\lbeq{split-LTs-N0}
	\Xi^{\ssc[0],\iota}_{z}(x)=\delta_{0,x} G_z(\ve[\iota]),
	\qquad
	\Xi^{\ssc[0]}_z(x)=\Psi^{\ssc[0],\iota}_z(x)=0,
	}
while, using \refeq{def-Eiota-forLTLA}, as well as \refeq{defPiIotaanimalIndic}--\refeq{defPiIotaanimal},
	\eqn{
	\lbeq{split-LTs-N0-Pi}
	\Pi^{\ssc[0],\iota,\kappa}_z(x)=\delta_{0,x} z \sum_{T\colon 0\in T} z^{|T|}\indic{x-\ve[\kappa]\nin T} \indic{\ve[\iota]\in T}.
	}
For LAs, several further contributions due to $x\neq 0$, which require a double connection, arise.
Using the above notation we obtain \refeq{Gx-in-ex-1}-\refeq{Gx-in-ex-2} by dividing the equations \refeq{animalbasic1} and \refeq{animalbasic2} by $g_z$. Thereby, we have completed the proof of Proposition \ref{prop-LE} and have identified the NoBLE coefficients appearing in it.
\qed

\subsection{Split for the coefficients}
\label{secSplitCoeff}
In this section, we define further terms that we extract from the NoBLE coefficients to be used in the analysis of \cite{FitHof13b}. By extracting these terms we reduce the size of the perturbation, which ultimately improves the numerical accuracy of our analysis. We give first the definition for LAs and then explain how terms simplify for LTs.

\paragraph{Definition of the explicit terms for $N=0$.}
For a LA $A$ and $x,y\in\Zd$, we denote by $x\dbct{A} y$ the event that $x$ and $y$ are doubly connected via bonds in $A$. By $d_A(x,y)$ we denote the intrinsic distance of $x,y$ in $A$, meaning the length of the shortest connection of $x$ and $y$ using bonds in $A$. For $N=0$, we define
	\begin{align}
	\lbeq{LA-Split-Def-1}
	\Xi^{\ssc[0]}_{\alpha,z}(x)&=(1-\delta_{0,x})\frac {1}{g_z}\sum_{A\colon 0,x\in A}\indic{0\dbct{A} x}\indic{d_A(0,x)=1}z^{|A|}, \\
   	 \Psi^{\ssc[0],\kappa}_{\alpha,{\sss I},z}(x)&=(1-\delta_{0,x})\frac {1}{g^\iota_z}\sum_{A\colon 0,x\in A}\indic{0\dbct{A} x}
    	\indic{x-\ve[\kappa]\nin A}\\
	&\qquad \times (\indic{d_A(0,x)=1}\indic{x=\ve[\kappa]}+\indic{|\ve[\kappa]-x|=1}\indic{d_A(0,x)=2}),\nn\\
    	\Psi^{\ssc[0],\kappa}_{\alpha,{\sss II},z}(x)&=(1-\delta_{0,x})\frac {1}{g^\iota_z}\sum_{A\colon 0,x\in A}\indic{0\dbct{A} x}\indic{d_A(0,x)=1} \indic{x-\ve[\kappa]\nin A},
	\end{align}
and
	\begin{align}
	\lbeq{LA-Split-Def-XiIotaAlphaa}
    	\Xi^{\ssc[0],\iota}_{\alpha,{\sss I},z}(x)&=\frac {1}{g_z}
    	\sum_{A\colon 0,x,\ve[\iota]\in A} z^{|A|}\indic{0\dbct{A} x} \left(\delta_{0,x}+\delta_{\ve[\iota],x} +(1-\delta_{0,x})\indic{d_A(\ve[\iota],x)=1}\right),\\
	\lbeq{LA-Split-Def-XiIotaAlphab}
    	\Xi^{\ssc[0],\iota}_{\alpha,{\sss II},z}(x)&=\frac {1}{g_z}\sum_{A\colon 0,x,\ve[\iota]\in A} z^{|A|}\indic{0\dbct{A} x} \left(\delta_{0,x}+\delta_{\ve[\iota],x}
	+(1-\delta_{\ve[\iota],x})\indic{d_A(0,x)=1}\right),\\
	\Pi^{\ssc[0],\iota,\kappa}_{\alpha,z}(x)&=
    	\sum_{A\colon 0,x,\ve[\iota]\in A} z^{|A|+1}\indic{0\dbct{A} x}\indic{x-\ve[\kappa]\nin A}\left(\delta_{\ve[\iota],x}+\delta_{\ve[\iota]+\ve[\kappa],x}\indic{d_A(\ve[\iota]+\ve[\kappa],\ve[\iota])=1}\right).
	\lbeq{LA-Split-Def-1a}
	\end{align}
For $N=1$, the terms are more involved and are given by
    	\begin{align}
    	\lbeq{LA-Split-Def-None}
    	\Xi^{\ssc[1]}_{\alpha,z}(x)&=\frac 1 {g_z}\sum_{\omega\in\Wcal(x)} z^{|\omega|} Z[0,|\omega|]J^{\ssc[1]}[0,|\omega|]\indic{\bb^\omega_0=0}\\
    	&\qquad\times \left(\delta_{0,x}+\delta_{|\omega|,1}\delta_{\tb^\omega_0,x}
    	+\indic{x\in S_0\cap S_{|\omega|}}\indic{d_{S_0}(0,x)=1} +
   	 \indic{0\in S_0\cap S_{|\omega|}\not \ni x}\indic{d_{S_{|\omega|}(0,x)=1}}\right),\nnb
    	\Psi^{\ssc[1],\kappa}_{\alpha,{\sss II},z}(x)&=\frac 1 {g^\iota_z}\sum_{\omega\in\Wcal(x)} z^{|\omega|} Z[0,|\omega|]J^{\ssc[1]}[0,|\omega|]\indic{\bb^\omega_0=0}
	\indic{x-\ve[\kappa]\nin S_{|\omega|}^{\omega}}\\
    	&\qquad \times\left(\delta_{0,x}+\delta_{|\omega|,1}\delta_{\tb^\omega_0,x}+\indic{x\in S_0\cap S_{|\omega|}}\indic{d_{S_0}(0,x)=1} +
    	\indic{0\in S_0\cap S_{|\omega|}\not \ni x}\indic{d_{S_{|\omega|}(0,x)=1}}\right),\nnb
     	\Psi^{\ssc[1],\kappa}_{\alpha,{\sss I},z}(x)&=(\delta_{0,x}+\delta_{x,\ve[\kappa]})
     	\Psi^{\ssc[1],\kappa}_{\alpha,{\sss II},z}(x)\lbeq{LA-Split-Def-10}\\
     	&\quad +\frac {2d D(x-\ve[\kappa])}{g^\iota_z}
     	\sum_{\omega\in\Wcal(x)} z^{|\omega|} Z[0,|\omega|]J^{\ssc[1]}[0,|\omega|]\indic{x-\ve[\kappa]\nin S_{|\omega|}^{\omega}}\delta_{|\omega|,2}
	\indic{\bb^\omega_0=0}\delta_{\tb^\omega_0,\bb^\omega_1}\delta_{\tb^\omega_1,x}.\nn
    	\end{align}
All sausage-walks contributing to these sums have the following in common: (a) The first step of the backbone $\bb^\omega_0$ starts at the origin.
(b) The sausages $S_0$ and $S_{|\omega|}$ intersect either at $0$ or at $x$. (c) At least one bond is explicitly present. Most of the times this bond is $(0,x)$, forced by $d_A(0,x)=1$.\\
We define the remainder terms arising through this split by
	\begin{align}
  	\Xi^{\ssc[0]}_{{\sss R},z}&=\Xi^{\ssc[0]}_{z}-\Xi^{\ssc[0]}_{\alpha,z},\qquad\qquad\quad
  	\Xi^{\ssc[1]}_{{\sss R},z}=\Xi^{\ssc[1]}_{z}-\Xi^{\ssc[1]}_{\alpha,z},\\
  	\Psi^{\ssc[0],\iota}_{{\sss R, I},z}&=\Psi^{\ssc[0],\iota}_{z}-\Psi^{\ssc[0],\iota}_{\alpha,{\sss I},z},\qquad\qquad
  	\Psi^{\ssc[0],\iota}_{{\sss R, II},z}=\Psi^{\ssc[0],\iota}_{z}-\Psi^{\ssc[0],\iota}_{\alpha,{\sss II},z},\\
  	\Psi^{\ssc[1],\iota}_{{\sss R, I},z}&=\Psi^{\ssc[1],\iota}_{z}-\Psi^{\ssc[1],\iota}_{\alpha,{\sss I},z},\qquad\qquad\
  	\Psi^{\ssc[1],\iota}_{{\sss R, II},z}=\Psi^{\ssc[1],\iota}_{z}-\Psi^{\ssc[1],\iota}_{\alpha,{\sss II},z},\\
  	\Xi^{\ssc[0],\iota}_{{\sss R, I},z}&=\Xi^{\ssc[0],\iota}_{z}-\Xi^{\ssc[0],\iota}_{\alpha,{\sss I},z},\qquad\qquad\ \
  	\Xi^{\ssc[0],\iota}_{{\sss R, II},z}=\Xi^{\ssc[0],\iota}_{z}-\Xi^{\ssc[0],\iota}_{\alpha,{\sss II},z},\\
  	\Pi^{\ssc[0],\iota,\kappa}_{{\sss R},z}&=\Pi^{\ssc[0],\iota,\kappa}_{z}-\Pi^{\ssc[0],\iota,\kappa}_{\alpha,z}.
	\end{align}
This completes the necessary split for LAs.

\paragraph{Lattice trees.}
We use the same definitions for LTs, where we sum over LTs $T$ instead of LAs $A$.
However, for LTs the terms simplify considerably, as double connections are not possible. This is especially true for $N=0$, for which, for all $x\in\Zd$ (recall \refeq{split-LTs-N0}),
	\begin{align}
	\lbeq{split-LTs-N0a}
	\Xi^{\ssc[0]}_{\alpha,z}(x)=\Psi^{\ssc[0],\iota}_{\alpha,{\sss I},z}(x)=\Psi^{\ssc[0],\iota}_{\alpha,{\sss II},z}(x)=0.
	\end{align}
Further, the split actually captures the complete contribution of $\Xi^{\ssc[0],\iota}_z$, since
	\begin{align}
	\lbeq{split-LTs-N0b}
	\Xi^{\ssc[0],\iota}_{z}(x)=\Xi^{\ssc[0],\iota}_{\alpha,{\sss I},z}(x)&=\Xi^{\ssc[0],\iota}_{\alpha,{\sss II},z}(x)=\delta_{0,x}  \frac {\bar G_z(\ve[\iota])}{g_z}=\delta_{0,x} G_z(\ve[\iota]).
	\end{align}
This completes the derivation of the split of the coefficients as used in \cite[Section 4]{FitHof13b}.

\subsection{Assumptions on the model}
\label{sec-gzAssumption}
In this section, we verify most of the assumptions necessary to apply the analysis of \cite{FitHof13b}.
We start by proving the assumption that are independent of the NoBLE: \cite[Assumptions 2.2, 2.3 and 2.4]{FitHof13b}.

\subsubsection{Assumption on the two-point function}
We begin with \cite[Assumption 2.4]{FitHof13b} as it will help us prove \cite[Assumptions 2.2]{FitHof13b}.
\paragraph{\cite[Assumption 2.4]{FitHof13b}:}  For $z\in[0,z_c)$, the functions $ z\mapsto\aabz$ and $z\mapsto\aaz$ are continuous.\\
To verify this assumption, we choose $\aabz=zg_z$ and $\aaz=zg_z^\iota$, with
	\begin{align}
	\gj=\sum_{A\colon A\ni 0}z^{|A|} \qquad \text{ and }\qquad \gj=\sum_{A\colon \ve[1]\nin A\ni 0}z^{|A|}.
	\end{align}
By Abel's Theorem, the one-point function is continuous within the radius of convergence of these power series, which is at most $z_c$, see \refeq{defLTLAsusceptibility} and the text thereafter. Thus, also $z\mapsto\aabz$ and $z\mapsto\aaz$ are continuous.

\paragraph{\cite[Assumption 2.2]{FitHof13b}:}There exists a $z_I\in[0,z_c)$ such that
	\begin{align}
	\lbeq{assInitialXBound}
	G_{z}(x)\leq B_{1/(2d-1)}(x)=\frac {2d-2}{2d-1} C_{1/2d}(x)
	\end{align}
for all $x\in \Zd$ and $z\in[0,z_I]$.\\

To prove this statement for LTs and LAs we adapt an argument used in \cite[Proof of Lemma 3.1]{HarSla90b}.
We know that each LT/LA containing $0$ and $x$ contains a path of bonds that connects $0$ and $x$. Each point of the path is connected to at most one rib/sausage.
The weight of all possible rib/sausages $z^{|S_i|}$ (see \refeq{defZanimal}) can be bounded by $g_z$, as each rib/sausage $S_i$ is itself a LT/LA.
For $x\neq 0$ we can improve this by bounding the weight by $\gj$ instead. This is possible as each rib/sausage needs to avoid at least the next/previous step of the path from $0$ to $x$. From this, we conclude for $x\neq 0$ that
	\begin{align}
	\lbeq{LTInitalBound-goal}
	\bar G_z(x)\leq \gj \sum_{\omega\in\Wcal^{\ssss[NBW]}(x)} (z\gj)^{|\omega|}=\gj B_{z\gj}(x),
	\end{align}
where $\Wcal^{\ssss[NBW]}(x)$ is the set of all NBWs (see the text above \refeq{NBWGen}) starting at $0$ and ending at $x$, and $|\omega|$ is the number of steps of the NBW $\omega$. The inequality \refeq{LTInitalBound-goal} then follows for all $z$ for which $z\gj\leq (2d-1)^{-1}$ and $x\neq 0$. We define
	\begin{align}
	\lbeq{DefinitionOfzI}
	z_I:=\sup\Big\{z\colon zg_{z}^\iota\leq \frac 1 {2d-1} \Big\}.
	\end{align}
This is well defined as $\aaz=zg_{z}^\iota$ is continuous and non-decreasing in $z$ with $zg_{z}^\iota=0$ when $z=0$. To complete the proof of \cite[Assumption 2.2]{FitHof13b}, we still need to prove that $z_I<z_c$.
For this we note that $g_z-g_z^\iota=\bar G_z(\ve[1])$ and use \refeq{LTInitalBound-goal} to obtain
	\begin{align}
	\lbeq{LTInitalBound-goal-step}
	\bar \chi(z)=\sum_x\bar G_z(x)\leq \bar G_z(0)+\sum_{x\neq 0} \gj B_{z\gj}(x)\leq \gj \big(1+ B_{z\gj}(\ve[1]) + \chi^{\sss \rm NBW}(z\gj)\big).
	\end{align}
From this we conclude that $z_I\leq z_c$ as otherwise $\chi(z_c)=\infty$, while $\chi^{\sss \rm NBW}(z_cg_{z_c}^\iota)$ on the right-hand side is finite.
To exclude that $z_I=z_c$, we see that \refeq{LTInitalBound-goal} also holds when we replace $\gj$ with
	\begin{align}
	\lbeq{gz'-def}
	g'_z=\min_{\kappa\neq 1} \sum_{A:\ve[1],\ve[\kappa]\not \in A\ni 0}z^{|A|}
	\end{align}
for all ribs/sausages except those at $0$ and $x$. We define $z'$ as the value $z$ such that $z'g'_{z'}=(2d-1)^{-1}$, and  conclude as in \refeq{LTInitalBound-goal-step} that $z_c \geq z'>z_I$.
Recalling $\bar G_z(x)=g_z G_z(x)$ and $\gj<g_z$ we  conclude \refeq{LTInitalBound-goal} from \refeq{assInitialXBound}. This concludes the proof of \cite[Assumption 2.4]{FitHof13b}.

Additionally, we prove the lower bound $z_I\geq (2d-1)^{-1}\e^{-1}$. As this uses  ideas not used elsewhere, we move the proof to Appendix \ref{sec-lemmaAnalysisLABound} (see Lemma \ref{lem-LBzI}).
\noindent

\paragraph{\cite[Assumption 2.3]{FitHof13b}: Growth of the two-point function.}
We need to show that for every $x\in \Zd$, the two-point functions $z\mapsto G_z(x)$ and $z\mapsto G^\iota_z(x)$
are non-decreasing and differentiable in $z\in(0,z_c)$. Further, we need to show that for all $\varepsilon>0$, there exists a constant $c_{\varepsilon}\geq 0$ such that for all $z\in(0,z_c-\varepsilon)$ and $x\in\Zd\setminus\{0\}$,
	\begin{eqnarray}
	\lbeq{assGzDiffBound}
	\frac d {dz} G_z(x)\leq c_{\varepsilon} (G_z\star D\star G_z)(x),
	\quad \text{ and therefore }\quad \frac d {dz} \hat G_z(0)\leq c_{\varepsilon} \hat G_z(0)^2.
	\end{eqnarray}
Finally, we need to show that for each $z\in(0,z_c)$, there exists a constant $K(z)<\infty$ such that $\sum_{x\in\Zd} |x|^2 G_{z}(x)<K(z)$. We will do this now.\\

As a generating function of a non-negative sequence (see \refeq{defLTLATwoPoint}), the two-point function is clearly non-decreasing in the parameter $z$
as well as differentiable in $z$ for $z\in(0,z_c)$. Next, we first prove the bound on the derivative in \refeq{assGzDiffBound} for LTs.
We know that a LT $T$ with $|T|$ edges contains $|T|+1$ vertices. We use this property to compute for $x\neq 0$
	\begin{align}
	\lbeq{LTDiffBound-tmp1}
	\sum_{T\ni 0,x} \frac d {dz} z^{|T|}=\sum_{T\ni 0,x} {|T|} z^{|T|-1}=\sum_{y\neq 0}\sum_{T\ni 0,x,y} z^{|T|-1}.
	\end{align}
As a LT $T\ni x,y$ contains no loops, the path from $x$ to $y$ is unique. We denote this path by $b^T(x,y)$.
By $u$ we denote the last vertex that the paths from $0$ to $x$ and from $0$ to $y$ have in common.
For $u\neq x$ we split the walk at $u$ into three individual trees and bound the contributions of these individual trees by two-point functions.
Doing this, we have to take into account that in \refeq{LTDiffBound-tmp1} the tree $T$ is only weighted by $z^{|T|-1}$,
so that we have to choose one bond of the tree that does not carry the weight $z$.
We choose the first step of the path from $u$ to $x$ to be this bond.
For $u=x$ there exists no first step. In this case, we choose the last step of the path from $0$ to $x$ to be the bond without weight $z$. Using that $\bar G_{z}(0)\geq 1$, we obtain the bound
	\begin{align}
	\sum_{T\ni 0,x} \frac d {dz} z^{|T|}\leq&\ \sum_{y\neq 0} \sum_{u\neq x,v}  \bar G^{(t)}_z(u)2d D(v-u)\bar G^{(t)}_{z}(x-v) \bar G^{(t)}_z(u-y)\lbeq{LTDiffBound-tmp2}\\
	&\ + \sum_{y\neq 0} \sum_{v} \bar G^{(t)}_z(v)2dD(v-x) \bar G^{(t)}_{z}(x-y) \nnb
	\leq &\ 4d (\bar G^{(t)}_z\star D\star \bar G^{(t)}_z)(x) \sum_y \bar G^{(t)}_z(y)= 4d \hat {\bar G}^{(t)}_z(0) (\bar G^{(t)}_z\star D\star \bar G^{(t)}_z)(x).
	\nn
	\end{align}
For the normalized two-point function, we conclude
	\begin{align}
	\frac d {dz} G^{(t)}_z(x)&= \frac d {dz} \frac {\bar G^{(t)}_z(x)}{g^{(t)}_z}=
	\frac 1 {g^{(t)}_z}\sum_{T\ni 0,x} \frac d {dz} z^{|T|} -\frac {\bar G^{(t)}_z(x)} {(g^{(t)}_z)^2} \sum_{T\ni 0} \frac d {dz} z^{|T|} \nnb
	&\leq 4d (g^{(t)}_z)^2 \hat G^{(t)}_z(0)(G^{(t)}_z\star D\star G^{(t)}_z)(x),
	\lbeq{LTDiffBound-tmp3}
	\end{align}
as required. To conclude such an inequality for LAs, we note that an animal with $|A|$ bonds contains at least $|A|/d$ vertices and as for the LT compute that
	\begin{align}
	\frac d {dz} G^{(a)}_z(x)&=\frac 1 {g^{(a)}_z}\sum_{A\ni 0,x} |A|z^{|A|-1}\leq \frac {d} {g^{(a)}_z}\sum_{A\ni 0,x,v} z^{|A|-1}\nnb
	&\leq 4d^2 (g^{(a)}_z)^2 \hat G^{(a)}_z(0)(G^{(a)}_z\star D\star G^{(a)}_z)(x).
	\lbeq{LTDiffBound-tmp4}
	\end{align}
	
As the last step we prove that for all $z< z_c$ there exists $K(z)<\infty$ such that
 $\sum_{x\in\Zd} |x|^2 G_{z}(x)<K(z)$. In \cite[(1.1)]{HarSla92c}, it is proved that the connectivity constants $\lambda=1/z_c$ can be used to prove that
	\begin{align}
	t_n(0) \leq  \lambda^n (n+1).
	\end{align}
Let $|x|_\infty:=\max_{i=1}^d |x_i|$ be the supremum norm and compute
	\begin{align}
	\bar G_z(x) = & \sum_{n=|x|_\infty}^\infty t_n(x)z^n\leq \sum_{n=|x|_\infty}^\infty t_n(0)z^n
	\leq \sum_{n=|x|_\infty}^\infty (z/z_c)^n (n+1).
	\end{align}
From this we conclude that $\bar G_z(x)$ decays exponentially, i.e., there exists $c,m(z)\in(0,\infty)$ such that
	\begin{align}
	\bar G_z(x) \leq&\sum_{n=|x|_\infty}^\infty (n+1)(z/z_c)^{n}\leq c \e^{-m(z) |x|_\infty}.
	\end{align}
We use this bound to conclude that
	\begin{align}
	\sum_{x\in\Zd} |x|^2\bar G_z(x) \leq&\sum_{x\in\Zd} |x|^2 c \e^{-m(z) |x|_\infty}
	\stackrel{|x|^2\leq d^2 |x|_\infty^2}\leq \sum_{n=1}^{\infty} (dn)^2c \e^{-m(z) n}\sum_{x\colon |x|_\infty=n}1\nnb
\leq&\ d^3 \sum_{n=1}^{\infty} n^{d+1}c \e^{-m(z) n}:=K(z)<\infty,
	\end{align}
which proves the desired statement.
\qed

\subsubsection{Assumptions on the NoBLE-coefficients}
In this section, we verify the assumptions on the NoBLE coefficients formulated in \cite[Assumptions 4.1, 4.2 and 4.4]{FitHof13b}.

\paragraph{\cite[Definition 2.5]{FitHof13b} Symmetry of the model.}{\it
We denote by $\mathcal{P}_d$ the set of all permutations of $\{1,2,\dots, d\}$. For $\nu\in \mathcal{P}_d$, $\delta\in\{-1,1\}^d$ and $x\in\Zd$, we define $p(x;\nu,\delta)\in\Zd$ to be the vector with entries $(p(x;\nu,\delta))_j=\delta_j x_{\nu_j}$. We say that a function $f\colon \Zd\mapsto \Rbold$ is {\em totally rotationally symmetric} (TRS) when $f(x)=f(p(x;\nu,\delta))$ for all $\nu\in \mathcal{P}_d$ and $\delta\in\{-1,1\}^d$.}\\[2mm]
Total rotational symmetry is natural on $\Zd$, e.g.\ the two-point function $G_z$ as well as the NBW and SRW two-point functions have this symmetry.
We next argue that the NoBLE coefficients have similar symmetries.\\[2mm]
{\bf \cite[Assumption 4.1]{FitHof13b}.} Let $\iota,\kappa\in\{\pm 1,\pm 2,\dots,\pm d\}$. {\it The following symmetries hold for all $x\in\Zd$, $z\leq z_c$, $N\in\Nbold$ and $\iota,\kappa$:
	\begin{eqnarray}
	\lbeq{Symmetrie-goal}
	\Xi^\ssc[N]_z(x)&=& \Xi^\ssc[N]_z(-x), \qquad  \qquad \qquad
	\Xi^{\ssc[N],\iota}_z(x)= \Xi^{\ssc[N],-\iota}_z (-x),\\
	\lbeq{Symmetrie-goal2}\Psi^{\ssc[N],\iota}_z(x)&=& \Psi^{\ssc[N],-\iota}_z(-x), \qquad \qquad\
	\Pi^{\ssc[N],\iota,\kappa}_z(x)= \Pi^{\ssc[N],-\iota,-\kappa}_z (-x).
	\end{eqnarray}
For all $N\in\Nbold$, the coefficients
	\begin{align}
	\lbeq{TRS-sums}
	\Xi^\ssc[N](x),\qquad \sum_{\iota}\Psi^{\ssc[N],\iota}_z(x),
	\qquad \sum_{\iota}\Xi^{\ssc[N],\iota}_z(x) \quad \text{and}\quad \sum_{\iota,\kappa}\Pi^{\ssc[N],\iota,\kappa}_z(x),
	\end{align}
as well as the remainder terms of the split
	\begin{align}
	\lbeq{TRS-sums2}\Xi^{\ssc[N]}_{{\sss R},z}(x),\quad \sum_{\iota}\Psi^{\ssc[N],\iota}_{{\sss R, I},z}(x),
	\quad \sum_{\iota}\Psi^{\ssc[N],\iota}_{{\sss R, II},z}(x),
	\quad \sum_{\iota}\Xi^{\ssc[0],\iota}_{{\sss R, I},z}(x)
	\quad \sum_{\iota}\Xi^{\ssc[0],\iota}_{{\sss R, II},z}(x),
	\quad \sum_{\iota,\kappa}\Pi^{\ssc[0],\iota,\kappa}_{{\sss R},z}(x),
	\end{align}
are totally rotationally symmetric functions of $x\in\Zd$.
Further, the dimensions are exchangeable, i.e., for all $\iota,\kappa$,
	\begin{align}
 	\lbeq{assCoefficentsDimInterchange}
	\hat \Psi^{\ssc[N],\iota}_z(0)=\ \hat \Psi^{\ssc[N],\kappa}_z(0),\qquad
	\hat \Xi^{\ssc[N],\iota}_z(0)=\ \hat \Xi^{\ssc[N],\kappa}_z(0),\qquad
	\sum_{\kappa'}\hat \Pi^{\ssc[N],\iota,\kappa'}_z(0)=\sum_{\iota'}\hat \Pi^{\ssc[N],\iota',\kappa}_z(0).
	\end{align}}
By the definition of the NoBLE coefficients in Section \ref{secExp}, it is easy to see that \refeq{Symmetrie-goal}, \refeq{Symmetrie-goal2} and \refeq{assCoefficentsDimInterchange} hold.
The TRS stated in \refeq{TRS-sums}, \refeq{TRS-sums2} might be less obvious.
The definition of $\Xi^\ssc[N]_z$ does not include constraints on specific directions, so it is not difficult to see that $x\mapsto \Xi^\ssc[N]_z(x)$ is TRS for all $N\geq 0$.
The other three NoBLE coefficients are not TRS as their definition includes constraints on one or two specific directions.
For example, the coefficient $\Psi^{\ssc[N],\kappa}_z(x)$ includes the constraint that $x-\ve[\kappa]$ is not in the last rib/sausage.
When we sum over $\kappa$, though, the directional constraint is averaged out and thus $\sum_{\kappa}\Psi^{\ssc[N],\kappa}_z(x)$ is TRS.
For the same reason, the sums over $\iota$ and $\iota,\kappa$ in \refeq{TRS-sums}, as well as the stated remainder terms, are TRS.
The arguments given above hold, when the coefficients are well defined, which is definitely the case for all $z<z_c$.

\cite[Assumption 4.1]{FitHof13b} states that the symmetry properties also hold for $z=z_c$, where it is not even obvious that these objects are well defined.
We verify the left-continuity in \cite[Assumption 4.4]{FitHof13b} below, from which the symmetries will follow also for $z=z_c$.
Further, inspection of the proof in \cite{FitHof13b} shows that the symmetries are only used for $z<z_c$.
\qed

\paragraph{\cite[Assumption 4.2]{FitHof13b} Relation between coefficients.}{\it
For all $x\in\Zd$, $z\leq z_c$, $N\in\Nbold$ and $\iota,\kappa\in\{\pm 1,\pm 2,\dots,\pm d\}$, the following bounds hold:
	\begin{align}
	\lbeq{XidominatespsiImproved}
	\Psi^{\ssc[N],\kappa}_z(x)\leq&\frac {\aabz}{ \aaz} \Xi^\ssc[N]_z(x),
	\qquad\qquad\Pi^{\ssc[N],\iota,\kappa}_z(x)\leq  \aabz \Xi^{\ssc[N],\iota}_z(x).
	\end{align}}
\medskip

We first compare $\bar \Xi^{\ssc[N]}$ and $\bar \Psi^{\ssc[N],\kappa}$, in \refeq{defXiNanimal},\refeq{defPsiNanimal} and see that they
differ by the additional condition $x-\ve[\kappa]\nin S^{\omega}_{|\omega|}$, so that $\bar \Xi^{\ssc[N]}(x)\leq \bar \Psi^{\ssc[N],\kappa}(x)$.
Considering the normalisation \refeq{LALTCoefficientNormalization}, we see that $\Xi^{\ssc[N]}$ was normalised using $g_z$, while $\Psi^{\ssc[N],\kappa}$ was normalised using $\gj$. As $\aabz/\aaz=g_z/\gj$ we conclude that the left inequality in \refeq{XidominatespsiImproved} holds.

Regarding $\Xi^{\ssc[N],\iota}$ and $\Pi^{\ssc[N],\iota,\kappa}$, defined in \refeq{defXiNIotaanimal} and \refeq{defPiNanimal}, note that $\Pi^{\ssc[N],\iota,\kappa}$ contains the extra factor $\1_{\iota}(\omega)=\1^{\Pi}_{\iota}(\omega,x,\kappa)$ for $N\geq 1$, as well as an additional factor $z$. Further, we do not normalize $\Pi^{\ssc[N],\iota,\kappa}$, while $\Xi^{\ssc[N],\iota}$ is normalized with a factor $1/g_z$, creating the stated factor $\aabz=zg_z$, which proves the bound for $N\geq 1$. To obtain the equation also for $N=0$, we review \refeq{defPiIotaanimal} and see that $\1_{\iota}(\omega)$ and $\1^{\Pi}_{\iota}(\omega,x,\kappa)$ only differ when $x\neq 0$, which is not possible for LTs, as double connections are not present. For LAs, we see that
	\begin{align}
	zg_z \Xi^{\ssc[0],\iota}(x)-\Pi^{\ssc[0],\iota,\kappa}(x)
	=&\sum_{A\ni 0,x} z^{|A|+1} \indic{0\dbct{A} x} \indic{\ve[\iota]\in A} \left( 1 -  \indic{x-\ve[\kappa]\nin A}\right)\nnb
	&-\indic{x-\ve[\kappa]=\ve[\iota]}\sum_{A\ni 0,x} z^{|A|+1} \indic{0\dbct{A} x} \indic{x\neq 0}\indic{x-\ve[\kappa]\nin A},
	\end{align}
and conclude the desired relation for $x\neq \ve[\iota]+\ve[\kappa]$, with $\iota\neq-\kappa$. For the remaining case $x=\ve[\iota]+\ve[\kappa]$ we compute
	\begin{align}
	zg_z \Xi^{\ssc[0],\iota}(\ve[\iota]+\ve[\kappa])-\Pi^{\ssc[0],\iota,\kappa}(\ve[\iota]+\ve[\kappa])
	=&\sum_{A\ni 0,\ve[\iota]+\ve[\kappa]} z^{|A|+1} \indic{0\dbct{A} \ve[\iota]+\ve[\kappa]} \left(\indic{\ve[\iota]\in A} - \indic{\ve[\iota]\nin A}\right).
	\end{align}
To verify that this is positive for all $d\geq 30$, we use the following very helpful rearrangement
	\begin{align}
	&\indic{\ve[\iota]\in A} - \indic{\ve[\iota]\nin A}=1 - \indic{\ve[\kappa]\nin A}- \indic{\ve[\iota]\nin A},\\
	&=1 -\indic{\ve[\iota]\nin A\ni \ve[\kappa]}-\indic{\ve[\kappa]\nin A\ni \ve[\iota]}-2\indic{\ve[\iota],\ve[\kappa]\in A}
	=\indic{\ve[\iota],\ve[\kappa]\in A}-\indic{\ve[\iota],\ve[\kappa]\nin A},\nn
	\end{align}
so that
	\begin{align}
	\lbeq{tech-cond-prob-iota-kappa}
	zg_z \Xi^{\ssc[0],\iota}(\ve[\iota]+\ve[\kappa])-\Pi^{\ssc[0],\iota,\kappa}(\ve[\iota]+\ve[\kappa])
	=&\sum_{A\ni 0,\ve[\iota]+\ve[\kappa]} z^{|A|+1} \indic{0\dbct{A} \ve[\iota]+\ve[\kappa]} \left(\indic{\ve[\iota],\ve[\kappa]\in A} - \indic{\ve[\iota], \ve[\kappa]\nin A}\right).
	\end{align}

While it is to be expected that this quantity is indeed positive, we could not find a direct injection between the different classes of LAs with the same number of bonds. Thus, we have to resort to a numerical verification of the inequality, using numerical upper and lower bounds.
We add this extra numerical verification in the LA mathematica notebook (see \cite[Improvement of Bounds, Technical condition]{FitNoblePage}), alongside three other numerical conditions that need to be verified for the analysis of \cite{FitHof13b}, alike $\hat G_z(k)\geq 0$ for all $k$.

 Alternatively, this condition could be dropped, at the expense of having to prove separate bounds for $\Pi^{\ssc[0],\iota,\kappa}(x)$, rather than relying on this simple relation to bound $\Pi^{\ssc[0],\iota,\kappa}(x)$ in terms of $ \Xi^{\ssc[0],\iota}(x)$. We refrain from following this route.
\qed

\paragraph{\cite[Assumption 4.4]{FitHof13b} Growth at the critical point.}{ \it
The functions $z\mapsto \hat \Xi_z(k),z\mapsto \hat  \Xi^{\iota}_z(k),z\mapsto \hat \Psi^{\kappa}_z(k),z\mapsto \hat  \Pi^{\iota,\kappa}_z(k)$ are continuous for $z\in(0,z_c)$. Further, let $\Gamma_1,\Gamma_2,\Gamma_3\geq 0$ be such that $f_i(z)\leq \Gamma_i$ and assume that \cite[Assumption 4.3]{FitHof13b} holds. Then, the functions stated above are left-continuous in $z_c$ with a finite limit when $z\nearrow z_c$ for all $x\in\Zd$. Further, for technical reasons, we assume that $z_c<1/2$.}
\medskip

The two-point function $G_z$ is defined as a power series that is clearly continuous in $z$ within its radius of convergence $z_c$. Since $G^\iota_z\leq G_z$, it has a radius of convergence that is at least $z_c$. The coefficients $\Xi^{\ssc[N]}_z,\Xi^{{\ssc[N]},\iota}_z,\Psi^{{\ssc[N]},\iota}_z$ and $\Pi^{{\ssc[N]},\iota,\kappa}_z$ are also power series and can be bounded in terms of many two-point functions. Since $\sum_x G_z(x)<\infty$ for all $z<z_c$, also the radii of convergence of these coefficients are at least $z_c$. [We believe that $G^\iota_z$ and the NoBLE coefficients have the same radii of convergence as $G_z$, but that is irrelevant here.]
Assuming that the bootstrap functions are uniformly bounded, all stated functions are uniformly bounded as well, which implies that they are left-continuous at $z_c$ by Abel's theorem.

To prove that $z_c<1/2$ consider a simple random walk that only takes steps in the positive direction. After $n$ steps there are $d^n$ possible trajectories, each being also a LT, so that $t_n(0)\geq d^n$.
As $z_c$ is defined as the radius of convergence of $\chi$ (see \refeq{defLTLAsusceptibility}), and
\begin{align*}
\chi(1/d)&=\hat {\bar G}_{1/d}(0)=\sum_{n=0}^\infty \sum_{x} t_n(x)d^{-n}\geq \sum_{n=0}^\infty t_n(0)d^{-n}= \sum_{n=0}^\infty d^n d^{-n}=\infty,
\end{align*}
we know that $z_c\leq 1/d<1/2$.


\def\picPiFourTreePictureRibs[#1]{
\begin{tikzpicture}[auto,scale=#1]

\draw[line width=2.5 pt,<->] (0,0) to (2,0);
\draw[line width=2.5 pt,->] (2,0) to (4,0);
\draw[line width=2.5 pt,->] (4,0) to (4,2);
\draw[line width=2.5 pt,->] (4,2) to (6,2);
\draw[line width=2.5 pt,->] (6,2) to (8,2);

\draw[line width=2 pt,blue!40,->] (4,0) to (5,-0.5);
\draw[line width=2 pt,red!70,<-] (5,-0.5) to (6,0);
\draw[line width=2 pt,red!70,->] (6,0) to (8,0);

\draw[line width=2 pt,red!70,<-] (0,2) to (2,2);
\draw[line width=2 pt,red!70,->] (2,2) to (3,2.5);
\draw[line width=2 pt,green!70,<-] (3,2.5) to (4,2);

\draw[line width=2 pt,green!60,->] (0,0) to (0,2);
\draw[line width=2 pt,red!70,->] (2,0) to (2,2);

\draw[line width=2 pt,red!70,<-] (6,0) to (6,2);
\draw[line width=2 pt,blue!40,<-] (8,0) to (8,2);

  \node[left]   at(0,0)      {$0$};
  \node[left]   at(0,2)      {$w_1$};
  \node[above]  at(3,2.5)    {$w_3$};
  \node[below]  at(5,-0.5)   {$w_2$};
  \node[right]  at(8,0)      {$w_4$};

  \node[above]  at(2,2)   {$u_1$};
  \node[below]  at(6,0)   {$u_4$};

  \node[below]  at(2,0)   {$x_1$};
  \node[below]  at(4,0)   {$x_2$};
  \node[above]  at(4,2)   {$x_3$};
  \node[above]  at(6,2)   {$x_4$};
  \node[right]  at(8,2)   {$x$};

  \tikzstyle{every node}=[font=\small]
  \node[right]  at(0,1)   {$2$};
  \node[below]  at(1,2)   {$4$};
  \node[left]  at(2,1)    {$3$};
  \node[above]  at(1,0)   {$1$};

  \node[below]  at(2.5,2.25)   {$7$};
  \node[below]  at(3.5,2.25)   {$8$};
  \node[left]  at(4,1)   {$6$};
  \node[above]  at(3,0)   {$5$};

  \node[above]  at(5,2)   {$9$};
  \node[left]  at(6,1)   {$10$};
  \node[above]  at(4.5,-0.25)   {$11$};
  \node[above]  at(5.5,-0.25)   {$12$};

  \node[below]  at(7,2)   {$13$};
  \node[left]  at(8,1)    {$15$};
  \node[above]  at(7,0)   {$14$};
  \tikzstyle{every node}=[font=\normalfont]
\end{tikzpicture}
}

\def\picPiFourTreePicture[#1]{
\begin{tikzpicture}[auto,scale=#1]
\node at(-0.5-.15,0.5-.15)   {$0$};
\fill[gray] (-0.5,0.5) -- (0,0)--(0,1);
\draw[line width=2.5 pt] (-0.5,0.5) to (0,0);
\draw[line width=1 pt] (-0.5,0.5) to (0,1);
\draw[line width=1 pt] (0,0) to (4,0);
\draw[line width=1 pt] (0,1) to (4,1);
\draw[line width=1 pt] (0,0) to (0,1);
\draw[line width=1 pt] (4,0) to (4,1);
\draw[line width=1 pt] (1,0) to (1,1);
\draw[line width=1 pt] (2,0) to (2,1);
\draw[line width=1 pt] (3,0) to (3,1);
\draw[line width=2.5 pt] (0,0) to (4,0);

\fill (0.5,1) circle (3pt);
\fill (1.5,1) circle (3pt);
\fill (2.5,1) circle (3pt);
\fill (4.0,1) circle (3pt);
\node at(4.15,-.15)   {$x$};

\begin{scope}[shift={(6,0)},rotate=0]
\node at(-0.5-.15,0.5-.15)   {$0$};
\fill[gray] (-0.5,0.5) -- (0,0)--(0,1);
\draw[line width=2.5 pt] (-0.5,0.5) to (0,0);
\draw[line width=1 pt] (-0.5,0.5) to (0,1);
\draw[line width=1 pt] (0,0) to (4,0);
\draw[line width=1 pt] (0,1) to (4,1);
\draw[line width=1 pt] (0,0) to (0,1);
\draw[line width=1 pt] (4,0) to (4,1);
\draw[line width=1 pt] (1,0) to (1,1);
\draw[line width=1 pt] (2,0) to (2,1);
\draw[line width=2.5 pt] (3,0) to (3,1);
\draw[line width=2.5 pt] (0,0) to (3,0);
\draw[line width=2.5 pt] (3,1) to (4,1);
\fill (0.5,1) circle (3pt);
\fill (1.5,1) circle (3pt);
\fill (2.5,1) circle (3pt);
\fill (4.0,0) circle (3pt);

\node at(4.15,1.15)   {$x$};
\end{scope}

\begin{scope}[shift={(12,0)},rotate=0]
\node at(-0.5-.15,0.5-.15)   {$0$};
\fill[gray] (-0.5,0.5) -- (0,0)--(0,1);
\draw[line width=2.5 pt] (-0.5,0.5) to (0,0);
\draw[line width=1 pt] (-0.5,0.5) to (0,1);
\draw[line width=1 pt] (0,0) to (4,0);
\draw[line width=1 pt] (0,1) to (4,1);
\draw[line width=1 pt] (0,0) to (0,1);
\draw[line width=1 pt] (4,0) to (4,1);
\draw[line width=1 pt] (1,0) to (1,1);
\draw[line width=2.5 pt] (2,0) to (2,1);
\draw[line width=2.5 pt] (3,0) to (3,1);
\draw[line width=2.5 pt] (0,0) to (2,0);
\draw[line width=2.5 pt] (2,1) to (3,1);
\draw[line width=2.5 pt] (3,0) to (4,0);
\fill (0.5,1) circle (3pt);
\fill (1.5,1) circle (3pt);
\fill (2.5,0) circle (3pt);
\fill (4.0,1) circle (3pt);

\node at(4.15,-.15)   {$x$};
\end{scope}

\begin{scope}[shift={(0,-2)},rotate=0]
\node at(-0.5-.15,0.5-.15)   {$0$};
\fill[gray] (-0.5,0.5) -- (0,0)--(0,1);
\draw[line width=2.5 pt] (-0.5,0.5) to (0,0);
\draw[line width=1 pt] (-0.5,0.5) to (0,1);
\draw[line width=1 pt] (0,0) to (4,0);
\draw[line width=1 pt] (0,1) to (4,1);
\draw[line width=1 pt] (0,0) to (0,1);
\draw[line width=1 pt] (4,0) to (4,1);
\draw[line width=1 pt] (1,0) to (1,1);
\draw[line width=2.5 pt] (2,0) to (2,1);
\draw[line width=1 pt] (3,0) to (3,1);
\draw[line width=2.5 pt] (0,0) to (2,0);
\draw[line width=2.5 pt] (2,1) to (4,1);
\fill (0.5,1) circle (3pt);
\fill (1.5,1) circle (3pt);
\fill (2.5,0) circle (3pt);
\fill (4.0,0) circle (3pt);
\node at(4.15,1.15)   {$x$};
\end{scope}

\begin{scope}[shift={(6,-2)},rotate=0]
\node at(-0.5-.15,0.5-.15)   {$0$};
\fill[gray] (-0.5,0.5) -- (0,0)--(0,1);
\draw[line width=2.5 pt] (-0.5,0.5) to (0,0);
\draw[line width=1 pt] (-0.5,0.5) to (0,1);
\draw[line width=1 pt] (0,0) to (4,0);
\draw[line width=1 pt] (0,1) to (4,1);
\draw[line width=1 pt] (0,0) to (0,1);
\draw[line width=1 pt] (4,0) to (4,1);
\draw[line width=2.5 pt] (1,0) to (1,1);
\draw[line width=2.5 pt] (2,0) to (2,1);
\draw[line width=1 pt] (3,0) to (3,1);
\draw[line width=2.5 pt] (0,0) to (1,0);
\draw[line width=2.5 pt] (1,1) to (2,1);
\draw[line width=2.5 pt] (2,0) to (4,0);
\fill (0.5,1) circle (3pt);
\fill (1.5,0) circle (3pt);
\fill (2.5,1) circle (3pt);
\fill (4.0,1) circle (3pt);
\node at(4.15,-.15)   {$x$};
\end{scope}

\begin{scope}[shift={(12,-2)},rotate=0]
\node at(-0.5-.15,0.5-.15)   {$0$};
\fill[gray] (-0.5,0.5) -- (0,0)--(0,1);
\draw[line width=2.5 pt] (-0.5,0.5) to (0,0);
\draw[line width=1 pt] (-0.5,0.5) to (0,1);
\draw[line width=1 pt] (0,0) to (4,0);
\draw[line width=1 pt] (0,1) to (4,1);
\draw[line width=1 pt] (0,0) to (0,1);
\draw[line width=1 pt] (4,0) to (4,1);
\draw[line width=2.5 pt] (1,0) to (1,1);
\draw[line width=2.5 pt] (2,0) to (2,1);
\draw[line width=2.5 pt] (3,0) to (3,1);
\draw[line width=2.5 pt] (0,0) to (1,0);
\draw[line width=2.5 pt] (1,1) to (2,1);
\draw[line width=2.5 pt] (2,0) to (3,0);
\draw[line width=2.5 pt] (3,1) to (4,1);
\fill (0.5,1) circle (3pt);
\fill (1.5,0) circle (3pt);
\fill (2.5,1) circle (3pt);
\fill (4.0,0) circle (3pt);
\node at(4.15,1.15)   {$x$};
\end{scope}
\begin{scope}[shift={(0,-4)},rotate=0]
\node at(-0.5-.15,0.5-.15)   {$0$};
\fill[gray] (-0.5,0.5) -- (0,0)--(0,1);
\draw[line width=2.5 pt] (-0.5,0.5) to (0,0);
\draw[line width=1 pt] (-0.5,0.5) to (0,1);
\draw[line width=1 pt] (0,0) to (4,0);
\draw[line width=1 pt] (0,1) to (4,1);
\draw[line width=1 pt] (0,0) to (0,1);
\draw[line width=1 pt] (4,0) to (4,1);
\draw[line width=2.5 pt] (1,0) to (1,1);
\draw[line width=1 pt] (2,0) to (2,1);
\draw[line width=2.5 pt] (3,0) to (3,1);
\draw[line width=2.5 pt] (0,0) to (1,0);
\draw[line width=2.5 pt] (1,1) to (3,1);
\draw[line width=2.5 pt] (3,0) to (4,0);
\fill (0.5,1) circle (3pt);
\fill (1.5,0) circle (3pt);
\fill (2.5,0) circle (3pt);
\fill (4.0,1) circle (3pt);
\node at(4.15,-.15)   {$x$};
\end{scope}
\begin{scope}[shift={(6,-4)},rotate=0]
\node at(-0.5-.15,0.5-.15)   {$0$};
\fill[gray] (-0.5,0.5) -- (0,0)--(0,1);
\draw[line width=2.5 pt] (-0.5,0.5) to (0,0);
\draw[line width=1 pt] (-0.5,0.5) to (0,1);
\draw[line width=1 pt] (0,0) to (4,0);
\draw[line width=1 pt] (0,1) to (4,1);
\draw[line width=1 pt] (0,0) to (0,1);
\draw[line width=1 pt] (4,0) to (4,1);
\draw[line width=2.5 pt] (1,0) to (1,1);
\draw[line width=1 pt] (2,0) to (2,1);
\draw[line width=1 pt] (3,0) to (3,1);
\draw[line width=2.5 pt] (0,0) to (1,0);
\draw[line width=2.5 pt] (1,1) to (4,1);
\fill (0.5,1) circle (3pt);
\fill (1.5,0) circle (3pt);
\fill (2.5,0) circle (3pt);
\fill (4.0,0) circle (3pt);
\node at(4.15,1.15)   {$x$};
\end{scope}
\end{tikzpicture}
}

\def\picPiSharedLabelings[#1]{
\begin{tikzpicture}[auto,scale=#1]
{\small

\begin{scope}[shift={(0,0)},rotate=0]
\node at(-1-.15,-1-.15)   {$0$};
\fill[gray] (-1,1) -- (0,0)--(0,2);
\draw[line width=0.7 pt] (-1,1) to (0,2);
\draw[line width=2 pt] (-1,1) to (0,0);
\node at(0,2.2)   {$m_0=-2$};
\fill (1,2) circle (4pt);
\end{scope}
\begin{scope}[shift={(0,0)},rotate=0]
\draw[line width=0.7 pt] (0,0) to (0,2);
\draw[line width=2 pt] (0,0) to (2,0);
\draw[line width=0.7 pt] (0,2) to (2,2);
\draw[line width=2 pt] (2,0) to (2,2);
\node at(2,2.2)   {$m_1=2$};
\end{scope}
\begin{scope}[shift={(2,0)},rotate=0]
\draw[line width=0.7 pt] (0,0) to (1,0);
\draw[line width=0.7 pt] (1,0) to (2,0.4);
\draw[line width=2 pt] (0,2) to (2,1.6);
\draw[line width=0.7 pt] (2,0.4) to (2,1.6);
\fill (1,0) circle (4pt);
\node at(2,2.2)    {$m_2=1$};
\end{scope}
\begin{scope}[shift={(4,0)},rotate=0]
\draw[line width=0.7 pt] (0,0.4) to (1,0);
\draw[line width=0.7 pt] (1,0) to (2,0);
\draw[line width=2 pt] (0,1.6) to (2,2);
\draw[line width=2 pt] (2,0) to (2,2);
\fill (1,0) circle (4pt);
\node at(2,2.2)     {$m_3=-2$};
\end{scope}
\begin{scope}[shift={(6,0)},rotate=0]
\draw[line width=2 pt] (0,0) to (2,0);
\draw[line width=0.7 pt] (0,2) to (1,2);
\draw[line width=0.7 pt] (2,0) to (1,2);
\node at(2,2.2)      {$m_4=0$};
\fill (1,2) circle (4pt);
\end{scope}
\begin{scope}[shift={(8,0)},rotate=0]
\draw[line width=2 pt] (0,0) to (2,0);
\draw[line width=0.7 pt] (0,0) to (1,2);
\draw[line width=0.7 pt] (1,2) to (2,2);
\draw[line width=0.7 pt] (2,0) to (2,2);
\node at(2,2.2)     {$m_5=-2$};
\fill (1,2) circle (4pt);
\end{scope}
\begin{scope}[shift={(10,0)},rotate=0]
\draw[line width=2 pt] (0,0) to (2,0);
\draw[line width=0.7 pt] (0,2) to (2,2);
\draw[line width=0.7 pt] (2,0) to (2,2);
\fill (2,2) circle (4pt);
\node at(2.15,-.15)   {$x$};
\end{scope}
}
\end{tikzpicture}
}

\def\picPiFourTreeDecomposePicture[#1]{
\begin{tikzpicture}[auto,scale=#1]
\begin{scope}[shift={(-1.3,0)},rotate=0]
\draw[line width=1 pt] (0,0-0.21) to (0,1+0.21);
\draw[line width=1 pt] (0,0-0.2) to (0.3,0-0.2);
\draw[line width=1 pt] (0,1+0.2) to (0.3,1+0.2);
\end{scope}

\begin{scope}[shift={(-0.5,0)},rotate=0]
\fill[gray] (-0.5,0.5) -- (0,0)--(0,1);
\draw[line width=1 pt] (-0.5,0.5) to (0,1);
\draw[line width=2 pt] (-0.5,0.5) to (0,0);
\draw[line width=1 pt] (0,0) to (0,1);
\draw[line width=2 pt] (0,0) to (1,0);
\draw[line width=1 pt] (1,0) to (1,1);
\draw[line width=1 pt] (0,1) to (1,1);
\end{scope}

\begin{scope}[shift={(1,0.5)},rotate=0]
\draw[line width=1 pt] (0.2,0) to (-0.2,0);
\draw[line width=1 pt] (0,0.2) to (0,-0.2);
\end{scope}

\begin{scope}[shift={(2,0)},rotate=0]
\fill[gray] (-0.5,0.5) -- (0,0)--(0,1);
\draw[line width=1 pt] (-0.5,0.5) to (0,1);
\draw[line width=2 pt] (-0.5,0.5) to (0,0);
\draw[line width=1 pt] (0,0) to (0,1);
\draw[line width=2 pt] (0,0) to (1,0);
\draw[line width=2 pt] (1,0) to (1,1);
\draw[line width=1 pt] (0,1) to (1,1);
\end{scope}
\begin{scope}[shift={(3.5,0)},rotate=0]
\draw[line width=1 pt] (0,0-0.2) to (0,1+0.2);
\draw[line width=1 pt] (-.3,0-0.2) to (0,0-0.2);
\draw[line width=1 pt] (-.3,1+0.2) to (0,1+0.2);
\end{scope}
\begin{scope}[shift={(4,0)},rotate=0]
\draw[line width=1 pt] (0,0-0.2) to (0,1+0.2);
\draw[line width=1 pt] (0,0-0.2) to (0.3,0-0.2);
\draw[line width=1 pt] (0,1+0.2) to (0.3,1+0.2);
\end{scope}

\begin{scope}[shift={(4.5,0)},rotate=0]
\draw[line width=2 pt] (0,0) to (1,0);
\draw[line width=1 pt] (1,0) to (1,1);
\draw[line width=1 pt] (0,1) to (0.5,1.25);
\draw[line width=1 pt] (0.5,1.25) to (1,1);
\end{scope}

\begin{scope}[shift={(6,0.5)},rotate=0]
\draw[line width=1 pt] (0.2,0) to (-0.2,0);
\draw[line width=1 pt] (0,0.2) to (0,-0.2);
\end{scope}

\begin{scope}[shift={(6.5,0)},rotate=0]
\draw[line width=2 pt] (0,0) to (1,0);
\draw[line width=2 pt] (1,0) to (1,1);
\draw[line width=1 pt] (0,1) to (0.5,1.25);
\draw[line width=1 pt] (0.5,1.25) to (1,1);
\end{scope}
\begin{scope}[shift={(8,0)},rotate=0]
\draw[line width=1 pt] (0,0-0.2) to (0,1+0.2);
\draw[line width=1 pt] (-.3,0-0.2) to (0,0-0.2);
\draw[line width=1 pt] (-.3,1+0.2) to (0,1+0.2);
\node[above]   at(0,1.4)      {$N-2$};
\end{scope}
\begin{scope}[shift={(8.8,0)},rotate=0]
\draw[line width=2 pt] (0,0) to (1,0);
\draw[line width=1 pt] (1,0) to (1,1);
\draw[line width=1 pt] (0,1) to (1,1);
\end{scope}
%
\end{tikzpicture}
}

\def\picPBoundGeneral[#1]{
\begin{tikzpicture}[line width=1pt,auto,scale=#1]
 \draw [-](0,0) to  (0,2);
 \draw [line width=2pt,-](0,0) to  (2,0);
 \draw [-](0,2) to    (2,2);
 \draw [line width=2pt](2,0) to    (2,2);
  \fill (0,2) circle (3pt);
  \fill (2,0) circle (3pt);
  \fill (2,2) circle (3pt);
  \fill (0,0) circle (3pt);
  \node[left]   at(0,0)      {$0$};
  \node[left]   at(0,2)      {$v$};
  \node[right]  at (2,0)     {$y$};
  \node[right]  at (2,2)     {$x$};
  \node[right]  at (2,1)     {$s$};
\end{tikzpicture}}

\def\picPBoundGeneralRib[#1]{
\begin{tikzpicture}[line width=1pt,auto,scale=#1]
 \draw [-](0,0) to  (0,2);
 \draw [line width=2pt,-](0,0) to  (2,0);
 \draw [-](0,2) to    (2,2);
 \draw [-](2,0) to    (2,2);
  \fill (0,2) circle (3pt);
  \fill (2,0) circle (3pt);
  \fill (2,2) circle (3pt);
  \fill (0,0) circle (3pt);
  \node[left]   at(0,0)      {$0$};
  \node[left]   at(0,2)      {$v$};
  \node[right]  at (2,0)     {$x$};
  \node[right]  at (2,2)     {$y$};
  \node[right]  at (2,1)     {$d_R$};
\end{tikzpicture}}

\def\picPBoundGeneralStep[#1]{
\begin{tikzpicture}[line width=1pt,auto,scale=#1]
 \draw [-](0,0) to  (0,2);
 \draw [line width=2pt,-](0,0) to  (2,0);
 \draw [-](0,2) to    (2,2);
 \draw [line width=2pt](2,0) to    (2,2);
  \fill (0,2) circle (3pt);
  \fill (2,0) circle (3pt);
  \fill (2,2) circle (3pt);
  \fill (0,0) circle (3pt);
  \fill (0.7,0) circle (4pt);

  \node[left]   at(0,0)      {$0$};
  \node[below]   at(0.7,0)   {$\ve[\iota]$};
  \node[left]   at(0,2)      {$v$};
  \node[right]  at (2,0)     {$y$};
  \node[right]  at (2,2)     {$x$};
  \node[right]  at (2,1)     {$s$};
\end{tikzpicture}}

\def\picPBoundGeneralRibStep[#1]{
\begin{tikzpicture}[line width=1pt,auto,scale=#1]
 \draw [-](0,0) to  (0,2);
 \draw [line width=2pt,-](0,0) to  (2,0);
 \draw [-](0,2) to    (2,2);
 \draw [-](2,0) to    (2,2);
  \fill (0,2) circle (3pt);
  \fill (0.7,0) circle (4pt);
  \fill (2,0) circle (3pt);
  \fill (2,2) circle (3pt);
  \fill (0,0) circle (3pt);
  \node[left]   at(0,0)      {$0$};
  \node[below]   at(0.7,0)   {$\ve[\iota]$};
  \node[left]   at(0,2)      {$v$};
  \node[right]  at (2,0)     {$x$};
  \node[right]  at (2,2)     {$y$};
  \node[right]  at (2,1)     {$d_R$};
\end{tikzpicture}}

\def\picPBoundGeneralIota[#1]{
\begin{tikzpicture}[line width=1pt,auto,scale=#1]
 \draw [-](0,0) to  (0.25,1.2);
 \draw [-](0.25,1.2) to  (1.4,2);
 \draw [-](1.4,2) to  (2.5,2);
 \draw [line width=2pt,-](0,0) to    (2.5,0);
 \draw [line width=2pt,-](2.5,0) to    (2.5,2);
 \draw [-](0.25,1.2) to  [out=135,in=90]   (-0.8,0);
  \fill (0,0) circle (3pt);
  \fill (0.25,1.2) circle (3pt);
  \fill (1.4,2) circle (3pt);
  \fill (2.5,2) circle (3pt);
  \fill (2.5,0) circle (3pt);
  \fill (-0.8,0) circle (3pt);
  \node[below]   at(-0.8,0)      {$\ve[\iota]$};
  \node[below]   at(0,0)      {$0$};
  \node[above]   at(0.25,1.2)      {$u$};
  \node[left]   at(1.5,1.5)      {$v$};
  \node[right]  at (2.5,0)     {$y$};
  \node[right]  at (2.5,2)     {$x$};
  \node[right]  at (2.5,1)     {$s$};
\end{tikzpicture}}

\def\picPBoundGeneralIotaRib[#1]{
\begin{tikzpicture}[line width=1pt,auto,scale=#1]
 \draw [-](0,0) to  (0.25,1.2);
 \draw [-](0.25,1.2) to  (1.4,2);
 \draw [-](1.4,2) to  (2.5,2);
 \draw [line width=2pt,-](0,0) to    (2.5,0);
 \draw [-](2.5,0) to    (2.5,2);
 \draw [-](0.25,1.2) to  [out=135,in=90]   (-0.8,0);
  \fill (0,0) circle (3pt);
  \fill (0.25,1.2) circle (3pt);
  \fill (1.4,2) circle (3pt);
  \fill (2.5,2) circle (3pt);
  \fill (2.5,0) circle (3pt);
  \fill (-0.8,0) circle (3pt);
  \node[below]   at(-0.8,0)      {$\ve[\iota]$};
  \node[below]   at(0,0)      {$0$};
  \node[above]   at(0.25,1.2)      {$u$};
  \node[left]   at(1.5,1.5)      {$v$};
  \node[right]  at (2.5,0)     {$x$};
  \node[right]  at (2.5,2)     {$y$};
  \node[right]  at (2.5,1)     {$d_R$};
\end{tikzpicture}}

\def\picPBoundxx[#1]{
\begin{tikzpicture}[line width=1pt,auto,scale=#1]
  \draw (0,0)-- (2,0)-- (1,1.5)-- cycle;
  \fill (0,0) circle (3pt);
  \fill (2,0) circle (3pt);
  \fill (1,1.5) circle (3pt);
  \node[left]   at(0,0)       {$0$};
  \node[right]   at (2,0)       {$x$};
  \node[left]   at (1,1.5)       {$v$};
\end{tikzpicture}}

\def\picPBoundxxZero[#1]{
\begin{tikzpicture}[line width=1pt,auto,scale=#1]
  \draw [-](0,0) to [out=90,in=90]    (2,0);
  \draw [-](0,0) to [out=270,in=270]    (2,0);
  \fill (0,0) circle (3pt);
  \node[left]   at(0,0)      {$0$};
  \node[right]  at (2,0)   {$\geq 4$};
\end{tikzpicture}}

\def\picPBoundxxOne[#1]{
\begin{tikzpicture}[line width=1pt,auto,scale=#1]
  \draw [-](0,1.5)-- (2,0);
  \draw [line width=1pt,-] (2,0) --(0,0);
  \draw [line width=1pt,-] (0,0) --(0,1.5);
  \fill (0,0) circle (3pt);
  \fill (0,1.5) circle (3pt);
  \fill (2,0) circle (3pt);
  \node[left]   at(0,0)      {$0$};
  \node[right]  at (2,0)     {$x$};
  \node[left]   at (0,1.5)   {$v$};
  \node[left]   at (0,.725)  {$=1$};
  \node[below]  at (1,0)     {$=1$};
  \node[above]  at (1.5,0.7)     {$\geq 2$};

  \draw [-](4,0)-- (6,0);
  \draw [-](5,1.5)-- (4,0);
  \draw [-](5,1.5)-- (6,0);
  \fill (4,0) circle (3pt);
  \fill (5,1.5) circle (3pt);
  \fill (6,0) circle (3pt);
  \node[left]   at(4,0)       {$0$};
  \node[right]   at (6,0)       {$x$};
  \node[left]   at (5,1.5)      {$v$};
  \node[below]  at (5,0)        {$=1$};
  \node[left]   at (4.6,.9)     {$\geq 2$};
  \node[right]  at (5.4,0.9)    {$\geq 1$};
\end{tikzpicture}}

\def\picPBoundxxTwo[#1]{
\begin{tikzpicture}[line width=1pt,auto,scale=#1]
  \draw [-](0,0) to [out=90,in=90]    (3,0);
  \draw [-](0,0) to [out=270,in=270]    (3,0);
  \fill (0,0) circle (3pt);
  \fill (3,0) circle (3pt);
  \node[left]   at(0,0)      {$0$};
  \node[right]  at (3,0)     {$x$};
  \node[above]  at (1,0.7)     {$\geq 3$};
  \node[below]  at (2,-0.7)     {$=1$};
  \begin{scope}[shift={(1,0)}]
  \draw [-](5,0) to [out=90,in=90]    (8,0);
  \draw [-](5,0) to [out=270,in=270]    (8,0);
  \fill (5,0) circle (3pt);
  \fill (8,0) circle (3pt);
  \node[left]   at(5,0)      {$0$};
  \node[right]  at (8,0)     {$x$};
  \node[above]  at (6,0.7)     {$\geq 2$};
  \node[below]  at (7,-0.7)     {$\geq 2$};
 \end{scope}
\end{tikzpicture}}

\def\picPBoundxxThree[#1]{
\begin{tikzpicture}[line width=1pt,auto,scale=#1]
  \draw [-] (0,0)--(1,1.5);
  \draw [-] (1,1.5)-- (2,0);
  \draw [-] (2,0)--(0,0);
  \fill (0,0) circle (3pt);
  \fill (1,1.5) circle (3pt);
  \fill (2,0) circle (3pt);
  \node[left]   at(0,0)       {$0$};
  \node[right]   at (2,0)       {$x$};
  \node[left]   at (1,1.5)      {$v$};
  \node[below]  at (1,0)        {$\geq 2$};
  \node[left]   at (0.6,.9)     {$\geq 1$};
  \node[right]  at (1.4,0.9)    {$\geq 1$};
\end{tikzpicture}}

\def\picPBoundxyOne[#1]{
\begin{tikzpicture}[line width=1pt,auto,scale=#1]
  \draw [-](0,0) to [out=90,in=90]    (3,0);
  \draw [-](0,0) to [out=270,in=270]    (3,0);
  \fill (0,0) circle (3pt);
  \fill (3,0) circle (3pt);
  \node[left]   at(0,0)      {$0$};
  \node[right]  at (3,0)     {$y$};
  \node[above]  at (1,0.7)     {$\geq 3$};
  \node[below]  at (2,-0.7)     {$= 1$};
\end{tikzpicture}}

\def\picPBoundxyTwo[#1]{
\begin{tikzpicture}[line width=1pt,auto,scale=#1]
  \draw [-](0,0) to [out=90,in=180]    (2,2);
  \draw [-](2,0) to    (2,2);
  \draw [-](0,0) to    (2,0);
  \fill (0,0) circle (3pt);
  \fill (2,0) circle (3pt);
  \fill (2,2) circle (3pt);
  \node[left]   at(0,0)      {$0$};
  \node[right]  at (2,0)     {$y$};
  \node[right]  at (2,2)     {$x$};
  \node[below]  at (0.6,0.3)  {$\geq 2$};
  \node[above]  at (1,1.8)   {$= 1$};
  \node[right]  at (2,1)     {$=1$};

  \draw [-](4,0) to [out=90,in=180]    (6,2);
  \draw [-](6,0) to    (6,2);
  \draw [-](4,0) to    (6,0);
  \fill (4,0) circle (3pt);
  \fill (6,0) circle (3pt);
  \fill (6,2) circle (3pt);
  \node[left]   at(4,0)      {$0$};
  \node[right]  at (6,0)     {$y$};
  \node[right]  at (6,2)     {$x$};
  \node[above]  at (4.4,1.5)  {$\geq 2$};
  \node[below]  at (5,-0.2)  {$\geq 1$};
  \node[right]  at (6,1)     {$=1$};
\end{tikzpicture}}

\def\picPBoundxyThree[#1]{
\begin{tikzpicture}[line width=1pt,auto,scale=#1]
 \draw [-](0,0) to  (0,2);
 \draw [-](0,0) to  (2,0);
 \draw [-](0,2) to    (2,2);
 \draw [-](2,0) to    (2,2);
  \fill (0,2) circle (3pt);
  \fill (2,0) circle (3pt);
  \fill (2,2) circle (3pt);
  \fill (0,0) circle (3pt);
  \node[left]   at(0,0)      {$0$};
  \node[left]   at(0,2)      {$v$};
  \node[right]  at (2,0)     {$y$};
  \node[right]  at (2,2)     {$x$};
  \node[left]  at (0,1)      {$\geq 1$};
  \node[below]  at (1,0)     {$\geq 1$};
  \node[above]  at (1,1.8)   {$\geq 1$};
  \node[right]  at (2,1)     {$=1$};
\end{tikzpicture}}

\def\picPBoundxyFour[#1]{
\begin{tikzpicture}[line width=1pt,auto,scale=#1]
  \draw [-](0,1.5)-- (2,0);
  \draw [line width=1pt,-] (2,0) --(0,0);
  \draw [line width=1pt,-] (0,0) --(0,1.5);
  \fill (0,0) circle (3pt);
  \fill (0,1.5) circle (3pt);
  \fill (2,0) circle (3pt);
  \node[left]   at(0,0)      {$0$};
  \node[right]  at (2,0)     {$x$};
  \node[left]   at (0,1.5)   {$v$};
  \node[left]   at (0,.725)  {$=1$};
  \node[below]  at (1,0)     {$=1$};
  \node[above]  at (1.5,0.7)     {$\geq 2$};

  \draw [-](4,0)--(5,1.5);
  \draw [-](5,1.5)-- (6,0);
  \draw [line width=1pt,-] (4,0) --(6,0);
  \fill (4,0) circle (3pt);
  \fill (5,1.5) circle (3pt);
  \fill (6,0) circle (3pt);
  \node[left]   at(4,0)       {$0$};
  \node[right]   at (6,0)       {$x$};
  \node[left]   at (5,1.5)      {$v$};
  \node[below]  at (5,0)        {$=1$};
  \node[left]   at (4.6,.9)     {$\geq 2$};
  \node[right]  at (5.4,0.9)    {$\geq 1$};
\end{tikzpicture}}

\def\picPBoundyyOne[#1]{
\begin{tikzpicture}[line width=1pt,auto,scale=#1]
  \draw [-](0,0) to [out=90,in=90]    (3,0);
  \draw [-](0,0) to [out=270,in=270]  (3,0);
  \fill (0,0) circle (3pt);
  \fill (3,0) circle (3pt);
 \node[left]   at(0,0)      {$0$};
  \node[right]  at (3,0)     {$y$};
  \node[above]  at (1,0.7)     {$\geq 1$};
  \node[below]  at (2,-0.7)     {$\geq 2$};
\end{tikzpicture}}

\def\picPBoundyyThree[#1]{
\begin{tikzpicture}[line width=1pt,auto,scale=#1]
  \draw [-](4,0) to [out=90,in=180]    (6,2);
  \draw [-](6,0) to    (6,2);
  \draw [-](4,0) to    (6,0);
  \fill (4,0) circle (3pt);
  \fill (6,0) circle (3pt);
  \fill (6,2) circle (3pt);
  \node[left]   at(4,0)      {$0$};
  \node[right]  at (6,0)     {$y$};
  \node[right]  at (6,2)     {$x$};
  \node[above]  at (3.5,0.5) {$\geq 1$};
  \node[below]  at (5,-0.3)  {$\geq 1$};
  \node[right]  at (6,1)     {$\geq 2$};
\end{tikzpicture}}

\def\picPBoundyyFour[#1]{
\begin{tikzpicture}[line width=1pt,auto,scale=#1]
  \draw [-](4,0) to     (4,2);
  \draw [-](4,2) to     (6,2);
  \draw [-](6,0) to    (6,2);
  \draw [-](4,0) to    (6,0);
  \fill (4,0) circle (3pt);
  \fill (4,2) circle (3pt);
  \fill (6,0) circle (3pt);
  \fill (6,2) circle (3pt);
  \node[left]   at(4,0)      {$0$};
  \node[left]   at(4,2)      {$v$};
  \node[right]  at (6,0)     {$y$};
  \node[right]  at (6,2)     {$x$};
  \node[left]  at (4,1) {$\geq 1$};
  \node[above]  at (5,2) {$\geq 0$};
  \node[below]  at (5,-0.2)  {$\geq 1$};
  \node[right]  at (6,1)     {$\geq 2$};
\end{tikzpicture}}

\def\picPBoundRibxyOne[#1]{
\begin{tikzpicture}[line width=1pt,auto,scale=#1]
  \draw [-](4,0) to [out=90,in=180]    (6,2);
  \draw [-](6,0) to    (6,2);
  \draw [-](4,0) to    (6,0);
  \fill (4,0) circle (3pt);
  \fill (6,0) circle (3pt);
  \fill (6,2) circle (3pt);
  \node[left]   at(4,0)      {$0$};
  \node[right]  at (6,0)     {$x$};
  \node[right]  at (6,2)     {$y$};
  \node[above]  at (4.4,1.5) {$\geq 2$};
  \node[below]  at (5,-0.2)  {$= 1$};
  \node[right]  at (6,1)     {$=1$};
\end{tikzpicture}}

\def\picPBoundRibxyTwo[#1]{
\begin{tikzpicture}[line width=1pt,auto,scale=#1]
  \draw [-](4,0) to     (4,2);
  \draw [-](4,2) to     (6,2);
  \draw [-](6,0) to    (6,2);
  \draw [-](4,0) to    (6,0);
  \fill (4,0) circle (3pt);
  \fill (4,2) circle (3pt);
  \fill (6,0) circle (3pt);
  \fill (6,2) circle (3pt);
  \node[left]   at(4,0)      {$0$};
  \node[left]   at(4,2)      {$v$};
  \node[right]  at (6,0)     {$x$};
  \node[right]  at (6,2)     {$y$};
  \node[left]  at (4,1) {$\geq1$};
  \node[above]  at (5,2) {$\geq1$};
  \node[below]  at (5,-0.2)  {$=1$};
  \node[right]  at (6,1)     {$=1$};
\end{tikzpicture}}

\def\picPBoundRibxyThree[#1]{
\begin{tikzpicture}[line width=1pt,auto,scale=#1]
  \draw [-](4,0) to [out=90,in=180]    (6,2);
  \draw [-](6,0) to    (6,2);
  \draw [-](4,0) to    (6,0);
  \fill (4,0) circle (3pt);
  \fill (6,0) circle (3pt);
  \fill (6,2) circle (3pt);
  \node[left]   at(4,0)      {$0$};
  \node[right]  at (6,0)     {$x$};
  \node[right]  at (6,2)     {$y$};
  \node[above]  at (4.4,1.5) {$\geq 1$};
  \node[below]  at (5,-0.2)  {$\geq 2$};
  \node[right]  at (6,1)     {$=1$};
\end{tikzpicture}}

\def\picPBoundRibxyFour[#1]{
\begin{tikzpicture}[line width=1pt,auto,scale=#1]
  \draw [-](4,0) to     (4,2);
  \draw [-](4,2) to     (6,2);
  \draw [-](6,0) to    (6,2);
  \draw [-](4,0) to    (6,0);
  \fill (4,0) circle (3pt);
  \fill (4,2) circle (3pt);
  \fill (6,0) circle (3pt);
  \fill (6,2) circle (3pt);
  \node[left]   at(4,0)      {$0$};
  \node[left]   at(4,2)      {$v$};
  \node[right]  at (6,0)     {$x$};
  \node[right]  at (6,2)     {$y$};
  \node[left]  at (4,1) {$\geq 1$};
  \node[above]  at (5,2) {$\geq 0$};
  \node[below]  at (5,-0.2)  {$\geq 2$};
  \node[right]  at (6,1)     {$=1$};
\end{tikzpicture}}

\def\picPBoundRibyyOne[#1]{
\begin{tikzpicture}[line width=1pt,auto,scale=#1]
  \draw [-](4,0) to [out=90,in=180]    (6,2);
  \draw [-](6,0) to    (6,2);
  \draw [-](4,0) to    (6,0);
  \fill (4,0) circle (3pt);
  \fill (6,0) circle (3pt);
  \fill (6,2) circle (3pt);
  \node[left]   at(4,0)      {$0$};
  \node[right]  at (6,0)     {$x$};
  \node[right]  at (6,2)     {$y$};
  \node[above]  at (4.4,1.5) {$\geq 1$};
  \node[below]  at (5,-0.2)  {$\geq 1$};
  \node[right]  at (6,1)     {$\geq 2$};
\end{tikzpicture}}

\def\picPBoundRibyyTwo[#1]{
\begin{tikzpicture}[line width=1pt,auto,scale=#1]
  \draw [-](4,0) to     (4,2);
  \draw [-](4,2) to     (6,2);
  \draw [-](6,0) to    (6,2);
  \draw [-](4,0) to    (6,0);
  \fill (4,0) circle (3pt);
  \fill (4,2) circle (3pt);
  \fill (6,0) circle (3pt);
  \fill (6,2) circle (3pt);
  \node[left]   at(4,0)      {$0$};
  \node[left]   at(4,2)      {$v$};
  \node[right]  at (6,0)     {$x$};
  \node[right]  at (6,2)     {$y$};
  \node[left]  at (4,1) {$\geq1$};
  \node[above]  at (5,2) {$\geq0$};
  \node[below]  at (5,-0.2)  {$\geq 1$};
  \node[right]  at (6,1)     {$\geq 2$};
\end{tikzpicture}}

\def\picPBoundxxStepOne[#1]{
\begin{tikzpicture}[line width=1pt,auto,scale=#1]
  \draw [-](0,0) to [out=90,in=90]    (3,0);
  \draw [-](0,0) to     (3,0);
  \fill (0,0) circle (3pt);
 \fill (3,0) circle (3pt);
 \node[left]   at(0,0)      {$0,x$};
  \node[right]  at (3,0)     {$\ve[\iota]$};
  \node[above]  at (1,0.7)     {$\geq 3$};
  \node[below]  at (2,0)     {$=1$};
\end{tikzpicture}}

\def\picPBoundxxStepTwo[#1]{
\begin{tikzpicture}[line width=1pt,auto,scale=#1]
 \draw [-](0,0) to (2,0);
 \draw [-](0,0) to (0,1.5);
 \draw [-](0,1.5) to  [out=0,in=90]   (2,0);
 \fill (0,0) circle (3pt);
 \fill (2,0) circle (3pt);
 \fill (0,1.5) circle (3pt);
 \node[left]   at(0,0)      {$0$};
 \node[right]  at (2,0)     {$\ve[\iota]$};
   \node[left]   at (0,1.5)      {$v$};
 \node[left]  at (0,1)      {$=1$};
 \node[below]  at (1,0)     {$=1$};
 \node[right]  at (1.2,1.5)     {$\geq2$};

  \draw[-] (4,0)to(5,1.5);
  \draw[-] (5,1.5) to (6,0);
  \draw[-] (4,0) --(6,0);
  \fill (4,0) circle (3pt);
  \fill (5,1.5) circle (3pt);
  \fill (6,0) circle (3pt);
  \node[left]   at(4,0)       {$0$};
  \node[right]   at (6,0)       {$\ve[\iota]$};
  \node[left]   at (5,1.5)      {$v$};
  \node[below]  at (5,0)        {$=1$};
  \node[left]   at (4.6,.9)     {$\geq 2$};
  \node[right]  at (5.4,0.9)    {$\geq 1$};
\end{tikzpicture}}

\def\picPBoundxxStepThree[#1]{
\begin{tikzpicture}[line width=1pt,auto,scale=#1]
 \draw [-](0,0) to (2,0);
 \draw [-](0,0) to (0,1.5);
 \draw [-](0,1.5) to  [out=0,in=90]   (2,0);
 \fill (0,0) circle (3pt);
 \fill (2,0) circle (3pt);
 \fill (0,1.5) circle (3pt);
 \node[left]   at(0,0)      {$0$};
 \node[right]  at (2,0)     {$\ve[\iota]$};
  \node[left]   at(0,1.5)      {$x$};
 \node[left]  at (0,0.6)      {$=1$};
 \node[below]  at (1,0)     {$=1$};
 \node[right]  at (1.2,1.5)     {$\geq2$};

  \draw[-] (4,0)to(5,1.5);
  \draw[-] (5,1.5) to (6,0);
  \draw [-] (4,0) --(6,0);
  \fill (4,0) circle (3pt);
  \fill (5,1.5) circle (3pt);
  \fill (6,0) circle (3pt);
  \node[left]   at(4,0)       {$0$};
  \node[right]   at (6,0)       {$\ve[\iota]$};
  \node[left]   at (5,1.5)      {$x$};
  \node[below]  at (5,0)        {$=1$};
  \node[left]   at (4.6,.9)     {$\geq 2$};
  \node[right]  at (5.4,0.9)    {$\geq 1$};
\end{tikzpicture}}

\def\picPBoundxxStepFour[#1]{
\begin{tikzpicture}[line width=1pt,auto,scale=#1]
  \draw [-](4,0) to     (4,2);
  \draw [-](4,2) to     (6,2);
  \draw [-](6,0) to    (6,2);
  \draw [-](4,0) to    (6,0);
  \fill (4,0) circle (3pt);
  \fill (4,2) circle (3pt);
  \fill (6,0) circle (3pt);
  \fill (6,2) circle (3pt);
  \node[left]   at(4,0)      {$0$};
  \node[left]   at(4,2)      {$v$};
  \node[right]  at (6,0)     {$\ve[\iota]$};
  \node[right]  at (6,2)     {$x$};
  \node[left]  at (4,1) {$\geq1$};
  \node[above]  at (5,2) {$\geq 1$};
  \node[below]  at (5,-0.2)  {$=1$};
  \node[right]  at (6,1)     {$\geq 1$};
\end{tikzpicture}}

\def\picPBoundxyStepOne[#1]{
\begin{tikzpicture}[line width=1pt,auto,scale=#1]
 \draw [-](0,0) to (2,0);
 \draw [-](0,0) to (0,1.5);
 \draw [-](0,1.5) to  [out=0,in=90]   (2,0);
 \fill (0,0) circle (3pt);
 \fill (2,0) circle (3pt);
 \fill (0,1.5) circle (3pt);
 \node[left]   at(0,0)      {$0$};
 \node[below]  at (2,0)     {$\ve[\iota]$};
  \node[left]   at(0,1.5)      {$v$};
 \node[left]  at (0,0.6)      {$=1$};
 \node[below]  at (0.6,0)     {$=1$};
 \node[right]  at (1.2,1.5)     {$\geq2$};

  \draw[-] (5,0) to(6,1.5);
  \draw[-](6,1.5) to (7,0);
  \draw [-] (5,0) --(7,0);
  \fill (5,0) circle (3pt);
  \fill (6,1.5) circle (3pt);
  \fill (7,0) circle (3pt);
  \node[left]   at(5,0)       {$0=y$};
  \node[right]   at (7,0)       {$\ve[\iota]$};
  \node[left]   at (6,1.5)      {$v$};
  \node[below]  at (6,0)        {$=1$};
  \node[left]   at (5.6,.9)     {$\geq 2$};
  \node[right]  at (6.4,0.9)    {$\geq 1$};
\end{tikzpicture}}

\def\picPBoundxyStepOnePointNine[#1]{
\begin{tikzpicture}[line width=1pt,auto,scale=#1]
 \draw [-](-1,0) to (1,0);
 \draw [-](-1,0) to [out=270,in=180] (0,-1);
  \draw [-](1,0) to [out=270,in=0] (0,-1);
 \fill (-1,0) circle (3pt);
 \fill (1,0) circle (3pt);
 \node[left]   at(-1,0)      {$0=x$};
 \node[above]  at (1.2,-0.1)     {$\ve[\iota]$};
 \node[above]  at (-.3,0)    {$=1$};
 \node[right]  at (0.6,-0.8)     {$\geq 3$};
\end{tikzpicture}}

\def\picPBoundxyStepTwo[#1]{
\begin{tikzpicture}[line width=1pt,auto,scale=#1]
 \draw [-](-1,0) to (1,0);
 \draw [-](-1,0) to [out=90,in=180] (1,1.5);
 \draw [-](1,1.5) to     (1,0);
 \fill (-1,0) circle (3pt);
 \fill (1,1.5) circle (3pt);
 \fill (1,0) circle (3pt);
 \node[left]   at(-1,0)      {$0$};
 \node[below]  at (1.3,0)     {$\ve[\iota]$};
  \node[right]   at(1,1.5)   {$x$};
 \node[right]  at (1,0.7)    {$=1$};
 \node[below]  at (0,0)     {$=1$};
 \node[left]  at (0,1.4)     {$\geq2$};
\end{tikzpicture}}

\def\picPBoundxyStepTwoPointFive[#1]{
\begin{tikzpicture}[line width=1pt,auto,scale=#1]
  \draw[-] (7,0)--(5,0);
  \draw[-] (5,0)--(5,1.5);
  \draw[-] (5,1.5)  -- (7,1.5);
  \draw[-] (7,0) --(7,1.5);
  \fill (5,0) circle (3pt);
  \fill (5,1.5) circle (3pt);
  \fill (7,1.5) circle (3pt);
  \fill (7,0) circle (3pt);
  \node[left]   at(5,0)       {$0$};
  \node[right]   at (7,0)       {$\ve[\iota]$};
  \node[right]   at (7,1.5)      {$x$};
  \node[left]   at (5,1.5)      {$v$};
  \node[below]  at (6,0)        {$=1$};
  \node[left]   at (5,.9)     {$\geq 1$};
  \node[above]   at (6,1.5)     {$\geq 1$};
  \node[right]  at (7,0.9)    {$= 1$};
\end{tikzpicture}}

\def\picPBoundxyStepThree[#1]{
\begin{tikzpicture}[line width=1pt,auto,scale=#1]
 \draw [-](0,0) to (2,0);
 \draw [-](0,0) to (0,1.5);
 \draw [-](0,1.5) to  [out=0,in=90]   (2,0);
 \fill (0,0) circle (3pt);
 \fill (2,0) circle (3pt);
 \fill (0,1.5) circle (3pt);
 \node[left]   at(0,0)      {$0=x$};
 \node[below]  at (2,0)     {$\ve[\iota]$};
  \node[left]   at(0,1.5)     {$y$};
 \node[left]  at (0,0.6)      {$=1$};
 \node[below]  at (0.6,0)     {$=1$};
 \node[right]  at (1.2,1.5)   {$\geq2$};
\end{tikzpicture}}

\def\picPBoundxyStepFour[#1]{
\begin{tikzpicture}[line width=1pt,auto,scale=#1]
 \draw [-](0,0) to (2,0);
 \draw [-](0,0) to (0,2);
 \draw [-](2,2) to (0,2);
 \draw [-](2,2) to  (2,0);
 \fill (0,0) circle (3pt);
 \fill (2,0) circle (3pt);
 \fill (2,2) circle (3pt);
 \fill (0,2) circle (3pt);
 \node[left]   at(0,0)      {$0$};
 \node[right]  at (2,0)     {$\ve[\iota]$};
 \node[left]   at(0,2)     {$x$};
  \node[right]   at(2,2)     {$y$};
 \node[left]  at (0,1)      {$\geq 1$};
  \node[above]  at (1,2)      {$=1$};
 \node[below]  at (1,0)     {$=1$};
 \node[right]  at (2,1)   {$\geq1$};
\end{tikzpicture}}

\def\picPBoundxyStepFive[#1]{
\begin{tikzpicture}[line width=1pt,auto,scale=#1]
 \draw [-](0,0) to (-1+0.25,1.5);
 \draw [-](0,0) to (2,0);
 \draw [-](2,0) to (3-0.25,1.5);
 \draw [-](3-0.25,1.5) to  (1,2.5);
 \draw [-](-1+0.25,1.5) to (1,2.5);
 \fill (0,0) circle (3pt);
 \fill (2,0) circle (3pt);
 \fill (3-0.25,1.5) circle (3pt);
 \fill (-1+0.25,1.5) circle (3pt);
 \fill (1,2.5) circle (3pt);

 \node[left]   at(0,0)      {$0$};
 \node[right]  at (2,0)     {$\ve[\iota]$};
 \node[above]   at(1,2.5)     {$x$};
 \node[right]   at(2.75,1.5)     {$y$};
 \node[left]   at(-0.75,1.5)     {$v$};

 \node[left]  at (-.5,0.6)      {$\geq 0$};
 \node[above]  at (2.3,2.1)      {$=1$};
 \node[below]  at (1,0)     {$=1$};
 \node[right]  at (2.5,0.6)   {$\geq1$};
 \node[above]  at (0,2.1)   {$\geq 1$};
\end{tikzpicture}}

\def\picPBoundyyStepOne[#1]{
\begin{tikzpicture}[line width=1pt,auto,scale=#1]
  \draw [-](4+0.8,0-0.8) to [out=90,in=180]    (6+0.8,2-0.8);
  \draw [-](6+0.8,0-0.8) to    (6+0.8,2-0.8);
  \draw [-](4+0.8,0-0.8) to    (6+0.8,0-0.8);
  \fill (4+0.8,0-0.8) circle (3pt);
  \fill (6+0.8,0-0.8) circle (3pt);
  \fill (6+0.8,2-0.8) circle (3pt);
  \node[left]   at(4+0.8,0-0.8)      {$0$};
  \node[right]  at (6+0.8,0-0.8)     {$\ve[\iota]$};
  \node[right]  at (6+0.8,2-0.8)     {$y$};
  \node[above]  at (3.8+0.8,1-0.8) {$\geq 2$};
  \node[below]  at (5+0.8,-0.1-0.8)  {$= 1$};
  \node[right]  at (6+0.8,1-0.8)     {$\geq0$};
\end{tikzpicture}}

\def\picPBoundyyStepTwo[#1]{
\begin{tikzpicture}[line width=1pt,auto,scale=#1]
  \draw [-](4,0) to (4,2);
  \draw [-](4,2) to (6,2);
  \draw [-](6,0) to    (6,2);
  \draw [-](4,0) to    (6,0);
  \fill (4,0) circle (3pt);
  \fill (6,0) circle (3pt);
  \fill (6,2) circle (3pt);
  \fill (6,0) circle (3pt);
  \node[left]   at (4,0)      {$0$};
  \node[left]   at (4,2)      {$x$};
  \node[right]  at (6,0)     {$\ve[\iota]$};
  \node[right]  at (6,2)     {$y$};
  \node[above]  at (5,2)  {$\geq 2$};
  \node[left]  at (4,1)  {$\geq 1$};
  \node[below]  at (5,-0.1)  {$= 1$};
  \node[right]  at (6,1)     {$\geq 0$};
\end{tikzpicture}}

\def\picPBoundyyStepThree[#1]{
\begin{tikzpicture}[line width=1pt,auto,scale=#1]
 \draw [-](0,0) to (-1+0.25,1.5);
 \draw [-](0,0) to (2,0);
 \draw [-](2,0) to (3-0.25,1.5);
 \draw [-](3-0.25,1.5) to  (1,2.5);
 \draw [-](-1+0.25,1.5) to (1,2.5);
 \fill (0,0) circle (3pt);
 \fill (2,0) circle (3pt);
 \fill (3-0.25,1.5) circle (3pt);
 \fill (-1+0.25,1.5) circle (3pt);
 \fill (1,2.5) circle (3pt);

 \node[left]   at(0,0)      {$0$};
 \node[right]  at (2,0)     {$\ve[\iota]$};
 \node[above]   at(1,2.5)     {$x$};
 \node[right]   at(2.75,1.5)     {$y$};
 \node[left]   at(-0.75,1.5)     {$v$};

 \node[left]  at (-.5,0.6)      {$\geq 0$};
 \node[above]  at (2.3,2.1)      {$\geq 2$};
 \node[below]  at (1,0)     {$=1$};
 \node[right]  at (2.5,0.6)   {$\geq0$};
 \node[above]  at (0,2.1)   {$\geq 1$};
\end{tikzpicture}}

\def\picPBoundxyStepRibZero[#1]{
\begin{tikzpicture}[line width=1pt,auto,scale=#1]
 \draw [-](-1,0) to (1,0);
 \draw [-](1,0) to [out=90,in=0] (-1,1.5);
 \draw [-](-1,1.5) to     (-1,0);
 \fill (-1,0) circle (3pt);
 \fill (-1,1.5) circle (3pt);
 \fill (1,0) circle (3pt);
 \node[left]   at(-1,0)      {$0=y$};
 \node[below]  at (1.3,0)     {$\ve[\iota]$};
 \node[left]   at(-1,1.5)   {$x$};
 \node[left]  at (-1,0.7)    {$=1$};
 \node[below]  at (0,0)     {$=1$};
 \node[right]  at (1,1.4)     {$\geq2$};
\end{tikzpicture}}

\def\picPBoundxyStepRibOne[#1]{
\begin{tikzpicture}[line width=1pt,auto,scale=#1]
 \draw [-](-1,0) to (1,0);
 \draw [-](-1,0) to [out=90,in=180] (1,1.5);
 \draw [-](1,1.5) to     (1,0);
 \fill (-1,0) circle (3pt);
 \fill (1,1.5) circle (3pt);
 \fill (1,0) circle (3pt);
 \node[left]   at(-1,0)      {$0$};
 \node[below]  at (1.3,0)     {$\ve[\iota]$};
  \node[right]   at(1,1.5)   {$y$};
 \node[right]  at (1,0.7)    {$=1$};
 \node[below]  at (0,0)     {$=1$};
 \node[left]  at (0,1.4)     {$\geq2$};

\begin{scope}[shift={(-6,-3.5)}]
  \draw[-] (7,0)--(5,0);
  \draw[-] (5,0)--(5,1.5);
  \draw[-] (5,1.5)-- (7,1.5);
  \draw[-]  (7,0) --(7,1.5);
  \fill (5,0) circle (3pt);
  \fill (5,1.5) circle (3pt);
  \fill (7,1.5) circle (3pt);
  \fill (7,0) circle (3pt);
  \node[left]   at(5,0)       {$0$};
  \node[right]   at (7,0)       {$\ve[\iota]$};
  \node[right]   at (7,1.5)      {$y$};
  \node[left]   at (5,1.5)      {$v$};
  \node[below]  at (6,0)        {$=1$};
  \node[left]   at (5,.9)     {$\geq 1$};
  \node[above]   at (6,1.5)     {$\geq 1$};
  \node[right]  at (7,0.9)    {$= 1$};
  \end{scope}
\end{tikzpicture}}

\def\picPBoundyyStepRibTwo[#1]{
\begin{tikzpicture}[line width=1pt,auto,scale=#1]
  \draw [-](4,0) to (4,2);
  \draw [-](4,2) to (6,2);
  \draw [-](6,0) to    (6,2);
  \draw [-](4,0) to    (6,0);
  \fill (4,0) circle (3pt);
  \fill (6,0) circle (3pt);
  \fill (6,2) circle (3pt);
  \fill (6,0) circle (3pt);
  \node[left]   at (4,0)      {$0$};
  \node[left]   at (4,2)      {$y$};
  \node[right]  at (6,0)     {$\ve[\iota]$};
  \node[right]  at (6,2)     {$x$};
  \node[above]  at (5,2)  {$=1$};
  \node[left]  at (4,1)  {$\geq 1$};
  \node[below]  at (5,-0.1)  {$= 1$};
  \node[right]  at (6,1)     {$\geq 1$};
\end{tikzpicture}}

\def\picPBoundyyStepRibThree[#1]{
\begin{tikzpicture}[line width=1pt,auto,scale=#1]
 \draw [-](0,0) to (-1+0.25,1.5);
 \draw [-](0,0) to (2,0);
 \draw [-](2,0) to (3-0.25,1.5);
 \draw [-](3-0.25,1.5) to  (1,2.5);
 \draw [-](-1+0.25,1.5) to (1,2.5);
 \fill (0,0) circle (3pt);
 \fill (2,0) circle (3pt);
 \fill (3-0.25,1.5) circle (3pt);
 \fill (-1+0.25,1.5) circle (3pt);
 \fill (1,2.5) circle (3pt);

 \node[left]   at(0,0)      {$0$};
 \node[right]  at (2,0)     {$\ve[\iota]$};
 \node[above]   at(1,2.5)     {$y$};
 \node[right]   at(2.75,1.5)     {$x$};
 \node[left]   at(-0.75,1.5)     {$v$};

 \node[left]  at (-.5,0.6)      {$\geq 1$};
 \node[above]  at (2.3,2.1)      {$=1$};
 \node[below]  at (1,0)     {$=1$};
 \node[right]  at (2.5,0.6)   {$\geq1$};
 \node[above]  at (0,2.1)   {$\geq 0$};
\end{tikzpicture}}

\def\picPBoundyyStepRibFour[#1]{
\begin{tikzpicture}[line width=1pt,auto,scale=#1]
  \draw [-](4,0) to (4,2);
  \draw [-](4,2) to (6,2);
  \draw [-](6,0) to    (6,2);
  \draw [-](4,0) to    (6,0);
  \fill (4,0) circle (3pt);
  \fill (6,0) circle (3pt);
  \fill (6,2) circle (3pt);
  \fill (4,2) circle (3pt);
  \node[left]   at (4,0)      {$0$};
  \node[left]   at (4,2)      {$v$};
  \node[right]  at (6,0)     {$\ve[\iota]$};
  \node[right]  at (6,2)     {$y$};
  \node[above]  at (5,2)  {$\geq 0$};
  \node[left]  at (4,1)  {$\geq 1$};
  \node[below]  at (5,-0.1)  {$= 1$};
  \node[right]  at (6,1)     {$\geq 2$};
\end{tikzpicture}}

\def\picPBoundyyStepRibFive[#1]{
\begin{tikzpicture}[line width=1pt,auto,scale=#1]
 \draw [-](0,0) to (-1+0.25,1.5);
 \draw [-](0,0) to (2,0);
 \draw [-](2,0) to (3-0.25,1.5);
 \draw [-](3-0.25,1.5) to  (1,2.5);
 \draw [-](-1+0.25,1.5) to (1,2.5);
 \fill (0,0) circle (3pt);
 \fill (2,0) circle (3pt);
 \fill (3-0.25,1.5) circle (3pt);
 \fill (-1+0.25,1.5) circle (3pt);
 \fill (1,2.5) circle (3pt);

 \node[left]   at(0,0)      {$0$};
 \node[right]  at (2,0)     {$\ve[\iota]$};
 \node[above]   at(1,2.5)     {$y$};
 \node[right]   at(2.75,1.5)     {$x$};
 \node[left]   at(-0.75,1.5)     {$v$};

 \node[left]  at (-.5,0.6)  {$\geq 0$};
 \node[above]  at (2.3,2.1) {$\geq 2$};
 \node[below]  at (1,0)     {$=1$};
 \node[right]  at (2.5,0.6) {$\geq 1$};
 \node[above]  at (0,2.1)   {$\geq 0$};
\end{tikzpicture}}

\def\picPBoundxxIotaOne[#1]{
\begin{tikzpicture}[line width=1pt,auto,scale=#1]
  \draw[-] (0,0)--(1,1.5);
  \draw[-] (1,1.5)--(2,0);
  \draw[-] (2,0)--(0,0);
  \draw[-] (0,0)--(-1,0);
  \fill (-1,0) circle (3pt);
  \fill (0,0) circle (3pt);
  \fill (1,1.5) circle (3pt);
  \fill (2,0) circle (3pt);
  \node[above]   at(-1,0)       {$\ve[\iota]$};
  \node[below]   at(0,0)       {$0$};
  \node[right]   at (2,0)       {$x$};
  \node[left]   at (1,1.5)      {$v$};
 \end{tikzpicture}
 }

\def\picPBoundxxIotaTwo[#1]{
\begin{tikzpicture}[line width=1pt,auto,scale=#1]
  \draw [-](0,0) to [out=90,in=90]    (3,0);
  \draw [-](0,0) to [out=270,in=270]    (3,0);
  \fill (0,0) circle (3pt);
 \fill (3,0) circle (3pt);
 \node[left]   at(0,0)      {$0$};
  \node[right]  at (3,0)     {$\ve[\iota]$};
  \node[above]  at (1,0.7)     {$\geq 1$};
  \node[below]  at (2,-0.7)     {$\geq 3$};
\end{tikzpicture}}

\def\picPBoundxxIotaThree[#1]{
\begin{tikzpicture}[line width=1pt,auto,scale=#1]
  \draw [-](0,0) to [out=270,in=270]    (3,0);
  \draw [-](0,0) to (3,0);
  \draw [-](3,0) to [out=90,in=0]    (0,1);
  \fill (0,0) circle (3pt);
 \fill (3,0) circle (3pt);
  \fill (0,1) circle (3pt);
 \node[left]   at(0,0)      {$0$};
 \node[right]  at (3,0)     {$x$};
  \node[left]  at (0,1)     {$\ve[\iota]$};
\end{tikzpicture}}

\def\picPBoundxxIotaFour[#1]{
\begin{tikzpicture}[line width=1pt,auto,scale=#1]
  \draw[-] (0,0)-- (2,0);
  \draw[-] (0,0)-- (0,1.5);
  \draw [-](0,1.5) to [out=180,in=90]    (-2,0);
  \draw [-](0,1.5) to [out=0,in=90]    (2,0);
  \fill (0,0) circle (3pt);
  \fill (0,1.5) circle (3pt);
  \fill (2,0) circle (3pt);
  \fill (-2,0) circle (3pt);
  \node[left]   at(0,0)      {$0$};
  \node[right]  at (2,0)     {$x$};
  \node[above]   at (0,1.5)   {$v$};
  \node[below]  at (1,0)     {$=1$};
  \node[below]  at (-0.6,1.2)     {$=1$};
  \node[above]  at (2,0.8) {$\geq 2$};
  \node[above]  at (-2,1) {$\geq 2$};
\end{tikzpicture}}

\def\picPBoundxxIotaFive[#1]{
\begin{tikzpicture}[line width=1pt,auto,scale=#1]
  \draw[-] (0,0)--(1,2);
  \draw[-] (1,2)--(2,0);
  \draw[-] (2,0)--(0,0);
  \draw [-] (1,2) to [out=180,in=90] (-1,0);
  \fill (-1,0) circle (3pt);
  \fill (0,0) circle (3pt);
  \fill (1,2) circle (3pt);
  \fill (2,0) circle (3pt);
  \node[below]   at(-1,0)       {$\ve[\iota]$};
  \node[below]   at(0,0)       {$0$};
  \node[right]   at (2,0)       {$x$};
  \node[right]   at (1,2)      {$u$};
  \node[above]  at (2.2,0.6)        {$\geq 1$};
  \node[left]   at (-1,0.9)     {$\geq 2$};
  \node[below]  at (1.2,0)    {$=1$};

 \end{tikzpicture}
 }

\def\picPBoundxxIotaSix[#1]{
\begin{tikzpicture}[line width=1pt,auto,scale=#1]
  \draw [-] (0,0)--(1,1.5);
  \draw [-] (2,0)--(1,1.5);
  \draw [-] (2,0)--(0,0);
  \draw [-] (1,1.5) to [out=180,in=90] (-1,0);
  \fill (-1,0) circle (3pt);
  \fill (0,0) circle (3pt);
  \fill (1,1.5) circle (3pt);
  \fill (2,0) circle (3pt);
  \node[below]   at(-1,0)       {$\ve[\iota]$};
  \node[below]   at(0,0)       {$0$};
  \node[right]   at (2,0)       {$x$};
  \node[right]   at (1,1.5)      {$u$};
  \node[above]  at (2.2,0.4)        {$\geq 1$};
  \node[left]   at (-1,0.7)     {$\geq 1$};
  \node[below]  at (1.2,0)    {$\geq 2$};
 \end{tikzpicture}
 }

 \def\picPBoundxxIotaSeven[#1]{
\begin{tikzpicture}[line width=1pt,auto,scale=#1]
  \draw [-](4,0) to (4,2);
  \draw [-](4,2) to (6,2);
  \draw [-](6,0) to    (6,2);
  \draw [-](4,0) to    (6,0);
  \draw [-](4,2) to  [out=180,in=90] (2,0);
  \fill (2,0) circle (3pt);
  \fill (4,0) circle (3pt);
  \fill (6,0) circle (3pt);
  \fill (6,2) circle (3pt);
  \fill (6,0) circle (3pt);
  \node[below]   at (3,0)      {$\ve[\iota]$};
  \node[below]   at (4,0)      {$0$};
  \node[above]   at (4,2)      {$u$};
  \node[right]  at (6,0)     {$x$};
  \node[right]  at (6,2)     {$v$};
  \node[above]  at (5,2)  {$\geq 1$};
  \node[left]  at (4.3,1)  {$\geq 1$};
  \node[below]  at (5,-0.1)  {$\geq 1$};
  \node[right]  at (6,1)     {$\geq 1$};
\end{tikzpicture}}

 \def\picPBoundxyIotaOne[#1]{
\begin{tikzpicture}[line width=1pt,auto,scale=#1]
  \draw [-](4,0) to (4,2);
  \draw [-](4,2) to (6,2);
  \draw [-](6,0) to    (6,2);
  \draw [-](4,0) to    (6,0);
  \draw [-](4,0) to  (2,0);
  \fill (2,0) circle (3pt);
  \fill (4,0) circle (3pt);
  \fill (6,0) circle (3pt);
  \fill (6,2) circle (3pt);
  \fill (4,2) circle (3pt);
  \node[below]   at (2,0)      {$\ve[\iota]$};
  \node[below]   at (4,0)      {$0$};
  \node[above]   at (4,2)      {$v$};
  \node[right]  at (6,0)     {$y$};
  \node[right]  at (6,2)     {$x$};
  \node[right]  at (6,1)     {$=1$};
\end{tikzpicture}}

 \def\picPBoundxyIotaTwo[#1]{
\begin{tikzpicture}[line width=1pt,auto,scale=#1]
  \draw [-](0,0) to [out=270,in=180] (2,-2);
  \draw [-](2,-2) to    (2,0);
  \draw [-](0,0) to  [out=90,in=90] (2,0);
  \fill (2,0) circle (3pt);
  \fill (0,0) circle (3pt);
  \fill (2,-2) circle (3pt);
  \node[right]   at (2,0)      {$\ve[\iota]$};
  \node[left]   at (0,0)      {$0$};
  \node[right]   at (2,-2)      {$y$};
  \node[above]  at (1,0.5)  {$\geq 1$};
  \node[below]  at (-0.5,-1)  {$\geq 2$};
  \node[right]  at (2,-1)  {$=1$};
\end{tikzpicture}}

 \def\picPBoundxyIotaSix[#1]{
\begin{tikzpicture}[line width=1pt,auto,scale=#1]
 \draw [-](0,0) to (0,2);
 \draw [-](0,0) to (2,0);
 \draw [-](2,0) to (2,2);
 \draw [-](0,2) to (2,2);

 \fill (0,0) circle (3pt);
 \fill (2,0) circle (3pt);
 \fill (0,2) circle (3pt);
 \fill (2,2) circle (3pt);

 \node[left]  at (0,0)      {$0$};
 \node[right]  at (2,0)     {$y$};
 \node[right]  at (2,2)     {$x$};
 \node[left]  at (0,2)     {$\ve[\iota]=v$};

 \node[left]  at (0,1)   {$\geq 1$};
 \node[right]  at (2,1)   {$= 1$};
 \node[below]  at (1,0)   {$= 1$};
 \node[above]  at (1,2)   {$\geq 1$};

\begin{scope}[shift={(-6.5,-5)},rotate=0]
 \draw [-](6.5,0) to (6.5-1+0.25,1.5);
 \draw [-](6.5,0) to (6.5+2,0);
 \draw [-](6.5+2,0) to (6.5+3-0.25,1.5);
 \draw [-](6.5+3-0.25,1.5) to  (6.5+1,2.5);
 \draw [-](6.5+-1+0.25,1.5) to (6.5+1,2.5);

 \fill (6.5+0,0) circle (3pt);
 \fill (6.5+2,0) circle (3pt);
 \fill (6.5+3-0.25,1.5) circle (3pt);
 \fill (6.5-1+0.25,1.5) circle (3pt);
 \fill (6.5+1,2.5) circle (3pt);

 \node[below]   at(6.5+0,0)      {$0$};
 \node[right]  at (6.5+2,0)     {$y$};
 \node[above]  at(6.5+1,2.5)     {$v$};
 \node[right]  at(6.5+2.75,1.5)     {$x$};
 \node[left]   at(6.5-0.75,1.5)     {$\ve[\iota]$};

 \node[left]  at (6.5+0,0.3)  {$\geq 1$};
 \node[above]  at (6.5+2.3,2-0.1) {$\geq 0$};
 \node[below]  at (6.5+1.4,0)     {$=1$};
 \node[right]  at (6.5+2.5,0.6) {$=1$};
 \node[above]  at (6.5-0.3,2.1)   {$\geq 1$};
\end{scope}
\end{tikzpicture}}

  \def\picPBoundxyIotaEight[#1]{
\begin{tikzpicture}[line width=1pt,auto,scale=#1]
 \draw [-](0,0) to (-1+0.25,1.5);
 \draw [-](0,0) to (2,0);
 \draw [-](2,0) to (3-0.25,1.5);
 \draw [-](3-0.25,1.5) to  (1,2.5);
 \draw [-](-1+0.25,1.5) to (1,2.5);
 \draw [-](-1+0.25,1.5) to [out=220,in=90] (-2.5,0);

 \fill (0,0) circle (3pt);
 \fill (2,0) circle (3pt);
 \fill (3-0.25,1.5) circle (3pt);
 \fill (-1+0.25,1.5) circle (3pt);
 \fill (1,2.5) circle (3pt);
 \fill (-2.5,0) circle (3pt);

 \node[below]   at(0,0)      {$0$};
 \node[right]  at (2,0)     {$y$};
 \node[above]  at(1,2.5)     {$v$};
 \node[right]  at(2.75,1.5)     {$x$};
 \node[left]   at(-0.75,1.5)     {$u$};
 \node[below]  at (-2.5,0)     {$\ve[\iota]$};

  \node[above]  at (-4,0.6) {$\geq 1$};
 \node[above]  at (2.3,2-0.1) {$\geq 0$};
 \node[below]  at (1.4,0)     {$= 1$};
 \node[right]  at (2.5,0.6) {$=1$};
 \node[above]  at (-0.3,2.1)   {$\geq 0$};
\end{tikzpicture}}

 \def\picPBoundxyIotaTen[#1]{
\begin{tikzpicture}[line width=1pt,auto,scale=#1]
 \draw [-](0,0) to (-1+0.25,1.5);
 \draw [-](0,0) to (2,0);
 \draw [-](2,0) to (3-0.25,1.5);
 \draw [-](3-0.25,1.5) to  (1,2.5);
 \draw [-](-1+0.25,1.5) to (1,2.5);
 \draw [-](-1+0.25,1.5) to [out=220,in=90] (-2.5,0);

 \fill (0,0) circle (3pt);
 \fill (2,0) circle (3pt);
 \fill (3-0.25,1.5) circle (3pt);
 \fill (-1+0.25,1.5) circle (3pt);
 \fill (1,2.5) circle (3pt);
 \fill (-2.5,0) circle (3pt);

 \node[below]   at(0,0)      {$0$};
 \node[right]  at (2,0)     {$y$};
 \node[above]  at(1,2.5)     {$v$};
 \node[right]  at(2.75,1.5)     {$x$};
 \node[left]   at(-0.75,1.5)     {$u$};
 \node[below]  at (-2.5,0)     {$\ve[\iota]$};

  \node[above]  at (-3,0.6) {$\geq 0$};
 \node[above]  at (2.3,2-0.1) {$\geq 0$};
 \node[below]  at (1.4,0)     {$\geq 2$};
 \node[right]  at (2.5,0.6) {$=1$};
 \node[above]  at (-0.3,2.1)   {$\geq 0$};
\end{tikzpicture}}

 \def\picPBoundyyIotaOne[#1]{
\begin{tikzpicture}[line width=1pt,auto,scale=#1]
  \draw [-](4,0) to (4,2);
  \draw [-](4,2) to (6,2);
  \draw [-](6,0) to    (6,2);
  \draw [-](4,0) to    (6,0);
  \draw [-](4,0) to  (2,0);
  \fill (2,0) circle (3pt);
  \fill (4,0) circle (3pt);
  \fill (6,0) circle (3pt);
  \fill (6,2) circle (3pt);
  \fill (6,0) circle (3pt);
  \node[below]   at (2,0)      {$\ve[\iota]$};
  \node[below]   at (4,0)      {$0$};
  \node[above]   at (4,2)      {$v$};
  \node[right]  at (6,0)     {$y$};
  \node[right]  at (6,2)     {$x$};
  \node[right]  at (6,1)     {$\geq 2$};
\end{tikzpicture}}

 \def\picPBoundyyIotaThree[#1]{
\begin{tikzpicture}[line width=1pt,auto,scale=#1]
 \draw [-](0,0) to (-1+0.25,1.5);
 \draw [-](0,0) to (2,0);
 \draw [-](2,0) to (3-0.25,1.5);
 \draw [-](3-0.25,1.5) to  (1,2.5);
 \draw [-](-1+0.25,1.5) to (1,2.5);
 \draw [-](-1+0.25,1.5) to [out=220,in=90] (-2.5,0);

 \fill (0,0) circle (3pt);
 \fill (2,0) circle (3pt);
 \fill (3-0.25,1.5) circle (3pt);
 \fill (-1+0.25,1.5) circle (3pt);
 \fill (1,2.5) circle (3pt);
 \fill (-2.5,0) circle (3pt);

 \node[below]   at(0,0)      {$0$};
 \node[right]  at (2,0)      {$y$};
 \node[above]  at(1,2.5)     {$v$};
 \node[right]  at(2.75,1.5)  {$x$};
 \node[left]   at(-0.75,1.5) {$u$};
 \node[below]  at (-2.5,0)   {$\ve[\iota]$};

  \node[above]  at (-3,0.6)   {$\geq 0$};
 \node  at (0.4,0.8)          {$\geq 1$};
 \node[above]  at (2.3,2-0.1) {$\geq 0$};
 \node[below]  at (1.4,0)     {$\geq 1$};
 \node[right]  at (2.5,0.6)   {$\geq 2$};
 \node[above]  at (-0.3,2.1)  {$\geq 0$};
\end{tikzpicture}}

 \def\picPBoundxyIotaRibOne[#1]{
\begin{tikzpicture}[line width=1pt,auto,scale=#1]
  \draw [-](4,0) to (4,2);
  \draw [-](4,2) to (6,2);
  \draw [-](6,0) to    (6,2);
  \draw [-](4,0) to    (6,0);
  \draw [-](4,0) to  (2,0);
  \fill (2,0) circle (3pt);
  \fill (4,0) circle (3pt);
  \fill (4,2) circle (3pt);
  \fill (6,0) circle (3pt);
  \fill (6,2) circle (3pt);
  \fill (6,0) circle (3pt);
  \node[below]   at (2,0)      {$\ve[\iota]$};
  \node[below]   at (4,0)      {$0$};
  \node[above]   at (4,2)      {$v$};
  \node[right]  at (6,0)     {$x$};
  \node[right]  at (6,2)     {$y$};
  \node[right]  at (6,1)     {$=1$};
\end{tikzpicture}}

 \def\picPBoundxyIotaRibTwo[#1]{
\begin{tikzpicture}[line width=1pt,auto,scale=#1]
  \draw [-](0,0) to  (2,2);
  \draw [-](0,0) to  (2,0);
  \draw [-](2,0) to    (2,2);
  \draw [-](2,2) to  [out=180,in=90] (-2,0);
  \fill (2,0) circle (3pt);
  \fill (0,0) circle (3pt);
  \fill (2,2) circle (3pt);
  \fill (-2,0) circle (3pt);
  \node[left]   at (-2,0)      {$\ve[\iota]$};
  \node[left]   at (0,0)      {$0$};
  \node[right]   at (2,2)      {$y$};
  \node[right]   at (2,0)      {$x$};-
  \node  at (0,0.9)   {$\geq 2$};
  \node[below]  at (1,0)  {$=1$};
  \node at (-1.5,2)  {$\geq 1$};
  \node[right]  at (2,1)  {$=1$};
\end{tikzpicture}}

 \def\picPBoundxyIotaRibThree[#1]{
\begin{tikzpicture}[line width=1pt,auto,scale=#1]
  \draw [-](0,0) to  (0,2);
  \draw [-](0,2) to  (2,2);
  \draw [-](0,0) to  (2,0);
  \draw [-](2,0) to    (2,2);
  \draw [-](0,2) to  [out=180,in=90] (-2,0);
  \fill (2,0) circle (3pt);
  \fill (0,0) circle (3pt);
  \fill (2,2) circle (3pt);
  \fill (0,2) circle (3pt);
  \fill (-2,0) circle (3pt);
  \node[left]   at (-2,0)      {$\ve[\iota]$};
  \node[left]   at (0,0)      {$0$};
  \node[right]   at (2,2)      {$y$};
  \node[right]   at (2,0)      {$x$};-
  \node  at (0.2,0.9)   {$\geq 1$};
  \node[above]  at (1,2)   {$\geq 1$};
  \node[below]  at (1,0)  {$=1$};
  \node at (-1.5,2)  {$\geq 1$};
  \node[right]  at (2,1)  {$=1$};
\end{tikzpicture}}

\def\picPBoundxyIotaRibFour[#1]{
\begin{tikzpicture}[line width=1pt,auto,scale=#1]
 \draw [-](0,0) to (-1+0.25,1.5);
 \draw [-](0,0) to (2,0);
 \draw [-](2,0) to (3-0.25,1.5);
 \draw [-](3-0.25,1.5) to  (1,2.5);
 \draw [-](-1+0.25,1.5) to (1,2.5);
 \draw [-](-1+0.25,1.5) to [out=220,in=90] (-2.5,0);

 \fill (0,0) circle (3pt);
 \fill (2,0) circle (3pt);
 \fill (3-0.25,1.5) circle (3pt);
 \fill (-1+0.25,1.5) circle (3pt);
 \fill (1,2.5) circle (3pt);
 \fill (-2.5,0) circle (3pt);

 \node[below]   at(0,0)      {$0$};
 \node[right]  at (2,0)      {$x$};
 \node[above]  at(1,2.5)     {$v$};
 \node[right]  at(2.75,1.5)  {$y$};
 \node[left]   at(-0.75,1.5) {$u$};
 \node[below]  at (-2.5,0)   {$\ve[\iota]$};

  \node[above]  at (-3,0.6)   {$\geq 0$};
 \node  at (0.5,0.85)          {$\geq 1$};
 \node[above]  at (2.3,2-0.1) {$\geq 0$};
 \node[below]  at (1.4,0)     {$=1$};
 \node[right]  at (2.5,0.6)   {$=1$};
 \node[above]  at (-0.3,2.1)  {$\geq 1$};
\end{tikzpicture}}

\def\picPBoundxyIotaRibFive[#1]{
\begin{tikzpicture}[line width=1pt,auto,scale=#1]
 \draw [-](0,0) to (-1+0.25,1.5);
 \draw [-](0,0) to (2,0);
 \draw [-](2,0) to (3-0.25,1.5);
 \draw [-](3-0.25,1.5) to  (1,2.5);
 \draw [-](-1+0.25,1.5) to (1,2.5);
 \draw [-](-1+0.25,1.5) to [out=220,in=90] (-2.5,0);

 \fill (0,0) circle (3pt);
 \fill (2,0) circle (3pt);
 \fill (3-0.25,1.5) circle (3pt);
 \fill (-1+0.25,1.5) circle (3pt);
 \fill (1,2.5) circle (3pt);
 \fill (-2.5,0) circle (3pt);

 \node[below]   at(0,0)      {$0$};
 \node[right]  at (2,0)      {$x$};
 \node[above]  at(1,2.5)     {$v$};
 \node[right]  at(2.75,1.5)  {$y$};
 \node[left]   at(-0.75,1.5) {$u$};
 \node[below]  at (-2.5,0)   {$\ve[\iota]$};

  \node[above]  at (-3,0.6)   {$\geq 0$};
 \node  at (0.5,0.85)          {$\geq 1$};
 \node[above]  at (2.3,2-0.1) {$\geq 0$};
 \node[below]  at (1.4,0)     {$\geq 2$};
 \node[right]  at (2.5,0.6)   {$=1$};
 \node[above]  at (-0.3,2.1)  {$\geq 0$};
\end{tikzpicture}}

 \def\picPBoundyyIotaRibOne[#1]{
\begin{tikzpicture}[line width=1pt,auto,scale=#1]
  \draw [-](4,0) to (4,2);
  \draw [-](4,2) to (6,2);
  \draw [-](6,0) to    (6,2);
  \draw [-](4,0) to    (6,0);
  \draw [-](4,0) to  (2,0);
  \fill (2,0) circle (3pt);
  \fill (4,0) circle (3pt);
  \fill (6,0) circle (3pt);
  \fill (6,2) circle (3pt);
  \fill (6,0) circle (3pt);
  \node[below]   at (2,0)      {$\ve[\iota]$};
  \node[below]   at (4,0)      {$0$};
  \node[above]   at (4,2)      {$v$};
  \node[right]  at (6,0)     {$x$};
  \node[right]  at (6,2)     {$y$};
  \node[right]  at (6,1)     {$\geq 2$};
\end{tikzpicture}}

\def\picPBoundyyIotaRibTwo[#1]{
\begin{tikzpicture}[line width=1pt,auto,scale=#1]
 \draw [-](0,0) to (-1+0.25,1.5);
 \draw [-](0,0) to (2,0);
 \draw [-](2,0) to (3-0.25,1.5);
 \draw [-](3-0.25,1.5) to  (1,2.5);
 \draw [-](-1+0.25,1.5) to (1,2.5);
 \draw [-](-1+0.25,1.5) to [out=220,in=90] (-2.5,0);

 \fill (0,0) circle (3pt);
 \fill (2,0) circle (3pt);
 \fill (3-0.25,1.5) circle (3pt);
 \fill (-1+0.25,1.5) circle (3pt);
 \fill (1,2.5) circle (3pt);
 \fill (-2.5,0) circle (3pt);

 \node[below]   at(0,0)      {$0$};
 \node[right]  at (2,0)     {$x$};
 \node[above]  at(1,2.5)     {$v$};
 \node[right]  at(2.75,1.5)     {$y$};
 \node[left]   at(-0.75,1.5)     {$u$};
 \node[below]  at (-2.5,0)     {$\ve[\iota]$};

  \node[above]  at (-3,0.6)   {$\geq 0$};
 \node  at (0.4,0.8)          {$\geq 1$};
 \node[above]  at (2.3,2-0.1) {$\geq 0$};
 \node[below]  at (1.4,0)     {$\geq 1$};
 \node[right]  at (2.5,0.6)   {$\geq 2$};
 \node[above]  at (-0.3,2.1)  {$\geq 0$};
\end{tikzpicture}}

\def\picMiddleBoundxx[#1]{
\begin{tikzpicture}[line width=1pt,auto,scale=#1]
  \draw (0,0)-- (2,0)-- (1,1.5)-- cycle;
  \fill (0,0) circle (3pt);
  \fill (2,0) circle (3pt);
  \fill (1,1.5) circle (3pt);
  \node[left]   at(0,0)       {$0$};
  \node[right]   at (2,0)       {$x$};
  \node[left]   at (1,1.5)       {$v$};
\end{tikzpicture}}

\def\picMiddleBoundxxZero[#1]{
\begin{tikzpicture}[line width=1pt,auto,scale=#1]
  \draw [-](0,0) to [out=90,in=90]    (2,0);
  \draw [-](0,0) to [out=270,in=270]    (2,0);
  \fill (0,0) circle (3pt);
  \node[left]   at(0,0)      {$0$};
  \node[right]  at (2,0)   {$\geq 4$};
\end{tikzpicture}}

\def\picMiddleBoundxxOne[#1]{
\begin{tikzpicture}[line width=1pt,auto,scale=#1]
  \draw (0,1.5)-- (2,0);
  \draw [line width=1pt] (2,0) --(0,0);
  \draw [line width=1pt] (0,0) --(0,1.5);
  \fill (0,0) circle (3pt);
  \fill (0,1.5) circle (3pt);
  \fill (2,0) circle (3pt);
  \node[left]   at(0,0)      {$0$};
  \node[right]  at (2,0)     {$x$};
  \node[left]   at (0,1.5)   {$v$};
  \node[left]   at (0,.725)  {$=1$};
  \node[below]  at (1,0)     {$=1$};
  \node[above]  at (1.5,0.7)     {$\geq 2$};

  \draw (4,0)--(5,1.5)-- (6,0);
  \draw [line width=1pt] (4,0) --(6,0);
  \fill (4,0) circle (3pt);
  \fill (5,1.5) circle (3pt);
  \fill (6,0) circle (3pt);
  \node[left]   at(4,0)       {$0$};
  \node[right]   at (6,0)       {$x$};
  \node[left]   at (5,1.5)      {$v$};
  \node[below]  at (5,0)        {$=1$};
  \node[left]   at (4.6,.9)     {$\geq 2$};
  \node[right]  at (5.4,0.9)    {$\geq 1$};
\end{tikzpicture}}

\def\picMiddleBoundxxTwo[#1]{
\begin{tikzpicture}[line width=1pt,auto,scale=#1]
  \draw [-](0,0) to [out=90,in=90]    (3,0);
  \draw [-](0,0) to [out=270,in=270]    (3,0);
  \fill (0,0) circle (3pt);
  \fill (3,0) circle (3pt);
  \node[left]   at(0,0)      {$0$};
  \node[right]  at (3,0)     {$x$};
  \node[above]  at (1,0.7)     {$\geq 3$};
  \node[below]  at (2,-0.7)     {$=1$};

  \draw [-](5,0) to [out=90,in=90]    (8,0);
  \draw [-](5,0) to [out=270,in=270]    (8,0);
  \fill (5,0) circle (3pt);
  \fill (8,0) circle (3pt);
  \node[left]   at(5,0)      {$0$};
  \node[right]  at (8,0)     {$x$};
  \node[above]  at (6,0.7)     {$\geq 2$};
  \node[below]  at (7,-0.7)     {$\geq 2$};
\end{tikzpicture}}

\def\picMiddleBoundxxThree[#1]{
\begin{tikzpicture}[line width=1pt,auto,scale=#1]
  \draw (0,0)--(1,1.5)-- (2,0)--(0,0);
  \fill (0,0) circle (3pt);
  \fill (1,1.5) circle (3pt);
  \fill (2,0) circle (3pt);
  \node[left]   at(0,0)       {$0$};
  \node[right]   at (2,0)       {$x$};
  \node[left]   at (1,1.5)      {$v$};
  \node[below]  at (1,0)        {$\geq 2$};
  \node[left]   at (0.6,.9)     {$\geq 1$};
  \node[right]  at (1.4,0.9)    {$\geq 1$};
\end{tikzpicture}}

\def\picMiddleBoundxyOne[#1]{
\begin{tikzpicture}[line width=1pt,auto,scale=#1]
  \draw [-](0,2) to [out=0,in=90]    (1,1);
  \draw [-](0,0) to [out=0,in=270]    (1,1);
  \draw [dotted](0,2) to     (0,0);
  \fill (0,0) circle (3pt);
  \fill (0,2) circle (3pt);
  \node[left]   at(0,2.1)      {$u=0$};
  \node[left]  at (0,-0.2)     {$x=v$};
  \node[right]  at (1,1)     {$\geq 3$};
  \node[left]  at (-0.3,1.3)     {$= 1$};
\end{tikzpicture}}

\def\picMiddleBoundxyTwo[#1]{
\begin{tikzpicture}[line width=1pt,auto,scale=#1]
  \draw [-] (0,1.5) to (2,1.5);
  \draw [line width=1pt,dotted] (0,0) --(0,1.5);
  \draw [-] (2,0) --(2,1.5);
  \draw [-] (0,0) --(2,0);

  \fill (0,0) circle (3pt);
  \fill (0,1.5) circle (3pt);
  \fill (2,0) circle (3pt);
  \fill (2,1.5) circle (3pt);
  \node[left]   at(0,1.5)      {$0$};
  \node[right]  at (2,0)     {$w$};
  \node[left]   at (0,0)   {$v$};
  \node[right]   at (2,1.5)   {$x$};
  \node[left]   at (0,.725)  {$=1$};
  \node[below]  at (1,0)     {$\geq 0$};
  \node[above]  at (1,1.5)     {$=1$};
  \node[right]   at (2,.725)  {$\geq 1$};

\end{tikzpicture}}

\def\picMiddleBoundxyThree[#1]{
\begin{tikzpicture}[line width=1pt,auto,scale=#1]
  \draw [-] (0,1.5) to (2,1.5);
  \draw [line width=1pt,dotted] (0,0) --(0,1.5);
  \draw [-] (2,0) --(2,1.5);
  \draw [-] (0,0) --(2,0);

  \fill (0,0) circle (3pt);
  \fill (0,1.5) circle (3pt);
  \fill (2,0) circle (3pt);
  \fill (2,1.5) circle (3pt);
  \node[left]   at(0,1.5)      {$0$};
  \node[right]  at (2,0)     {$w$};
  \node[left]   at (0,0)   {$v$};
  \node[right]   at (2,1.5)   {$x\neq v$};
  \node[left]   at (0,.725)  {$=1$};
  \node[below]  at (1,0)     {$\geq 0$};
  \node[above]  at (1,1.5)     {$\geq 2$};
  \node[right]   at (2,.725)  {$\geq 0$};
  \end{tikzpicture}}

\def\picMiddleBoundyyOne[#1]{
\begin{tikzpicture}[line width=1pt,auto,scale=#1]
  \draw [-] (0,1.5) to (2,1.5);
  \draw [line width=1pt,dotted] (0,0) --(0,1.5);
  \draw [-] (2,0) --(2,1.5);
  \draw [-] (0,0) --(2,0);

  \fill (0,0) circle (3pt);
  \fill (0,1.5) circle (3pt);
  \fill (2,0) circle (3pt);
  \fill (2,1.5) circle (3pt);
  \node[left]   at(0,1.5)      {$0$};
  \node[right]  at (2,0)     {$w$};
  \node[left]   at (0,0)   {$v$};
  \node[right]   at (2,1.5)   {$x$};
  \node[left]   at (0,.725)  {$\geq 2$};
  \node[below]  at (1,0)     {$\geq 0$};
  \node[above]  at (1,1.5)     {$=1$};
  \node[right]   at (2,.725)  {$\geq 1$};
\end{tikzpicture}}

\def\picMiddleBoundyyTwo[#1]{
\begin{tikzpicture}[line width=1pt,auto,scale=#1]
  \draw [-] (0,1.5) to (2,1.5);
  \draw [line width=1pt,dotted] (0,0) --(0,1.5);
  \draw [-] (2,0) --(2,1.5);
  \draw [-] (0,0) --(2,0);

  \fill (0,0) circle (3pt);
  \fill (0,1.5) circle (3pt);
  \fill (2,0) circle (3pt);
  \fill (2,1.5) circle (3pt);
  \node[left]   at(0,1.5)      {$0$};
  \node[right]  at (2,0)     {$w$};
  \node[left]   at (0,0)   {$v$};
  \node[right]   at (2,1.5)   {$x$};
  \node[left]   at (0,.725)  {$\geq 2$};
  \node[below]  at (1,0)     {$\geq 0$};
  \node[above]  at (1,1.5)     {$\geq 2$};
  \node[right]   at (2,.725)  {$\geq 0$};
  \end{tikzpicture}}

\def\picMiddleBarBoundNxxOne[#1]{
\begin{tikzpicture}[line width=1pt,auto,scale=#1]

  \draw [-](0,0) to [out=90,in=90]    (3,0);
  \draw [dotted](0,0) to [out=270,in=270]    (3,0);
  \fill (0,0) circle (3pt);
  \fill (3,0) circle (3pt);
  \node[left]   at(0,0)      {$0$};
  \node[right]  at (3,0)     {$y$};
  \node[above]  at (1,0.7)     {$\geq 3$};
  \node[below]  at (2,-0.7)     {$= 1$};

  \end{tikzpicture}}

\def\picMiddleBarBoundNxxTwo[#1]{
\begin{tikzpicture}[line width=1pt,auto,scale=#1]
  \draw [-] (0,1.5) to (2,1.5);
  \draw [-] (0,0) --(0,1.5);
  \draw [dotted] (2,0) --(2,1.5);
  \draw [-] (0,0) --(2,0);

  \fill (0,0) circle (3pt);
  \fill (0,1.5) circle (3pt);
  \fill (2,0) circle (3pt);
  \fill (2,1.5) circle (3pt);
  \node[left]   at(0,1.5)      {$0\neq x$};
  \node[right]  at (2,0)     {$x$};
  \node[left]   at (0,0)   {$w$};
  \node[right]   at (2,1.5)   {$y$};
  \node[left]   at (0,.725)  {$\geq 0$};
  \node[below]  at (1,0)     {$\geq 0$};
  \node[above]  at (1,1.5)     {$=1$};
  \node[right]   at (2,.725)  {$=1$};
  \end{tikzpicture}}

\def\picMiddleBarBoundNxxThree[#1]{
\begin{tikzpicture}[line width=1pt,auto,scale=#1]
  \draw [-] (0,1.5) to (2,1.5);
  \draw [-] (0,0) --(0,1.5);
  \draw [dotted] (2,0) --(2,1.5);
  \draw [-] (0,0) --(2,0);

  \fill (0,0) circle (3pt);
  \fill (0,1.5) circle (3pt);
  \fill (2,0) circle (3pt);
  \fill (2,1.5) circle (3pt);
  \node[left]   at(0,1.5)      {$0\neq x$};
  \node[right]  at (2,0)     {$x$};
  \node[left]   at (0,0)   {$w$};
  \node[right]   at (2,1.5)   {$y$};
  \node[left]   at (0,.725)  {$\geq 0$};
  \node[below]  at (1,0)     {$\geq 0$};
  \node[above]  at (1,1.5)     {$\geq 2$};
  \node[right]   at (2,.725)  {$=1$};
  \end{tikzpicture}}

\def\decomposePi4Tree{
\begin{tikzpicture}[line width=1pt,auto,scale=0.6]
  \fill (0,0) circle (3pt);
  \fill (0,2) circle (3pt);
  \fill (3,0) circle (3pt);
  \fill (3,2) circle (3pt);
  \fill (4.5,3) circle (3pt);
  \fill (6,0) circle (3pt);
  \fill (6,2) circle (3pt);
  \fill (7.5,3) circle (3pt);
  \fill (9,0) circle (3pt);
  \fill (9,2) circle (3pt);
  \fill (12,0) circle (3pt);
  \fill (12,2) circle (3pt);

 \draw [line width=3pt,-](0,0) to  (3,0);
 \draw [line width=3pt,-](3,0) to  (6,0);
 \draw [line width=3pt,-](6,0) to  (9,0);
 \draw [line width=3pt,-](9,0) to  (12,0);
 \draw [-](0,0) to  (0,2);
 \draw [-](3,0) to  (3,2);
 \draw [-](6,0) to  (6,2);
 \draw [-](9,0) to  (9,2);
 \draw [-](12,0) to  (12,2);
 \draw [-](0,2) to  (3,2);
 \draw [-](3,2) to  (4.5,3);
 \draw [-](6,2) to  (4.5,3);
 \draw [-](6,2) to  (7.5,3);
 \draw [-](9,2) to  (7.5,3);
 \draw [-](9,2) to  (12,2);
  \node   at(-2,1)      {$\sum_x\sum_{v_i}$};
  \node[below]   at(0,0)      {$0$};
  \node[above]   at(0,2)      {$v_1$};
  \node[above]   at(3,2)      {$v_2$};
  \node[below]   at(3,0)      {$v_3$};
  \node[above]   at(4.5,3)    {$v_4$};
  \node[above]   at(6,2)      {$v_5$};
  \node[below]   at(6,0)      {$v_6$};
  \node[above]   at(7.5,3)    {$v_7$};
  \node[above]   at(9,2)      {$v_8$};
  \node[below]   at(9,0)      {$v_9$};
  \node[above]   at(12,2)      {$v_{10}$};
  \node[below]   at(12,0)      {$x$};
\end{tikzpicture} \\
$\leq$
\begin{tabular}{ c c c c}

\begin{tikzpicture}[line width=1pt,auto,scale=0.5]
  \fill (0,0) circle (3pt);
  \fill (0,2) circle (3pt);
  \fill (3,0) circle (3pt);
  \fill (3,2) circle (3pt);

 \draw [line width=3pt,-](0,0) to  (3,0);
 \draw [-](0,0) to  (0,2);
 \draw [-](3,0) to  (3,2);
 \draw [-](0,2) to  (3,2);

  \node   at(-1.5,1)      {$\sum_{v_1,v_2,v_3}$};
  \node[below]   at(0,0)      {$0$};
  \node[above]   at(0,2)      {$v_1$};
  \node[above]   at(3,2)      {$v_2$};
  \node[below]   at(3,0)      {$v_3$};
\end{tikzpicture}
&
\begin{tikzpicture}[line width=1pt,auto,scale=0.5]
  \fill (3,0) circle (3pt);
  \fill (3,2) circle (3pt);
  \fill (4.5,3) circle (3pt);
  \fill (6,0) circle (3pt);
  \fill (6,2) circle (3pt);

 \draw [line width=3pt,-](3,0) to  (6,0);
 \draw [-](6,0) to  (6,2);
 \draw [-](3,2) to  (4.5,3);
 \draw [-](6,2) to  (4.5,3);
  \node   at(1,1)      {$\sup_{v_2,v_3} \sum_{v_4,v_5,v_{6}}$};
  \node[above]   at(3,2)      {$v_2$};
  \node[below]   at(3,0)      {$v_3$};
  \node[above]   at(4.5,3)    {$v_4$};
  \node[above]   at(6,2)      {$v_5$};
  \node[below]   at(6,0)      {$v_6$};
\end{tikzpicture}

&
\begin{tikzpicture}[line width=1pt,auto,scale=0.5]
  \fill (6,0) circle (3pt);
  \fill (6,2) circle (3pt);
  \fill (7.5,3) circle (3pt);
  \fill (9,0) circle (3pt);
  \fill (9,2) circle (3pt);

 \draw [line width=3pt,-](6,0) to  (9,0);

 \draw [-](6,2) to  (7.5,3);
 \draw [-](9,2) to  (7.5,3);
  \draw [-](9,0) to  (9,2);
  \node   at(4,1)      {$\sup_{v_5,v_6}\sum_{v_7,v_8,v_9}$};
  \node[above]   at(6,2)      {$v_5$};
  \node[below]   at(6,0)      {$v_6$};
  \node[above]   at(7.5,3)    {$v_7$};
  \node[above]   at(9,2)      {$v_8$};
  \node[below]   at(9,0)      {$v_9$};
\end{tikzpicture}
&

\begin{tikzpicture}[line width=1pt,auto,scale=0.5]
  \fill (9,0) circle (3pt);
  \fill (9,2) circle (3pt);
  \fill (12,0) circle (3pt);
  \fill (12,2) circle (3pt);

 \draw [line width=3pt,-](9,0) to  (12,0);
 \draw [-](12,0) to  (12,2);
 \draw [-](9,2) to  (12,2);
  \node   at(7,1)      {$\sup_{v_8,v_9}\sum_{v_{10},x}$};
  \node[above]   at(9,2)      {$v_8$};
  \node[below]   at(9,0)      {$v_9$};
  \node[above]   at(12,2)      {$v_{10}$};
  \node[below]   at(12,0)      {$x$};
\end{tikzpicture}
\end{tabular}
$\leq$ ClosedSquare (OpenSquare$)^2$ OpenTriangle
}

\def\picPOne[#1]{
\begin{tikzpicture}[line width=1pt,auto,scale=#1]
  \node   at(-3,1)      {$P^{\ssc[1],b}(x,y):=\sum_v$};
 \draw [-](0,0) to  (0,2);
 \draw [line width=2pt,-](0,0) to  (2,0);
 \draw [-](0,2) to    (2,2);
 \draw [line width=2pt,-](2,0) to    (2,2);
  \fill (0,2) circle (2pt);
  \fill (2,0) circle (2pt);
  \fill (2,2) circle (2pt);
  \fill (0,0) circle (2pt);
  \node[left]   at(0,0)      {$0$};
  \node[left]   at(0,2)      {$v$};
  \node[right]  at (2,0)     {$y$};
  \node[right]  at (2,2)     {$x$};
  \node[right]  at (2,1)     {$b$};
\end{tikzpicture}}

\def\picPOneLA[#1]{
\begin{tikzpicture}[line width=1pt,auto,scale=#1]
  \node   at(-3-0.2,)      { $P^{\ssc[1],m}(x,y):=\sum_w$};
 \begin{scope}
 \draw [-](0,0) to  (0,2);
 \draw [line width=2pt,-](0,0) to  (2,0);
 \draw [-](0,2) to    (2,2);
 \draw [line width=2pt,-](2,0) to    (2,2);
  \fill (0,2) circle (2pt);
  \fill (2,0) circle (2pt);
  \fill (2,2) circle (2pt);
  \fill (0,0) circle (2pt);
  \node[left]   at(0,0)      {$0$};
  \node[left]   at(0,2)      {$w$};
  \node[right]  at (2,0)     {$y$};
  \node[right]  at (2,2)     {$x$};
  \node[right]  at (2,1)     {$m$};
  \end{scope}

  \begin{scope}[shift={(7.5,0)}]
 \node   at(-3,1-.1)      {{\Large $+\sum_{w,u,v}$}};
 \draw [-](-0.8,1) to  (0,2);
 \draw [-](-0.8,1) to  (0,0);
 \draw [-](0,0) to  (0,2);
 \draw [line width=2pt,-](0,0) to  (2,0);
 \draw [-](0,2) to    (1,2.5);
 \draw [-](1,2.5) to    (2,2);
 \draw [line width=2pt,-](2,0) to    (2,2);
  \fill (0,2) circle (2pt);
  \fill (2,0) circle (2pt);
  \fill (2,2) circle (2pt);
  \fill (0,0) circle (2pt);
  \fill (-0.8,1) circle (2pt);
  \fill (1,2.5) circle (2pt);

  \node[left]   at(-0.8,1)      {$0$};
  \node[below]   at (0,0)      {$u$};
  \node[above]   at(0,2)      {$v$};
  \node[below]   at(1,2.4)      {$w$};
  \node[right]  at (2,0)     {$y$};
  \node[right]  at (2,2)     {$x$};
  \node[right]  at (2,1)     {$m$};
  \end{scope}

\end{tikzpicture}}

\def\picPOneRib[#1]{
\begin{tikzpicture}[line width=1pt,auto,scale=#1]
  \node   at(-3,1)      {$P^{\ssc[1],-m}(x,y):=\sum_v$};
  \draw [-](0,0) to  (0,2);
 \draw [line width=2pt,-](0,0) to  (2,0);
 \draw [-](0,2) to    (2,2);
 \draw [-](2,0) to    (2,2);
  \fill (0,2) circle (2pt);
  \fill (2,0) circle (2pt);
  \fill (2,2) circle (2pt);
  \fill (0,0) circle (2pt);
  \node[left]   at(0,0)      {$0$};
  \node[left]   at(0,2)      {$v$};
  \node[right]  at (2,0)     {$x$};
  \node[right]  at (2,2)     {$y$};
  \node[right]  at (2,1)     {$m$};
\end{tikzpicture}}

\def\picPOneRibLA[#1]{
\begin{tikzpicture}[line width=1pt,auto,scale=#1]
  \node   at(-3-0.4,1)      {$P^{\ssc[1],-m}(x,y):=\sum_w$};
 \begin{scope}
 \draw [-](0,0) to  (0,2);
 \draw [line width=2pt,-](0,0) to  (2,0);
 \draw [-](0,2) to    (2,2);
 \draw [-](2,0) to    (2,2);
  \fill (0,2) circle (2pt);
  \fill (2,0) circle (2pt);
  \fill (2,2) circle (2pt);
  \fill (0,0) circle (2pt);
  \node[left]   at(0,0)      {$0$};
  \node[left]   at(0,2)      {$w$};
  \node[right]  at (2,0)     {$x$};
  \node[right]  at (2,2)     {$y$};
  \node[right]  at (2,1)     {$m$};
  \end{scope}

  \begin{scope}[shift={(7.5,0)}]
 \node   at(-3,1-.1)      {{\Large $+\sum_{w,u,v}$}};
 \draw [-](-0.8,1) to  (0,2);
 \draw [-](-0.8,1) to  (0,0);
 \draw [-](0,0) to  (0,2);
 \draw [line width=2pt,-](0,0) to  (2,0);
 \draw [-](0,2) to    (1,2.5);
 \draw [-](1,2.5) to    (2,2);
 \draw [-](2,0) to    (2,2);
  \fill (0,2) circle (2pt);
  \fill (2,0) circle (2pt);
  \fill (2,2) circle (2pt);
  \fill (0,0) circle (2pt);
  \fill (-0.8,1) circle (2pt);
  \fill (1,2.5) circle (2pt);

  \node[left]   at(-0.8,1)      {$0$};
  \node[below]   at (0,0)      {$u$};
  \node[above]   at(0,2)      {$v$};
  \node[below]   at(1,2.4)      {$w$};
  \node[right]  at (2,0)     {$x$};
  \node[right]  at (2,2)     {$y$};
  \node[right]  at (2,1)     {$m$};
  \end{scope}

\end{tikzpicture}}

\def\picPOneStep[#1]{
\begin{tikzpicture}[line width=1pt,auto,scale=#1]
  \node   at(-4,1)      {$P^{\ssc[1],\iota,\text{step},b}(x,y):=\sum_v$};
  \draw [-](0,0) to  (0,2);
  \draw [line width=2pt,-](0,0) to  (0.7,0);
  \draw [line width=2pt,-](0.7,0) to  (2,0);
  \draw [-](0,2) to    (2,2);
  \draw [line width=2pt,-](2,0) to    (2,2);
  \fill (0,2) circle (2pt);
  \fill (2,0) circle (2pt);
  \fill (2,2) circle (2pt);
  \fill (0,0) circle (2pt);
  \fill (0.7,0) circle (4pt);

  \node[left]   at(0,0)      {$0$};
  \node[below]   at(0.7,0)   {$\ve[\iota]$};
  \node[left]   at(0,2)      {$v$};
  \node[right]  at (2,0)     {$y$};
  \node[right]  at (2,2)     {$x$};
  \node[right]  at (2,1)     {$b$};
\end{tikzpicture}}

\def\picPOneStepRib[#1]{
\begin{tikzpicture}[line width=1pt,auto,scale=#1]
    \node   at(-4,1)      {$P^{\ssc[1],\iota,\text{step},-b}(x,y):=\sum_v$};
  \draw [-](0,0) to  (0,2);
  \draw [line width=2pt,-](0,0) to  (0.7,0);
  \draw [line width=2pt,-](0.7,0) to  (2,0);
  \draw [-](0,2) to    (2,2);
  \draw [-](2,0) to    (2,2);
  \fill (0,2) circle (2pt);
  \fill (0.7,0) circle (4pt);
  \fill (2,0) circle (2pt);
  \fill (2,2) circle (2pt);
  \fill (0,0) circle (2pt);
  \node[left]   at(0,0)      {$0$};
  \node[below]   at(0.7,0)   {$\ve[\iota]$};
  \node[left]   at(0,2)      {$v$};
  \node[right]  at (2,0)     {$x$};
  \node[right]  at (2,2)     {$y$};
  \node[right]  at (2,1)     {$b$};
\end{tikzpicture}}

\def\picPOneIota[#1]{
\begin{tikzpicture}[line width=1pt,auto,scale=#1]
  \node   at(-4,1)      {$P^{\ssc[1],\iota,\text{rib},b}(x,y):=\sum_{u,v}$};
 \draw [-](0,0) to  (0.25,1.2);
 \draw [-](0.25,1.2) to  (1.4,2);
 \draw [-](1.4,2) to  (2.5,2);
 \draw [line width=2pt,-](0,0) to    (2.5,0);
 \draw [line width=2pt,-](2.5,0) to    (2.5,2);
 \draw [-](0.25,1.2) to  [out=135,in=90]   (-0.8,0);
  \fill (0,0) circle (2pt);
  \fill (0.25,1.2) circle (2pt);
  \fill (1.4,2) circle (2pt);
  \fill (2.5,2) circle (2pt);
  \fill (2.5,0) circle (2pt);
  \fill (-0.8,0) circle (2pt);
  \node[below]   at(-0.8,0)      {$\ve[\iota]$};
  \node[below]   at(0,0)      {$0$};
  \node[above]   at(0.25,1.2)      {$u$};
  \node[left]   at(1.5,1.5)      {$v$};
  \node[right]  at (2.5,0)     {$y$};
  \node[right]  at (2.5,2)     {$x$};
  \node[right]  at (2.5,1)     {$b$};
  \end{tikzpicture}}

\def\picPOneIotaRib[#1]{
\begin{tikzpicture}[line width=1pt,auto,scale=#1]
  \node   at(-4,1)      {$P^{\ssc[1],\iota,\text{rib},-b}(x,y):=\sum_{u,v}$};
 \draw [-](0,0) to  (0.25,1.2);
 \draw [-](0.25,1.2) to  (1.4,2);
 \draw [-](1.4,2) to  (2.5,2);
 \draw [line width=2pt,-](0,0) to    (2.5,0);
 \draw [-](2.5,0) to    (2.5,2);
 \draw [-](0.25,1.2) to  [out=135,in=90]   (-0.8,0);
  \fill (0,0) circle (2pt);
  \fill (0.25,1.2) circle (2pt);
  \fill (1.4,2) circle (2pt);
  \fill (2.5,2) circle (2pt);
  \fill (2.5,0) circle (2pt);
  \fill (-0.8,0) circle (2pt);
  \node[below]   at(-0.8,0)      {$\ve[\iota]$};
  \node[below]   at(0,0)      {$0$};
  \node[above]   at(0.25,1.2)      {$u$};
  \node[left]   at(1.5,1.5)      {$v$};
  \node[right]  at (2.5,0)     {$x$};
  \node[right]  at (2.5,2)     {$y$};
  \node[right]  at (2.5,1)     {$b$};
\end{tikzpicture}}

\def\picAGeneral[#1]{
\begin{tikzpicture}[line width=1pt,auto,scale=#1]
  \node   at(-3.2,0.25)      {$A^{l,m}(u,v,x,y):=\sum_{w}$};

  \draw [-](1,4) to  (1,-3);
  \draw [-](5,4) to  (5,-3);
  \draw [-](0,3) to  (9,3);
  \draw [-](0,0) to  (9,0);
  \node at(0,1.5)      {$l\leq 0$};
  \node at(0,-1.5)      {$l> 0$};

  \node at(3,3.5)      {$m\leq 0$};
  \node at(7,3.5)      {$m> 0$};

  \begin{scope}[shift={(2,0.5)},rotate=0]
 \draw [line width=2pt,-](0,0) to  (2,0);
 \draw [-](2,0) to  (2,1.5);
 \draw [dotted](0,0) to    (0,1.5);
 \draw [-](0,1.5) to    (1,2);
 \draw [-](1,2) to    (2,1.5);

  \fill (0,1.5) circle (2pt);
  \fill (2,1.5) circle (2pt);
  \fill (2,0) circle (2pt);
  \fill (1,2) circle (2pt);
  \fill (0,0) circle (2pt);
  \node[left]   at(0,0)      {$u$};
  \node[left]   at(0,0.75)      {$l$};
  \node[left]   at(0,1.5)      {$v$};
  \node[right]  at (2,0)     {$x$};
  \node[right]  at (2,0.75)     {$m$};
  \node[right]  at (2,1.5)     {$y$};
  \node[below]  at (1,1.8)   {$w$};
  \end{scope}

    \begin{scope}[shift={(2,-2)},rotate=0]
 \draw [line width=2pt,-](0,1.5) to  (2,1.5);
 \draw [line width=2pt,-](2,0) to  (2,1.5);
 \draw [dotted](0,0) to    (0,1.5);
 \draw [-](0,0) to    (1,-.5);
 \draw [-](1,-.5) to    (2,0);

  \fill (0,1.5) circle (2pt);
  \fill (2,1.5) circle (2pt);
  \fill (2,0) circle (2pt);
  \fill (1,-.5) circle (2pt);
  \fill (0,0) circle (2pt);
  \node[left]   at(0,0)      {$v$};
  \node[left]   at(0,0.75)      {$l$};
  \node[left]   at(0,1.5)      {$u$};
  \node[right]  at (2,0)     {$x$};
  \node[right]  at (2,0.75)     {$m$};
  \node[right]  at (2,1.5)     {$y$};
  \node[above]  at (1,-.5)   {$w$};
  \end{scope}

  \begin{scope}[shift={(6,0.5)},rotate=0]
 \draw [line width=2pt,-](0,0) to  (2,0);
 \draw [line width=2pt,-](2,0) to  (2,1.5);
 \draw [dotted](0,0) to    (0,1.5);
 \draw [-](0,1.5) to    (1,2);
 \draw [-](1,2) to    (2,1.5);

  \fill (0,1.5) circle (2pt);
  \fill (2,1.5) circle (2pt);
  \fill (2,0) circle (2pt);
  \fill (1,2) circle (2pt);
  \fill (0,0) circle (2pt);
  \node[left]   at(0,0)      {$u$};
  \node[left]   at(0,0.75)      {$l$};
  \node[left]   at(0,1.5)      {$v$};
  \node[right]  at (2,0)     {$y$};
  \node[right]  at (2,0.75)     {$b$};
  \node[right]  at (2,1.5)     {$x$};
  \node[below]  at (1,1.8)   {$w$};
  \end{scope}

 \begin{scope}[shift={(6,-2)},rotate=0]
 \draw [line width=2pt,-](0,1.5) to  (2,1.5);
 \draw [-](2,0) to  (2,1.5);
 \draw [dotted](0,0) to    (0,1.5);
 \draw [-](0,0) to    (1,-.5);
 \draw [-](1,-.5) to    (2,0);

  \fill (0,1.5) circle (2pt);
  \fill (2,1.5) circle (2pt);
  \fill (2,0) circle (2pt);
  \fill (1,-.5) circle (2pt);
  \fill (0,0) circle (2pt);
  \node[left]   at(0,0)      {$v$};
  \node[left]   at(0,0.75)      {$l$};
  \node[left]   at(0,1.5)      {$u$};
  \node[right]  at (2,0)     {$y$};
  \node[right]  at (2,0.75)     {$m$};
  \node[right]  at (2,1.5)     {$x$};
  \node[above]  at (1,-.5)   {$w$};
  \end{scope}

 \end{tikzpicture}}

\def\picAbarGeneral[#1]{
\begin{tikzpicture}[line width=1pt,auto,scale=#1]
  \node   at(-3.5,0.75)      {$\bar A^{l,m}(u,v,x,y):=\sum_{w}$};
 \draw [line width=2pt,-](0,1.5) to  (2,1.5);
 \draw [dotted](2,0) to  (2,1.5);
 \draw [dotted](0,0) to    (0,1.5);
 \draw [-](0,0) to    (1,-.5);
 \draw [-](1,-.5) to    (2,0);

  \fill (0,1.5) circle (2pt);
  \fill (2,1.5) circle (2pt);
  \fill (2,0) circle (2pt);
  \fill (1,-.5) circle (2pt);
  \fill (0,0) circle (2pt);
  \node[left]   at(0,0)      {$v$};
  \node[left]   at(0,0.75)      {$l$};
  \node[left]   at(0,1.5)      {$u$};
  \node[right]  at (2,0)     {$y$};
  \node[right]  at (2,0.75)     {$m$};
  \node[right]  at (2,1.5)     {$x$};
  \node[above]  at (1,-.5)   {$w$};
\end{tikzpicture}}

\def\picDGeneral[#1]{
\begin{tikzpicture}[line width=1pt,auto,scale=#1]
  \node   at(-3.2,0.25)      {$C^{l,m}(u,v,x,y):=\sum_{w}$};

  \draw [-](1,4) to  (1,-3);
  \draw [-](5,4) to  (5,-3);
  \draw [-](0,3) to  (9,3);
  \draw [-](0,-1) to  (9,-1);
  \node at(0,1.5)      {$l\leq 0$};
  \node at(0,-1.5)      {$l> 0$};

  \node at(3,3.5)      {$m\leq 0$};
  \node at(7,3.5)      {$m> 0$};

  \begin{scope}[shift={(2,0.5)},rotate=0]
 \draw [line width=2pt,-](0,0) to  (2,0);
 \draw [dotted](2,0) to  (2,1.5);
 \draw [dotted](0,0) to    (0,1.5);
 \draw [-](0,1.5) to    (1,2);
 \draw [-](1,2) to    (2,1.5);

  \fill (0,1.5) circle (2pt);
  \fill (2,1.5) circle (2pt);
  \fill (2,0) circle (2pt);
  \fill (1,2) circle (2pt);
  \fill (0,0) circle (2pt);
  \node[left]   at(0,0)      {$u$};
  \node[left]   at(0,0.75)      {$l$};
  \node[left]   at(0,1.5)      {$v$};
  \node[right]  at (2,0)     {$x$};
  \node[right]  at (2,0.75)     {$m$};
  \node[right]  at (2,1.5)     {$y$};
  \node[below]  at (1,1.8)   {$w$};
  \node  at (1,-0.4)     {$|x-u|^2$};
  \end{scope}

    \begin{scope}[shift={(2,-3)},rotate=0]
 \draw [line width=2pt,-](0,1.5) to  (2,1.5);
 \draw [line width=2pt,dotted](2,0) to  (2,1.5);
 \draw [dotted](0,0) to    (0,1.5);
 \draw [-](0,0) to    (1,-.5);
 \draw [-](1,-.5) to    (2,0);

  \fill (0,1.5) circle (2pt);
  \fill (2,1.5) circle (2pt);
  \fill (2,0) circle (2pt);
  \fill (1,-.5) circle (2pt);
  \fill (0,0) circle (2pt);
  \node[left]   at(0,0)      {$v$};
  \node[left]   at(0,0.75)      {$l$};
  \node[left]   at(0,1.5)      {$u$};
  \node[right]  at (2,0)     {$x$};
  \node[right]  at (2,0.75)     {$m$};
  \node[right]  at (2,1.5)     {$y$};
  \node[above]  at (1,-.5)   {$w$};
  \node[below]  at (1,-0.2)     {$|x-v|^2$};
  \end{scope}

  \begin{scope}[shift={(6,0.5)},rotate=0]
 \draw [line width=2pt,-](0,0) to  (2,0);
 \draw [line width=2pt,dotted](2,0) to  (2,1.5);
 \draw [dotted](0,0) to    (0,1.5);
 \draw [-](0,1.5) to    (1,2);
 \draw [-](1,2) to    (2,1.5);

  \fill (0,1.5) circle (2pt);
  \fill (2,1.5) circle (2pt);
  \fill (2,0) circle (2pt);
  \fill (1,2) circle (2pt);
  \fill (0,0) circle (2pt);
  \node[left]   at(0,0)      {$u$};
  \node[left]   at(0,0.75)      {$l$};
  \node[left]   at(0,1.5)      {$v$};
  \node[right]  at (2,0)     {$y$};
  \node[right]  at (2,0.75)     {$m$};
  \node[right]  at (2,1.5)     {$x$};
  \node[below]  at (1,1.8)   {$w$};
    \node  at (1,-0.4)     {$|y-u|^2$};
  \end{scope}

 \begin{scope}[shift={(6,-3)},rotate=0]
 \draw [line width=2pt,-](0,1.5) to  (2,1.5);
 \draw [dotted](2,0) to  (2,1.5);
 \draw [dotted](0,0) to    (0,1.5);
 \draw [-](0,0) to    (1,-.5);
 \draw [-](1,-.5) to    (2,0);

  \fill (0,1.5) circle (2pt);
  \fill (2,1.5) circle (2pt);
  \fill (2,0) circle (2pt);
  \fill (1,-.5) circle (2pt);
  \fill (0,0) circle (2pt);
  \node[left]   at(0,0)      {$v$};
  \node[left]   at(0,0.75)      {$l$};
  \node[left]   at(0,1.5)      {$u$};
  \node[right]  at (2,0)     {$y$};
  \node[right]  at (2,0.75)     {$m$};
  \node[right]  at (2,1.5)     {$x$};
  \node[above]  at (1,-.5)   {$w$};
  \node[below]  at (1,-0.2)     {$|y-v|^2$};
  \end{scope}

 \end{tikzpicture}}

\def\picDeltaMentryGeneral[#1]{
\begin{tikzpicture}[line width=1pt,auto,scale=#1]
  \node   at(-4,0.75)      {$\Delta M^{-a,b}(u,v,x,y):=\sum_{w}$};
 \draw [line width=2pt,-](0,0) to  (2,0);
 \node[below]  at (1,0)     {$[1-\cos(k\cdot (y-u))]$};
 \draw [dotted](2,0) to  (2,1.5);
 \draw [dotted](0,0) to    (0,1.5);
 \draw [-](0,1.5) to    (1,2);
 \draw [-](1,2) to    (2,1.5);

  \fill (0,1.5) circle (2pt);
  \fill (1,2) circle (2pt);
  \fill (2,0) circle (2pt);
  \fill (2,1.5) circle (2pt);
  \fill (0,0) circle (2pt);
  \node[left]   at(0,0)      {$u$};
  \node[left]   at(0,0.75)      {$a$};
  \node[left]   at(0,1.5)      {$v$};
  \node[right]  at (2,0)     {$y$};
  \node[right]  at (2,0.75)     {$b$};
  \node[right]  at (2,1.5)     {$x$};
  \node[above]  at (1,2)   {$w$};

\begin{scope}[shift={(0,-3)}]
 \node   at(-4,0.75)      {$\Delta M^{a,b}(u,v,x,y):=\sum_{w}$};
 \draw [line width=2pt,-](0,0) to  (1,-0.5);
  \draw [line width=2pt,-](1,-0.5) to  (2,0);
 \node[below]  at (1,-.5)     {$[1-\cos(k\cdot (y-v))]$};
 \draw [dotted](2,0) to  (2,1.5);
 \draw [dotted](0,0) to    (0,1.5);
 \draw [-](0,1.5) to    (2,1.5);
  \fill (0,1.5) circle (2pt);
  \fill (1,-0.5) circle (2pt);
  \fill (2,0) circle (2pt);
  \fill (2,1.5) circle (2pt);
  \fill (0,0) circle (2pt);
  \node[left]   at(0,0)      {$v$};
  \node[left]   at(0,0.75)      {$a$};
  \node[left]   at(0,1.5)      {$u$};
  \node[right]  at (2,0)     {$y$};
  \node[right]  at (2,0.75)     {$b$};
  \node[right]  at (2,1.5)     {$x$};
  \node[right]  at (1.1,-.5)   {$w$};
  \end{scope}

\end{tikzpicture}}

\def\picHGeneral[#1]{
\begin{tikzpicture}[line width=1pt,auto,scale=#1]
  \draw [-](5,4) to  (5,-0.5);
  \draw [-](1.5,3) to  (13.5,3);
  \node at(3,3.5)      {$a\leq 0$};
  \node at(9,3.5)      {$a> 0$};

  \begin{scope}[shift={(2,0.5)},rotate=0]
 \draw [line width=2pt,-](0,0) to  (2,0);
 \draw [dotted](2,0) to  (2,1.5);
 \draw [dotted](0,0) to    (0,1.5);
 \draw [-](0,1.5) to    (1,2);
 \draw [-](1,2) to    (2,1.5);

  \fill (0,1.5) circle (2pt);
  \fill (2,1.5) circle (2pt);
  \fill (2,0) circle (2pt);
  \fill (1,2) circle (2pt);
  \fill (0,0) circle (2pt);
  \node[left]   at(0,0)      {$u$};
  \node[left]   at(0,0.75)      {$a$};
  \node[left]   at(0,1.5)      {$v$};
  \node[right]  at (2,0)     {$x$};
  \node[right]  at (2,0.75)     {$b$};
  \node[right]  at (2,1.5)     {$y$};
  \node[below]  at (1,1.8)   {$w$};
{\small    \node[below]  at (1,-0.5)   {$[1-\cos(k\cdot (u-x))]$};}
  \end{scope}

 \begin{scope}[shift={(6.5,0.75)},rotate=0]
 \draw [-](0,1.5) to  (2,1.5);
 \draw [dotted](2,0) to  (2,1.5);
 \draw [dotted](0,0) to    (0,1.5);
 \draw [line width=2pt,-](0,0) to    (1,-.5);
 \draw [-](1,-.5) to    (2,0);

  \fill (0,1.5) circle (2pt);
  \fill (2,1.5) circle (2pt);
  \fill (2,0) circle (2pt);
  \fill (1,-.5) circle (2pt);
  \fill (0,0) circle (2pt);
  \node[left]   at(0,0)      {$v$};
  \node[left]   at(0,0.75)      {$a$};
  \node[left]   at(0,1.5)      {$u$};
  \node[right]  at (2,0)     {$x$};
  \node[right]  at (2,0.75)     {$b$};
  \node[right]  at (2,1.5)     {$y$};
  \node[above]  at (1,-.5)   {$w$};
{\small  \node[below]  at (0.7,-0.5)   {$[1-\cos(k\cdot (v-w))]$};}

   \draw [-](3,0.75+0.25) to  (3,0.75-0.25);
   \draw [-](3+0.25,0.75) to  (3-0.25,0.75);
  \end{scope}

 \begin{scope}[shift={(10.5,0.75)},rotate=0]
 \draw [-](0,1.5) to  (2,1.5);
 \draw [dotted](2,0) to  (2,1.5);
 \draw [dotted](0,0) to    (0,1.5);
 \draw [-](0,0) to    (1,-.5);
 \draw [line width=2pt,-](1,-.5) to    (2,0);

  \fill (0,1.5) circle (2pt);
  \fill (2,1.5) circle (2pt);
  \fill (2,0) circle (2pt);
  \fill (1,-.5) circle (2pt);
  \fill (0,0) circle (2pt);
  \node[left]   at(0,0)      {$v$};
  \node[left]   at(0,0.75)      {$a$};
  \node[left]   at(0,1.5)      {$u$};
  \node[right]  at (2,0)     {$x$};
  \node[right]  at (2,0.75)     {$b$};
  \node[right]  at (2,1.5)     {$y$};
  \node[above]  at (1,-.5)   {$w$};
{\small  \node[below]  at (1.4,-0.5)   {$[1-\cos(k\cdot (w-y))]$};}
  \end{scope}

 \end{tikzpicture}}

\def\picRootContreeAndFirstIntersection[#1]{
\begin{tikzpicture}[auto,scale=#1]

\begin{scope}[shift={(-1+0.07,0.07)},rotate=0]
\draw[red,line width=1 pt] (3,-1) to (3,3);
\draw[red,line width=1 pt] (3,-1) to (4,-1);
\draw[red,line width=1 pt] (3,0) to (4,0);
\draw[red,line width=1 pt] (3,1) to (4,1);
\draw[red,line width=1 pt] (4,2) to (4,1);
\draw[red,line width=1 pt] (2,2) to (4,2);
\draw[red,line width=1 pt] (2,2) to (2,3)   ;
\node[red] at(4,3)   {$B_{t_N}^\omega$};
\fill[red] (3,0) circle (3pt);
\node[red] at(2.5,-0.2)   {$\overline b^\omega_{t_N}$};
\fill[red] (2,2) circle (2pt);
\fill[red] (3,1) circle (2pt);
\end{scope}

\begin{scope}[shift={(-1,0)},rotate=0]
\fill (0,0) circle (3pt);
\draw[line width=1 pt] (0,0) to (0,2);
\draw[line width=1 pt] (-1,1) to (3,1);
\draw[line width=1 pt] (-1,2) to (0,2);
\draw[line width=1 pt] (1,1) to (1,2);
\draw[line width=1 pt] (1,2) to (4,2);
\draw[line width=1 pt] (2,2) to (2,3);
\draw[line width=1 pt] (1,3) to (2,3);
\node at(0,3)   {$B_{s_N}^\omega$};
\node at(-0.3,-0.2)   {$\overline b^\omega_{s_N}$};
\fill (2,2) circle (2pt);
\fill (3,1) circle (2pt);
\end{scope}
\begin{scope}[shift={(-1.3,-0.25)},rotate=0]
\node at(2,2)   {$w_a$};
\node at(3,1)   {$w_b$};
\end{scope}

\begin{scope}[shift={(7.07,0.07)},rotate=0]
\draw[red,line width=1 pt] (4,1) to (1,1);
\draw[red,line width=1 pt] (3,0) to (3,2);
\draw[red,line width=1 pt] (2,0) to (2,3);
\draw[red,line width=1 pt] (1,1) to (1,2);
\draw[red,line width=1 pt] (1,2) to (-1,2);
\draw[red,line width=1 pt] (0,2) to (0,3);
\node[red] at(3,3)   {$B_{t_N}^\omega$};
\fill[red] (4,1) circle (3pt);
\node[red] at(4.5,1.2)   {$\overline b^\omega_{t_N}$};
\fill[red] (2,0) circle (2pt);
\fill[red] (1,1) circle (2pt);
\fill[red] (0,2) circle (2pt);
\end{scope}

\begin{scope}[shift={(8,-1)},rotate=0]
\fill (0,0) circle (3pt);
\draw[line width=1 pt] (0,0) to (0,1);
\draw[line width=1 pt] (0,1) to (1,1);
\draw[line width=1 pt] (0,1) to (-1,1);
\draw[line width=1 pt] (-1,1) to (-1,3);
\draw[line width=1 pt] (0,1) to (0,2);
\draw[line width=1 pt] (0,2) to (1,2);
\draw[line width=1 pt] (0,2) to (0,3);
\node at(-2,1)   {$B_{s_N}^\omega$};
\node at(-0.3,-0.2)   {$\overline b^\omega_{s_N}$};
\fill (1,1) circle (2pt);
\fill (0,2) circle (2pt);
\fill (-1,3) circle (2pt);
\end{scope}
\begin{scope}[shift={(8-0.2,-1-0.25)},rotate=0]
\node at(1.6,1)   {$w_a$};
\node at(-0.2,2)   {$w_b$};
\node at(-1.2,3)   {$w_c$};
\end{scope}

\end{tikzpicture}
}

\newcommand{\LTDispoffset}{0.15}

\def\picLtchoiceOfXXX[#1]{
\begin{tikzpicture}[auto,scale=#1]

\begin{scope}[shift={(0,0)},rotate=0]
\node at(-.65,.5)   {$0$};
\fill[gray] (0,0) -- (-0.5,0.5)--(0,1);
\draw[line width=1 pt] (-0.5,0.5) to (0,0);
\draw[line width=1 pt] (-0.5,0.5) to (0,1);
\draw[line width=1 pt] (0,0) to (4,0);
\draw[line width=1 pt] (0,1) to (4,1);
\draw[line width=1 pt] (0,0) to (0,1);
\draw[line width=1 pt] (4,0) to (4,1);
\draw[line width=1 pt] (1,0) to (1,1);
\draw[line width=1 pt] (2,0) to (2,1);
\draw[line width=1 pt] (3,0) to (3,1);
\draw[line width=2.5 pt] (0,0) to (4,0);
\draw[line width=1 pt] (3,1) to (4,1);
\fill (1.5,1) circle (3pt);
\fill (2.5,1) circle (3pt);
\node at(4.3,0)   {$x$};

\draw[red,line width=1 pt] (-0.46 ,0.5-0.15) to (0,-0.1);
\draw[red,line width=1 pt] (0,-0.1) to (1-\LTDispoffset,-0.1);
\draw[red,line width=1 pt] (1+\LTDispoffset,-0.1) to (2-\LTDispoffset,-0.1);
\draw[red,line width=1 pt] (2+\LTDispoffset,-0.1) to (3-\LTDispoffset,-0.1);
\draw[red,line width=1 pt] (3+\LTDispoffset,-0.1) to (4-\LTDispoffset,-0.1);

\node[red] at(0.25,-0.5)   {$x_1$};
\node[red] at(1.5,-0.5)   {$x_2$};
\node[red] at(2.5,-0.5)   {$x_3$};
\node[red] at(3.5,-0.5)   {$x_4$};
\end{scope}

\begin{scope}[shift={(6,0)},rotate=0]
\node at(-.65,.5)   {$0$};
\fill[gray] (0,0) -- (-0.5,0.5)--(0,1);
\draw[line width=1 pt] (-0.5,0.5) to (0,0);
\draw[line width=1 pt] (-0.5,0.5) to (0,1);
\draw[line width=1 pt] (0,0) to (4,0);
\draw[line width=1 pt] (0,1) to (4,1);
\draw[line width=1 pt] (0,0) to (0,1);
\draw[line width=1 pt] (4,0) to (4,1);
\draw[line width=1 pt] (1,0) to (1,1);
\draw[line width=1 pt] (2,0) to (2,1);
\draw[line width=2.5 pt] (3,0) to (3,1);
\draw[line width=2.5 pt] (0,0) to (3,0);
\draw[line width=2.5 pt] (3,1) to (4,1);
\fill (1.5,1) circle (3pt);
\fill (2.5,1) circle (3pt);
\node at(4.1,1.2)   {$x$};

\draw[red,line width=1 pt] (-0.46 ,0.5-0.15) to (0,-0.1);
\draw[red,line width=1 pt] (0,-0.1) to (1-\LTDispoffset,-0.1);
\draw[red,line width=1 pt] (1+\LTDispoffset,-0.1) to (2-\LTDispoffset,-0.1);
\draw[red,line width=1 pt] (2+\LTDispoffset,-0.1) to (3-\LTDispoffset,-0.1);
\draw[red,line width=1 pt] (3+\LTDispoffset,-0.1) to (4+0.1,-0.1);
\draw[red,line width=1 pt] (4+0.1,-0.1) to (4+0.1,1);

\node[red] at(0.25,-0.5)   {$x_1$};
\node[red] at(1.5,-0.5)   {$x_2$};
\node[red] at(2.5,-0.5)   {$x_3$};
\node[red] at(4,-0.5)      {$x_4$};
\end{scope}

\begin{scope}[shift={(0,-2.5)},rotate=0]
\node at(-.65,.5)   {$0$};
\fill[gray] (0,0) -- (-0.5,0.5)--(0,1);
\draw[line width=1 pt] (-0.5,0.5) to (0,0);
\draw[line width=1 pt] (-0.5,0.5) to (0,1);
\draw[line width=1 pt] (0,0) to (4,0);
\draw[line width=1 pt] (0,1) to (4,1);
\draw[line width=1 pt] (0,0) to (0,1);
\draw[line width=1 pt] (4,0) to (4,1);
\draw[line width=2.5 pt] (1,0) to (1,1);
\draw[line width=2.5 pt] (2,0) to (2,1);
\draw[line width=2.5 pt] (3,0) to (3,1);
\draw[line width=2.5 pt] (0,0) to (1,0);
\draw[line width=2.5 pt] (1,1) to (2,1);
\draw[line width=2.5 pt] (2,0) to (3,0);
\draw[line width=2.5 pt] (3,1) to (4,1);
\fill (1.5,0) circle (3pt);
\fill (2.5,1) circle (3pt);
\node at(4.1,1.2)   {$x$};

\draw[red,line width=1 pt] (-0.46 ,0.5-0.15) to (0,-0.1);
\draw[red,line width=1 pt] (0,-0.1) to (1-\LTDispoffset,-0.1);
\draw[red,line width=1 pt] (1+\LTDispoffset,-0.1) to (2-\LTDispoffset,-0.1);
\draw[red,line width=1 pt] (2+\LTDispoffset,-0.1) to (3-\LTDispoffset,-0.1);
\draw[red,line width=1 pt] (3+\LTDispoffset,-0.1) to (4+0.1,-0.1);
\draw[red,line width=1 pt] (4+0.1,-0.1) to (4+0.1,1);

\node[red] at(0.25,-0.5)   {$x_1$};
\node[red] at(1.5,-0.5)   {$x_2$};
\node[red] at(2.5,-0.5)   {$x_3$};
\node[red] at(4,-0.5)   {$x_4$};

\end{scope}
\begin{scope}[shift={(6,-2.5)},rotate=0]
\node at(-.65,.5)   {$0$};
\fill[gray] (0,0) -- (-0.5,0.5)--(0,1);
\draw[line width=1 pt] (-0.5,0.5) to (0,0);
\draw[line width=1 pt] (-0.5,0.5) to (0,1);
\draw[line width=1 pt] (0,0) to (4,0);
\draw[line width=1 pt] (0,1) to (4,1);
\draw[line width=1 pt] (0,0) to (0,1);
\draw[line width=1 pt] (4,0) to (4,1);
\draw[line width=2.5 pt] (1,0) to (1,1);
\draw[line width=1 pt] (2,0) to (2,1);
\draw[line width=2.5 pt] (3,0) to (3,1);
\draw[line width=2.5 pt] (0,0) to (1,0);
\draw[line width=2.5 pt] (1,1) to (3,1);
\draw[line width=2.5 pt] (3,0) to (4,0);
\fill (1.5,0) circle (3pt);
\fill (2.5,0) circle (3pt);
\node at(4.3,0)   {$x$};

\draw[red,line width=1 pt] (-0.46 ,0.5-0.15) to (0,-0.1);
\draw[red,line width=1 pt] (0,-0.1) to (1-\LTDispoffset,-0.1);
\draw[red,line width=1 pt] (1+\LTDispoffset,-0.1) to (2-\LTDispoffset,-0.1);
\draw[red,line width=1 pt] (2+\LTDispoffset,-0.1) to (3-\LTDispoffset,-0.1);
\draw[red,line width=1 pt] (3+\LTDispoffset,-0.1) to (4-\LTDispoffset,-0.1);

\node[red] at(0.25,-0.5)   {$x_1$};
\node[red] at(1.5,-0.5)   {$x_2$};
\node[red] at(2.5,-0.5)   {$x_3$};
\node[red] at(3.5,-0.5)   {$x_4$};

\end{scope}

\end{tikzpicture}
}

\def\picPIotaWeightedGeneral[#1]{
\begin{tikzpicture}[line width=1pt,auto,scale=#1]

  \begin{scope}[shift={(0,0)},rotate=0]
   \draw [-](0,0) to  (0,2);
 \draw [line width=2pt,-](0,0) to  (2,0);
 \draw [-](0,2) to    (2,2);
 \draw [dotted](2,0) to    (2,2);
  \fill (0,2) circle (3pt);
  \fill (2,0) circle (3pt);
  \fill (2,2) circle (3pt);
  \fill (0,0) circle (3pt);
  \fill (0.7,0) circle (4pt);

  \node[left]   at(0,0)      {$0$};
  \node[below]   at(0.7,0)   {$\ve[\iota]$};
  \node[left]   at(0,2)      {$v$};
  \node[right]  at (2,0)     {$x$};
  \node[right]  at (2,2)     {$y$};
  \node[right]  at (2,1)     {$a$};
  \end{scope}

 \begin{scope}[shift={(5,0)},rotate=0]
 \draw [-](0,0) to  (0.25,1.2);
 \draw [-](0.25,1.2) to  (1.4,2);
 \draw [-](1.4,2) to  (2.5,2);
 \draw [line width=2pt,-](0,0) to    (2.5,0);
 \draw [dotted](2.5,0) to    (2.5,2);
 \draw [-](0.25,1.2) to  [out=135,in=90]   (-0.8,0);
  \fill (0,0) circle (3pt);
  \fill (0.25,1.2) circle (3pt);
  \fill (1.4,2) circle (3pt);
  \fill (2.5,2) circle (3pt);
  \fill (2.5,0) circle (3pt);
  \fill (-0.8,0) circle (3pt);
  \node[below]   at(-0.8,0)      {$\ve[\iota]$};
  \node[below]   at(0,0)      {$0$};
  \node[above]   at(0.25,1.2)      {$u$};
  \node   at(1.5,1.5)      {$v$};
  \node[right]  at (2.5,0)     {$x$};
  \node[right]  at (2.5,2)     {$y$};
  \node[right]  at (2.5,1)     {$a$};
  \end{scope}

 \end{tikzpicture}}

\def\picPIotaWeightedGeneralTricks[#1]{
\begin{tikzpicture}[line width=1pt,auto,scale=#1]

  \begin{scope}[shift={(0,0)},rotate=0]
   \draw [-](0,0) to  (0,2);
 \draw [line width=2pt,-](0,0) to  (2,0);
 \draw [-](0,2) to    (2,2);
 \draw [dotted](2,0) to    (2,2);
  \fill (0,2) circle (3pt);
  \fill (2,0) circle (3pt);
  \fill (2,2) circle (3pt);
  \fill (0,0) circle (3pt);
  \fill (0.7,0) circle (4pt);

  \node[left]   at(-0.5,1)      {{\huge$\sum_\iota$}};

  \node[left]   at(0,0)      {$0$};
  \node[below]   at(0.7,0)   {$\ve[\iota]$};
  \node[left]   at(0,2)      {$v$};
  \node[right]  at (2,0)     {$x$};
  \node[right]  at (2,2)     {$y$};
  \node[right]  at (2,1)     {$a$};
{\small  \node  at (1,-1)     {$[1-\cos(k\cdot x)]$};}
    \node[right]  at (3,1)     {$=$};
  \end{scope}

  \begin{scope}[shift={(6,0)},rotate=0]
   \draw [-](0,0.7) to  (0,2);
   \draw [line width=2pt,-](0,0) to  (0,0.7);
   \draw [line width=2pt,-](0,0) to  (2,0);
   \draw [-](0,2) to    (2,2);
   \draw [dotted](2,0) to    (2,2);
  \fill (0,2) circle (3pt);
  \fill (2,0) circle (3pt);
  \fill (2,2) circle (3pt);
  \fill (0,0) circle (3pt);
  \fill (0,0.7) circle (4pt);

  \node[left]   at(-0.5,1)      {{\huge$\sum_\iota$}};

  \node[left]   at(0,0)      {$0$};
  \node[right]   at(0,0.7)   {$\ve[\iota]$};
  \node[left]   at(0,2)      {$v$};
  \node[right]  at (2,0)     {$x$};
  \node[right]  at (2,2)     {$y$};
  \node[right]  at (2,1)     {$a$};
{\small  \node  at (1,-1)     {$[1-\cos(k\cdot (x-\ve[\iota]))]$};}
    \node[right]  at (3.5,1)     {$=$};
  \end{scope}

  \begin{scope}[shift={(12,0)},rotate=0]
   \draw [-](0,0) to  (0,2);
 \draw [line width=2pt,-](0,0) to  (2,0);
 \draw [-](0,2) to    (2,2);
 \draw [dotted](2,0) to    (2,2);
  \fill (0,2) circle (3pt);
  \fill (2,0) circle (3pt);
  \fill (2,2) circle (3pt);
  \fill (0,0) circle (3pt);

  \node[left]   at(0,0)      {$0$};
  \node[left]   at(0,2)      {$v$};
  \node[right]  at (2,0)     {$x$};
  \node[right]  at (2,2)     {$y$};
  \node[right]  at (2,1)     {$a$};
{\small  \node  at (1,-1)     {$[1-\cos(k\cdot x)]$};}
  \end{scope}

  \begin{scope}[shift={(0,-5)},rotate=0]
   \draw [-](0,0) to  (0,2);
 \draw [line width=2pt,-](0.7,0) to  (2,0);
 \draw [-](0,0) to  (0.7,0);
 \draw [-](0,2) to    (2,2);
 \draw [dotted](2,0) to    (2,2);
  \fill (0,2) circle (3pt);
  \fill (2,0) circle (3pt);
  \fill (2,2) circle (3pt);
  \fill (0,0) circle (3pt);
  \fill (0.7,0) circle (4pt);

  \node[left]   at(-1,1)      {{\huge$\sum_\iota$}};

  \node[left]   at(0,0)      {$0$};
  \node[below]   at(0.7,0)   {$\ve[\iota]$};
  \node[left]   at(0,2)      {$v$};
  \node[right]  at (2,0)     {$x$};
  \node[right]  at (2,2)     {$y$};
  \node[right]  at (2,1)     {$a$};
   \node[right]  at (-0.1,1)     {${\tiny \geq n}$};
{\small  \node  at (1,-1)     {$[1-\cos(k\cdot (x-\ve[\iota]))]$};}
    \node[right]  at (3,1)     {$=$};
  \end{scope}

  \begin{scope}[shift={(6,-5)},rotate=0]
   \draw [-](0,0.7) to  (0,2);
   \draw [-](0,0) to  (0,0.7);
   \draw [line width=2pt,-](0,0) to  (2,0);
   \draw [-](0,2) to    (2,2);
   \draw [dotted](2,0) to    (2,2);
  \fill (0,2) circle (3pt);
  \fill (2,0) circle (3pt);
  \fill (2,2) circle (3pt);
  \fill (0,0) circle (3pt);
  \fill (0,0.7) circle (4pt);

  \node[left]   at(-0.5,1)      {{\huge$\sum_\iota$}};
  \node[right]  at (-0.1,1.4)     {${\tiny \geq n}$};
  \node[left]   at(0,0)      {$0$};
  \node[right]   at(0,0.7)   {$\ve[\iota]$};
  \node[left]   at(0,2)      {$v$};
  \node[right]  at (2,0)     {$x$};
  \node[right]  at (2,2)     {$y$};
  \node[right]  at (2,1)     {$a$};
{\small  \node  at (1,-1)     {$[1-\cos(k\cdot x)]$};}
    \node[right]  at (3.5,1)     {$=$};
  \end{scope}

  \begin{scope}[shift={(12,-5)},rotate=0]
   \draw [-](0,0) to  (0,2);
 \draw [line width=2pt,-](0,0) to  (2,0);
 \draw [-](0,2) to    (2,2);
 \draw [dotted](2,0) to    (2,2);
  \fill (0,2) circle (3pt);
  \fill (2,0) circle (3pt);
  \fill (2,2) circle (3pt);
  \fill (0,0) circle (3pt);

  \node[left]   at(0,0)      {$0$};
  \node[left]   at(0,2)      {$v$};
  \node[right]  at (2,0)     {$x$};
  \node[right]  at (2,2)     {$y$};
  \node[right]  at (-0.2,1)     {${\Small \geq n+1}$};
  \node[right]  at (2,1)     {$a$};
{\small  \node  at (1,-1)     {$[1-\cos(k\cdot x)]$};}
  \end{scope}
 \end{tikzpicture}}

\def\picContplantAnimal[#1]{
\begin{tikzpicture}[auto,scale=#1]

\begin{scope}[shift={(0,0)},rotate=0]

\draw[line width=1 pt] (0,0) to (-1,0);
\draw[line width=1 pt] (-1,0) to (-2,0);
\draw[line width=1 pt] (-2,0) to (-2,1);
\draw[line width=1 pt] (0,0) to (0,-1);
\draw[line width=1 pt] (-1,0) to (-1,2);
\draw[line width=1 pt] (0,-1) to (0,-2);
\draw[line width=1 pt] (0,-1) to (-1,-1);
\draw[line width=1 pt] (-1,-1) to (-1,-2);
\draw[line width=1 pt] (0,0) to (1,0);
\draw[line width=1 pt] (1,0) to (2,0);
\draw[line width=1 pt] (1,0) to (1,-1);
\draw[line width=1 pt] (0,1) to (0,2);
\draw[line width=1 pt] (-1,2) to (0,2);
\draw[line width=1 pt] (0,1) to (-1,1);
\draw[red,line width=2 pt] (0,0) to (0,1);
\draw[red,line width=2 pt] (0,1) to (1,1);
\draw[red,line width=2 pt] (1,1) to (2,1);
\draw[red,line width=2 pt] (1,1) to (1,2);
\draw[red,line width=2 pt] (2,1) to (2,2);
\draw[red,line width=2 pt] (1,2) to (2,2);
\draw[red,line width=2 pt] (2,2) to (2,3);
\draw[red,line width=2 pt] (2,2) to (3,2);
\draw[red,line width=2 pt] (2,3) to (3,3);
\fill (0,0) circle (3pt);
\fill (2,2) circle (3pt);
\node at(-0.2,-0.2)   {$x$};
\node at(2.2,2.2)   {$y$};
\node[red] at(0.7,3)   {$B^A(x,y)$};
\end{scope}

\begin{scope}[shift={(7,0)},rotate=0]
\draw[line width=1 pt] (0,0) to (-1,0);
\draw[line width=1 pt] (-1,0) to (-2,0);
\draw[line width=1 pt] (-2,0) to (-2,1);
\draw[line width=1 pt] (0,0) to (0,-1);
\draw[line width=1 pt] (0,-1) to (2,-1);
\draw[line width=1 pt] (2,0) to (2,-1);
\draw[line width=1 pt] (0,-1) to (0,-2);
\draw[line width=1 pt] (0,-1) to (-1,-1);
\draw[line width=1 pt] (-1,-1) to (-1,-2);
\draw[line width=1 pt] (1,0) to (1,-1);
\draw[line width=1 pt] (1,0) to (2,0);
\draw[line width=1 pt] (0,0) to (1,0);
\draw[red,line width=2 pt] (0,0) to (1,0);
\draw[line width=1 pt] (0,1) to (1,1);
\draw[line width=1 pt] (0,1) to (0,2);
\draw[line width=1 pt] (0,0) to (0,1);
\draw[line width=1 pt] (-1,2) to (0,2);
\draw[line width=1 pt] (0,1) to (-1,1);
\draw[red,line width=2 pt] (1,0) to (1,1);
\draw[red,line width=2 pt] (2,1) to (2,2);

\draw[red,line width=2 pt] (1,1) to (1,2);
\draw[red,line width=2 pt] (1,2) to (2,2);
\draw[red,line width=2 pt] (2,2) to (2,3);
\draw[red,line width=2 pt] (2,2) to (3,2);
\draw[red,line width=2 pt] (2,3) to (3,3);
\fill (0,0) circle (3pt);
\fill (2,2) circle (3pt);
\node at(-0.2,-0.2)   {$x$};
\node at(2.2,2.2)   {$y$};
\node[red] at(0.7,3)   {$B^A(x,y)$};
\end{scope}
\end{tikzpicture}
}

\def\picXiNtwoStructureLA[#1]{
\begin{tikzpicture}[line width=1pt,auto,scale=#1]

 \begin{scope}[shift={(0,0)},rotate=0]
  \begin{scope}[shift={(0.5,0)},rotate=0]
 \fill[gray] (4,0) -- (5,1)--(5,-1);
 \draw [- latex](4,0) to  (5,1);
 \draw [latex-](4,0) to  (5,-1);
 \draw [-latex](5,1) to  (5,-1);
  \fill (4,0) circle (3pt);
  \node[left]   at(4,0)      {$0$};
 \draw [-](5+0.03,-1-0.2) to (5+0.03,-1+0.2);
 \draw [-](5+0.3-0.03,-1-0.2) to   (5+0.3-0.03,-1+0.2);
 \node[below]   at(5.2,-1)      {$b_1$};
  \node[above]   at(5,1)      {$u$};
     \fill (5,1) circle (3pt);
  \end{scope}
 \draw [-latex](5.5,1) to  (6.5,1.2);
 \draw [latex-](6.5,1.2) to  (7.5,1);
 \draw [-latex](5.8,-1) to  (7.5,-1);
 \draw [latex-](7.5,1) to  (7.5,-1);
 \draw [-](7.5+0.03,1-0.2) to (7.5+0.03,1+0.2);
 \draw [-](7.5+0.3-0.03,1-0.2) to   (7.5+0.3-0.03,1+0.2);
 \draw [-latex](7.5,-1) to  (9, -1);
 \draw [-latex](7.8,1) to  (9,  1);
  \draw [-latex](9,-1) to  (9,  1);
   \fill (9,-1) circle (3pt);
   \fill (9,1) circle (3pt);
   \fill (7.5,-1) circle (3pt);
   \fill (6.5,1.2) circle (3pt);
  \node[right]   at(9,1)      {$x$};
  \node[right]   at(9,-1)      {$w_2$};

  \node[above]   at(6.5,1.2)      {$w_1$};
  \node[above]   at(7.75,1)      {$b_{t_1+1}$};
  \node[below]   at(7.5,-1)      {$y$};

  \tikzstyle{every node}=[font=\footnotesize]
  \node[above]  at(4.9,0.5)    {$I_{S2}$};
  \node[below]  at(5,-0.5)   {$I_{S3}$};
  \node[right]  at(5.5,0)    {$I_{S4}$};
  \node[above]  at(6.5,-1)   {$II_B$};
  \node[below]  at(6.2,1.2)      {$I_{S1}$};
  \node[below]  at(7,1.1)      {$III_S$};
  \node[right]  at(7.4,0)        {$III_B$};
  \node[above]  at(8.5,-1)   {$II_S$};
  \node[below]  at(8.5,1)   {$IV_B$};
  \node[right]  at(9,0)    {$IV_S$};
  \end{scope}

 \begin{scope}[shift={(7,0)},rotate=0]
  \begin{scope}[shift={(0.5,0)},rotate=0]
 \fill[gray] (4,0) -- (5,1)--(5,-1);
 \draw [-latex](4,0) to  (5,1);
 \draw [latex-](4,0) to  (5,-1);
 \draw [-latex](5,1) to  (5,-1);
  \fill (4,0) circle (3pt);
  \node[left]   at(4,0)      {$0$};
 \draw [-](5+0.03,-1-0.2) to (5+0.03,-1+0.2);
 \draw [-](5+0.3-0.03,-1-0.2) to   (5+0.3-0.03,-1+0.2);
 \node[below]   at(5.2,-1)      {$b_1$};
  \node[above]   at(5,1)      {$u$};
     \fill (5,1) circle (3pt);
  \end{scope}
 \draw [-latex](5.5,1) to  (6.5,1.2);
 \draw [latex-](6.5,1.2) to  (7.5,1);
 \draw [-latex](5.8,-1) to  (7.5,-1);
 \draw [latex-](7.5,1) to  (7.5,-1);
 \draw [-](7.5+0.03,-1-0.2) to (7.5+0.03,-1+0.2);
 \draw [-](7.5+0.3-0.03,-1-0.2) to   (7.5+0.3-0.03,-1+0.2);
 \draw [-latex](7.8,-1) to  (9, -1);
 \draw [-latex](7.5,1) to  (9,  1);
  \draw [-latex](9,-1) to  (9,  1);
   \fill (9,-1) circle (3pt);
   \fill (9,1) circle (3pt);
   \fill (7.5,1) circle (3pt);
   \fill (6.5,1.2) circle (3pt);
  \node[right]   at(9,-1)      {$x$};
  \node[right]   at(9,1)      {$w_2$};

  \node[above]   at(6.5,1.2)      {$w_1$};
  \node[below]   at(7.75,-1)      {$b_{t_1+1}$};
  \node[above]   at(7.5,1)      {$y$};

  \tikzstyle{every node}=[font=\footnotesize]
  \node[above]  at(4.9,0.5)    {$I_{S2}$};
  \node[below]  at(5,-0.5)   {$I_{S3}$};
  \node[right]  at(5.5,0)    {$I_{S4}$};
  \node[above]  at(6.5,-1)   {$II_B$};
  \node[below]  at(6.2,1.2)      {$I_{S2}$};
  \node[below]  at(7,1.1)      {$II_{S1}$};
  \node[right]     at(7.4,0)        {$II_{S3}$};
  \node[above]  at(8.5,-1)   {$III_B$};
  \node[below]  at(8.5,1)   {$II_{S1}$};
  \node[right]  at(9,0)    {$III_S$};
  \end{scope}
\end{tikzpicture}}

\def\picXiNanimaldetalotaStructure[#1]{
\begin{tikzpicture}[line width=1pt,auto,scale=#1]
 \tikzstyle{every node}=[font=\small]

 \begin{scope}[shift={(0,0)},rotate=0]
 \draw [-](5,1) to  (5,-1);
  \node[left]   at(5,-1)      {$0$};
 \draw [-](5+0.03,-1-0.2) to (5+0.03,-1+0.2);
 \draw [-](5+0.5-0.03,-1-0.2) to   (5+0.5-0.03,-1+0.2);
 \node[below]   at(5.5,-1)      {$e_1$};
  \node[above]   at(5,1)      {$w_1$};
  \fill (5,1) circle (3pt);
   \fill (7,1) circle (3pt);
  \fill (7,-1) circle (3pt);
  \draw [-](5,1) to  (7,1);
  \draw [-](5.5,-1) to  (7,-1);
  \node[right]   at(7,1)      {$y$};
  \node[right]   at(7,-1)      {$x$};
  \end{scope}

 \begin{scope}[shift={(6,0)},rotate=0]
 \draw [-](4,0) to  (5,1);
 \draw [-](4,0) to  (5,-1);
 \draw [-](5,1) to  (5,-1);
  \fill (4,0) circle (3pt);
  \node[left]   at(4,0)      {$0$};
 \draw [-](5+0.03,-1-0.2) to (5+0.03,-1+0.2);
 \draw [-](5+0.5-0.03,-1-0.2) to   (5+0.5-0.03,-1+0.2);
 \node[below]   at(5.5,-1)      {$e_1$};
 \node[left]   at(5,-1)      {$\bb_0\neq 0$};

  \node[above]   at(5,1)      {$v$};
  \fill (5,1) circle (3pt);
  \fill (6,1.2) circle (3pt);
  \fill (7,1) circle (3pt);
  \fill (7,-1) circle (3pt);

 \draw [-](5,1) to  (6,1.2);
 \draw [-](6,1.2) to  (7,1);
 \draw [-](5.5,-1) to  (7,-1);
  \node[right]   at(7,1)      {$y$};
  \node[right]   at(7,-1)      {$x$};
  \node[above]   at(6,1.2)      {$w_1$};
  \end{scope}

 \begin{scope}[shift={(12,0)},rotate=0]
 \draw [-](4,-1) to  (5,0.5);
 \draw [-](5,0.5) to  (5,-1);
  \fill (4,-1) circle (3pt);
  \fill (5,-1) circle (3pt);
  \node[below]   at(5,-1)      {$0$};
  \node[below]   at(4,-1)      {$e_1$};
 \draw [-](5+0.03,-1-0.2) to (5+0.03,-1+0.2);
 \draw [-](5+0.5-0.03,-1-0.2) to   (5+0.5-0.03,-1+0.2);
  \node[above]   at(5-0.3,0.5)      {$u$};
  \fill (5,0.5) circle (3pt);
  \fill (6,1.2) circle (3pt);
  \fill (7,1) circle (3pt);
  \fill (7,-1) circle (3pt);

 \draw [-](5,0.5) to  (6,1.2);
 \draw [-](6,1.2) to  (7,1);
 \draw [-](5.5,-1) to  (7,-1);
  \node[right]   at(7,1)      {$y$};
  \node[right]   at(7,-1)      {$x$};
  \node[above]   at(6,1.2)      {$w_1$};
  \end{scope}

 \begin{scope}[shift={(1,-5)},rotate=0]
 \draw [-](4,0) to  (5,1);
 \draw [-](4,0) to  (5,-1);
 \draw [-](3,0) to  (4,0);
 \draw [-](5,1) to  (5,-1);
  \fill (4,0) circle (3pt);
  \node[above]   at(4,0)      {$0$};
  \node[below]   at(3,0)      {$e_1$};
 \draw [-](5+0.03,-1-0.2) to (5+0.03,-1+0.2);
 \draw [-](5+0.5-0.03,-1-0.2) to   (5+0.5-0.03,-1+0.2);
  \node[above]   at(5,1)      {$v$};
  \fill (5,1) circle (3pt);
  \fill (6,1.2) circle (3pt);
  \fill (7,1) circle (3pt);
  \fill (7,-1) circle (3pt);

 \draw [-](5,1) to  (6,1.2);
 \draw [-](6,1.2) to  (7,1);
 \draw [-](5.5,-1) to  (7,-1);
  \node[right]   at(7,1)      {$y$};
  \node[right]   at(7,-1)      {$x$};
  \node[above]   at(6,1.2)      {$w_1$};
  \end{scope}

 \begin{scope}[shift={(7,-5)},rotate=0]
 \draw [-](4,0) to  (5,1);
 \draw [-](4,0) to  (5,-1);
 \draw [-](4.5,0.5) to  [out=135,in=90]   (3,0);
  \draw [-](5,1) to  (5,-1);
  \fill (4,0) circle (3pt);
  \fill (4.5,0.5) circle (3pt);
  \node[below]   at(4,0)      {$0$};
  \node[above]   at(4.5,0.5)      {$u$};
  \node[below]   at(3,0)      {$e_1$};
 \draw [-](5+0.03,-1-0.2) to (5+0.03,-1+0.2);
 \draw [-](5+0.5-0.03,-1-0.2) to   (5+0.5-0.03,-1+0.2);
  \node[above]   at(5,1)      {$v$};
  \fill (5,1) circle (3pt);
  \fill (6,1.2) circle (3pt);
  \fill (7,1) circle (3pt);
  \fill (7,-1) circle (3pt);

 \draw [-](5,1) to  (6,1.2);
 \draw [-](6,1.2) to  (7,1);
 \draw [-](5.5,-1) to  (7,-1);
  \node[right]   at(7,1)      {$y$};
  \node[right]   at(7,-1)      {$x$};
  \node[above]   at(6,1.2)      {$w_1$};
  \end{scope}

 \begin{scope}[shift={(13,-5)},rotate=0]
 \draw [-](4,0) to  (5,1);
 \draw [-](4,0) to  (5,-1);
 \draw [-](4.5,-0.5) to  [out=225,in=270]   (3,0);
  \draw [-](5,1) to  (5,-1);
  \fill (4,0) circle (3pt);
  \fill (4.5,-0.5) circle (3pt);
  \node[above]   at(4,0)      {$0$};
  \node[below]   at(4.5,-0.5)      {$u$};
  \node[above]   at(3,0)      {$e_1$};
 \draw [-](5+0.03,-1-0.2) to (5+0.03,-1+0.2);
 \draw [-](5+0.5-0.03,-1-0.2) to   (5+0.5-0.03,-1+0.2);
  \node[above]   at(5,1)      {$v$};
  \fill (5,1) circle (3pt);
  \fill (6,1.2) circle (3pt);
  \fill (7,1) circle (3pt);
  \fill (7,-1) circle (3pt);

 \draw [-](5,1) to  (6,1.2);
 \draw [-](6,1.2) to  (7,1);
 \draw [-](5.5,-1) to  (7,-1);
  \node[right]   at(7,1)      {$y$};
  \node[right]   at(7,-1)      {$x$};
  \node[above]   at(6,1.2)      {$w_1$};
  \end{scope}

\end{tikzpicture}}

\def\picXilotaInitialSketch[#1]{
\begin{tikzpicture}[line width=1pt,auto,scale=#1]
 \tikzstyle{every node}=[font=\small]
\node   at(-3-0.4,1)      {$P^{\ssc[1],\iota,-m}(x,y):=\sum_w$};
 \begin{scope}[shift={(0,0)},rotate=0]

 \node   at(1,0)      {$\sum_{v,w,u,\bb_1}$};
 \fill[gray] (4,0) -- (5,1)--(5,-1);
 \draw [-](4,0) to  (5,1);
 \draw [line width=2pt](4,0) to  (5,-1);
 \draw [dotted](4.5,0.5) to  [out=135,in=45]   (3,0);
 \draw [dotted](4.5,-0.5) to  [out=225,in=270]   (3,0);
 \draw [dotted](3,0) to  [out=345,in=210] (5,0);
 \draw [dotted](3,0) to  [out=90,in=150] (5.5,1.1);
  \draw [-](5,1) to  (5,-1);
  \fill (4,0) circle (3pt);
  \fill (5,-1) circle (3pt);
  \fill[gray,draw=black] (4.5,0.5) circle (3pt);
  \fill[gray,draw=black] (4.5,-0.5) circle (3pt);
  \fill[gray,draw=black] (5,0) circle (3pt);
  \fill[gray,draw=black] (5.5,1.1) circle (3pt);

  \node[below]   at(4,0)      {$0$};
  \node   at(4.6,0)      {$u$};
  \node[below]   at(3,0)      {$e_\iota$};
  \node[above]   at(5,1)      {$v$};
  \node[below]   at(5,-1)      {$\bb_1$};
  \fill (5,1) circle (3pt);
  \fill (6,1.2) circle (3pt);
  \fill (7,1) circle (3pt);
  \fill (7,-1) circle (3pt);
  \fill (3,0) circle (3pt);

 \draw [-](5,1) to  (6,1.2);
 \draw [-](6,1.2) to  (7,1);
 \draw [line width=2pt](5,-1) to  (7,-1);
 \draw [-](7,-1) to  (7,1);
  \node[right]   at(7,0)      {$m$};
  \node[right]   at(7,1)      {$y$};
  \node[right]   at(7,-1)      {$x$};
  \node[above]   at(6,1.2)      {$w$};
  \end{scope}

 \begin{scope}[shift={(7,0)},rotate=0]

 \node   at(2,0)      {$+\sum_{v,w,\bb_1}$};
 \fill[gray] (4,0) -- (5,1)--(5,-1);
 \draw [-](4,0) to  (5,1);
 \draw [line width=2pt](4,0) to  (5,-1);
 \draw [-](5,1) to  (5,-1);

  \fill (4,0) circle (3pt);
  \node[left]   at(4,0)      {$0$};
 \draw [-](5+0.03,-1-0.2) to (5+0.03,-1+0.2);
 \draw [-](5+0.5-0.03,-1-0.2) to   (5+0.5-0.03,-1+0.2);
 \node[below]   at(5.5,-1)      {\small $e_\iota$};
 \node[left]   at(5.1,-1.2)      {\small  $\bb_1$};

  \node[above]   at(5,1)      {\small $v$};
  \fill (5,1) circle (3pt);
  \fill (6,1.2) circle (3pt);
  \fill (7,1) circle (3pt);
  \fill (7,-1) circle (3pt);

 \draw [-](5,1) to  (6,1.2);
 \draw [-](6,1.2) to  (7,1);
 \draw [line width=2pt](5.5,-1) to  (7,-1);
 \draw [-](7,1) to  (7,-1);

  \node[right]   at(7,1)      {\small $y$};
  \node[right]   at(7,0)      {\small $m$};
  \node[right]   at(7,-1)      {\small $x$};
  \node[above]   at(6,1.2)      {\small $w$};
  \end{scope}

\end{tikzpicture}}

\section{An overview to the NoBLE diagrammatic bounds}
\label{secBoundsExplained}
In this section we discuss how we bound the NoBLE coefficients. We explain the concepts that we use to obtain sharp bound on these diagrams, each of which forms a numerical improvement on the bounds on classical lace-expansion coefficients. Then, we follow that up with the technical definitions of these concepts. We end this section with a discussion of how we bound double connections for LAs, as these are quite central in our bounds for LAs.

\subsection{Basic bounds and their improvements}
\label{sec-BoundIdea-Heuristic}
We explain here how we obtain sharp bounds on
	\begin{align}
	\bar \Xi^{\ssc[1]}_z(x) =& \sum_{\omega\in\Wcal(x)} z^{|\omega|} Z[0,|\omega|]  J^\ssc[1][0,|\omega|]\nnb
                                              =& \sum_{\omega\in\Wcal(x)}  z^{|\omega|}  \prod_{i=0}^{|\omega|} z^{|S^{\omega}_i|}
                                              \indic{S^{\omega}_s\text{ and }S^{\omega}_t\text{ are only intersecting for $s=0$ and $t=|\omega|$}   },
	\lbeq{bound-Xi1-1}
	\end{align}
by using four ideas, one at a time. We restrict here to LTs, as these give rise to the simplest diagrams, which already display all the main ingredients to our bounds.

We compare the improvement that we obtain in our bounds by providing numerical bounds for the initial point $z=z_I$ resulting from the numerical analysis performed in our Mathematica notebooks. Our proof is based upon the fact that the values at $z_I$ and at $z_c$  are actually quite close. Thus, even though $z_c>z_I$,
such bounds should give us a clear idea of how large such bounds are, and what the effect of the improvements is.

We recall \refeq{LTInitalBound-goal} and combine it with the bound $g^\iota_{z_I}\leq \e$, see \refeq{gjgzrelation-Initial}, to conclude, for $z_I=(2d-1)^{-1}\e^{-1}$,
	\begin{align}
	\lbeq{bound-Gzi}
	\bar G_{z_I}(x)\leq g_{z_I}^\iota B_{1/(2d-1)}(x)=\e B(x),
	\end{align}
where $\e$ is the Euler number and $B(x)$ critical two-point function of the NBW, see \refeq{NBWGen}. This bound is independent of the lace expansion and we can compute its value numerically. While this bound does not hold at the critical point $z_c$, its value is good enough to compare the bounds discussed in this section. All stated numerical results are computed for dimension $d=18$, and are aimed to convey the numerical improvements in the bounds.

\paragraph{Plain-vanilla bound.}
We focus on $\bar \Xi^{\ssc[1]}_z$ and start with the simplest possible approach, which relies on the bounds as they have been stated up to now in the literature.
The coefficient $\bar \Xi^{\ssc[1]}_z$ involves three points: $0$, $x$ and a point where $S^{\omega}_0$ and $S^{\omega}_{|\omega|}$ intersect, which we denote by $w$.
Bounding the connections between each of the points by individual independent LTs we obtain
	\begin{align}
	\lbeq{bound-Xi1-2pre}
    	\bar \Xi^{\ssc[1]}_z(x) \leq & \bar G_z(x)\sum_{w} \bar G_z(w-x) \bar G_z(w),
	\end{align}
which, combined with \refeq{bound-Gzi}, implies, for $z\leq z_I$,
	\begin{align}
	\lbeq{bound-Xi1-2}
	 \sum_x\bar \Xi^{\ssc[1]}_{z_I}(x) \leq & (\bar G_{z_I})^{\star 3}(0)\leq \e^3  B^{\star 3}(0)=22.322\dots
	\end{align}
This is of course a very bad bound that does not allow us to successfully apply the lace-expansion method, so we improve it using the four ideas explained below.

\paragraph{First improvement: extracting trivial lines.}
A simple way to improve this bound is to consider four cases depending on which lines are trivial: $x=0$, then $w=x\neq 0$ and $w=0\neq x$ and the remaining cases.
In the case $x=0$, we extract the contribution of the first rib $g_z^\iota$ and the first bond of the backbone
with all its possible directions $2dz$, and bound the remaining rib-walk, that goes back to the origin, by $\bar G_z(\ve[1])$.
Thus, the case $x=0$ is bounded by $2d zg_z^\iota(D\star \bar G_z)(0)$.

For $w=x\neq 0$ and $w=0\neq x$,  $2$ non-trivial two-point functions are required. For $w\nin\{0,x\}$, instead, all three connections are non-trivial.
We bound the non-trivial connections by $\bar G_{z_I}(x) \leq \e B(x)\leq \tfrac {2d\e}{2d-1}(D\star B)(x)$ for $x\neq 0$ and conclude that
	\begin{align}
 	\sum_x\bar \Xi^{\ssc[1]}_{z_I}(x) \leq &
 	\e\Big(\frac {2d}{2d-1}\Big)^2 (D^{\star 2}\star B)(0)+ 2\e^2 \Big(\frac {2d}{2d-1}\Big)^2 (D^{\star 2}\star B^{\star 2})(0)\nnb
 	&+ \e^3\Big(\frac {2d}{2d-1}\Big)^3 (D^{\star 3}\star B^{\star 3})(0)=0.85541\ldots
	\lbeq{bound-Xi1-3}
	\end{align}
This is an enormous improvement compared to  \refeq{bound-Xi1-2}. The simple idea of extracting trivial lines, by splitting between different cases, is a basic technique that is used for all coefficients in Section \ref{secBoundProof}.

\paragraph{Second improvement: improved counting of one-point functions by rib allocation.}
We observe that we extract a single one-point function at every line in an intersection point. This gives rise to the three factors of $\e$ in \refeq{bound-Xi1-3}.
However, when the different lines come from the {\em same} rib, there in fact is only {\em one} one-point function, so that we are overcounting one one-point function (leading to a factor $\e$)
every time we split a rib. Thus, for the diagram in \refeq{bound-Xi1-2}, we would need only {\em one} factor $\e$ instead of three.
Let us explain one way to obtain a bound without overcounting.

In \refeq{bound-Xi1-2pre}, the terms $\bar G_z(w)$ and $\bar G_z(w-x)$ bound the ribs $S^{\omega}_0$ and $S^{\omega}_{|\omega|}$, which, being ribs, are just LTs.
The indicator $J^\ssc[1][0,|\omega|]$ in \refeq{bound-Xi1-1} guarantees that only those $\omega$ contribute for which all in-between ribs do not intersect, so that the rib-walk between
$\tb_1^\omega=u$ and $\tb^{\omega}_{|\omega|}=v$ also describes a LT and can be bounded by $\bar G_z(v-u)$.
Extracting the case that $|\omega|=1$, for which $\tb_1^\omega=\tb^{\omega}_{|\omega|}$,
so that there are no in-between ribs, we obtain
	\begin{align}
	\lbeq{bound-Xi1-4pre}
    	\bar \Xi^{\ssc[1]}_z(x) \leq
    	&\indic{x=0}2d zg_z^\iota(D\star \bar G_z)(0) \\
    	&+\indic{x\neq 0} \sum_w \bar G_z(w) \bar G_z(w-x)
  	\Big( \indic{|x|=1} z+ z^2  \sum_{\stackrel{u,v\in\Zd}{|u|=|x-v|=1}}
      	\bar G_z(v-u)\Big),\nn
	\end{align}
where $z$ and $z^2$ correspond to the weight of the first and last bond of the original backbone.
Considering the four cases for $0,x,w$, used in the first improvement, we obtain
	\begin{align}
	\sum_x\bar \Xi^{\ssc[1]}_{z_I}(x) \leq &
 	3\e\Big(\frac {2d}{2d-1}\Big)^2 (D^{\star 2}\star B)(0) + \e \Big(\frac {2d}{2d-1}\Big)^3 (D^{\star 3}\star B^{\star 2})(0)\nnb
 	&+ 2\e\Big(\frac {2d}{2d-1}\Big)^3 (D^{\star 3}\star B^{\star 2})(0) +\e\Big(\frac {2d}{2d-1}\Big)^4 (D^{\star 3}\star B^{\star 3})(0)=0.5725\dots
	\lbeq{bound-Xi1-4}
	\end{align}
The improvement of \refeq{bound-Xi1-3} to \refeq{bound-Xi1-4} is mostly realized by the proper handling
of the weights of the one-point function $\gj\approx \e$ at the origin and at $x$, which are bounded twice in \refeq{bound-Xi1-3}.

\begin{remark}[Rib-weight allocation]
\label{rem-rib-weight}
{\rm In general, the diagram of $\Xi^{\ssc[N]}_z$ has $4N-1$ connections for LTs. Allocating the one-point functions properly will save us approximately a factor $\e^{4N-2}$ for all $N$. This allocation of rib weights is a central technical problem, especially as we require various different allocations for the weights
when we bound the coefficient with spatial weight $|x|^2$. Using splitting arguments of such spatial terms, that are standard in lace-expansion analyses, we split the weight along a path connecting $0$ and $x$,
e.g.\ using $|x|^2\leq N\sum_{i=1}^N |x_i|^2$, where the $(x_i)_{i=1}^N$ are the displacements along the bottom lines of the diagram,
so that $x=\sum_{i=1}^N x_i$. For each of these partial weights $|x_i|^2$, we want to bound the corresponding connection
by $\bar G_z$, as $f_3$ only provides us with a bound on $|x_i|^2 \bar G_z(x_i)$.
This imposes severe restrictions on our rib-weight allocation, and we need to be really careful in such arguments.
We discuss this in more detail in Section \ref{secBoundsOnePointF}, see also Remark \ref{rem-rib-weight2}.}
 \end{remark}

\paragraph{Third improvement: using the non-backtracking nature of the diagram.}
Now, we use the non-backtracking property of the rib-walk for the first time. This property implies that any loop needs at least $4$ bonds,
as direct reversals are prohibited and loops on the lattice require an even number of bonds.
For example the self-loop of the rib-walk back to $x=0$ will take at least $4$ steps and can be bounded by
	\begin{align}
	(2d {z_I}g_{z_I}^\iota)^3 (D^{\star 3}\star \bar G_{z_I})(0)
	\leq \left(\frac {2d }{2d-1}\right)^3 \frac {2d\e} {2d-1} (D^{\star 4}\star B)(0).
	\end{align}
Further, we can omit the case $w\in\{0,x\}$ for the case that $b_1^\omega=(0,x)$, see property (v) of Definition \ref{defLTRibwalks}.
Going through $8$ different combinations of how the $4$ steps can be distributed, we obtain
	\begin{align}
	\sum_x\bar \Xi^{\ssc[1]}_{z_I}(x) \leq &
 	4\e\big(\frac {2d}{2d-1}\big)^4 (D^{\star 4}\star B)(0) + 3\e \big(\frac {2d}{2d-1}\big)^4 (D^{\star 4}\star B^{\star 2})(0)\nnb
 	&+\e\big(\frac {2d}{2d-1}\big)^4 (D^{\star 3}\star B^{\star 3})(0)=0.0755\dots
	\lbeq{bound-Xi1-5}
	\end{align}
The improvement from \refeq{bound-Xi1-4} to \refeq{bound-Xi1-5} demonstrate the power of the NoBLE and its non-backtracking property. To use this property for the other diagrams, we will split depending on the lengths of lines shared by two loops. This splitting into a total of five cases of lengths of shared lines will create {\em matrix-based bounds.} This will become apparent in Section \ref{secBuildingBlocks}.

\paragraph{Fourth improvement: using repulsiveness in diagrams.} We can actually improve this bound even further. For this we have to understand that the paths connecting $0$, $x$ and $w$ do not intersect, so that the loop is actually a self-avoiding polygon. We call such a polygon a {\em repulsive diagram}. For this to work, we have to choose the point $w_1$ as a {\it first intersection point} and bound it by a {\it repulsive} diagram. Both terms are defined within the next two sections. Using these repulsive diagrams, we improve the bound \refeq{bound-Xi1-5} to
	\begin{align}
	\lbeq{bound-Xi1-6}
	\sum_x\bar \Xi^{\ssc[1]}_{z_I}(x) \leq & 0.02562\dots
	\end{align}
This is explained in more detail in Section \ref{secBoundsRepdia}.

\paragraph{Summary.} By the above four improvements combined, we have now improved our estimate on this simple diagram $\sum_x\bar \Xi^{\ssc[1]}_{z_I}(x)$ by a factor that is close to 1000. Since we need all the numerical precision that we can get our hands upon, so as to be able to apply our methods in the lowest possible dimensions, such estimates are crucial to the success of our method.

\paragraph{Extension to lattice animals.}
Lattice animals (LAs) allow for double connections. Despite this additional difficulty, for $\bb_0=0$ and $\tb_{|\omega|}=x$, we can use the same bounds \refeq{bound-Xi1-2pre}-\refeq{bound-Xi1-6} based on exactly the same arguments. For $\bb_0\neq 0$ we need to bound the sausage, which connects the points $0,\bb_0$, as well as the first intersection point $w_1$. We explain this in Section \ref{secBoundsOnePointF}, where we also define the required notions. In Section \ref{secBoundsDouble}, we discuss how we use a symmetry argument to improve the bound for the important case $w_1\in\{0,\bb_0\}$ of a double connection.

\subsection{Allocation of one-point functions}
\label{secBoundsOnePointF}
In the preceding section, we disassembled the rib-walk into up to three pieces: the initial rib, the final rib and the walk in between, and have argued that rib-weight allocation is a crucial ingredient to the success of our method. In the following, we also need to split individual ribs/sausages. Below we give a rigorous description of the split of a rib/LT and explain how we ensure that we do not overcount one-point functions by an optimal rib-weight allocation.

After this, we explain the allocation of one-point functions for the more involved diagrams. In the process we define the concept of \emph{first intersection point} of two sausages, that will allow use to choose bond-disjoint paths connecting the corner points of the loops. In the next section we use this choice of bond-disjoint paths to define {\em repulsive diagrams}.

\paragraph{Splitting the weight of a rib: planted animals/trees.}
Our aim is to bound the contribution $z^{|T|}$ of LTs that contain the vertices $0,x,v\in\Zd$ efficiently. We identify the last vertex $u$ that the path $0\leftrightarrow x$ and the path $0\leftrightarrow v$ have in common, split the tree into three trees and ignore the avoidance constraints between them to obtain
	\begin{align}
	\sum_{T} \indic{0,x,v\in T} z^{|T|}&\leq \sum_{T_1,T_2,T_3}\sum_{u\in\Zd}
	\indic{0,u\in T_1}\indic{x,u\in T_2}\indic{u,v\in T_3} z^{|T_1|+|T_2|+|T_3|}\nnb
	\lbeq{ribweightLT-1}
	&=\sum_{u} \bar G_{z}(u) \bar G_{z}(x-u)\bar G_{z}(v-u),
	\end{align}
which corresponds to the plain-vanilla bound of the last section. Thus, it also counts {\em three} ribs at $u$, whereas there really only is one.

To avoid this, we define an adapted version of the two-point function that only counts the trivial first rib/sausage $A^\omega_0=\varnothing$, i.e.,
	\begin{eqnarray}
	\lbeq{defTwoPointFunctionTildeTotal}
	\tilde G_{z}(x)&=&\sum_{\omega\in\Wcal(x)}z^{|\omega|} Z[0,|\omega|]K[0,|\omega|]\indic{A^\omega_0=\varnothing}.
	\end{eqnarray}
Using this two-point function, we obtain the bound
	\begin{align}
	\lbeq{ribweightLT-2}
	\sum_{T} \indic{0,x,v\in T} z^{|T|}\leq \sum_{u} \bar G_{z}(u)\tilde G_{z}(x-u)\tilde G_{z}(v-u),
	\end{align}
which is approximately $\e^2$ smaller than \refeq{ribweightLT-1}.

\paragraph{Planted animals.}
For LTs, the point $u$ at which we split the tree is unique, which is not the case for LA. We define the concepts of \emph{backbone}, \emph{connecting planted animal} and
 \emph{first point of intersection}, to provide a rigorous way to split the sausage for LAs. We only give the definition for LAs and sausage-walks, as each LT is also a LA and each rib-walk is also a sausage-walk.

Before defining these notions, we discuss the {\em lexicographic order} on paths, which is a way to identify paths uniquely. We say that $x\in \Zd$ has a lower lexicographic order than $y\in \Zd$, if there exists an $i\in \{1,2,\dots,d \}$ with $x_i<y_i$ and $x_j=y_j$ for $j=1,2,\dots, i-1$. A bond $b$ is defined as a tuple of two vertices $b=(\bb,\tb)$, so that we can view $b$ also to be a vector in $\Zbold^{2d}$ and use the same order relations as for vertices. A path $(b_1,b_2,\dots,b_n)$ has a lower lexicographic order than the path $(t_1,t_2,\dots,t_m)$, if either there exists a $i\in \{1,2,\dots, \min(n,m)\}$, such that $b_i$ has a lower order than $t_i$ and $b_j=t_j$ for $j=1,2,\dots, i-1$ or if $n<m$ and $b_i=t_i$ for $i=1,2,\dots, n$.

\begin{definition}[Backbone]
\label{defLTbackbone}
Let $A$ be a lattice animal containing $x,y\in\Zd$, with $x\neq y$.
A path from $x$ to $y$ in $A$ is a sequence of bonds $(b_i)_{i=1,\dots,N}$ such that $b_i\in A,
\tb_i\neq \tb_j$ for all $i\neq j$, and $\bb_1=x,\tb_i=\bb_{i+1},\tb_N=y$ for $i=1,\dots,N-1$.
We define $b^A(x,y)$ to be the path from to $x$ to $y$ in $A$ with the lowest lexicographic order.
For $x=y\in A$, we define $b^A(x,y)=\varnothing$.
\end{definition}

\begin{definition}[Connecting planted animal]
\label{defPlantedAnimal}
Let $A$ be a lattice animal containing $x,y\in\Zd$, with $x\neq y$, and let $p$ be a path from $x$ to $y$ in $A$.
Let $S$ be the subset of $A\setminus\{p\}$ of all bonds for which at least one vertex of the bond is connected to $x$ via bonds in $A\setminus \{ p\}$.
We call $B^A(x,y;p)=A\setminus S$ the \emph{connecting planted animal} from $x$ to $y$ along $p$.\\
Unless stated otherwise, we consider the path $p$ to be the backbone $p=b^A(x,y)$ and omit $p$ from the notation.
For a sausage-/rib-walk $\omega$, we define $B^\omega_i(x,y)=B^{A^\omega_i}(x,y)$.
\end{definition}
\begin{figure}[ht!]
\begin{center}
{\Large
\picContplantAnimal[0.85]
\\ }
\caption{Two examples of a connecting planted animal.}
\label{BoundLA-Figure-connectingtree}
\end{center}
\end{figure}

We next define the notion of a {\em first intersection point.} We do this in a more general setting, encoded by a lace $L$, where some of the LAs are forced to intersect each other, while others are not. This will make sure that we can use this definition for all NoBLE lace-expansion coefficients:

\begin{definition}[First intersection point]
\label{defLTFirstPointIntersect}
For a sausage-walk $\omega$, a lace $L=\{s_1t_1,s_2t_2,$ $\dots,$ $s_Mt_M\}$ and $i\leq M$, we define $W_i(\omega)$  to be the set of vertices $w$ that are contained in both $A^\omega_{s_i}$ and $A^\omega_{t_i}$, such that $b^{A^\omega_{s_i}}(\tb^\omega_{s_i},w)$ and $A^\omega_{t_i}(\tb^\omega_{t_i},w)$ only intersect at $w$. If $W_i(\omega)$ is non-empty, then we define $w_i(\omega)$ to be the smallest in the set $W_i(\omega)$.
We call $w_i(\omega)$ the \emph{first intersection point} of $A^\omega_{s_i}$ and $A^\omega_{t_i}$.
\end{definition}
Note that if the end of the $s_i$th pivotal bond $\tb^\omega_{s_i}$ is in $S^\omega_{t_i}$, then $W_i(\omega)=\{\tb^\omega_{s_i}\}$.
We illustrate the concept of a first intersection point in Figure \ref{BoundLT-Figure-connectingtree}.

\begin{figure}[ht!]
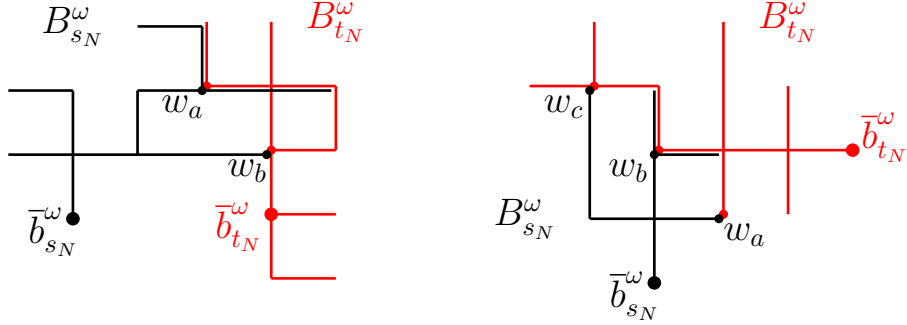

\begin{center}
{\Large
\picRootContreeAndFirstIntersection[0.85]}
\caption{
Two examples of the set of first intersection points $W_N(\omega)=\{w_a,w_b,w_c\}$.
In the picture the $s_N$th and $t_N$th sausage of the sausage-walk $\omega$ intersect.
We choose the point $w_N(\omega)$ to be the unique smallest representative of $W_N(\omega)$ in the lexicographic order.}
\label{BoundLT-Figure-connectingtree}
\end{center}
\end{figure}

\paragraph{Non-trivial sausages.}
If we wish to extend the analysis to LAs, then we need to deal with non-trivial sausages. A lattice animal $A$ is called a {\em non-trivial sausage} if it contains three vertices $0,x,w$ (with $x\neq 0$), such that $0$ and $x$ are doubly connected in $A$, as shown in the first image of Figure \ref{split-LA-sausage-Gbar-Table}. We now show how we bound these sausages.
In Section \ref{secBoundsDouble}, we explain how we obtain improved bounds for the special case $w\in\{0,x\}$, which is a major contribution for LAs.

We define the point $u$ to be a point that every path from $0$ to $w$ and $x$ to $w$ shares. By construction, the points $0,u,x$ are pairwise doubly connected, so that there exist four bond-disjoint paths: $p_1(u,w), p_2(u,x), p_3(0,u)$ and $p_4(0,x)$ connecting the indicated points in $A$. We bound the contribution of such LAs by
	\begin{align}
	\sum_{A\colon 0,x,w\in A} \indic{0\dbct{A} x}z^{|A|}\leq& 2d z \sum_{u}\tilde G_{z}(u-w)(D\star \tilde G_{z})(u-x)\tilde G_{z}(u)\bar G_{z}(x).
	\lbeq{split-LA-sausage-Gbar}
	\end{align}
We obtain this bound by decomposing each LA $A$ into four parts as shown in Figure \ref{split-LA-sausage-Gbar-Table}. Using the pictured decomposition, we know that $A_4$ is an animal containing $x$, which we bound by $\bar G_z(x)$. Further, $A_1$ and $A_3$, respectively, only include one bond that contains $u$ and $0$, respectively.
Thus, we bound them using the modified two-point function $\tilde G_z$. The animal $A_2$ is the interesting case as $u$ and $x$ are only connected by one bond each. If $u\neq x$, then we bound this contribution by the step involved by $2d z D(\cdot)$ and the remainder by $\tilde G_z$, creating the bound $2dz(D\star \tilde G_z)(x-u)$.
If $u=x$ then $A_2=\varnothing$ and $A_3$ only has one bond at the origin and one at $x$, so we bound $A_3$ by $2dz(D\star \tilde G_z)(x)$.

{\centering
\begin{figure}[ht]
\begin{tabular}{cc}
\begin{subfigure}{0.4\textwidth}\centering
\begin{tikzpicture}
    \node[anchor=south west,inner sep=0] at (0,0) {\includegraphics[height=6 cm]{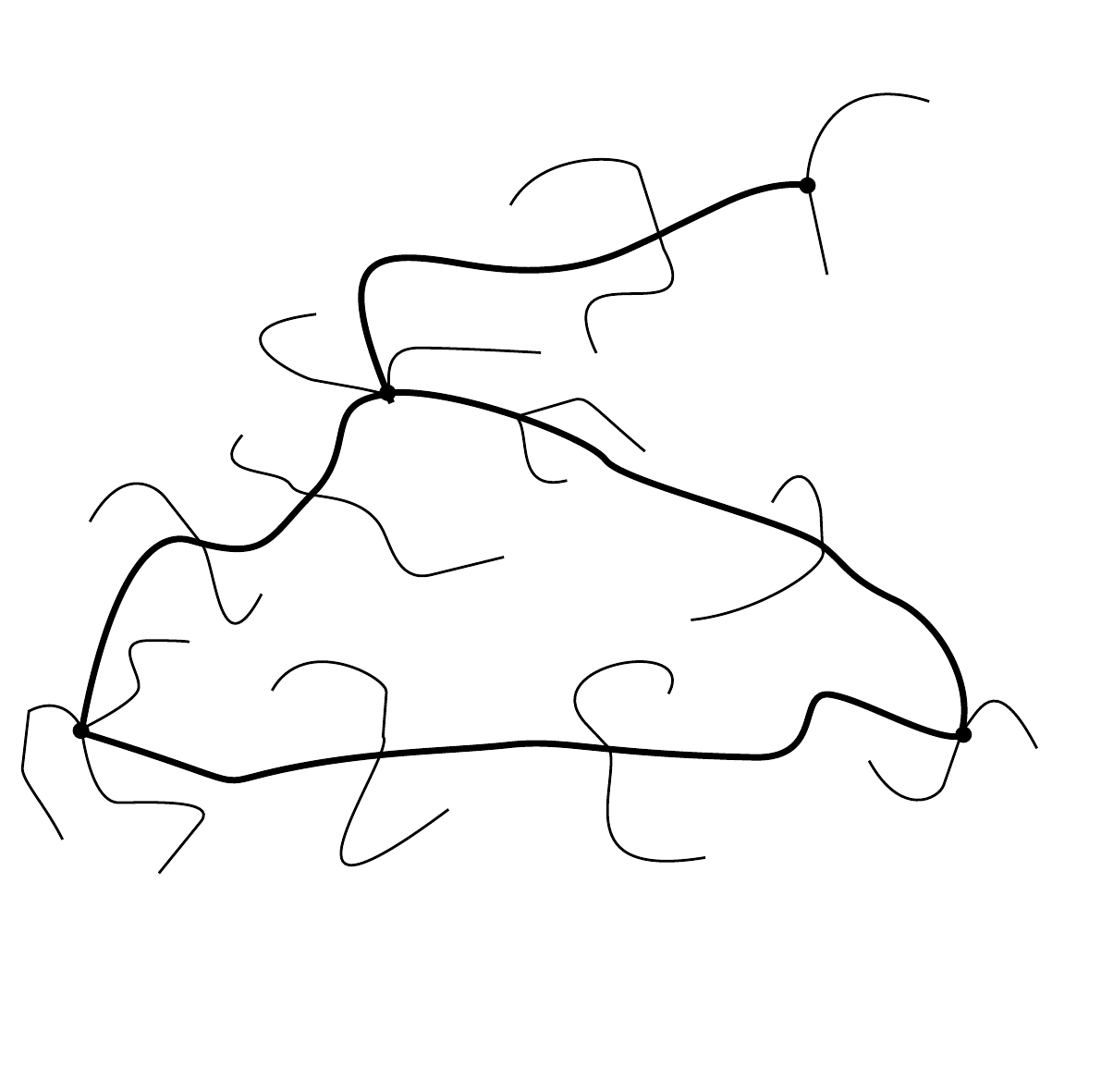}};
    \node at(0.1,2.4)   {$0$};
    \node at(5.5,1.7)   {$x$};
    \node at(2.2,3.6)   {$u$};
    \node at(4.8,5.1)   {$w$};
    \node at(1,4.5)   {$A$};
\end{tikzpicture}
\caption*{$A$ containing $0,x,w$ and the last common point $u$ on the paths from $0$ to $w$ and $x$ to $w$.}
\end{subfigure}&
\begin{subfigure}{0.4\textwidth}\centering
\begin{tikzpicture}
    \node[anchor=south west,inner sep=0] at (0,0) {\includegraphics[height=6 cm]{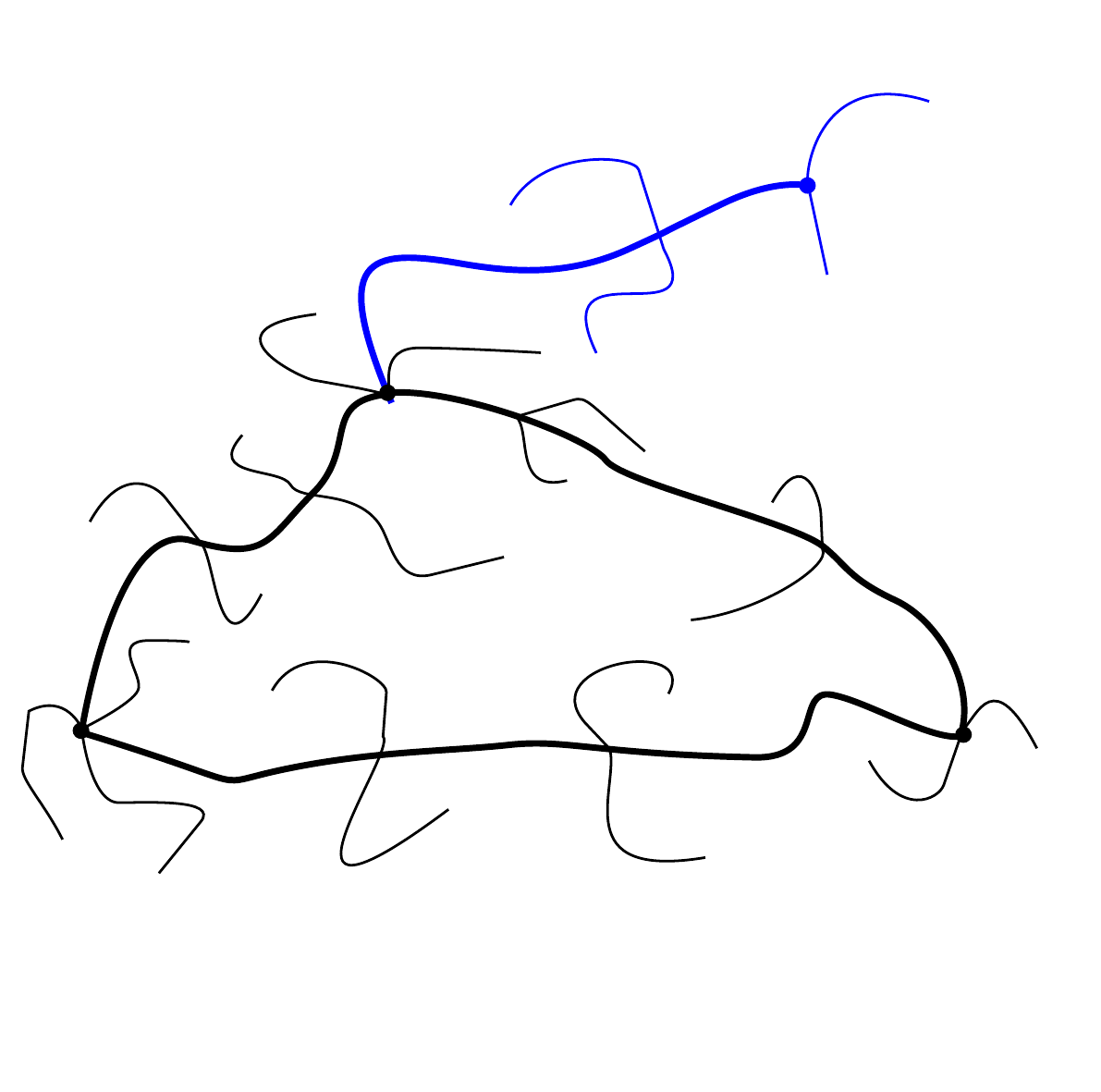}};
    \node at(0.1,2.4)   {$0$};
    \node at(5.5,1.7)   {$x$};
    \node at(2.2,3.6)   {$u$};
    \node at(4.8,5.1)   {$w$};
    \node[blue] at(2.2,5.0)   {$A_1$};
    \node at(0,3.1)   {$R_1$};
\end{tikzpicture}
\caption*{Let $A_1=B^{A}(u,w; p_1(u,w))$ and $R_1=A\setminus A_1$.}
\end{subfigure}\\
\newline
\begin{subfigure}{0.4\textwidth}\centering
\begin{tikzpicture}
    \node[anchor=south west,inner sep=0] at (0,0) {\includegraphics[height=6 cm]{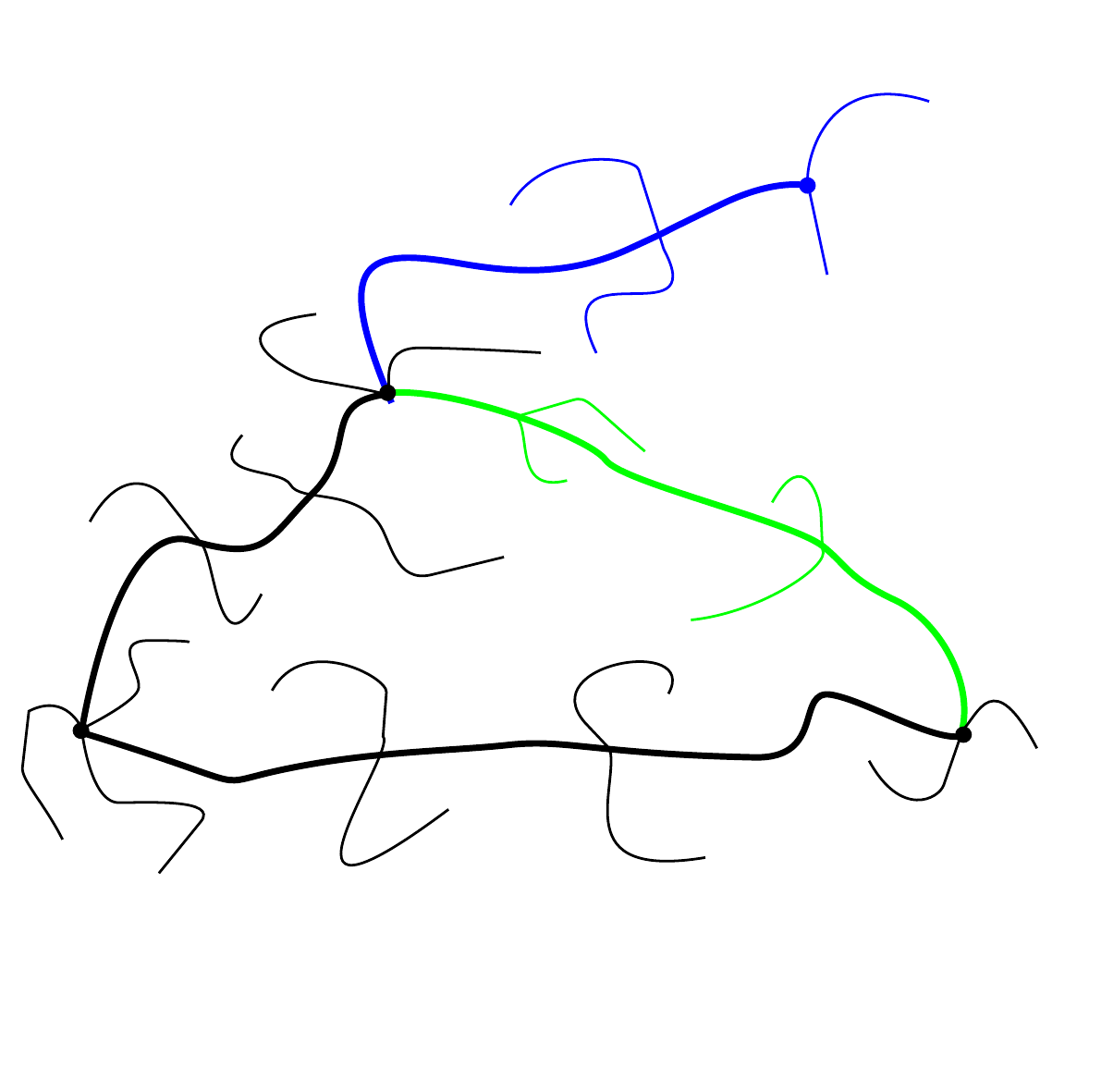}};
    \node at(0.1,2.4)   {$0$};
    \node at(5.5,1.7)   {$x$};
    \node at(2.2,3.6)   {$u$};
    \node at(4.8,5.1)   {$w$};
   \node[blue] at(2.2,5.0)   {$A_1$};
   \node[green] at(5,3)   {$A_2$};
    \node at(0,3.1)   {$R_2$};
\end{tikzpicture}
\caption*{Let $A_2=B^{R_1}(u,x; p_2(u,x))$ and $R_2=R_1\setminus A_2$.}
\end{subfigure}&
\begin{subfigure}{0.4\textwidth}\centering
\begin{tikzpicture}
    \node[anchor=south west,inner sep=0] at (0,0) {\includegraphics[height=6 cm]{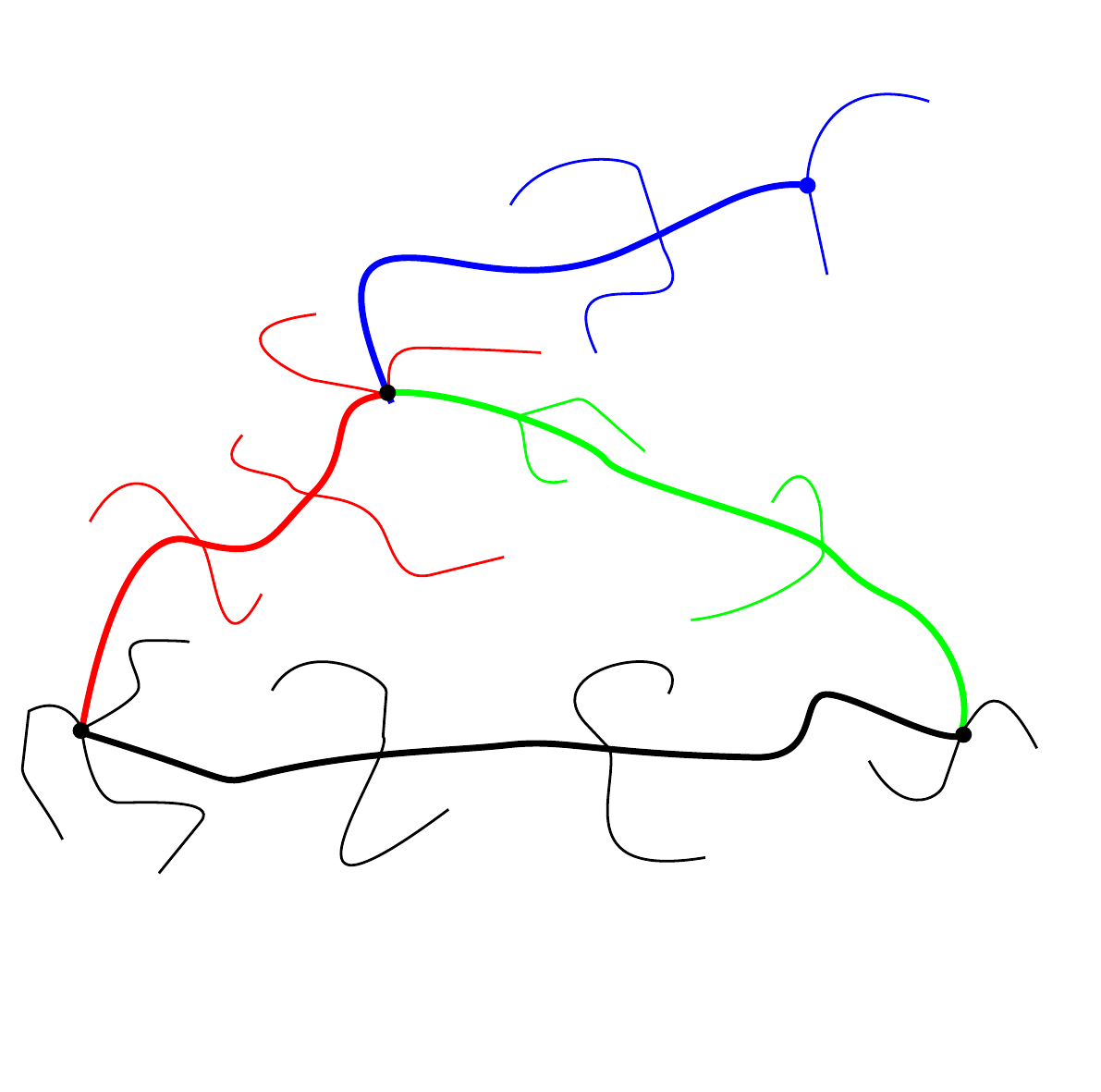}};
    \node at(0.1,2.4)   {$0$};
    \node at(5.5,1.7)   {$x$};
    \node at(2.2,3.6)   {$u$};
    \node at(4.8,5.1)   {$w$};
   \node[blue] at(2.2,5.0)   {$A_1$};
   \node[green] at(5,3)   {$A_2$};
   \node[red] at(0,3.1)   {$A_3$};
    \node at(2.7,1)   {$A_4$};
\end{tikzpicture}
\caption*{Let $A_3=B^{R_2}(0,u; p_3(0,u))$ and $A_4=A\setminus(A_1\cup A_2 \cup A_3)$.}
\end{subfigure}\\
\end{tabular}
\caption{The way that we split a LA to create the bound \refeq{split-LA-sausage-Gbar}.}
\label{split-LA-sausage-Gbar-Table}
\end{figure}
}

There are multiple ways to split such animals, and depending on the situation we might want to use a different combination of where to attribute the rib weights. To give an example, we can create animals as discussed above in which each line is bounded by $\tilde G_z$. For this, we modify the split shown in Figure \ref{split-LA-sausage-Gbar-Table} as follows:
\begin{enumerate}[i.)]
\item Let $R_3=B^{R_2}(x,0; p_4(x,0))$, and $R_4=R_2\setminus R_3$.
\item Let $\bar A_3=B^{R_2}(0,u; p_3(0,u))$ and $\bar A_4=R_3\setminus \bar A_3$.
\item Let $\bar A_2=R_4\cup A_2$.
\end{enumerate}
The split into $A_1,\bar A_2,\bar A_3$ and $\bar A_4$ is very much alike the one shown in Figure \ref{split-LA-sausage-Gbar-Table}, the only difference
being that every bond incident to $x$ that is not part of the path $p_4(x,0)$ is now part of $\bar A_2$ instead of $\bar A_4$.
This creates the bound
	\begin{align}
	\lbeq{split-LA-sausage-second}
	\sum_{A\colon 0,x,w\in A} \indic{0\dbct{A} x}z^{|A|}\leq& \sum_{u}\tilde G_{z}(u-w)\tilde G_{z}(u-x)\tilde G_{z}(u)\tilde G_{z}(x).
	\end{align}
Which of the different bounds is chosen depends on the precise nature of the bound that we are deriving. In fact, optimizing over such decisions is numerically an important ingredient of our method.

\paragraph{Bound on the sausage-walk in $\Xi^{\ssc[1]}_z(x)$.}
Here we continue our bound on $\Xi^{\ssc[1]}_z(x)$ started in Section \ref{sec-BoundIdea-Heuristic}, now specializing to LAs with a non-trivial first sausage.
We have bounded the case $\bb_1=0$ by splitting the sausage-walk into three parts:
\begin{enumerate}[(a)]
  \item The first sausage $S^\omega_0$ that connects $0$ and $w_1(\omega)$, bounded by $\bar G_z(w)$;
  \item The last sausage $S^\omega_{|\omega|}$ that connects $x$ and $w_1(\omega)$, bounded by $\bar G_z(w-x)$;
  \item The remaining sausage-walk without the first and last sausage, bounded by $2dz(D\star\tilde G_z)(x)$.
\end{enumerate}
For $\bb_1\neq 0$, which is not possible for LTs, we need to split the first sausage $S^\omega_0$.
The first sausage contains $0,\bb_1$ and $w_1(\omega)$ and doubly connects $0$ and $\bb_1$.
Thus, $S^\omega_0$ is a non-trivial sausage as discussed in the previous paragraph.
Combining these bounds we arrive at
	\begin{align}
	\lbeq{bound-xi1-structure}
	\Xi^{\ssc[1]}_z(x)\leq& 2dz \sum_{w} \bar G_{z}(w) \bar G_{z}(w-x)(D\star \tilde G_{z})(x)\\
	&+    2dz\sum_{w,\bb,u}\bar G_{z}(w-x)(D\star \tilde G_{z})(x-\bb)\tilde G_{z}(u-w)\tilde G_{z}(u-x)\tilde G_{z}(u)\tilde G_{z}(\bb),\nn
	\end{align}
where the second term is absent for LTs.

\paragraph{The allocation of one-point functions.}
As can be seen above, we have some choice of how we split the sausages. This becomes quite complex when more lines are involved.
Let us discuss the example $\Xi^{\ssc[2]}_z$, which is the simplest diagram of the more involved cases.
\begin{figure}[ht!]
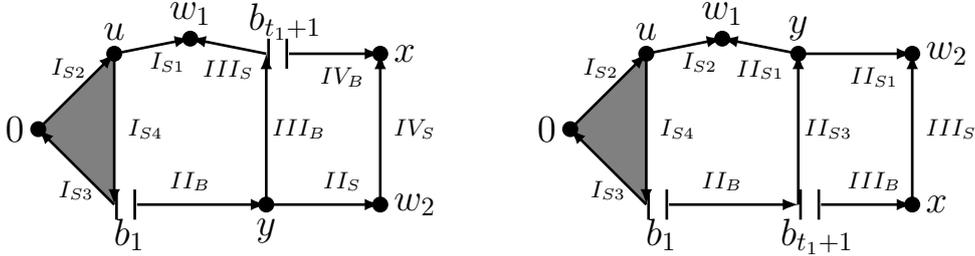

\begin{center}
{\Large
\picXiNtwoStructureLA[1]
}
\caption{The diagram of a contribution to $\Xi^{\ssc[2]}_z(x)$. We indicate the lace by $L=\{0t_1, s_2|\omega|\}$.
The left diagram correspond to $s_2<t_1$ and the right diagram to $s_2=t_1$.
For reference we label the individual lines by roman capital numbers.}
\label{BoundLA-Figure-Xitwo}
\end{center}
\end{figure}

In the following, we define how we split the diagram of $\Xi^{\ssc[2]}_z(x)$ for $s_2<t_1$, as shown in Figure \ref{BoundLA-Figure-Xitwo}.
Consider $\omega$ and a lace $L$ that contributes to $\Xi^{\ssc[2]}_z(x)$.
We know that $S^\omega_{0},S^\omega_{t_1}$ and $S^\omega_{s_2},S^\omega_{|\omega|}$ intersect,
at their first intersection points $w_1(\omega)$ and $w_2(\omega)$ (recall Definition \ref{defLTFirstPointIntersect}).

Let $u$ be the last common point of the paths $0\conn w_1(\omega)$ and $\bb_1\conn w_1(\omega)$ in $S_0^\omega$.
Let $y$ be the last point that the paths $\tb_1\conn w_2(\omega)$ and $\bb_{t_1+1}\conn w_1(\omega)$ in $\omega$ share.

For $\bb_1=0$, we have that $u=0$ and define $I_{S1}:=B^{S_0^{\omega}}_0(0,w_1)$ and $I_{S2}=I_{S3}=I_{S4}=\varnothing$.

For $\bb_1\neq 0$, the first sausage is non-trivial. We use the split that created \refeq{split-LA-sausage-second}, which formally corresponds to
	\begin{align*}
  	I_{S1}=&B^{S_0^{\omega}}_0(u,w_1), \quad R_1=S_0^{\omega}\setminus I_{S1},\quad  A_2=B^{R_1}_0(u,x),\\
      	R_2=&R_1\setminus A_2,\quad R_3=B^{R_2}_0(x,0), \quad R_4=R_2\setminus R_3,\\
      	I_{S2}=&B^{R_2}_0(0,u), \quad I_{S3}=R_3\setminus I_{S2}, \quad \text{and}\quad  I_{S4}=R_4\cup A_2.
	\end{align*}
Then, we split the three branches from the backbone, that create two intersections at $w_1$ and $w_2$, namely
$II_S:=B^ \omega_{t_1}(y,w_2)$, $III_S:=B^\omega_{s_2}(\bb_{s_2+1},w_1)$ and $IV_S:=B^\omega_{|\omega|}(x,w_2)$.
The trimmed sausage-walk then create $II_B$, $III_B$, $IV_B$.
By this construction, all ten paths connecting the labeled points are bond-disjoint and each piece is a LA.

All pieces with sub-index $S$ are planted animals and can be bounded by $\tilde G_z$.
For the backbone lines $II_B,III_B,IV_B$, we have to decide where to bound the sausages at the labeled points.
We can bound two of the pieces using $\tilde G_z$, taking into account only the weight of either the first or final sausage.
The third piece has to be bounded using $\bar G_z$, which bounds the sausages on both ends.
Thus, we bound only one connection by $\bar G_z$, rather than all four, which is a major improvement.
Using a slightly different split of the sausage-walk, we are actually free to choose any line to be bounded by $\bar G_z$.
This will be useful when dealing with weighted diagrams, we discuss this in detail in Section \ref{secLTProofBoundN1}.\\
For the case $t_1=s_2$, as shown in the right diagram of Figure \ref{BoundLA-Figure-Xitwo}, we need to
split the sausage $S^\omega_{t_1}$. For this we define $y$ to the be point that all connections in  $S^\omega_{t_1}$
from $\bb_{t_1+1}$ to $w_1$  and from $\bb_{t_1+1}$ to $w_2$ share and define
	\begin{align*}
  	II_{S1}=&B^{S_{t_1}^{\omega}}_0(y,w_1),\qquad II_{S2}=&B^{S_{t_1}^{\omega}}_0(y,w_2), \quad
 	II_{S3}=&B^{S_{t_1}^{\omega}\setminus(II_{S1}\cup II_{S2}) }_0(\bb_{t_1+1},y).
	\end{align*}
The rest is spit in the same way as described for the case $s_2<t_1$.
\tikzset{>=latex}
\begin{figure}[ht!]
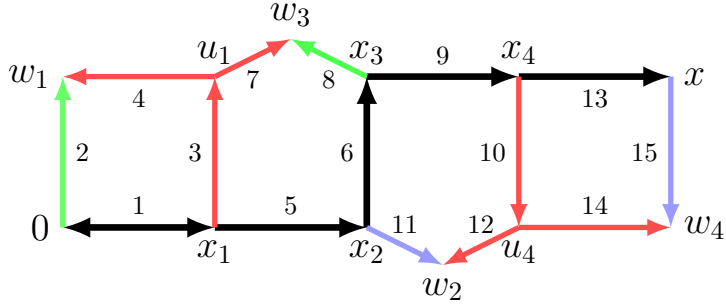

\begin{center}
{\Large
\picPiFourTreePictureRibs[1]
}
\caption{Picture of a possible $\Xi^{\ssc[4]}_z$ diagram. The backbone is marked by a thicker line and the ribs of $x_i$ have different colors. Arrow tips indicate that the one-point function is allocated at this endpoint. Thus, all lines except for one are $\tilde G$'s. Only one line has two arrow tips, which corresponds to a factor $\bar{G}$. This line can be located anywhere in the diagram. The numbers are labels for the corresponding two-point functions.
}
\label{LTXi4Ribweight}
\end{center}
\end{figure}

\begin{remark}[One-point function allocation]
\label{rem-rib-weight2}
{\rm In Figure \ref{BoundLA-Figure-Xitwo}, we have ten connections, nine of which can be bounded by $\tilde G_z$, and one by $\bar G_z(k)$.
In general, the diagram of $\Xi^{\ssc[N]}_z$ has $4N+2$ connections, and only one of them needs to be bounded by $\bar G_z$.
This means that by bounding each line by $\bar G_z$ instead of $\tilde G_z$,  as done previously in literature,
any bound on $\hat \Xi^{\ssc[N]}_z$ is unnecessarily a factor $\e^{4N+1}$ too large, at least for $\bb_1\neq 0$.
For $\bb_1=0$ three lines are trivial ($I_{S2}=I_{S3}=I_{S4}=\varnothing$), so that this unnecessary factor is still $\e^{4N-2}$.
We can and will choose which of the lines is bounded by $\bar G_z(k)$ in a way that sensitively depends on the precise structure of a diagram.
In Section \ref{secLTProofBoundN1}, we explain in detail how we do this for the lace-expansion coefficients with the spatial weight $|x|^2$,
for which it is the most relevant. For bounds without the spatial weights, we usually use the split shown in Figure \ref{LTXi4Ribweight}.
There, we use the bound $\bar G_z(x)\leq \gj \tilde G_z(x)$ for $x\neq 0$ for the single $\bar G_z$ factor, which simply extracts one
one-point function $\gj$, after which all connections of the diagram are given by $\tilde G_z$.
}
\end{remark}
\subsection{Repulsive diagrams}
\label{secBoundsRepdia}
In this section, we define the basic diagrams used to bound the NoBLE coefficients.
The ribs/sausages of the NoBLE coefficients have numerous avoidance constraints.
We incorporate some of these in our diagrams to improve our numerical bounds.
Here we explain how this can be done using so-called {\em repulsive diagrams}.

To motivate the precise definition of these diagrams, we review the avoidance structure of the first square $(1,2,3,4)$ in Figure \ref{LTXi4Ribweight}.
In the diagrams the vertices $x_i$ denote points on the backbone, the vertices $w_i$ denote the $i$th first intersection point (recall Definition \ref{defLTFirstPointIntersect}), and $u_i$ the vertices where we have to split a sausage/rib.
Alike the split of $\Xi^{\ssc[N]}_z$ given above, the split given in Figure \ref{LTXi4Ribweight} has the following interpretation:
\begin{itemize}
\item[$\rhd$] Line $2$ is the connected planted animal $B^{\omega}_0(0,w_1)$;\\[-6mm]
\item[$\rhd$] Line $3$ is the connected planted animal $B^{\omega}_{t_1}(x_1,w_1)$;\\[-6mm]
\item[$\rhd$] Line $4$ is the connected planted animal $B^{\omega}_{t_1}(w_1,u_1)$;\\[-6mm]
\item[$\rhd$] Line $1$ is the backbone $b^\omega(0,x_1)$, together with the sausages $(A^\omega_{i})_{i=1,\dots,t_1-1}$ and the (trimmed) animal
$A^{\omega}_{t_1}\setminus (B^{\omega}_{t_1}(x_1,w_1)\cup B^{\omega}_{t_1}(w_1,u_1))$.
\end{itemize}

These lines obey the following avoidance constraints:
\begin{itemize}
\item[$\rhd$] Relation between $1$ and $2$: By the definition of $J[0,|\omega|]$, the sausage $A_{t_1}^\omega$ is the first sausage to intersect with the sausage $A_0^\omega$.
    Thus, lines $2$ and $1$ only have the origin as starting point in common and are otherwise vertex-disjoint;\\[-6mm]
\item[$\rhd$] Relation between $1$ and $3,4$: The connections $3$ and $4$ are all in $A_{t_1}^\omega$.
    The indicator $J[0,|\omega|]$ ensures that $A_{t_1}^\omega$ does not intersect any of the intermediate sausages $(A_i^\omega)_{i=1,2,\dots,t_1-1}$.
    Further, $1$ bounds the part of $A_{t_1}^\omega$ that has not been split off as planted animal into 3, indicated by the tip of the arrow in  Figure \ref{LTXi4Ribweight}.
    Thus, $1$ is bond-disjoint from $3$ and $4$;\\[-6mm]
\item[$\rhd$] Relation between $3$ and $4$: Both are part of the sausage $A_{t_1}^\omega$ and are thus bond-disjoint;\\[-6mm]
\item[$\rhd$] Relation between $2$ and $3,4$: We know that $A_{0}^\omega$ and $A_{t_1}^\omega$ intersect in at least one point.
   We choose $w_1=w_1(\omega)$ to be the \emph{first intersection point}, as defined in Definition \ref{defLTFirstPointIntersect}.
   By this choice the backbones $b^{A^\omega_0}(0,w_1)$ and $b^{A^\omega_{t_1}}(x_1,w_1)$ only intersect at $w_1$.
\end{itemize}
We define repulsive diagrams to bound such diagrams. As seen for $2$ and $3,4$, the constraint that the pieces are bond-disjoint is too strong.
For this reason, we denote by repulsiveness of a diagram that we can find bond-disjoint paths connecting the corner points of the diagram.

In the following, we define the {\em skeleton} of a diagram that encodes the avoidance constraints of the backbone lines, the lengths of the backbone lines
and the information whether the start or the end points of a connection is the root of a planted animal. This is formalized as follows:

\begin{definition}[Mutually avoiding skeleton]
\label{defSkeletonLA}
Let $x_0=0\in\Zd$ and $n\in\{1,2,3,4\}$. For each $i\in\{1,\dots,n\}$, let $x_i\in\Zd, l_i\in\Nbold$, the index $j_i\in\{l_i,\underline l_i\}$, $s_i\in\{+,-\}$, and let $\omega^i$ be a sausage-walk from $x_{i-1}$ to $x_i$. We define $S^{s_1,\dots,s_n}_{j_1,\dots,j_n}(\omega^1,\dots,\omega^n)$ to be the indicator that the following holds:
\begin{enumerate}[(1)]
\item There exists a sequence of paths $(p^i)_{i=1}^n$ such that $p^i$ is a path from $x_{i-1}$ to $x_i$ and $p^i$ uses only edges in $\omega^i$, each path $p^i$ describes a self-avoiding walk and all paths are pairwise bond-disjoint;
\item For each $i\in\{1,\dots,n\}$, we can choose $p^i$ in item (1) such that, $|p^i|\geq l_i$ when $j_i=l_i$, while $|p^i|=l_i$ when $j_i=\underline l_i$;
\item For each $i\in\{1,\dots,n\}$, the sausages of the walk $(A^{\omega^i}_j)_{j=0,\dots,|\omega^i|}$ do not intersect, so that $\omega^i$ describes a LA.
\item For each $i\in\{1,\dots,n\}$,
if $s_i=+$ then $A^{\omega^i}_{0}=\varnothing$, while $A^{\omega^i}_{|\omega^i|}=\varnothing$ when $s_i=-$.
\end{enumerate}
\end{definition}

\noindent
We next use this skeleton to define the repulsive bubble, triangle and square diagrams as follows:
\begin{definition}[Repulsive diagrams]
\label{def-repulsive-diamgrams}
For $i=1,2,3,4$ let $x_i\in \Zd,\ l_i\in\Nbold,\ j_i\in\{l_i,\underline l_i\}$ and $x_0=0$. We define
	\begin{align}
	\diagRepulsiveLetter{B}_{j_1,j_2}&(x_1,x_2)
 	\lbeq{defLTBubble}
	=\max_{s_1,s_2\in\{+,-\}}\sum_{\omega^1\in\Wcal(x_0,x_1)} \sum_{\omega^2\in\Wcal(x_1,x_2)}
	\left(\prod_{i=1}^2 z^{|\omega^i|} Z[0,|\omega^i|]\right) S^{s_1,s_2}_{j_1,j_2}(\omega^1,\omega^2),\\
	\diagRepulsiveLetter{T}_{j_1,j_2,j_3}&(x_1,x_2,x_3)\nnb
	=&\max_{s_1,s_2,s_3\in\{+,-\}}
	\sum_{\footnotesize\begin{array}{c}
        \omega^1\in\Wcal(x_0,x_1)\\
        \omega^2\in\Wcal(x_1,x_2)\\
        \omega^3\in\Wcal(x_2,x_3)
      	\end{array}}
	\left(\prod_{i=1}^3 z^{|\omega^i|} Z[0,|\omega^i|]\right) S^{s_1,s_2,s_3}_{j_1,j_2,j_3}(\omega^1,\omega^2,\omega^3),
 	\lbeq{defLTTriangle} \\
	\diagRepulsiveLetter{S}_{j_1,j_2,j_3,j_4}&(x_1,x_2,x_3,x_4)\nnb
	=&\max_{s_1,s_2,s_3,s_4\in\{+,-\}}\sum_{\omega^1\in\Wcal(x_0,x_1)}\sum_{\omega^2\in\Wcal(x_1,x_2)}\sum_{\omega^3\in\Wcal(x_2,x_3)}\sum_{\omega^4\in\Wcal(x_3,x_4)}\nnb
	&\qquad\qquad\qquad\qquad\times\left(\prod_{i=1}^4 z^{|\omega^i|} Z[0,|\omega^i|]\right) S^{s_1,s_2,s_3,s_4}_{j_1,j_2,j_3,j_4}(\omega^1,\omega^2,\omega^3,\omega^4).
	 \lbeq{defLTSquare}
	\end{align}
Further, for $i\in\{1,\dots,n\}$ and $x\in\Zd$, we define
	\begin{align}
	 \lbeq{defLTDouble}
    	\diagRepulsiveLetter{D}_i(x)=\sum_{A\colon 0,x\in A} \indic{0\dbct{A} x}\indic{d_A(0,x)\geq i}z^{|A|}.
	\end{align}
\end{definition}
Each line of these repulsive diagrams in \refeq{defLTBubble}, \refeq{defLTTriangle} and \refeq{defLTSquare} are planted animals and can therefore be bounded using the modified two-point functions $\tilde G_{z}$. The double connection in \refeq{defLTDouble} is special, and we will treat this separately in Section \ref{secBoundsDouble} below.

To use the constraints on the lengths of connections arising from loops consisting of at least 4 bonds in the NoBLE, we further define, for $l\in\mathbb{N} $ and $j\in {l,\underline l}$,
	\begin{eqnarray}
	\tilde G_{j,z}(x)=\sum_{\omega\in\Wcal(x_0,x)}  z^{|\omega|} Z[0,|\omega|] S^{+}_{j}(\omega),
	\end{eqnarray}
and note that $\tilde G_{0,z}(x)=\tilde G_{z}(x)$.
Using $\tilde G_{j,z}$, we can obviously bound the repulsive diagrams by ignoring the self-avoidance constraints as
	\begin{eqnarray}
	\lbeq{badBoundSquareLTLA}
	\diagRepulsiveLetter{T}_{j_1,j_2,j_3}(x_1,x_2,x_3)&\leq &\tilde G_{j_1,z}(x_1)\tilde G_{j_2,z}(x_2-x_1)\tilde G_{j_3,z}(x_2-x_3).
	\end{eqnarray}
However, as demonstrated by the fourth improvement in Section \ref{sec-BoundIdea-Heuristic}, the use of
repulsiveness improves our numerical bounds drastically, so we do not use the worse bounds as in \refeq{badBoundSquareLTLA}.

\subsection{Bounds on double connections}
\label{secBoundsDouble}
The bound on contributions in which two-points are doubly connected, that is, connected by two-bond disjoint paths,
is of central interest for LAs. This is especially relevant since a large contribution to the NoBLE lace-expansion coefficients for LAs is given by
	\begin{align}
	\lbeq{split-LA-doublePrimitiveRecall}
	\bar \Xi^{\ssc[0]}_z(x)= \sum_{A\colon 0,x\in A} \indic{0\dbct{A} x}z^{|A|}.
	\end{align}
Using our notation we bound this $\diagRepulsiveLetter{D}_0(x)$, defined in \refeq{defLTDouble}.\\
Let us briefly discuss a numerical aspect of our bounds. Using the techniques described in the previous section, we obtain
\begin{align}
	\lbeq{split-LA-doublePrimitive}
	\Xi^{\ssc[0]}_z(x)\leq  \min\big\{ 2d z (D\star \tilde G_{z})(x)\bar G_{z}(x), \tilde G_{z}(x)^2\big\}.
	\end{align}

We now explain how we can improve this bound by a factor close to $2$. This is possible as the two paths that created the double connection are interchangeable, so that
\refeq{split-LA-doublePrimitive} bounds each possible animal {\em twice}.

For the technical derivation of this improvement, we consider yet another split of the LA $A$ given as follows:
\begin{enumerate}[(i)]
\item Let $p_1=(b^1_1,b^1_2,\dots,b^1_{n})$ and $p_2(0,x)=(b^2_1,b^2_2,\dots,b^2_{m})$ be two bond-disjoint paths between $0$ and $x$ in $A$,
that we choose in some unique way for each $A$, e.g.\ $p_1$ is the smallest path in the lexicographic order such that there exists such a bond-disjoint $p_2$, and then $p_2$ is the smallest of these bond-disjoint paths, so that $p_1$ is smaller in lexicographic order than $p_2$ as we will assume from now on.
\item Let $A_1=B^{A}(0,b^1_{n};p_1)$. 
\end{enumerate}
Then,
	\begin{align}
	\sum_{A\colon 0,x\in A} \indic{0\dbct{A} x}z^{|A|}=&\sum_{A} \indic{0\dbct{A} x}z^{|A_1|}z^{|A\setminus A_1|}
	\lbeq{split-LA-double-tmp}
	\leq 2d z (D\star \tilde G_{z})(x)\bar G_{z}(x).
	\end{align}
By symmetry, we could also have performed the split of the animal $A$ according to the path $p_2$, instead.

We next improve upon this bound. Since the paths $p_1$ and $p_2$ are bond-disjoint, we conclude that $\tb^1_1=\ve[\iota]\neq \tb^2_1=\ve[\kappa]$, with $\iota<\kappa$ since $p_1$ is smaller than $p_2$ in the lexicographic order. Remove the bonds in the paths $p_1$ and $p_2$ from $A$. That divides the LA into at most $n+m+2$ disjoint LAs containing the vertices of $p_1$ and $p_2$. Thus, now defining $A_0$ the LA containing $0$, $A_1=B^{A'}(b^1_1,x;p_1)\setminus A_0$, where $A'$ is the LA obtained from $A_2$ by removing $p_2$ and keeping the parts that do not contain vertices in $p_2$. Then, by definition of $B^{A'}(b^1_1,x;p_1)$, $A_1$ is disjoint from $p_2$, and the sausage at $\tb_1^1$ is trivial in $A_1$. Finally, define $A_2=A\setminus (A_0\cup A_1)$ (which contains the bonds in $p_2$). Note that, by construction,  the sausage of $x$ in $A_2$ is trivial, so that $A_2$ is a planted LA, while also the sausage at $0$ is trivial. Then,
	\begin{align}
	\sum_{A\colon 0,x\in A} \indic{0\dbct{A} x}z^{|A|}=
	&\sum_{\iota}\sum_{\kappa>\iota}  \sum_{A} \indic{0\dbct{A} x}\indic {\tb^1_1=\ve[\iota]}
	\indic {\tb^2_1=\ve[\kappa]}z^{|A_0|} z^{|A_1|} z^{|A\setminus (A_0\cup A_1)|}\nnb
	\leq & \sum_{\iota}\sum_{\kappa>\iota} z^2  \gj \tilde G_{z}(x-\ve[\iota]) G_{z}(x-\ve[\kappa])\nnb
	\leq & \frac{1}{2} (2d z  \gj)^2 (D\star \tilde G_{z})(x)(D\star \bar G_{z})(x),
	\lbeq{split-LA-double-tmp2}
	\end{align}
as the bound \refeq{split-LA-double-tmp} does not assume anything on the order of the paths. Here, in the first inequality, $z^2$ corresponds to the weight of the bonds $b_1^1, b_1^2$, the factor $\gj$ comes from the LA $A_0$, the factor $\tilde G_{z}$ from the LA $A_1$, and the factor $\bar G_{z}$ from the LA $A_2$.
Using the same arguments we obtain the factor $1/2$ for all other terms in \refeq{split-LA-doublePrimitive}. In our final numerical bounds, we will drastically improve upon this bound by using the repulsive nature of the two paths $p_1, p_2$.

\section{The diagrammatic bounds}
\label{secBounds}
In the preceding section we have explained the technical details of our method of bounding LT/LA diagrams.
In this section we first combine the repulsive diagrams \refeq{defLTBubble}-\refeq{defLTSquare}, to larger diagrams that we call {\em building blocks}.
Then, we combine these building blocks to the create bounds on the NoBLE coefficients.
We close this section with an overview of where to find all bounds required for the analysis, as stated in \cite[Assumption 4.3]{FitHof13b}.

In Section \ref{secBoundProof} we prove the bounds for $N=0$ and sketch how the remaining bounds are proven in the example of $N=1,2$. The rigorous proof of these bounds is tedious and is based on techniques that are relatively standard in lace-expansion analyses. For this reason we omit many details.

\subsection{Building blocks}
\label{secBuildingBlocks}
In this section we combine the simple diagram, defined in \refeq{defLTBubble}-\refeq{defLTSquare}, to larger diagrams, that we call {\em building blocks}. These building blocks are designed
to encode the following information:
\begin{itemize}
  \item[$\rhd$] the repulsiveness of the paths involved;
  \item[$\rhd$] the lengths of lines shared by two loops;
  \item[$\rhd$] the role of the lines shared by two loops, meaning whether it is part of the backbone or not, and, if so, how the backbone is connected to the line in question.
\end{itemize}
We encode the repulsiveness by defining the blocks using the repulsive diagram \refeq{defLTBubble}-\refeq{defLTSquare}. In the following we discuss the other two features.

\paragraph{Aim: encoding backbone-paths.}
In Figure \ref{LTXi4Ribweight}, we have shown an example of the diagram contributing to $\Xi^\ssc[4]_z$.
However, the form shown in Figure \ref{LTXi4Ribweight} is only one of eight possible patterns that the backbone can form. These patterns arise by the different cases for the underlying lace, i.e. whether $s_{i+1}=t_i$ or $s_{i+1}<t_i$ for $i=1,2,3$. In Figure \ref{LTXi4Struc} we show all the possible backbone patterns for the $\Xi^{\ssc[4]}$-diagrams. One function of the building block will be to encode all possible patterns at once, so as to encode the roles that the shared lines play compared to the backbone.
\begin{figure}[h!]
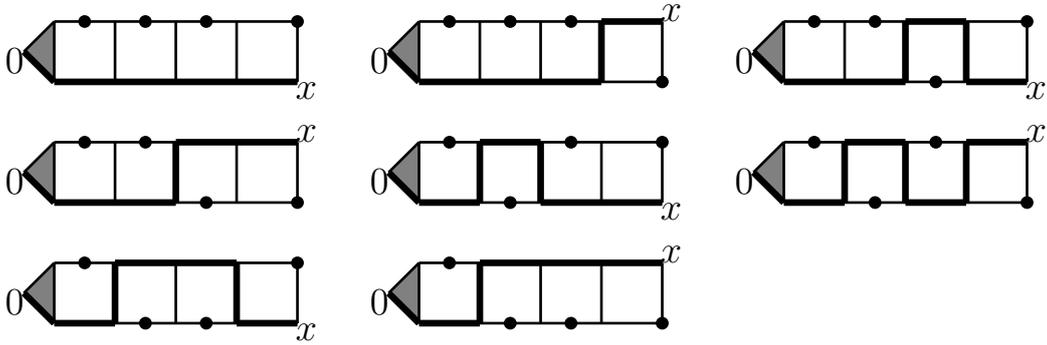

\begin{center}
{\Large
\picPiFourTreePicture[0.8]}
\caption{Picture of all possible $\Xi^{\ssc[4]}_z$ diagrams.
The backbone is marked by a thick line and the four \emph{first intersection points} are marked as dots.
The leftmost diagram on the second row corresponds to the diagram shown in Figure \ref{LTXi4Ribweight}.}
\label{LTXi4Struc}
\end{center}
\end{figure}

\noindent
\paragraph{Aim: encoding of the line lengths.}
When creating a bound on the NoBLE coefficients we want to make full use of all the avoidance constraints of a sausage-walk, as well as the fact that any closed loop consists of at least four bonds. To be able to do this, we consider different cases for the lengths of paths that are shared by two squares. Here, the length of a path is given by the number of bonds it uses.

\paragraph{Our solution.}
We use the indices $l,m\in\{-2,-1,0,1,2\}$ to encode the information of the length and role of a shared line.
We explain their meaning at the example of the diagram given in Figure \ref{LTXisharedLabeling}.
\begin{figure}[h!]
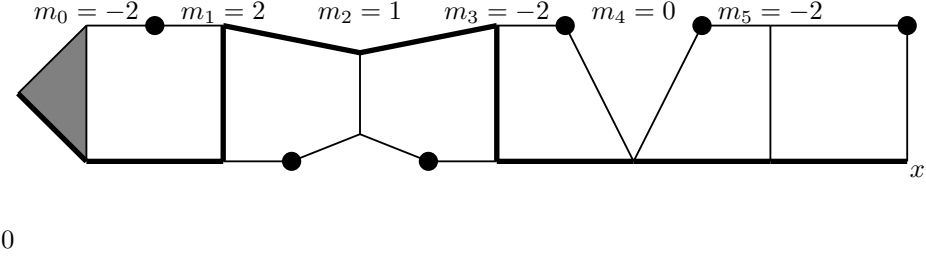

\begin{center}
{\Large
\picPiSharedLabelings[0.9]}
\caption{Example of a description of the backbone-path and the length of the shared line using the indices $(a_i)_i$. }
\label{LTXisharedLabeling}
\end{center}
\end{figure}
As indicated in Figures \ref{LTXi4Ribweight}, \ref{LTXi4Struc} and \ref{LTXisharedLabeling}, we always draw the diagrams in such a way that the origin is on the left side and the first piece of the sausage-walk from $\bb_0$ to $\bb_{s_2}$ is the lower line of the first square. A positive index $m_i$ indicates that the backbone ends on the upper part of a piece of the diagram.
Interactions between neighboring squares will be particularly important, in which case there will be two indices, for example in $A^{m_{i-1},m_i}$. In this case, $m_{i-1}>0$ indicates that the backbone in the $(i-1)$st square starts on the upper part of the diagram, while $m_i>0$ indicates that the backbone ends on the upper part of a piece of the diagram.
By $|m_i|=1$, we denote that the vertical line consists of exactly one bond and $|m_i|=2$ denotes that
the vertical line consists of at least two bonds. If $m_i=0$, then the vertical line is trivial, i.e., it consists of no bonds.
In the case that $m_i=0$, we draw the following piece of the backbone on the lower part of the following square.

\paragraph{Informal definition.}
In this section, we informally define the diagrams, also called {\em building blocks}, that we use to bound the NoBLE coefficients. The precise properties of the building blocks originate from the properties of the NoBLE coefficients, as we will see when we bound the coefficients in Section \ref{secBoundProof}. The list below consists of 9 basic building blocks:
\begin{enumerate}[(i.)]
\item Open square $A^{l,m}$.\\
We use this building block to bound the intermediate pieces of the diagrams appearing in the NoBLE coefficients. It consists of an open square, with 4 of the 5 vertices fixed and one being summed over.
\begin{center}
\picAGeneral[0.6]
\end{center}
We label the vertices in the diagram $A^{l,m}$ such that $u$ is the starting vertex and $x$ is the ending vertex of the backbone.
The indices $l,m$ identify the lengths of the connections between $u,v$ and $x,y$, respectively, i.e., the number of bonds in these connections.
For example: If $m=0$ then $x=y$. If $|m|=1$ then $x$ and $y$ are direct neighbors and are directly connected. The case where $|m|=2$ summarizes all the other cases.

A positive index $l,m$ indicates that the backbone starts/ends at the top. This means that if $lm\geq 0$, then the connection $x$ to $y$ is not part of the backbone.
The horizontal parts of the backbones (i.e., the paths from $u$ to $x$ or from $u$ to $y$) are always non-trivial, in that they use at least one bond.
If the backbone $b^A(u,x)$ consists of only one bond, then we know that $w\nin \{u,x\}$.
Closing the diagram, by including $u\conn v$, always creates a diagram with at least four bonds.

\item Double open triangle $\bar A^{l,m}$.

\begin{center}
\picAbarGeneral[0.6]
\end{center}
This diagram follows the same rules as the open square $A^{l,m}$, the difference being that the connection $x$ to $y$ does not contribute (i.e., that line is absent and used in another building block). We only draw the example with $l,m>0$. The role of $l$ and $m$ in $\bar A^{l,m}$ is the same as for $A^{l,m}$.
\item Weighted open triangle $C^{l,m}$.

\begin{center}
\begin{figure}[ht]
\picDGeneral[0.8]
\caption{This diagram follows the same rules as the double open triangle $\bar A^{l,m}$.
Additionally, the lower connections are not planted animals. This means that the connection with spatial weight
might have a non-trivial one-point function at both end-points, thus should be bounded by $\bar G_z$.}
\label{C-ml-def}
\end{figure}
\end{center}

\item Closed fundamental block $P^{\ssc[1],m}.$\\
This diagram is used to bound the first square of $\Xi^{\ssc[N]}_z$ and has the following shape for $m>0$:
\begin{center}
\begin{figure}[ht]
\picPOneLA[0.6] \\
\picPOneRibLA[0.6]
\caption{For $m=0$, we have that $x=y$, apart from that, the diagram remains the same.}
\label{P1pieceSketch}
\end{figure}
\end{center}
Each triangle, square and pentagon consists of at least four bonds. The first diagram characterizes the case where the first step of the backbone starts at the origin, i.e., $\bb_1^\omega=0$.
In the following, we use $u=\bar{b}_1^\omega$ to denote the first step of the backbone. Note that only $\bb_1^\omega= 0$ contributes for LT.
If the backbone between $u$ and $x$ consists of only one step, then $w\nin\{u,x\}$.
The pieces of the backbone $b^A(u,x)$, $b^A(0,y)$ and $b^A(x,y)$, respectively, are non-trivial.
Note that the left diagram is the same diagram as $A^{0,m}(0,0,x,y)$.
\item Open fundamental block $\bar P^{\ssc[1],m}.$\\
This diagram corresponds to the closed fundamental block $P^{\ssc[1],m}$, in which the connection $x$ to $y$ does not contribute (but we keep information about its length using the index $m$).

\item Weighted fundamental block $\Delta^{\text {start},m}$.\\
This diagram corresponds to the open fundamental block $\bar P^{\ssc[1],m}$, where the lower connection is weighted by its displacement, i.e., $|x|^2$ for $m\leq 0$ and $|y|^2$ for $m\geq 0$.

\item Weighted terminal block $\Delta^{\text {end},m}$.\\
This diagram is defined by $\Delta^{\text {end},m}(u,v,x)=C^{m,0}(u,v,x,x)$ and is used to bound the last square of $\Xi^{\ssc[N]}_z$ and $\Xi^{\ssc[N],\iota}_z$.

\item Iota-fundamental block $P^{\ssc[1],\iota,m}.$\\
This diagram is used to bound the first square of $\Xi^{\ssc[N],\iota}_z$ and has the following structure for $m<0$:
\begin{center}
\picXilotaInitialSketch[0.6]
\end{center}
This diagram is similar to $P^{\ssc[1],m}$. The sausage-walks that we bound with this diagram have the property that either $\ve[\iota]$ is part of the first sausage or that $\bb_1^\omega$ ends in $\ve[\iota]$. If $\ve[\iota]$ is in the first sausage, then we extend the idea of \refeq{split-LA-sausage-second}.
We identify the last point $v$ that the double connection $0$ to $\bb_1$ and a connection $0$ to $w$ share.
Then, we identify the last common point $u$ of the connection $0$ to $\ve[\iota]$ to the rest.
This point $u$ could be on all possible connections $0$ to $\bb_1$, $0$ to $v$, $\bb_1$ to $v$ as well as $v$ to $w$, as indicated in the figure above.

\item Weighted iota-fundamental block $\Delta^{\text {iota},{\sss I},m},\Delta^{\text {iota},{\sss II},m}$.\\
These diagrams have the same shape as $P^{\ssc[1],\iota,m}$, but the connection $x$ to $y$ does not contribute. Further, $\Delta^{\text {iota},{\sss I},m}$ is weighted by a factor $|x|^2$ and $\Delta^{\text {iota},{\sss II},m}$ by a factor $|x-\ve[\iota]|^2$ for $m\leq 0$.
\end{enumerate}

\iflongversion
The formal definition of these diagrams is given in Appendix \ref{sec-formaldefinition}.
\else
In this version we omit the detailed definition of these diagrams, as they do not create additional  insight and are rather lengthly. There definition is given in the appendix of the extended version \cite{FitHof18ext}.
\fi

\paragraph{Elements of the bounds.}
Here we define the quantities that we use to state the bounds on the NoBLE coefficients.
In the following, we define the vectors $\vec P,\vec P^{\iota},\vec E,\vec \Delta^{\rm{start}},\vec \Delta^{\rm{end}},\vec \Delta^{\iota}\in\Rbold^{5}$
and the matrices ${ \bf {\bar A}},{\bf A},{\bf C}\in\Rbold^{5\times 5}$ by their entries. For convenience, we use $\{-2,-1,0,1,2\}$ as indices for the entries of the vectors and matrices.
Let
	\begin{align}
	\lbeq{defOfBoundingElement1}
	(\vec P)_m&= \sum_{x,y} P^{\ssc[1],m}(x,y),\qquad\qquad
	(\vec P^\iota)_m= \sum_{x,y} P^{\ssc[1],\iota,m}(x,y),\\
	(\vec \Delta^{\rm{start}})_m&= \sup_{y}\sum_{x} \Delta^{\rm{start},m}(x,x+y), \quad
	(\vec \Delta^{\rm{end}})_m= \sup_{y} \sum_{x} \Delta^{\rm{start},m}(x,x+y),\\
	(\vec \Delta^{\iota,{\sss I}})_m &=\sup_{y} \sum_{x} \Delta^{\iota,{\sss I},m}(x,x+y),\qquad
	(\vec \Delta^{\iota,{\sss II}})_m =\sup_{y} \sum_{x} \Delta^{\iota,{\sss II},m}(x,x+y),
	\end{align}
and
	\begin{align}
    \lbeq{defOfBoundingElement2a}
	({\bf A})_{l,m}&= \sup_{v} \sum_{x,y} A^{l,m}(0,v,x,y),\qquad\qquad
	({\bf C})_{l,m}= \sup_{v,y} \sum_{x} C^{l,m}(0,v,x,x+y),\\
	(\vec E^{\rm{open}})_m&= \sum_{v,x} A^{m,0}(x,x,v,0), \qquad\qquad (\vec E^{\rm{closed}})_m= \sum_{v,x} A^{0,m}(x,x,v,0).
	\lbeq{defOfBoundingElement2}
	\end{align}
Using the bootstrap assumption $f_i(z)<\Gamma_i$ we can compute explicit numerical bound for each of these quantities. We use these quantities to formulate the bounds on the NoBLE coefficients.
\subsection{The bounds}
\label{secStatingtheBounds}
Now we state the required bound on the coefficients. For this we use the notation $\indAnimal$, as used in \refeq{conj-zc}, where $\indAnimal$ is $1$ in bounds for LAs and $0$ in bounds for LTs. We have to be especially careful in our bounds for $N=0$ and $N=1$, which turn out to be the major contributions to the NoBLE coefficients. Therefore, we distinguish between the bounds for $N=0$, $N=1$ and $N\geq 2$. We start by formulating the bounds on $\Xi^{\ssc[0]}$ and $\Psi^{\ssc[0],\kappa}$:

\begin{lemma}[Bounds on $\Xi^{\ssc[0]}$ and $\Psi^{\ssc[0],\kappa}$]
\label{BoundNZero-Xi}
Let $0\leq z<z_c$. Then,
	\begin{align}
	\lbeq{Bound-Xi0}
	\sum_{x\in\Zd}\Xi^{\ssc[0]}_z(x)&\leq \frac {\indAnimal } {g_z} \sum_x\diagRepulsiveLetter{D}_{1}(x),\qquad\quad
	\sum_{x\in\Zd}\Xi^{\ssc[0]}_{{\sss R},z}(x) \leq \frac {\indAnimal} {g_z} \diagRepulsiveLetter{D}_{2}(x),\\
	\lbeq{Bound-Psi0RemII}
	\sum_{x\in\Zd} \Psi^{\ssc[0],\kappa}_{{\sss R, II},z}(x)&\leq \frac {2d-2}{2d} \frac {\indAnimal} {g^\iota_z} \diagRepulsiveLetter{D}_{2}(x,0),
	\end{align}
and
	\begin{align}
	\lbeq{Bound-Xi0W}
	\sum_{x\in\Zd}|x|^2\Xi^{\ssc[0]}_z(x)&\leq \frac {\indAnimal} {2}
	\sum_x|x|^2  G_z(x) 2dz (D \star \tilde G_z)(x),\\
	\lbeq{Bound-Xi0RemW}
	\sum_{x\in\Zd}|x|^2\Xi^{\ssc[0]}_{{\sss R},z}(x)&\leq \frac {\indAnimal} {2}
	\sum_x|x|^2  G_z(x) (2dz)^2g_z (D^{\star 2} \star \tilde G_z)(x),\\
	\lbeq{Bound-Psi0RemIIW}
	\sum_{x\in\Zd} |x|^2 \Psi^{\ssc[0],\kappa}_{{\sss R, II},z}(x)&\leq \frac {2d-2}{2d} \frac {\indAnimal} {2}
	\sum_x|x|^2  G_z(x)  (2dz)^2g_z (D^{\star 2} \star \tilde G_z)(x).
	\end{align}
Further,
	\begin{align}
	\lbeq{Bound-Psi0RemI}
	\sum_{x\in\Zd} \Psi^{\ssc[0],\kappa}_{{\sss R, I},z}(x)\leq& \frac {\indAnimal} {g_z^\iota}\big(
	(2d-1) \diagRepulsiveLetter{B}_{\underline 1,3}(\ve[1],0)+\frac {2d-2}{2d} \sum_x\diagRepulsiveLetter{B}_{2,2}(x,0)\big),\\
	\lbeq{Bound-Psi0RemIW}
	\sum_{x\in\Zd} |x-\ve[\kappa]|^2 \Psi^{\ssc[0],\kappa}_{{\sss R, I},z}(x)\leq &
	4d\indAnimal \diagRepulsiveLetter{B}_{\underline 1,3}(\ve[1],0)\\
	&+ \frac {\indAnimal} {2}\sum_x (1+|x|^2)(2dz)^2g_z (D^{\star 2} \star \tilde G_z)(x) G_z(x).\nn
	\end{align}
	\end{lemma}

We continue by formulating the bounds on $\Xi^{\ssc[0],\iota}$ and $\Pi^{\ssc[0],\iota,\kappa}$. To state the bound on $\Xi^{\ssc[0],\iota}$ and the remainder of the split $\Xi^{\ssc[0],\iota}_{{\sss R, I},z}$ and $\Xi^{\ssc[0],\iota}_{{\sss R, II},z}$, we use the abbreviations $\alpha_{n,m}$, $\beta_{n}$ and $\gamma_{n}$, that will only be used in this lemma, and which, for $n,m\in\{0,1,2\}$, are given by
	\begin{align}
	\lbeq{BoundNZero-XiIota-alpha}
	\alpha_{n,m}=&
	\frac {\indAnimal}{g_z}\tilde G_z(\ve[\iota]) \sum_{x} \diagRepulsiveLetter{D}_{n}(x)
	+\frac {\indAnimal}{g_z}  \indic{m=0} \diagRepulsiveLetter{D}_{1}(\ve[\iota])
    +\frac {\indAnimal}{g_z}  \sum_{x} \diagRepulsiveLetter{T}_{1,\max\{1,m\},n} (\ve[\iota],x,0)\\
	&+\frac {\indAnimal}{g_z}
\big( \indic{n\leq 1}  \sup_{y} \tilde G_{2,z}(y) (2d-1)zg_z \tilde G_{3,z}(\ve[\iota])
+\sup_{y} \tilde G_{m,z}(y) \sum_{x\neq 0} \diagRepulsiveLetter{D}_{2}(x)\big)\nnb
&+\frac {\indAnimal}{g_z}   \sum_{u}\diagRepulsiveLetter{B}_{1,1}(u,\ve[\iota])  \sup_{u} \sum_{x\neq 0 }\diagRepulsiveLetter{B}_{n,1}(x,u)\big),\nnb
	\lbeq{BoundNZero-XiIota-beta}
	\beta_{n}=&2dz\frac {\indAnimal}{g_z}\tilde G_z(\ve[\iota]) \sum_{x} |x|^2\bar G_z(x)
	\big(\frac 1 2  (D\star \tilde G_{n,z})(x)+(D\star \tilde G_z)(x-\ve[\iota])\big)\\
	&+2dz\frac {\indAnimal}{g_z}  \sup_{y} \tilde G_{1,z}(y)\sum_{x}|x|^2\bar G_z(x)
	\big(\frac 1 2 (D\star \tilde G_{1,z})(x)+(D\star \tilde G_{z}\star \tilde G_{1,z})(x)\big),\nnb
	\lbeq{BoundNZero-XiIota-gamma}
	\gamma_n=& \frac {\indAnimal}{g_z}\tilde G_z(\ve[\iota]) \sum_{x}
	\big(\diagRepulsiveLetter{D}_{1}(x)+|x|^2\bar G_z(x) dz (D\star \tilde G_z)(x)\big)\\
	&+\frac {\indAnimal}{g_z} \big(\tilde G_z(\ve[\iota])+ \sup_{y} G_{n,z}(y)\big) \sum_{x} |x-\ve[\iota]|^2\bar G(x-\ve[\iota]) 2dz (D\star \tilde G_z)(x)\nnb
	&+2\frac {\indAnimal}{g_z}  \sum_{u}\diagRepulsiveLetter{B}_{1,1}(u,\ve[\iota])\sup_{u} \sum_{x}\diagRepulsiveLetter{B}_{1,1}(x,u)\nnb
	&+4dz \frac {\indAnimal}{g_z}  \sup_{y} \tilde G_{1,z}(y)\sum_{x}|x|^2\bar G_z(x)(D\star \tilde G_z\star \tilde G_{1,z})(x).\nn
\end{align}


Then the bounds on $\Xi^{\ssc[0],\iota}$ and $\Pi^{\ssc[0],\iota,\kappa}$ read as follows:
\begin{lemma}[Bounds on $\Xi^{\ssc[0],\iota}$ and $\Pi^{\ssc[0],\iota,\kappa}$]
\label{BoundNZero-XiIota}
Let $0\leq z< z_c$. Then,
	\begin{align}
	\lbeq{Bound-XiIota0}
	\sum_{x\in\Zd}\Xi^{\ssc[0],\iota}_z(x)\leq& G_z(\ve[\iota]) +\alpha_{1,0},\\
	\lbeq{Bound-XiIota0Wa}
	\sum_{x\in\Zd}|x|^2\Xi^{\ssc[0],\iota}_z(x)\leq&\beta_{0},\\
	\lbeq{Bound-XiIota0Wb}
	\sum_{x\in\Zd}|x-\ve[\iota]|^2\Xi^{\ssc[0],\iota}_z(x)\leq& G_z(\ve[\iota]) +\gamma_1.
	\end{align}
The remainder terms of the split are bounded in a similar way as
	\begin{align}
	\lbeq{Bound-XiIota0RemI}
	\sum_{x\in\Zd}\Xi^{\ssc[0],\iota}_{{\sss R, I},z}(x)\leq& \alpha_{1,2},\qquad \qquad
	\sum_{x\in\Zd}|x-\ve[\iota]|^2\Xi^{\ssc[0],\iota}_{{\sss R, I},z}(x)\leq \gamma_2,\\
	\lbeq{Bound-XiIota0RemII}
	\sum_{x\in\Zd}\Xi^{\ssc[0],\iota}_{{\sss R, II},z}(x)\leq& \alpha_{2,1},\qquad\qquad
	\sum_{x\in\Zd}|x|^2\Xi^{\ssc[0],\iota}_{{\sss R, II},z}(x)\leq\beta_{1}.
	\end{align}
The explicit terms of the splits are bounded by
	\begin{align}
	\lbeq{Bound-XiIotaAlphaI}
	\Xi^{\ssc[0],\iota}_{\alpha,{\sss I},z}(\ve[\iota])\leq &\frac {\indAnimal}{g_z}\diagRepulsiveLetter{D}_{1}(\ve[\iota]),\quad\qquad\qquad
	\Xi^{\ssc[0],\iota}_{\alpha,{\sss II},z}(0)\leq G_z(\ve[\iota]),\\
    	\lbeq{Bound-XiIotaAlphaIWa}
	\sum_\kappa \Xi^{\ssc[0],\iota}_{\alpha,{\sss I},z}(\ve[\kappa]+\ve[\iota])\leq&
    	G_z(\ve[\iota])+ \frac {\indAnimal}{g_z}\sum_{\kappa\colon \kappa\neq -\iota}\diagRepulsiveLetter{T}_{1,\underline 1,2}(\ve[\iota],\ve[\kappa]+\ve[\iota],0)
    	+\frac {zg_z^\iota}{g_z}\indAnimal\sum_{\kappa\colon \kappa\neq -\iota}\diagRepulsiveLetter{B}_{2,2}(\ve[\iota]+\ve[\kappa],0),\\
    	\lbeq{Bound-XiIotaAlphaIIWb}
    	\sum_\kappa \Xi^{\ssc[0],\iota}_{\alpha,{\sss II},z}(\ve[\kappa])\leq&
   	\frac {\indAnimal}{g_z}\diagRepulsiveLetter{D}_{1}(\ve[\iota])+\frac {\indAnimal z g_z^\iota}{g_z}\sum_{\kappa\colon \kappa\neq -\iota}\sup_{x\neq 0}\tilde G_{1,z}(x)\sum_{u}\diagRepulsiveLetter{B}_{1,1}(u,\ve[\kappa]),\\
   	&+ \indAnimal\sum_{\kappa\colon \kappa\neq -\iota} \Big(\frac{1}{g_z}\diagRepulsiveLetter{T}_{1,2,\underline 1} (\ve[\iota],\ve[\kappa],0)
    	+ \indAnimal  z G_{3,z}(\ve[\kappa])\big(\tilde G_{1,z}(\ve[\iota])+\tilde G_{2,z}(\ve[\kappa]-\ve[\iota])\big)\Big).\nn
	\end{align}
We bound the terms in the split of $\Pi^{\ssc[0],\iota,\kappa}_z$ as
	\begin{align}
	\lbeq{Bound-Pi0AlphaI}
	\sum_{\kappa}\Pi^{\ssc[0],\iota,\kappa}_{\alpha,z}(\ve[\iota])\leq&  \indAnimal(2d-1)z\diagRepulsiveLetter{D}_{1}(\ve[\iota]),\\
	\lbeq{Bound-Pi0RemI}
	\sum_{x\in\Zd}\sum_{\iota}\Pi^{\ssc[0],\iota,\kappa}_{{\sss R},z}(x)\leq& (2d-1)zg_z G_z(\ve[\iota]) +zg_z\alpha_{1,1},\\
	\lbeq{Bound-Pi0RemIW}
	\sum_{x,\iota,\kappa}|x-\ve[\iota]|^2\Pi^{\ssc[0],\iota,\kappa}_{{\sss R},z}(x+\ve[\kappa])\leq&
	2(2d)^2zg_z G_z(\ve[\iota]) +(2d)^2zg_z\gamma_1+4d(d+1)zg_z\alpha_{1,1}.
	\end{align}
\end{lemma}

\begin{lemma}[Bounds on the coefficients for $N=1$]
\label{LemmaBoundXiLTOne}
Let $0\leq z< z_c$. Then,
	\begin{align}
	\lbeq{Bound-Xi1}
	\sum_{x\in\Zd}\Xi^{\ssc[1]}_z(x)&\leq \frac {g_z^\iota}{g_z} (\vec P)_0,\\
	\lbeq{Bound-Xi1Rem}
	\sum_{x\in\Zd}\Xi^{\ssc[1]}_{{\sss R},z}(x) &\leq \frac {g_z^\iota}{g_z}\Big( (\vec P)_0
	- \sum_{x\in\Zd} A^{0,0}(0,0,x,x)+2\diagRepulsiveLetter{B}_{2,2}(x,0)+\sum_w \diagRepulsiveLetter{T}_{2,1,1}(x,w,0)
	 \Big),\\
	\lbeq{Bound-Psi1RemI}
	\sum_{x\in\Zd} \Psi^{\ssc[1],\kappa}_{{\sss R, I},z}(x)&\leq
	\frac {g_z^\iota}{g_z} \Big( (\vec P)_0 - \sum_{x\in\Zd} A^{0,0}(0,0,x,x)+2\indic{\ve[\kappa]\neq x}\diagRepulsiveLetter{B}_{\underline 1,3}(x,0)\Big)\\
	&+\frac {g_z^\iota}{g_z} \sum_{x\in\Zd} \Big(2\diagRepulsiveLetter{B}_{2,2}(x,0)+\sum_w \diagRepulsiveLetter{T}_{1,1,1}(x,w,0)\Big),\nnb
	\lbeq{Bound-Psi1RemII}
	\sum_{x\in\Zd} \Psi^{\ssc[1],\kappa}_{{\sss R, II},z}(x)&\leq \frac {g_z^\iota}{g_z} \Big( (\vec P)_0
	- \sum_{x\in\Zd} A^{0,0}(0,0,x,x)+2\diagRepulsiveLetter{B}_{2,2}(x,0)+\sum_w \diagRepulsiveLetter{T}_{2,1,1}(x,w,0)
	 \Big),
	\end{align}
and
	\begin{align}
	\lbeq{Bound-Xi1W}
	\sum_{x\in\Zd}|x|^2\Xi^{\ssc[1]}_z(x)&\leq (\vec \Delta^{\rm{start}})_0, \qquad
	\sum_{x\in\Zd}|x|^2\Xi^{\ssc[1]}_{{\sss R},z}(x)\leq(\vec \Delta^{\rm{start}})_0,\\
	\sum_{x\in\Zd} |x|^2 \Psi^{\ssc[1],\kappa}_{{\sss R, II},z}(x)&\leq(\vec \Delta^{\rm{start}})_0,\qquad
	\lbeq{Bound-Psi1RemIW}
	\sum_{x\in\Zd} |x-\ve[\kappa]|^2 \Psi^{\ssc[1],\kappa}_{{\sss R, I},z}(x)\leq(\vec \Delta^{\rm{start}})_0+\sum_{x\in\Zd} \Psi^{\ssc[1],\kappa}_{{\sss R, I},z}(x).
	\end{align}
Further, $\Xi^{\ssc[1],\iota}_z$ is bounded by
	\begin{align}
	\lbeq{Bound-XiIota1}
	\sum_{x\in\Zd}\Xi^{\ssc[1],\iota}_z(x)\leq& \frac {g_z}{\gj}(\vec P^\iota)_0,\\
	\lbeq{Bound-XiIota1Wa}
	\sum_{x\in\Zd}|x|^2\Xi^{\ssc[1],\iota}_z(x)\leq&(\vec \Delta^{\iota,{\sss I}})_0,\\
	\lbeq{Bound-XiIota1Wb}
	\sum_{x\in\Zd}|x-\ve[\iota]|^2\Xi^{\ssc[1],\iota}_z(x)\leq &(\vec \Delta^{\iota,{\sss II}})_0.
\end{align}
\end{lemma}

\begin{prop}[Bounds on the coefficients for $N\geq 2$]
\label{LemmaBoundXiLTTwo}
Let $0\leq z<z_c$ and $N\geq 2$. Then,
\begin{align}
\lbeq{Bound-Xi2}
\sum_{x\in\Zd}\Xi^{\ssc[N]}_z(x)&\leq \frac {g_z}{\gj} \vec P^T{\bf A}^{N-2}{\vec E}^{\rm{open}}, \quad \sum_{x\in\Zd}\Xi^{\ssc[N],\iota}_z(x)\leq \frac {g_z}{\gj} (\vec P^{\iota})^T{\bf A}^{N-2}{\vec E}^{\rm{open}},\\
\sum_{x\in\Zd}|x|^2\Xi^{\ssc[N]}_z(x)\leq& \ N  \left( (\vec \Delta^{\rm{start}})^T{\bf A}^{N-2}{\vec E^{\rm{closed}}}+\vec P^T{\bf A}^{N-2}\vec \Delta^{\rm{end}}\right)\nn \\[-5mm]
&\qquad\qquad\qquad+N\indic{N> 2}\sum_{M=0}^{N-3} \vec P {\bf A}^{M} {\bf C} {\bf A}^{N-3-M}\vec E^{\rm{closed}},
\lbeq{Bound-Xi2W}
\\
\sum_{x\in\Zd}|x|^2\Xi^{\ssc[N],\iota}_z(x)\leq& \ N \left((\vec \Delta^{\iota,{\sss I}})^T{\bf A}^{N-2}{\vec E^{\rm{closed}}}+(\vec P^{\iota})^T{\bf A}^{N-2}\vec \Delta^{\rm{end}}\right)\nn \\[-5mm]
&\qquad\qquad\qquad+N\indic{N> 2}\sum_{M=0}^{N-3} \vec P^\iota {\bf A}^{M} {\bf C} {\bf A}^{N-3-M}\vec E^{\rm{closed}},
\lbeq{Bound-XiIota2Wa}
\\
\sum_{x\in\Zd}|x-\ve[\iota]|^2\Xi^{\ssc[N],\iota}_z(x)\leq& \ N  \left((\vec \Delta^{\iota,{\sss II}})^T{\bf A}^{N-2}{\vec E^{\rm{closed}}}+(\vec P^{\iota})^T{\bf A}^{N-2}\vec \Delta^{\rm{end}}\right)\nn \\[-5mm]
&\qquad\qquad\qquad+N\indic{N> 2}\sum_{M=0}^{N-3} \vec P^\iota {\bf A}^{M} {\bf C} {\bf A}^{N-3-M}\vec E^{\rm{closed}}.
\lbeq{Bound-XiIota2Wb}
\end{align}
\end{prop}
\begin{figure}[ht!]
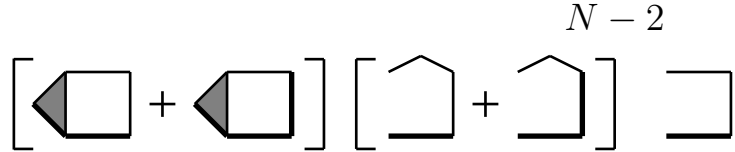

\begin{center}
{\Large
\picPiFourTreeDecomposePicture[0.85]}
\caption{Schematic representation of the bound on $\hat \Xi^{\ssc[N]}_z(0)$, given in \refeq{Bound-Xi2}. The backbone is marked by a thicker line.
See also Figures \ref{BoundLA-Figure-Xitwo} and \ref{LTXi4Ribweight}.}
\label{TreeXiFourDeomposedStructure}
\end{center}
\end{figure}

Having stated our main bounds on the NoBLE coefficients, we next state some additional bounds that improve the numerical accuracy of our method:

\begin{lemma}[Bounds on differences]
\label{lemmapercboundXi0minus1}
Let $0\leq z<z_c$. Then,
	\begin{align}
	\lbeq{Differencebound-1}
	\Xi^\ssc[0]_{\alpha,z}(0)-\Xi^\ssc[1]_{\alpha,z}(0)\leq&0,\qquad \qquad \Xi^\ssc[0]_{\alpha,z}(\ve[1])-\Xi^\ssc[1]_{\alpha,z}(\ve[1])\leq 0,\\
	\lbeq{Differencebound-2}
\Xi^\ssc[1]_{\alpha,z}(0)-\Xi^\ssc[0]_{\alpha,z}(0)\leq & 2dz\gj G_{3,z}(\ve[\iota]),\\
	\Xi^\ssc[1]_{\alpha,z}(\ve[1])-\Xi^\ssc[0]_{\alpha,z}(\ve[1])\leq&
	\frac {g_z^\iota} {g_z} \Big(2\diagRepulsiveLetter{B}_{3,\underline 1}(\ve[1],0)+\sum_v \diagRepulsiveLetter{T}_{\underline 1,1,1}(\ve[1],v,0)\Big),
\lbeq{Differencebound-3}
	\end{align}
and
	\begin{align}
	\lbeq{Differencebound-4}
	&\sum_\iota \big(\Psi^{\ssc[1],1}_{\alpha,{\sss I},z}(\ve[1]+\ve[\iota])-\Psi^{\ssc[0],1}_{\alpha,{\sss I},z}(\ve[1]+\ve[\iota])\big) \leq\sum_\iota\sum_w\diagRepulsiveLetter{T}_{0,0,\underline 2}(w,\ve[1]+\ve[\iota],0),
\\
	\lbeq{Differencebound-5}
	&\sum_\iota \big(\Psi^{\ssc[0],1}_{\alpha,{\sss I},z}(\ve[1]+\ve[\iota])  -\Psi^{\ssc[1],1}_{\alpha,{\sss I},z}(\ve[1]+\ve[\iota]) \big)\leq {\indAnimal}{g^\iota_z}\sum_\iota\diagRepulsiveLetter{B}_{2,\underline 2}(\ve[1]+\ve[\iota],0).
	\end{align}
Further,
	\begin{align}
	\sum_\iota \big(\Psi^{\ssc[1],1}_{\alpha,{\sss II},z}&(\ve[\iota])  -\Psi^{\ssc[0],1}_{\alpha,{\sss II},z}(\ve[\iota])\big)\leq
	(2d-1) \Big(2\diagRepulsiveLetter{B}_{3,\underline 1}(\ve[1],0)+\sum_v \diagRepulsiveLetter{T}_{\underline 1,1,1}(\ve[1],v,0)\Big),
	\lbeq{Differencebound-6}\\
	\lbeq{Differencebound-7}
	\sum_\iota \big(\Psi^{\ssc[0],1}_{\alpha,{\sss II},z}&(\ve[\iota])  -\Psi^{\ssc[1],1}_{\alpha,{\sss II},z}(\ve[\iota])\big)
	\leq 0.
	\end{align}
\end{lemma}

\subsection{Summary of the bounds}
\label{secBoundsSummary}
In this section, we give an overview of where to find the bounds stated in \cite[Assumption 4.3]{FitHof13b}. We emphasize that, next to the diagrammatic bounds proven in this document, \cite[Assumption 4.3]{FitHof13b} also requires a method to conclude numerical bounds. In \cite[Section 5]{FitHof13b}, we describe this method. In \cite{FitNoblePage}, we provide the source code that computes the numerical bounds. This source code allows us to apply the analysis of \cite{FitHof13b} to obtain our mean-field results. We give a summary of where to find most bounds in Table \ref{TableLinktoBounds}.

\begin{table}
\centering
\begin{tabular}{c|c || c|c || c|c }
  Bound             & defined in & Bound &  defined in & Bound &  defined in \\
  \hline
$\beta_{\sss \Xi}^\ssc[0]$ & \refeq{Bound-Xi0} &
$\beta_{\sss \Xi}^\ssc[1]$ &  \refeq{Bound-Xi1} &
$\beta_{\sss \Xi}^\ssc[N],N\geq 2$ & \refeq{Bound-Xi2}\\
\hline
$\beta_{\sss \Xi^\iota}^\ssc[0]$ & \refeq{Bound-XiIota0}&
$\beta_{\sss \Xi^\iota}^\ssc[1]$ & \refeq{Bound-XiIota1}&
$\beta_{\sss \Xi^\iota}^\ssc[N],N\geq 2$ & \refeq{Bound-Xi2}\\
\hline
$\beta_{{\sss \Delta \Xi}}^\ssc[0]$ & \refeq{Bound-Xi0W}&
$\beta_{{\sss \Delta \Xi}}^\ssc[1]$ & \refeq{Bound-Xi1W}&
$\beta_{{\sss \Delta \Xi}}^\ssc[N],N\geq 2$ &  \refeq{Bound-Xi2W}\\
\hline
$\beta_{{\sss \Delta \Xi^{\iota}},0}^\ssc[0]$ & \refeq{Bound-XiIota0Wa}&
$\beta_{{\sss \Delta \Xi^{\iota}},0}^\ssc[1]$ & \refeq{Bound-XiIota1Wa}&
$\beta_{{\sss \Delta \Xi^{\iota}},0}^\ssc[N],N\geq 2$ & \refeq{Bound-XiIota2Wa}\\
\hline
$\beta_{{\sss \Delta \Xi^{\iota}},\iota}^\ssc[0]$ & \refeq{Bound-XiIota0Wb}&
$\beta_{{\sss \Delta \Xi^{\iota}},\iota}^\ssc[1]$ & \refeq{Bound-XiIota1Wb}&
$\beta_{{\sss \Delta \Xi^{\iota}},\iota}^\ssc[N],N\geq 2$ &   \refeq{Bound-XiIota2Wb}\\
\hline
$ \beta_{{\sss\Xi_\alpha(0)}}^\ssc[0-1]$ &\refeq{Differencebound-1}&
$ \beta_{{\sss\Xi_\alpha(\ve[1])}}^\ssc[0-1]$ &\refeq{Differencebound-1}&
$ \beta_{{\sss\Xi_\alpha(0)}}^\ssc[1-0]$ &\refeq{Differencebound-2}\\
 \hline
 $\beta_{{\sss\Xi_\alpha(\ve[1])}}^\ssc[1-0]$&\refeq{Differencebound-3}&
$\beta_{{\sss\Xi^{\iota}_\alpha,I}}^\ssc[0]$ & 	\refeq{Bound-XiIotaAlphaI} &
$\beta_{{\sss \sum}{\sss\Xi^{\iota}_\alpha,I}}^\ssc[0]$& \refeq{Bound-XiIotaAlphaIWa} \\
\hline
$\beta_{{\sss\Xi^{\iota}_\alpha,II}}^\ssc[0]$ & 	\refeq{Bound-XiIotaAlphaI} &
$\beta_{ {\sss\sum \Xi^{\iota}_\alpha,II}}^\ssc[0]$ & 	\refeq{Bound-XiIotaAlphaIIWb}  &
$\beta_{{\sss\sum \Psi^{\iota}_\alpha,I}}^\ssc[0-1]$  & \refeq{Differencebound-5} \\
\hline
$\beta_{{\sss\sum \Psi^{\iota}_\alpha,II}}^\ssc[0-1]$ & \refeq{Differencebound-7} &
$\beta_{{\sss\sum \Psi^{\iota}_\alpha,I}}^\ssc[1-0]$  & \refeq{Differencebound-4} &
$\beta_{{\sss\sum \Psi^{\iota}_\alpha,II}}^\ssc[1-0]$  & \refeq{Differencebound-6} \\
\hline
$\beta_{\sss\sum \Pi_\alpha}^\ssc[0]$& \refeq{Bound-Pi0AlphaI} &&&&\\
\hline
$\beta_{\sss\Xi,R}^\ssc[0]$ & \refeq{Bound-Xi0} &
$\beta_{\Delta \sss\Xi,R}^\ssc[0]$ &  \refeq{Bound-Xi0RemW}&
$\beta_{\sss\Psi,R,I}^\ssc[0]$ & \refeq{Bound-Psi0RemI}\\
\hline
$\beta_{\Delta \sss\Psi,R,I}^\ssc[0]$ & \refeq{Bound-Psi0RemIW}&
$\beta_{\sss\Psi,R,II}^\ssc[0]$ &  \refeq{Bound-Psi0RemII}&
$\beta_{\Delta \sss\Psi,R,II}^\ssc[0]$ & \refeq{Bound-Psi0RemIIW}\\
\hline
$\beta_{\sss\Xi,R}^\ssc[1]$ &     \refeq{Bound-Xi1Rem} &
$\beta_{\Delta \sss\Xi,R}^\ssc[1]$ &  \refeq{Bound-Xi1W}&
$\beta_{\sss\Psi,R,I}^\ssc[1]$ & \refeq{Bound-Psi1RemI} \\
\hline
$\beta_{\Delta \sss\Psi,R,I}^\ssc[1]$ & \refeq{Bound-Psi1RemIW} &
$\beta_{\sss\Psi,R,II}^\ssc[1]$ & \refeq{Bound-Psi1RemII} &
$\beta_{\Delta \sss\Psi,R,II}^\ssc[1]$ & \refeq{Bound-Psi1RemIW} \\
\hline
$\beta_{\sss\Xi^\iota,R,I}^\ssc[0]$ & \refeq{Bound-XiIota0RemI}&
$\beta_{\Delta \sss\Xi^\iota,R,I}^\ssc[0]$ & \refeq{Bound-XiIota0RemI}&
$\beta_{\sss\Xi^\iota,R,II}^\ssc[0]$ & \refeq{Bound-XiIota0RemII}\\
\hline
$\beta_{\Delta \sss\Xi^\iota,R,II}^\ssc[0]$ & \refeq{Bound-XiIota0RemII}&
$\beta_{\sss\Pi,R}^\ssc[0]$ &\refeq{Bound-Pi0RemI}&
$\beta_{\Delta \sss\Pi,R}^\ssc[0]$ & \refeq{Bound-Pi0RemIW}
\end{tabular}
\caption{An overview where to find the bounds stated in  \cite[Assumption 4.3]{FitHof13b}.}
\label{TableLinktoBounds}
\end{table}

In the following, we discuss some bounds that are not stated in Table \ref{TableLinktoBounds}.  These are also required for the analysis, but are in general very easy to obtain.
We start by bounding $\aabz/\aaz$, as formulated in \cite[first inequality in (4.30) in Assumption 4.3]{FitHof13b},
	\begin{align*}
	\frac {\aabz}{ \aaz}=\frac {zg_z}{zg_z^\iota}=1+\frac {g_z-g_z^\iota}{g_z^\iota}=1+\frac {\bar G_z(\ve[\iota])}{g_z^\iota} \leq 1 +2dz (D\star \bar G_z)(\ve[\iota])
	:=\betaaa.
	\end{align*}
Regarding $\betaaa$, we remark that by the definition of $z_I$ in \refeq{DefinitionOfzI}, for each $z>z_I$,
	\begin{align*}
	zg_z \geq {z_I}g_{z_I}=\frac 1 {(2d-1)\e}:=\underline{\beta}_{\mu},
	\end{align*}
as formulated in \cite[second inequality in (4.30) in Assumption 4.3]{FitHof13b}.
\cite[Assumption 4.3]{FitHof13b} also states lower bonds on $\Psi^\ssc[0]$ and $\Pi^\ssc[0]$, which we could obtain using simple combinatorics. We numerically found that the benefit of these extra bounds is minimal, so that we simply use the trivial lower $\underline{\beta}_{\sss \Psi}^\ssc[0]  = \underline{\beta}_{\sss \sum \Pi}^\ssc[1]=\underline{\beta}_{ \sss \sum \Pi_\alpha}^\ssc[0]=0$ for these coefficients, which are by definition non-negative.

Further, we assume that the geometric series bounding our NoBLE coefficients converge, which follows if their base is smaller than one as formulated in \cite[(4.34) in Assumption 4.3]{FitHof13b}, i.e.,
	\begin{align}
	\lbeq{analys-assumed-invertablecondition}
	\frac {(2d-1)\aabz}{1-\aaz}\sum_{N=0}^\infty\beta_{{\sss \Xi}^{\iota}}^\ssc[N]<1.
	\end{align}
For example, this shows that the matrix inverse in \refeq{lace-exp-eq} is well defined. Finally, the inequalities $\lowaf-\betadeltaRfzlow>0$ and $\lowcp-\betaap-\beta_{{\sss R,\Phi}}>0$,  which guarantees that the numerator and denominator of $\hat{G}_z(k)$ in \refeq{lace-exp-eq} are non-negative, see \cite[one-but-last sentence of Assumption 2.7]{FitHof13b}.
These are numerical conditions that are verified explicitly in the Mathematica notebooks. These conditions are relatively weak, in the sense that empirically, we observe that the bootstrap analysis, which in particular includes an improvement of bounds, is much more likely to fail.


\section{Proof of the bounds on the NoBLE coefficients}
\label{secBoundProof}

In this section, we explain how the bounds stated in the preceding section are proven. In Section \ref{secBoundsExplained}, we have already explained the basic idea of the bounds. We start by proving the bounds on the relatively simple coefficients arising for $N=0$ and $N=1$. Then, we explain how the coefficients for general $N\geq 2$ are bounded. After this, we discuss bounds on the difference of coefficients.

\subsection{Proof of bounds for $N=0$}
Here we prove Lemma \ref{BoundNZero-Xi} and \ref{BoundNZero-XiIota}.
 For $N=0$, the coefficients are created by LAs in which $0$ and $x$ are doubly connected, see \refeq{defJNanimal}. In Section
\ref{secBoundsDouble} we have explained how to bound the contribution of such LAs, so we add here only the missing details.

\begin{proof}[Proof of Lemma \ref{BoundNZero-Xi}.]
In Section \ref{secBoundsDouble} we have already proven the bound on $\Xi^{\ssc[0]}_{z}$.  From $\Xi^{\ssc[0]}_{{\sss R},z}$, we have extracted the contributions in which $0$ and $x$ are connected by a direct bond, see \refeq{LA-Split-Def-1}. Considering this, the bounds \refeq{Bound-Xi0} and \refeq{Bound-Xi0RemW} follow directly.
The only difference between $\Xi^{\ssc[0]}_{{\sss R},z}$ and $\Psi^{\ssc[0],\kappa}_{{\sss R, II},z}$ is that $x-\ve[\kappa]\nin A$ for $\Psi^{\ssc[0],\kappa}_{{\sss R, II},z}$.
Since each double connection requires at least two neighbors of $x$ to be part of $A$, the spatial symmetry of the LAs implies that
	\begin{align}
	\sum_{x\in \Zd} \Psi^{\ssc[0],\kappa}_{{\sss R, II},z}(x)=\frac 1 {2d} \sum_{x\in \Zd, \kappa} \Psi^{\ssc[0],\kappa}_{{\sss R, II},z}(x)\leq \frac {2d-2} {2d} \sum_{x\in \Zd}\Xi^{\ssc[0]}_{{\sss R, II},z}(x),
	\end{align}
which implies \refeq{Bound-Psi0RemII} and \refeq{Bound-Psi0RemIIW}.

The coefficient $\Psi^{\ssc[0],\kappa}_{{\sss R, I},z}(x)$ was created by extracting the shortest contributions with $x=\ve[\kappa]$ and $x$'s that are neighbors of $\ve[\kappa]$. The contribution of the remaining $(2d-1)$ neighbors of the origin is bounded using $\diagRepulsiveLetter{B}_{\underline 1,3}(\ve[1],0)$. For all other $x$, i.e., those with $|x|>1$, we use $\Psi^{\ssc[0],\kappa}_{{\sss R, II},z}(x)\leq \Psi^{\ssc[0],\kappa}_{{\sss R, I},z}(x)$ to conclude the bounds stated in \refeq{Bound-Psi0RemI}.

In \refeq{Bound-Psi0RemIW}, we bound the contribution due to all neighbors of origin by the first term.
The contribution due to all other $x$ are bounded by the second term. Regarding the neighbors $x=\ve[\iota]$ of the origin, with $\iota\in\{\pm1,\dots,\pm d\}$, we note that
	\begin{align}
	|x-\ve[\kappa]|^2=|\ve[\iota]-\ve[\kappa]|^2=
	\begin{cases}
	2 & \text{ for~} |\iota|\neq |\kappa|,\\
	4 & \text{ for } \iota=-\kappa,\\
	0 & \text{ for } \iota= \kappa.
	\end{cases}
	\end{align}
Summing over $\iota$ creates the stated factor $4d$.
For the other $x$, we know that any path to $x$ contains at least two bonds.
To be able to use a symmetry argument, we ignore that we extracted terms and as well as the condition $x-\ve[\kappa]\nin A$, to obtain the simple bound
	\begin{align}
	(2dz)^2g_z (D^{\star 2} \star \tilde G_z)(x)  \bar G_z(x)
	\end{align}
Then, we use $|x-\ve[\kappa]|^2=(|x|^2 -2x_\kappa+1)$ and that the term involving $-2x_\kappa$ sums to zero by symmetry, to obtain
	\begin{align}
	\lbeq{Bound-Psi0RemIIW-proofLine}
	\sum_x |x-\ve[\kappa]|^2 (2dz)^2g_z (D^{\star 2} \star \tilde G_z)(x)  \bar G_z(x)=\sum_x (|x|^2+1 ) (2dz)^2g_z (D^{\star 2} \star \tilde G_z)(x)  \bar G_z(x),
	\end{align}
which is the second term in \refeq{Bound-Psi0RemIW}.
This proves \refeq{Bound-Psi0RemIW}, and completes the proof of Lemma \ref{BoundNZero-Xi}.
\end{proof}

\begin{proof}[Proof of Lemma \ref{BoundNZero-XiIota}.]
By definition,
	\begin{align}
	\Xi^{\ssc[0],\iota}_z(x)= \frac {1}{g_z}    \sum_{A\colon 0,x,\ve[\iota]\in A} z^{|A|}\indic{0\dbct{A} x}.
	\end{align}
For LTs, this simplifies to
	\begin{align}
	\Xi^{\ssc[0],\iota}_z(x)= \frac {\delta_{0,x}}{g_z}  \bar G_z(\ve[\iota])=\delta_{0,x}G_z(\ve[\iota]).
	\end{align}
For LAs we recall that we have bounded such animals already between \refeq{split-LA-sausage-Gbar} and \refeq{split-LA-sausage-second}.

We improve these bounds by using {\em repulsive diagrams} (recall Section \ref{secBoundsRepdia}) and by considering four cases for the point $u$, where $u$ is the summand appearing in \refeq{split-LA-sausage-Gbar}, thus yielding the bound $\alpha_{1,1}$ (see \refeq{BoundNZero-XiIota-alpha}). Let us explain the four cases for $u$:\\
(a) For $u=0$, the diagram is given by a double connection from $0$ to $x$ and a simple path from $0$ to $\ve[\iota]$ that we bound by
	\begin{align}
	G_z(\ve[\iota])+ \frac {1}{g_z} \tilde G_z(\ve[\iota]) \sum_{x\neq 0} \diagRepulsiveLetter{D}_{1}(x,0).
	\end{align}
The first term is given explicitly in \refeq{Bound-XiIota0} and the second term corresponds to the first term of $\alpha_{1,1}$.\\
(b) For $u=\ve[\iota]$, the point $\ve[\iota]$ is on the double connection between $0$ and $x$.
For $x=\ve[\iota]$, we have a simple double connection to $\ve[\iota]$, otherwise a repulsive triangle, as stated in the second and third term of $\alpha_{1,1}$.\\
(c) In the contributions in which $u=x\nin \{0,\ve[\iota]\}$ we find a double connection between $0$ and $x$, bounded by $\diagRepulsiveLetter{D}_{1}(x,0)$  and a connection from $x$ to $\ve[\iota]$, bounded by $\tilde G_{m,z}(x-\ve[\iota])$. We improve this bound slightly using the parity of the lattice:
If $0$ and $x$ are directly connected, the length of the connection between $x$ to $\ve[\iota]$ is at least $2$.\\
(d) In the remaining contributions, we have that $u\neq \{0,x,\ve[\iota]\}$, so that these LAs can be bounded by
	\begin{align}
 	\sum_{u}\diagRepulsiveLetter{B}_{1,1}(u,\ve[\iota])  \sup_{u} \sum_{x}\diagRepulsiveLetter{B}_{1,1}(x,u),
	\end{align}
which corresponds to the last term of $\alpha_{1,1}$. This completes the proof of \refeq{Bound-XiIota0}.

The diagrams with the additional weights $|x|^2$ and $|x-\ve[\iota]|^2$ are bounded in a similar manner. As the numerical bounds, using the bootstrap function $f_3$, require a $G_z$ to carry the spatial weight, we express the bounds using two-point functions instead of repulsive diagrams. For the weight $|x-\ve[\iota]|^2$, two additional steps are made: Firstly, for $u=0\neq x$ we first split the weights using symmetry, as shown in \refeq{Bound-Psi0RemIIW-proofLine}, and then bound the two arising terms. This leads to the first line of $\gamma_1$ in \refeq{BoundNZero-XiIota-gamma}. Secondly, for $u\neq \{0,\ve[\iota],x\}$, we use that $|x-\ve[\iota]|^2\leq 2(|x|^2+1)$, see \cite[(2.24)]{FitHof13b}, and bound the arising terms individually. This completes the proofs of \refeq{Bound-XiIota0Wa} and \refeq{Bound-XiIota0Wb}.

For the bounds on $\Xi^{\ssc[0],\iota}_{{\sss R, I},z}$ and $\Xi^{\ssc[0],\iota}_{{\sss R, II},z}$ stated in \refeq{Bound-XiIota0RemI}-\refeq{Bound-XiIota0RemII}, we review which terms have been extracted from $\Xi^{\ssc[0],\iota}_z(x)$ to yield $\Xi^{\ssc[0],\iota}_{{\sss R, I},z}$ and $\Xi^{\ssc[0],\iota}_{{\sss R, II},z}$. For $\Xi^{\ssc[0],\iota}_{{\sss R, II},z}(x)$, the LAs in which $x$ is directly connected to the origin have been removed. For this reason, in $\Xi^{\ssc[0],\iota}_{{\sss R, II},z}$ all paths from $0$ to $x$ have a length of at least $2$ instead of $1$. This is  expressed by the subscripts $n$ of $\alpha_{n,m}$ and $\beta_{n}$. To obtain $\Xi^{\ssc[0],\iota}_{{\sss R, I},z}(x)$ from  $\Xi^{\ssc[0],\iota}_z(x)$, we have extracted all LAs in which $x$ is directly connected to $\ve[\iota]$. This means that in $\Xi^{\ssc[0],\iota}_{{\sss R, I},z}(x)$, the path from $\ve[\iota]$ to $x$ in contributing LAs have length at least $2$.
The bounds \refeq{Bound-XiIota0RemI} and \refeq{Bound-XiIota0RemII} are created by repeating the consideration of cases for $u$, as given above for the bound on bound $\Xi^{\ssc[0],\iota}_z(x)$, while keeping these additional restrictions.

We continue with the bounds on $\Xi^{\ssc[0],\iota}_{\alpha,{\sss I},z}$ and $\Xi^{\ssc[0],\iota}_{\alpha,{\sss II},z}$ in \refeq{Bound-XiIotaAlphaI}--\refeq{Bound-XiIotaAlphaIIWb}. Note that for LTs, most of these terms are zero due to the lack of double connections, so we restrict to LAs. Recall \refeq{split-LTs-N0}, as well as \refeq{split-LTs-N0a}--\refeq{split-LTs-N0b}.

By their definitions in \refeq{LA-Split-Def-XiIotaAlphaa} and \refeq{LA-Split-Def-XiIotaAlphab},
	\begin{align}
	\Xi^{\ssc[0],\iota}_{\alpha,{\sss I},z}(\ve[\iota])= &  \frac {\indAnimal}{g_z} \sum_{A\colon 0,\ve[\iota]\in A} z^{|A|}\indic{0\dbct{A} \ve[\iota]}
	\leq \frac {\indAnimal}{g_z}\diagRepulsiveLetter{D}_{1}(\ve[\iota]),\\
	\Xi^{\ssc[0],\iota}_{\alpha,{\sss II},z}(0)=&\Xi^{\ssc[0],\iota}_{\alpha,{\sss I},z}(0)=\frac 1 {g_z} \bar G_z(\ve[\iota])= G_z(\ve[\iota]),
	\end{align}
which proves \refeq{Bound-XiIotaAlphaI}. To obtain the bound in \refeq{Bound-XiIotaAlphaIWa}, we again use \refeq{LA-Split-Def-XiIotaAlphaa}, and remark that, for $x\neq 0$ with $|x-\ve[\iota]|^2=1$,
we can write $x=\ve[\kappa]+\ve[\iota]$ with $\kappa\neq -\iota$ to get
	\begin{align}
	\Xi^{\ssc[0],\iota}_{\alpha,{\sss I},z}(\ve[\kappa]+\ve[\iota])=&\frac {1}{g_z} \sum_{A\colon 0,\ve[\kappa]+\ve[\iota],\ve[\iota]\in A} z^{|A|}\indic{0\dbct{A} \ve[\kappa]+\ve[\iota]} \indic{d_A(\ve[\iota],\ve[\kappa]+\ve[\iota])=1}.
	\end{align}
Thus, $\ve[\kappa]+\ve[\iota]$ and $\ve[\iota]$ are directly connected in the LA. We distinguish between the case where $\ve[\iota]$ is on the double connection or not to obtain
	\begin{align}
	\Xi^{\ssc[0],\iota}_{\alpha,{\sss I},z}(\ve[\kappa]+\ve[\iota])\leq&\frac {\indAnimal}{g_z}\diagRepulsiveLetter{T}_{1,\underline 1,2}(\ve[\iota],\ve[\kappa]+\ve[\iota],0)
	+\frac {zg_z^\iota \indAnimal}{g_z}(1-\delta_{\kappa,-\iota})\diagRepulsiveLetter{B}_{2,2}(\ve[\iota]+\ve[\kappa],0).
	\end{align}
Summing over all neighbors $\ve[\kappa]+\ve[\iota]$ of $\ve[\iota]$, and also using that the contribution of $\kappa=-\iota$ is $\Xi^{\ssc[0],\iota}_{\alpha,{\sss I},z}(0)=G_z(\ve[\iota])$, proves \refeq{Bound-XiIotaAlphaIWa}.

For $\Xi^{\ssc[0],\iota}_{\alpha,{\sss II},z}(\ve[\kappa])$, we recall \refeq{LA-Split-Def-XiIotaAlphab}, and consider $\kappa\neq \iota$, distinguish between the cases where $\ve[\iota]$ is part of the double connection from $0$ to $\ve[\kappa]$ or not. For this, let $u$ be the first point that any path between $0$ and $\ve[\iota]$, and the path between $\ve[\kappa]$ and $\ve[\iota]$, share.
We split between $u=0, \ve[\kappa],\ve[\iota]$ and $u\not\in \{0,\ve[\kappa], \ve[\iota]\}$, to obtain
	\begin{align}
    	\Xi^{\ssc[0],\iota}_{\alpha,{\sss II},z}(\ve[\kappa])=&\frac {\indAnimal}{g_z}\sum_{A\colon 0,\ve[\kappa],\ve[\iota]\in A} z^{|A|}\indic{0\dbct{A} \ve[\kappa]} \indic{d_A(0,\ve[\kappa])=1}\\
    	\leq & z \indAnimal G_{3,z}(\ve[\kappa]) \big(\tilde G_{1,z}(\ve[\iota])+\tilde G_{2,z}(\ve[\kappa]-\ve[\iota])\big)\nn\\
    	&+\frac {\indAnimal}{g_z}\diagRepulsiveLetter{T}_{1,2,\underline 1} (\ve[\iota],\ve[\kappa],0)+z \frac{g_z^\iota}{g_z} \indAnimal \sup_{x\neq 0}\tilde G_{1,z}(x)\sum_{u}\diagRepulsiveLetter{B}_{1,1}(u,\ve[\kappa]).\nn
	\end{align}
Summing this over all neighbors $\ve[\kappa]$ of the origin, and using that the contribution for $\kappa=\iota$ can be bounded by $\frac {\indAnimal}{g_z}\diagRepulsiveLetter{D}_{1}(\ve[\iota])$, proves \refeq{Bound-XiIotaAlphaIIWb}.


For the bounds on $\Pi^{\ssc[0],\iota,\kappa}_{z}, \Pi^{\ssc[0],\iota,\kappa}_{\alpha,z}$  and  $\Pi^{\ssc[0],\iota,\kappa}_{{\sss R},z}$ stated in \refeq{Bound-Pi0AlphaI}--\refeq{Bound-Pi0RemIW}, it is useful to write out the definition of $\Pi^{\ssc[0],\iota,\kappa}_{z}$ in \refeq{defPiNanimal}, and simplify it for the special case $N=0$. This leads to
	\begin{align}
	\Pi^{\ssc[0],\iota,\kappa}_z(x)&=
	\sum_{A\colon 0,x\in A} z^{|A|+1}\indic{0\dbct{A} x}
	\big(\indic{\ve[\iota]\in A_0^\omega}+ \indic{|\omega|=0}\indic{x\neq 0}\indic{x-\ve[\kappa]=\ve[\iota]}\big)\indic{x-\ve[\kappa]\nin A}.
	\end{align}
Using the split $\Pi^{\ssc[0],\iota,\kappa}_{\alpha,z}$ (see \refeq{LA-Split-Def-1a}), we extract from this the contributions where $x=\ve[\iota]$, as well as those with
$x=\ve[\iota]+\ve[\kappa]$ in which the edge $(\ve[\iota],\ve[\iota]+\ve[\kappa])$ is in the LA. For $x=\ve[\iota]$, we see that
	\begin{align}
	\lbeq{no-problems-Xi-Pi}
	\Pi^{\ssc[0],\iota,\kappa}_{\alpha,z}(\ve[\iota])&=\Pi^{\ssc[0],\iota,\kappa}_z(\ve[\iota])= \indAnimal\sum_{A\colon 0,\ve[\iota]\in A} z^{|A|+1}\indic{0\dbct{A} \ve[\iota]}\indic{\ve[\iota]-\ve[\kappa]\nin A}\nnb
	&\leq \indAnimal (1-\delta_{\iota,-\kappa}) z \diagRepulsiveLetter{D}_{1}(\ve[\iota]),
	\end{align}
as the LAs contain a simple double connection to a neighbor of the origin. Summing this over $\kappa$ proves \refeq{Bound-Pi0AlphaI}.
The bounds \refeq{Bound-Pi0RemI} and \refeq{Bound-Pi0RemIW} basically follow from the fact that $\Pi^{\ssc[0],\iota,\kappa}_z(x)\leq zg_z\Xi^{\ssc[0],\iota}_z(x)$, see \refeq{XidominatespsiImproved}.
To understand \refeq{Bound-Pi0RemI} we compare the bound with \refeq{Bound-XiIota0}.
When summing over $\iota$ for $x=0$, $\iota=-\kappa$ does not contribute explaining the factor $(2d-1)$.
Further, we have extracted the case that $x=\ve[\iota]$, thus we can us $\alpha_{1,1}$ instead of $\alpha_{1,0}$ for this bound.

For \refeq{Bound-Pi0RemIW}, we need to redistribute the weights. We write
	\eqan{
	\lbeq{Pi(0)-bd}
	\sum_{x,\iota,\kappa}|x-\ve[\iota]|^2\Pi^{\ssc[0],\iota,\kappa}_{{\sss R},z}(x+\ve[\kappa])
	&=\sum_{\iota,\kappa}\sum_{x\neq \ve[\iota]}|x-\ve[\iota]-\ve[\kappa]|^2\Pi^{\ssc[0],\iota,\kappa}_{{\sss R},z}(x)\\
	&\leq\sum_{\iota,\kappa} |\ve[\iota]+\ve[\kappa]|^2\Pi^{\ssc[0],\iota,\kappa}_{{\sss R},z}(0)\nn\\
	&\qquad+\indAnimal zg_z
\sum_{\iota,\kappa}\sum_{x\nin\{0,\ve[\iota]\}}[|x-\ve[\iota]|^2+1-2(x_\kappa-\delta_{\iota,\kappa})]\Xi^{\ssc[0],\iota}_z(x),\nn
	}
since $\Pi^{\ssc[0],\iota,\kappa}_z(x)\leq zg_z\Xi^{\ssc[0],\iota}_z(x)$. The first contribution is bounded by $2(2d)^2zg_z G_z(\ve[\iota])$. We use \refeq{Bound-XiIota0Wb}
for the $|x-\ve[\iota]|^2$ contribution.
We further use \refeq{Bound-XiIota0} (without the term created by $x\in\{0,\ve[\iota]\}$) to obtain a bound $(2d)^2zg_z\alpha_{1,1}$ on the contribution due to $+1$. Finally, $\sum_{\iota} \Xi^{\ssc[0],\iota}_z(x)$ is symmetric, so that the $x_\kappa$ contribution vanishes, while the $\delta_{\iota,\kappa}$ contribution gives $2\sum_{\iota, x\nin\{0,\ve[\iota]\}} \Xi^{\ssc[0],\iota}_z(x)$, which can be taken together with the $+1$ contribution using \refeq{Bound-XiIota0}.
\end{proof}

\subsection{Proof of bounds for $N=1$}
\label{secLTProofBoundN1}
In this section we explain how we prove the bounds stated in Lemma \ref{LemmaBoundXiLTOne}.
We explain in detail how we obtain the bounds on $\Xi^{\ssc[1]}_{z}$, stated in
\refeq{Bound-Xi1} and \refeq{Bound-Xi1W}, and then discuss how to
modify the statements to derive the bounds on the other coefficients.

The first step to obtain the bounds are the following $x$-space bounds:

\begin{lemma}[Pointwise bounds]
\label{pDominatespTreeNOne}
For all $x\in\Zd$, $\iota$ and $0\leq z\leq z_c$,
	\begin{eqnarray}
	\lbeq{PnDominatesXiNOne}
	\Xi^{\ssc[1]}_{z}(x)&\leq& \frac {g_z}{\gj} P^{\ssc[1],0}_{z}(x,x),\\
	\lbeq{PnDominatesXiIotaNOne}
	\Xi^{\ssc[1],\iota}_{z}(x)&\leq& \frac {g_z}{\gj} P^{\ssc[1],\iota,0}_{z}(x,x).
	\end{eqnarray}
\end{lemma}
\noindent
We define $P^{\ssc[N],0}$ in Table \ref{boundXiOne}.
\iflongversion
The definition of $P^{\ssc[1],\iota,0}$ is deferred to Appendix \ref{Appendix-def-Initial-XiIota}.
\else
The formal definition of $P^{\ssc[1],\iota,0}$ is provided in Appendix C.2.2 of the extended version \cite{FitHof13g-ext}.
\fi

\begin{table}[ht]
\begin{center}
\begin{tikzpicture}[auto,scale=1]

\begin{scope}[shift={(1,-1.25)},rotate=0]
\draw[line width=0.5 pt] (6,2) to (5,2.5);
\draw[line width=0.5 pt] (5,2.5) to (6,3);
\draw[line width=0.5 pt] (6,2) to (6,3);
\draw[line width=0.5 pt] (6,3) to (7,3);
\draw[line width=0.5 pt] (6,2) to (7,2);
\draw[line width=0.5 pt] (7,3) to (7,2);
\node[right]   at(7,3)      {$w$};
\node[right]   at(7,2)      {$x$};
\node[left]   at(5,2.5)      {$0$};
\node[above]   at(6,3)      {$u$};
\node   at(6.15,2.15)      {$\underline b$};
\end{scope}

\draw[line width=1 pt] (9,1) to (9,-4.5);
\draw[line width=1 pt] (0,0.5) to (15,0.5);

\draw[dotted] (0,-0.5) to (14,-0.5);
\draw[dotted] (0,-1.5) to (14,-1.5);
\draw[dotted] (0,-2.5) to (14,-2.5);
\draw[dotted] (0,-3.5) to (14,-3.5);

 \begin{scope}[shift={(0,1.3)},rotate=0]
    \node[right]   at(1.6,0)      {Terms of $P^{\ssc[1],0}_{z}(x,x)$};
    \node[right]   at(9.5,0)    {represent the case};
   \end{scope}

   \begin{scope}[shift={(0,0)},rotate=0]
    \node[right]   at(0,0)      {$ \sum_{w\in\Zd}
    \left(\diagRepulsiveLetter{T}_{\underline 1,1,1}(x,w,0)+\diagRepulsiveLetter{T}_{2,0,0}(x,w,0)\right)$};
    \node[right]   at(9.25,0)      {$\bb=0$ with $\tb=x$ or not};
   \end{scope}

\begin{scope}[shift={(0,-1)},rotate=0]
    \node[right]   at(0,0)      {$\sum_{w,\bb\in\Zd }
    \diagRepulsiveLetter{B}_{1,1}(0,\bb)\diagRepulsiveLetter{T}_{\underline 1,1,1}(x-\bb,w-\bb,u-\bb)$};
    \node[right]   at(9.25,0)      {$\bb\neq 0$, $v\in\{0,\bb\}$, $\tb=x$};
\end{scope}

\begin{scope}[shift={(0,-2)},rotate=0]
    \node[right]   at(0,0)      {$\sum_{w,\bb\in\Zd }
    \diagRepulsiveLetter{B}_{1,1}(0,\bb)\diagRepulsiveLetter{T}_{2,0,0}(x-\bb,w-\bb,u-\bb)$};
    \node[right]   at(9.25,0)      {$\bb\neq 0$ with $u\in\{0,\bb\}$ and $\tb\neq x$};
\end{scope}

 \begin{scope}[shift={(0,-3)},rotate=0]
    \node[right]   at(0,0)      {$\sum_{w,\bb,u\in\Zd}
\diagRepulsiveLetter{T}_{1,1,1}(0,\bb,u)\diagRepulsiveLetter{T}_{\underline 1,1,0}(x-\bb,w-\bb,u-\bb)$};
\node[right]   at(9.25,0)      {$\bb\neq 0$ with $u\nin\{0,\bb\}$ and $\tb=x$};
 \end{scope}
 \begin{scope}[shift={(0,-4)},rotate=0]
    \node[right]   at(0,0)      {$\sum_{w,\bb,u\in\Zd} \diagRepulsiveLetter{T}_{1,1,1}(0,\bb,u)\diagRepulsiveLetter{T}_{2,0,0}(x-\bb,w-\bb,u-\bb)$};
    \node[right]   at(9.25,0)      {$\bb\neq 0$ with $v\nin\{0,\bb\}$ and $\tb\neq x$};
 \end{scope}
\end{tikzpicture}
\caption{$P^{\ssc[1],0}_{z}(x,x)$ is defined as the sum of the quantities in the left column, and is our bound on $\Xi^{\ssc[1]}_z(x)$.}
\label{boundXiOne}
\end{center}
\end{table}

\begin{proof}
The coefficients that we consider here are defined by sausage-walks for which only the first and last sausage intersect. In Section \ref{sec-BoundIdea-Heuristic}, we have discussed how to bound $\Xi^{\ssc[1]}_z(x)$ for LTs and later obtained the bound \refeq{bound-xi1-structure} for LAs. Using that all connections are repulsive, we can improve \refeq{bound-xi1-structure} to
	\begin{align}
	\Xi^{\ssc[1]}_z(x) \leq& \frac {\gj}{g_z} \sum_{w\in\Zd}\diagRepulsiveLetter{T}_{1,0,0}(x,w,0)
	+\frac 1 {g_z} \sum_{u,w,\bb}
	\diagRepulsiveLetter{S}_{1,0,0,0}(x-\bb,w-\bb,u-\bb,0)
	\diagRepulsiveLetter{B}_{1,0}(-\bb,w-\bb).\nn
	\end{align}
As discussed in the third improvement in Section \ref{sec-BoundIdea-Heuristic}, each loop in this diagram consists of at least four
bonds, and we have that $w\nin b$. All this is implied by the non-backtracking condition in the NoBLE.
We make use of this by considering some special cases for $b$ and $u$, as summarized in Table \ref{boundXiOne},
which improves the bound to obtain \refeq{Bound-Xi1}:
	\begin{align}
	\Xi^{\ssc[1]}_z(x) \leq& \frac {\gj}{g_z} P^{\ssc[1],0}_{z}(x,x).
	\end{align}
\paragraph{Bounds on $\Xi^{\ssc[1],\iota}$.}
For $\Xi^{\ssc[1],\iota}$ we know that either $\tb_1=\ve[\iota]$ or $\ve[\iota]\in A_0^\omega$ as enforced by the indicator $\1_\iota(\omega)$ defined in \refeq{def-Eiota-forLTLA}.
To create a bound we consider $16$ possible forms of the diagram and then bound each of them separately. Then, we define $P^{\ssc[1],\iota,0}_{z}(x,x)$, as well as $\Delta^{\iota,{\sss I},0}$ and $\Delta^{\iota,{\sss II},0}$, as the sum of these $16$ individual diagrams.
\iflongversion
The necessary explanation and definitions of these diagrams are presented in Appendices \ref{Appendix-def-Initial-XiIota}-\ref{Appendix-def-Initial-XiIota-Delta}
of this extended version.
\else
The necessary derivation and definitions are only given in the extended version of this article.
\fi

This completes the proof of Lemma \ref{pDominatespTreeNOne}.
\end{proof}
\medskip

We continue with the proof of Lemma \ref{LemmaBoundXiLTOne}:
\medskip

{\it Proof of Lemma \ref{LemmaBoundXiLTOne}.}
The bounds in \refeq{Bound-Xi1} and \refeq{Bound-XiIota1} follow by summing the pointwise bounds in Lemma \ref{pDominatespTreeNOne} over $x$.

\paragraph{Weighted diagrams for $\Xi_z^{\ssc[1]}(x)$.}
Now we derive a bound for the weighted diagram $|x|^2\Xi_z^{\ssc[1]}(x)$.
Below \refeq{bound-xi1-structure}, we have argued that we can split the sausage-walk such that (only) one line is bounded by $\bar G_z$.
Multiplying such a bound with $|x|^2$ we obtain
\begin{align}
	\lbeq{BoundLT-tmp-2}
	\sum_{x}|x|^2 \bar \Xi^{\ssc[1]}_z(x) \leq&
	\sum_{w,x} |x|^2 \bar G_{z}(x) \tilde G_{z}(w) \tilde G_{z}(w-x)\\
	&+\indAnimal\sum_{x,w,u} \sum_{\bb_1\neq 0} |x|^2  \bar G_{z}(u-\bb_1) 2dz\gj (D\star \tilde G_{z})(u)\nnb
    &\qquad \qquad \times \tilde G_{z}(u-w) \bar G_{z}(\bb_1) \tilde G_{z}(w-x)\tilde G_{z}(x-\bb_1).
	\lbeq{BoundLT-tmp-2.5}
\end{align}
For the second term, which is not present for LTs, we have the problem
that there is no single line connecting $0$ and $x$ to which we allocate the weight $|x|^2$. This is problematic as the bootstrap function $f_3$ only allows us to bound $|x|^2\bar G_z(x)$.
For this reason, we go back to an earlier step in the bounding procedure.
We first apply $|x|^2\leq 2 |\bb_1|^2+2 |x-\bb_1|^2$ when still summing over the sausage-walks. Only, then we split the sausages and chose which line is bounded by $\bar G_z$. This creates the bound
	\begin{align}
	\refeq{BoundLT-tmp-2.5}\leq &
	2 \indAnimal \sum_{u,w,x} \sum_{\bb_1\neq 0}
(2dz\gj) (D\star \tilde G_{z})(u) \tilde G_{z}(u-\bb_1)\tilde G_{z}(x-w)\tilde G_{z}(w-u)\nnb
& \qquad \times \big( |\bb_1|^2 \bar G_{z}(\bb_1) \tilde G_{z}(x-\bb_1)+\tilde G_{z}(\bb_1) |x-\bb_1|^2\bar G_{z}(x-\bb_1) \big)
	\lbeq{BoundLT-tmp-2.6}
	\end{align}
This term can be bounded using our bootstrap functions.
We improve this bound using three ideas, that we now explain.
To improve the first line \refeq{BoundLT-tmp-2}, we consider the cases $w=0$, $w=x$ and $w\not\in\{0,x\}$ to obtain the bound
	\begin{align}
	\sum_{w,x} |x|^2 \bar G_{z}(x) \tilde G_{z}(w) \tilde G_{z}(w-x)
	\leq \sum_{x}|x|^2  \bar G_{z}(x) \left(4dz (D\star G_z)(x)+ (2dz)^2 (D^{\star 2}\star G^{\star 2}_z)(x)\right).\nn
	\end{align}
To improve \refeq{BoundLT-tmp-2.6} we use a total of five special cases for $u,w$ and $\bb_1$ and will omit the details here.

The third idea is to use symmetry, of which we provide an example when $u=\bb_1$ and $x=y+u$. In this case, we split the weight using $|u+y|^2=|y|^2+|u|^2+2\sum_{i=1}^d u_iy_i$ and conclude a bound of the form
	\begin{align}
	&\sum_{u,y,w}  \left(|u|^2+|y|^2+2\sum_{i=1}^d u_iy_i\right)
 	\bar G_{z}(u)  \tilde G_{z}(u) \bar G_{z}(y) \tilde G_{z}(w)\tilde G_{z}(y-w)\nnb
 	&\qquad =\sum_{u} |u|^2\bar G_{z}(u)  \tilde G_{z}(u) \sum_{y,w} \bar G_{z}(y) \tilde G_{z}(w)\tilde G_{z}(y-w)\nnb
	&\qquad\qquad +\sum_{u} \bar G_{z}(u)  \tilde G_{z}(u) \sum_{y,w} |y|^2\bar G_{z}(y) \tilde G_{z}(w)\tilde G_{z}(y-w).
	\lbeq{BoundLT-tmp-3}
	\end{align}
The mixed weights $u_iy_i$ cancel since the bubble to $u$, as well as the triangle to $y$, are symmetric in each coordinate. This saves us a factor $2$, for $u=\bb_1$ when compared to \refeq{BoundLT-tmp-2.6}. The building block $(\vec \Delta^{\rm{start}})_0$ is defined as the sum of these terms. This proves \refeq{Bound-Xi1W}.

\paragraph{Bounds on $\Xi^{\ssc[1]}_{{\sss R},z},\Psi^{\ssc[1]}_{{\sss R, I},z},\Psi^{\ssc[1]}_{{\sss R, I},z}$.}
These coefficients have been obtained by extracting some explicit contributions from $\Xi^{\ssc[1]}_{z}$  resulting in very similar bounds as for $\Xi^{\ssc[1]}_{{\sss R},z}$. The diagrammatic representation of these coefficients are all similar to $P^{\ssc[1],0}_{z}$ depicted in Figure \ref{P1pieceSketch}. In all cases we have extracted contributions to the left square of Figure \ref{P1pieceSketch}, which is bounded by $A^{0,0}$.
We derive the bounds in \refeq{Bound-Xi1Rem}, \refeq{Bound-Psi1RemI} and \refeq{Bound-Psi1RemII} by removing the contribution $A^{0,0}$, and then add the parts of this square that are actually required.

For $\Xi^{\ssc[1]}_{{\sss R},z}$ and $\Psi^{\ssc[1]}_{{\sss R, II},z}$, we have extracted the case $x=0$, as well as all contributions in which $0$ and $x$ are directly connection via a bond, see \refeq{LA-Split-Def-None}. For this reason, all connections from $0$ to $x$ consist of at least $2$ steps, which is reflected in \refeq{Bound-Xi1Rem} and \refeq{Bound-Psi1RemII}.

For $\Psi^{\ssc[1]}_{{\sss R, I},z}$, we have only extracted the contribution of $x=0$ and $x=\ve[\kappa]$.
Thus, when comparing \refeq{Bound-Psi1RemI} and \refeq{Bound-Psi1RemII}, we see that we have to add to
\refeq{Bound-Psi1RemI} also the contribution of points $x\neq \ve[\kappa]$ that can be reached within one step.

Regarding the weighted bounds stated in \refeq{Bound-Xi1W}-\refeq{Bound-Psi1RemIW}, it (numerically) turns out that we cannot benefit from the fact that we have extracted terms. For this reason we simply note that $\Xi^{\ssc[1]}_{{\sss R},z}(x)\leq \Xi^{\ssc[1]}_{z}(x)$ and $\max\{\Psi^{\ssc[1]}_{{\sss R, I},z}(x),\Psi^{\ssc[1]}_{{\sss R, II},z}(x)\}\leq \Xi^{\ssc[1]}_{z}(x)$,
and apply the original bounds for these terms.

\paragraph{Weighted bounds on $\Xi^{\ssc[1], \iota}_{z}$.} The bounds on the weighted diagrams of $\Xi^{\ssc[1], \iota}_{z}$ in \refeq{Bound-XiIota1Wa}--\refeq{Bound-XiIota1Wb} follow by carefully analysing the weighted version of $P^{\ssc[1],\iota,0}_{z}(x,x)$, which appear as sums of $32$ individual diagrams.
\iflongversion
The necessary explanation and definitions of these diagrams are again presented in Appendix \ref{Appendix-def-Initial-XiIota-Delta}
of this extended version.
\else
The necessary derivation and definitions are only given in the extended version of this article.
\fi
\medskip

The above bounds prove Lemma \ref{LemmaBoundXiLTOne}.
\qed

\subsection{Proof of bounds for $N\geq 2$}
\label{secLTProofBound}
In this section we will explain how we bound the NoBLE coefficients with $N\geq2$ and thereby complete the proof of Proposition \ref{LemmaBoundXiLTTwo}.

\paragraph{Decomposition of weights.}
We bound the diagram for $N\geq 2$ by splitting the diagram into individual pieces, as seen in Figure \ref{TreeXiFourDeomposedStructure}. When adding the spatial weight $|x|^2$ we need to distributed the weight along the connections of the diagram, as already seen in the case of $\Xi^{\ssc[1]}$ in \refeq{BoundLT-tmp-2.5}-\refeq{BoundLT-tmp-2.6}.
For this, we define $(x_i)_i$ as displayed in Figure \ref{BoundLT-splitchoice} and use
	\begin{align}
	    \lbeq{Split-weight-ineq}
	   \big|\sum_{i=1}^J x_i\big|^2=
	&\sum_{i=1}^J |x_i|^2+ \sum_{i=2}^J x_i^T \Big(\sum_{j=1}^{i-1} x_j\Big),
	\qquad
	  &
	  \big|\sum_{i=1}^J x_i\big|^2\leq& J\sum_{i=1}^J|x_i|^2.
	\end{align}

\begin{figure}[ht!]
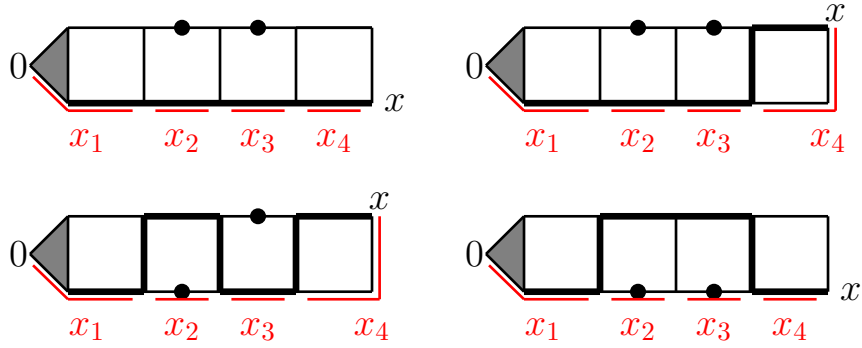

\begin{center}
{\Large
\picLtchoiceOfXXX[1]
\\}
\caption{The choice of the element for the rewrite of $x_i$ at the example of the diagrammatic representation of $\Xi^{\ssc[4]}_z(x)$. The wide line represents the path of the backbone.
We choose to distribute the displacement along the bottom-line of the diagram.}
\label{BoundLT-splitchoice}
\end{center}
\end{figure}

After applying \refeq{Split-weight-ineq} we bound each term of the sum with its partial spatial weight $|x_i|^2$ individually. This will be done in \refeq{Split-weight-bound} below. As the partial spatial weight depends on the lattice walk and lace, writing $|x_i|^2\Xi^{\ssc[N]}_{z}(x)$ is not meaningful. To fix this technicality, we define
for $N\geq 2$ and $i=1,\dots,N$,
	\begin{align}
	\Xi^{\ssc[N]}_z(x,\Delta_i) &= \frac {1}{g_z} \sum_{\omega\in\Wcal(x)} z^{|\omega|} Z[0,|\omega|]  \nnb
	&\qquad\times \sum_{L\in\Lcal^\ssc[N][0,|\omega|]} \prod_{st\in L}(-\Ucal_{st})\prod_{s't'\in \Ccal(L)}(1+\Ucal_{s't'})
	|x_i(\omega,L)|^2
	\end{align}
with $|x_i(\omega,L)|^2$  shown in Figure \ref{BoundLT-splitchoice}.
We define $\Xi^{\ssc[N],\iota}_z(x,\Delta_i)$ accordingly with the additional indicator $\1_{\iota}(\omega)$
(compare \refeq{defXiNanimal} and \refeq{defXiNIotaanimal}). It is possible to give the formal definition of $x_i(\omega,L)$ using the lace $L={s_1t_1,\dots, s_Nt_N}$, the backbone $(b^\omega_j)_{j=1}^{|\omega|}$, and the sausages that are being split $S_{s_i}^\omega$, $S_{t_i}^\omega$. As this would not be very informative, we omit this and just refer to Figure \ref{BoundLT-splitchoice} instead. Figure \ref{BoundLT-splitchoice} allows for a straightforward adaptation for all $N\geq 3$. Then we obtain, using \refeq{Split-weight-ineq},
	\eqn{
	\lbeq{weighted-diag-Xi}
	\sum_{x\in \Zd} |x|^2 \Xi^{\ssc[N]}_z(x)\leq N\sum_{i=1}^N \Xi^{\ssc[N]}_z(x,\Delta_i).
	}

\paragraph{Construction of the bounding diagram.}
We combine the building blocks, see Section \ref{secBuildingBlocks}, to construct the bounds on the NoBLE coefficients: We recursively define
	\begin{eqnarray}
	\lbeq{PnConstructLT}
	P^{\ssc[N],m}(x,y)&=&\sum_{u,v\in\Zd} \sum_{l=-2}^2 P^{\ssc[N-1],l}(u,v) A^{l,m}(u,v,x,y),\\
	P^{\ssc[N],\iota,m}(x,y)&=&\sum_{u,v\in\Zd} \sum_{l=-2}^2 P^{\ssc[N-1],\iota,l}(u,v) A^{l,m}(u,v,x,y),\\
	Q^{\ssc[1],m}(t,y;x)&=& A^{0,m}(x,x,t,y),
	\lbeq{Q1ConstructLT}\\
	Q^{\ssc[N],m}(t,y;x)&=& \sum_{u,v\in\Zd}\sum_{l=-2}^2 Q^{\ssc[N-1],l}(u,v;x) A^{l,m}(u,v,t,y),
	\lbeq{QnConstructLT}
\lbeq{PnConstructLT-Qm}
	\end{eqnarray}
for $N\geq 2$ and $m\in\{-2,-1,\dots,2\}$.
\iflongversion
The formal definition of $P^{\ssc[1],m}$ and $P^{\ssc[1],\iota,m}$, which serve as the initialisation of $P^{\ssc[N],m}$ and $P^{\ssc[N],\iota,m}$ for $N=1$ is given in Appendix \ref{Appendix-def-Initial-Xi} and \ref{Appendix-def-Initial-XiIota}.
\else
The formal definition of $P^{\ssc[1],m}$ and $P^{\ssc[1],\iota,m}$ are only provided in the
Appendix B of the extended version \cite{FitHof13g-ext}. Here we give only the definition of $P^{\ssc[N],0}$ in Figure \ref{boundXiOne} as example.
\fi

Using these constructs we bound the coefficients as follows:
\begin{lemma}[Pointwise bounds]
\label{pDominatespTree}
For all $x\in\Zd$, $\iota$ and $0\leq z\leq z_c$, and $N\geq 2$
	\begin{eqnarray}
	\lbeq{PnDominatesXiNTwo}
	\Xi^{\ssc[N]}_{z}(x)&\leq& \frac {\gj}{g_z} \sum_{u,v,l} P^{\ssc[N-1],l}(u,v) A^{l,0}(u,v,x,x),\\
	\lbeq{PnDominatesXiNTwoWeight}
	\Xi^{\ssc[N]}_{z}(x,\Delta_N)&\leq&  \sum_{u,v,l} P^{\ssc[N-1],l}(u,v) \Delta^{{\rm end},l}(u,v,x,x),\\
	\lbeq{PnDominatesXiNTwoPrime}
	\Xi^{\ssc[N]}_{z}(x,\Delta_1)&\leq& \sum_{u,v,l} \Delta^{{\rm start},l} (x_1,v) Q^{\ssc[N-1],l}(u,v;x).
	\end{eqnarray}
Further, for  $N,M\geq 1$,
	\begin{align}
	\lbeq{PnDominatesXiN}
	\Xi^{\ssc[N+M+1]}_{z}(x,\Delta_{N+1})&\leq& \sum_{u,v,w,y}\sum_{l,m}P^{\ssc[N],l}(u,w)C^{l,m}(u,v,w,y)Q^{\ssc[M],m}(w,y;x).
	\end{align}
The same bounds hold for $\Xi^{\ssc[N],\iota}$, when $P^{\ssc[N-1],l}$ and $\Delta^{{\rm start},l} $
 are replaced by $P^{\ssc[N-1],\iota,l}$ and $\Delta^{{\rm iota},{\sss I},l} $ or $\Delta^{{\rm iota},{\sss II},l}$, respectively.
\end{lemma}
Combining all these bounds we are able to bound the weighted diagrams. First, we split the weight $|x|^2$ using
\refeq{Split-weight-ineq}, as formulated in \refeq{weighted-diag-Xi}. Then, we use the bounds in Lemma \ref{pDominatespTree} to obtain
	\begin{align}
	\sum_{x\in\Zd}|x|^2 \Xi^{\ssc[N]}(x)\leq
	&N \sum_{l,x_1,v,x} \Delta^{{\rm start},l} (x_1,v) Q^{\ssc[N-1],l}(x_1,v;x)\nnb
	&+N\sum_{M=1}^{N-2}
	\sum_{\footnotesize\begin{array} {c} l,m,u,\\ v,w,x,y\end{array}}
	P^{\ssc[M],l}(u,v) C^{l,m}(u,v,w,y) Q^{\ssc[N-M-1],m}(w,y;x)\nnb
	&+N\sum_{l,x_N,v,x}  P^{\ssc[N-1],l}(x-x_N,x+v) \Delta^{{\rm end},l}(0,v,x_N,x_N),
	\lbeq{Split-weight-bound}
	\end{align}
which implies the bound stated in \refeq{Bound-Xi2W}.

Let us highlight once more how crucial the order is in which we create our bounds:
First, we group each combination $(\omega,L)$ according to the length and role of shared lines $(l_i)_{i=1,\dots,N-1}$.
Secondly, we identify the first intersection points at which we split the sausage-walks.
Thirdly, we use \refeq{Split-weight-ineq} to split the weight along the bottom lines, as shown in
 Figure \ref{BoundLT-splitchoice}.
In the fourth step, we apply, depending on the weight $|x_i|^2$, one of the bounds of  Lemma \ref{pDominatespTree}.
Then we also allocate the one-point functions such that the weight is along a line corresponding to a two-point function $G_z$.

\begin{figure}[ht]
\begin{center}
\begin{tikzpicture}[auto,scale=1.5]
\begin{scope}
\draw[line width=2.5 pt] (0,0) to (2,0);
\draw[line width=1 pt] (0,1) to (2,1);
\draw[line width=1 pt] (0,0) to (0,1);
\draw[line width=1 pt] (1,0) to (1,1);
\draw[line width=1 pt] (2,0) to (2,1);
 \fill (0,1) circle (2pt);
 \fill (2,1) circle (2pt);
 \node[left]   at(0,1)      {$w_1$};
 \node[left]   at(0,0)      {$0$};
 \node[right]   at(2,1)      {$w_2$};
 \node[right]   at(2,0)      {$x$};
  \node[above]   at(1,1)      {$u$};
    \node[below]   at(1,0)      {$\underline b^\omega_{s_2}=\underline b^\omega_{t_1}=y$};
   \node[right]   at(1,0.5)      {$=-m$};
   \end{scope}

   \begin{scope}[shift={(4,0)},rotate=0]
\draw[line width=2.5 pt] (0,0) to (1,0);
\draw[line width=1 pt] (1,0) to (2,0);
\draw[line width=1 pt] (0,1) to (1,1);
\draw[line width=2.5 pt] (1,1) to (2,1);
\draw[line width=1 pt] (0,0) to (0,1);
\draw[line width=2.5 pt] (1,0) to (1,1);
\draw[line width=1 pt] (2,0) to (2,1);
 \fill (0,1) circle (2pt);
 \fill (2,1) circle (2pt);
 \node[left]   at(0,1)      {$w_1$};
 \node[left]   at(0,0)      {$0$};
 \node[right]   at(2,1)      {$x$};
 \node[right]   at(2,0)      {$w_2$};
  \node[above]   at(1,1)      {$\underline b^\omega_{t_1}=u$};
    \node[below]   at(1,0)      {$\underline b^\omega_{s_2}=y$};
   \node[right]   at(1,0.5)      {$=m$};
   \end{scope}
\end{tikzpicture}
\caption{Labeling in the proof of the bound on $\Xi^\ssc[2]_z$ for $\bb_0=0$.
The left diagram corresponds to $s_2=t_1$ and the right to $s_2<t_1$.}
\label{structureInProveTreeNtwo}
\end{center}
\end{figure}
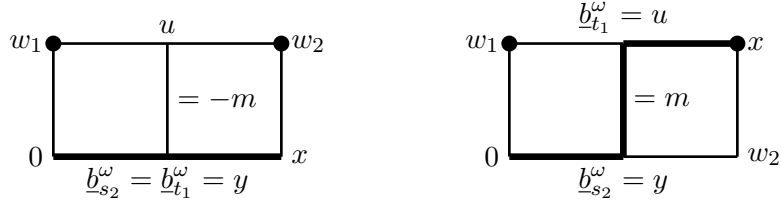
\noindent

\begin{proof}[Proof of Lemma \ref{pDominatespTree}] We divide the proof depending on whether $N=2$ and $N\geq 3$.

\paragraph{Proof for $N=2$.} As we have already discussed how to bound the non-trivial sausage for $\Xi^\ssc[1]_z$, occurring for $\bb_1\neq 0$, we will restrict to $\bb_1= 0$ in our discussion for $\Xi^\ssc[2]_z $. We have to consider the coefficients as shown in Figure \ref{structureInProveTreeNtwo}: For a sausage-walk $\omega$ and a lace $L=\{0t_1,s_2|\omega|\}$ to which $\omega$ contributes we define $w_1=w_1(\omega),\ w_2=w_2(\omega)$ and $y=\tb^\omega_{s_2}$. If $s_2<t_1$, then we define $u=\omega_{s_1}$ and $l=\max\{d_\omega(u,y),2\}$. For $s_2=t_1$ we identify the last point $u$ that the paths $b^{S^\omega_{s_2}}(\bb^\omega_{s_2+1}, v)$ and $b^{S^\omega_{s_2}}(\bb^\omega_{s_2+1}, w)$ have in common and define $l=-\max\{d_{A^\omega_{s_2}}(\bb^\omega_{s_2+1},u),2\}$.

By definition of a sausage-walk (Definition \ref{defLASausagewalks}), using the non-backtracking condition, and the selection of the first intersection points $w_1,w_2$, we know that each square is repulsive and consists of at least four bonds. Especially, if $t_1=1$, so that $S^\omega_{0}$ and $S^\omega_{1}$ intersect, then the non-backtracking condition states that this intersection cannot occur at the backbone bond, i.e., $w_1\nin b_1^\omega$. For the same reason we know that $w_2\nin b_{|\omega|}^\omega$ in the cases of $|\omega|-s_2=1$ and $\tb_{|\omega|}^\omega=x$.
This last observations look like a minute special case, but it is numerically a major contribution.

To obtain \refeq{PnDominatesXiNTwo}, we let the line $u,y$ contribute to the left square.
When we add the weight $|x-y|^2$, we obtain \refeq{PnDominatesXiNTwoWeight}.
When considering the weight $|y|^2$, we choose the split of the diagram such that
the line $u,y$ contributes to the right square instead. This creates the bound \refeq{PnDominatesXiNTwoPrime}.
To prove \refeq{PnDominatesXiNTwo}-\refeq{PnDominatesXiNTwoPrime}, we have to go through all cases for $l$ and
check that all possible sausage-walks are captured by the bounds on the right-hand sides. These tedious arguments
follow the same steps as in the proof of the bounds on the coefficients for $N=1$ in the previous section, and we omit the details.
The bound on $\Xi^{\ssc[2],\iota}_z$ is obtained in the same way, using the initial piece  $P^{\ssc[1],\iota,m}_{z}$, sketched in Section \ref{secBuildingBlocks}.

\paragraph{Proof for $N\geq 3$.}
The bounds for $N\geq 3$ are proven following the same procedure as described for $N=2$. For each sausage-walk $\omega$ and lace $L$,
we identify the points at which we split $\omega$, these are $\bb_{s_i}$, $\bb_{t_i}$, as well as the first intersection points $w_i(\omega)$.
Whenever $s_i=t_{i-1}$, we also identify the point at which we split the sausage $S_{s_i}^\omega$, see Figures \ref{BoundLA-Figure-Xitwo} and \ref{LTXi4Ribweight}.
Then, we define the lengths of the shared lines of neighboring loops $l_i$ accordingly.
For a bound on the weighted diagrams, we then split the weight using \refeq{Split-weight-ineq}, leading to a sum of weighted diagrams as in \refeq{weighted-diag-Xi}, and for each of these summands, we derive a bound as in \refeq{PnDominatesXiNTwoWeight}-\refeq{PnDominatesXiN} depending on which line is weighted.

We now give some insight into how these bounds can be proven. As all these arguments do not offer any further
insight, we omit their formal proof and refer the interested reader to \cite[Section 4.3.5]{Fit13}, where more details are given.
It involves defining skeletons that encode (i) the length of the pieces of shared lines of neighboring loops; (ii) the location of the
backbone; and (iii) the partial displacements $(x_i)_{i=1}^{N}$.
Then, we classify all pairs or rib-/sausage-walks and laces $(\omega,L)$ uniquely according to their skeleton and
prove the bound using induction on $N$. The proof strategy uses ideas already presented above and techniques described
in the standard reference \cite{Slad06}.

\paragraph{Bound on $\Xi^{\ssc[N],\iota}_z$.}
The coefficients $\Xi^{\ssc[N]}_z$ and  $\Xi^{\ssc[N],\iota}_z$ differ only by the first sausage,  as informally explained in the foundation block $P^{\ssc[1],m}$ and $P^{\ssc[1],\iota,m}$.
Thus, also the bounds differ only in the first block. The block $P^{\ssc[1],\iota,m}$ can have very many forms and the bounds on its weighted versions $\Delta^{\iota,{\sss I},m}$ and
$\Delta^{\iota,{\sss II},m}$ have to be done with care, as they are numerically important.

\iflongversion
We provide a complete definition of the bounding diagrams $P^{\ssc[1],\iota,m}$, $\Delta^{\iota,{\sss I},m}$ and
$\Delta^{\iota,{\sss II},m}$ that list all possible cases, and rigorously derive their corresponding bounds in
Appendices \ref{Appendix-def-Initial-XiIota}-\ref{Appendix-def-Initial-XiIota-Delta}.
\else
In the extended version of this article \cite{FitHof13g-ext}, we provide the complete definitions of $P^{\ssc[1],\iota,m}$, $\Delta^{\iota,{\sss I},m}$ and $\Delta^{\iota,{\sss II},m}$ that list all possible cases/forms and their corresponding bounds. See \cite[Appendices C.2.2-C.2.4]{FitHof13g-ext} As this takes 19 pages we omit these definition here. In this version we only added the sketch of $P^{\ssc[1],\iota,m}$ in Section \ref{secBuildingBlocks}.

\fi
\end{proof}

\paragraph{Decomposition using matrices.}
In the recursive definitions of $P^{\ssc[N],m}$, $P^{\ssc[N],\iota,m}$, $Q^{\ssc[N],m}$, in
\refeq{PnConstructLT}-\refeq{PnConstructLT-Qm}, as well as in the bounds of Lemma \ref{pDominatespTree}, we sum over the lengths of lines that are shared by two loops. We do this as it allows us to make use of the fact that any loop consists of at least four bonds.
We convert this sum over line lengths into a product of matrices to obtain the bounds stated in Proposition \ref{LemmaBoundXiLTTwo}.
Here we explain how we do this. We start at the definition of $P^{\ssc[N],m}$ in \refeq{PnConstructLT} and obtain
	\begin{align}
	\lbeq{BoundLT-tmp-7}
	P^{\ssc[N],m}_z&=\sum_{x,y\in\Zd}P^{\ssc[N],m}_z(x,y)=	\sum_{l=-2}^2 \sum_{u,v,x,y\in\Zd}  P^{\ssc[N-1],l}_z(u,u+v) A^{l,m}(u,u+v,x,y)\\
	&\leq \sum_{l=-2}^2 \sum_{u,v\in\Zd}  P^{\ssc[N-1],l}(u,v)
	\sup_{v\in\Zd} \sum_{x,y} A^{l,m}(0,v,x,y)=\sum_{l} P^{\ssc[N-1],m}_z ({\bf A})_{l,m},\nn
	\end{align}
where we recall \refeq{defOfBoundingElement2a}. Using this recursion, we can prove by induction and using \refeq{defOfBoundingElement1}, that
	\begin{align}
	P^{\ssc[N],m}_z\leq (\vec P^T{\bf A}^{N-1})_m.\nn
	\end{align}
Combining this with \refeq{PnDominatesXiNTwo},
we obtain the bound on $\Xi^{\ssc[N]}_{z}$ as stated in \refeq{Bound-Xi2} by bounding
	\begin{align}
	\sum_x \Xi^{\ssc[N]}_{z}(x)\leq& \frac {\gj}{g_z} \sum_{u,v,l,x} P^{\ssc[N-1],l}_{z}(u,v) A^{l,0}(u,v,x,x)\nnb
	=&\frac {\gj}{g_z}  \sum_{u,v,l,x} P^{\ssc[N-1],l}_{z}(u,u+v) A^{l,0}(0,v,x,x)\nnb
\leq & \frac {\gj}{g_z} \sum_l \left( \sum_{u,v} P^{\ssc[N-1],l}_{z}(u,u+v)\right) \left( \sup_{v\in\Zd} \sum_{x} A^{l,0}(0,v,x,x)\right)\nnb
    \stackrel{\refeq{defOfBoundingElement2}}{\leq} &
   \frac {\gj}{g_z}  \sum_l (\vec P^T{\bf A}^{N-2})_l (\vec E^{\rm{closed}})_l= \frac {\gj}{g_z}  \vec P^T{\bf A}^{N-2}{\vec E}^{\rm{open}}.
 	\lbeq{BoundLT-tmp-5}
 	\end{align}
The bound on $\Xi^{\ssc[N],\iota}_{z}$ is obtained in the same way, where $P^{\ssc[N],l}_{z}$ is replaced with $P^{\ssc[N],\iota,l}_{z}$.\\
Using translation invariance for ${Q}_z^{\ssc[M],m}$, defined in \refeq{Q1ConstructLT}--\refeq{QnConstructLT}, we obtain
	\begin{align}
	\lbeq{BoundLT-tmp-8} \sum_{v,w\in\Zd}{Q}_z^{\ssc[M],m}(w,y;x)=&\sum_{w,y\in\Zd}Q_z^{\ssc[M],m}(w,y;0)=\sum_{y,x\in\Zd}Q_z^{\ssc[M],m}(0,y;x)\\
\stackrel{\refeq{PnConstructLT-Qm}}=  &\sum_{l=-2}^2 \sum_{u,v,x,y\in\Zd} Q_z^{\ssc[M-1],l}(0,v;x) A^{l,m}(u,u+v,0,y)\nn\\
	\leq& \sum_{l=-2}^2 \left(\sum_{x,v\in\Zd}  Q^{\ssc[M-1],l}(0,v;x)\right)\left(\sup_{v\in\Zd} \sum_{u,y} A^{l,m}(u,u+v,0,y)\right)\nn\\
	\leq& ({\bf A}^{M-1}\vec E^{\rm{closed}})_m.\nn
	\end{align}
Combining \refeq{BoundLT-tmp-7}, \refeq{BoundLT-tmp-8} with \refeq{Split-weight-bound}, we obtain
	\begin{align}
	\sum_{x\in\Zd}|x|^2 \Xi_z^{\ssc[N]}(x)\leq
	&N \sum_{l} \Big(\sup_v \sum_{x_1} \Delta^{{\rm start},l} (x_1,v)\Big)
	({\bf A}^{N-2}\vec E^{\rm{closed}})_l\nnb
	&+N\sum_{M=1}^{N-2}
	(\vec P^T{\bf A}^{M-2})_m  \Big(\sup_{v,y}\sum_{w} C^{m,l}(0,v,w,w+y)\Big) ({\bf A}^{N-M-2}\vec E^{\rm{closed}})_l\nnb
	&+N\sum_{l} (\vec P^T{\bf A}^{N-2})_l  \Big(\sup_v \sum_{x_N}\Delta^{{\rm end},l}(0,v,x_N,x_N)\Big)\nnb
	\stackrel{\refeq{defOfBoundingElement1},\refeq{defOfBoundingElement2}}{=}&
	N  (\vec \Delta^{\rm{start}})^T{\bf A}^{N-2}{\vec E^{\rm{closed}}}+N\sum_{M=1}^{N-2} \vec P^T {\bf A}^{M-2} {\bf C} {\bf A}^{N-M-2}\vec E^{\rm{closed}}\nnb
	&+N\vec P^T{\bf A}^{N-2}\vec \Delta^{{\rm end}},
	\end{align}
which proves \refeq{Bound-Xi2W}. The bound on $\Xi^{\ssc[N],\iota}$ are proven in the same way, using other initial diagrams.

Let us briefly mention a numerical aspect here. For $|m|=1$, we only count $u,v$ that are neighbors.
By spatial symmetry,
	\begin{align}
	\sup_{v}\sum_{u} P_z^{\ssc[N],m}(u,u+v)=\sum_{u}P_z^{\ssc[N],m}(u,u+\ve[1])=\frac 1 {2d}\sum_{u,w} P_z^{\ssc[N],m}(u,u+w).
	\lbeq{BoundLT-tmp-8.1}
	\end{align}
This means that the step of taking the supremum over $v\in\Zd$, while disassembling $P_z^{\ssc[N],m}$ and ${Q}_z^{\ssc[M],m}$ in \refeq{BoundLT-tmp-7} and \refeq{BoundLT-tmp-8}, actually holds with an equality for $|m|\leq 1$. We also apply this method to $P^{\ssc[N],\iota,m}_{z}$, when we simply sum over $\iota$, to obtain
	\begin{align}
	\sum_{u,\iota} P^{\ssc[N],\iota,m}_{z}(u,u+\ve[\kappa])=\sum_{u,\iota}P^{\ssc[N],\iota, m}_{z}(u,u+\ve[1]).
	\end{align}

\subsection{Bounds on differences: Proof of Lemma \ref{lemmapercboundXi0minus1}}

Now we turn to the bounds stated in Lemma \ref{lemmapercboundXi0minus1}.
By definition $\Xi^\ssc[0]_{\alpha,z}(0)=0$, see \refeq{LA-Split-Def-1}.
Further, we recall that, now by \refeq{LA-Split-Def-None},
	\begin{align}
	\Xi^\ssc[1]_{\alpha,z}(x)
	=\frac 1 {g_z} \sum_{\omega\in\Wcal(x)} z^{|\omega|} Z[0,|\omega|]J^{\ssc[1]}[0,|\omega|]\indic{\bb^\omega_0=0}.
	\end{align}
For $x=0$, the sausage-walk needs to return to the origin.
Note that the indicator $J^{\ssc[1]}[0,|\omega|]$ enforces an intersection between the zeroth and the last sausage, which in fact is implied also by the fact that $\omega\in\Wcal(0)$ and $\bb^\omega_0=0$. However, $J^{\ssc[1]}[0,|\omega|]$ also forces the other sausages to avoid each other. Thus, the sausage-walk, without the initial sausage and $b^\omega_0$, form a lattice animal and can be bounded by $\bar G_{3,z}(\tb^\omega_0)$. We bound the first sausage by the factor $z\gj$, resulting in the bound
	\begin{align}
		0\leq \Xi^\ssc[1]_{\alpha,z}(0)\leq \frac {2dz\gj} {g_z} \bar G_{3,z}(\ve[\iota] ),
	\end{align}
which, together with $\Xi^\ssc[0]_{\alpha,z}(0)=0$, shows the bounds in \refeq{Differencebound-1} for $x=0$, as well as \refeq{Differencebound-2}. Using the same reasoning we obtain
	\begin{align}
	\Xi^\ssc[1]_{\alpha,z}(\ve[1])\leq  &\frac {g_z^\iota} {g_z} \Big(2\diagRepulsiveLetter{B}_{3,\underline 1}(\ve[1],0)+\sum_v \diagRepulsiveLetter{T}_{\underline 1,1,1}(\ve[1],v,0)\Big),
	\end{align}
which, together with $\Xi^\ssc[1]_{\alpha,z}(\ve[1])\geq 0$, implies \refeq{Differencebound-3}.

To prove \refeq{Differencebound-1} for $x=\ve[1]$, we now show that $\Xi^\ssc[0]_{\alpha,z}(\ve[1])<\Xi^\ssc[1]_{\alpha,z}(\ve[1])$. We do this by arguing that each LA contributing to $\Xi^\ssc[0]_{\alpha,z}(\ve[1])$ can be decomposed into a unique contribution to $\Xi^\ssc[1]_{\alpha,z}(\ve[1])$. For any  LA that contributes to $\Xi^\ssc[0]_{\alpha,z}(\ve[1])$, choose one of the connections between $0$ and $\ve[1]$ to be the backbone, and allocate the rest of the LA to be sausages of the walk in any unique way you want (as long as it is unique, the map will yield an injection). As there is a double connection in the original LA, the first and last sausages in the created sausage-walks intersect, and thus provide a contribution to $\Xi^\ssc[1]_{\alpha,z}(\ve[1])$.
The bounds \refeq{Differencebound-6}-\refeq{Differencebound-7} are obtained in a similar way.

In $\Psi^{\ssc[0],1}_{\alpha,{\sss I},z}(\ve[1]+\ve[\iota])$ and $\Psi^{\ssc[1],1}_{\alpha,{\sss I},z}(\ve[1]+\ve[\iota])$, the origin and $\ve[1]+\ve[\iota]$ are connected using two disjoint paths, each consisting of two bonds.
Thus,
	\begin{align}
	0\leq \Psi^{\ssc[0],1}_{\alpha,{\sss I},z}(\ve[1]+\ve[\iota])\leq& \indAnimal \diagRepulsiveLetter{B}_{2,\underline 2}(\ve[1]+\ve[\iota],0),\\
	0\leq  \Psi^{\ssc[1],1}_{\alpha,{\sss I},z}(\ve[1]+\ve[\iota])\leq& \sum_{w_1}\diagRepulsiveLetter{T}_{0,0,\underline 2}(w_1,\ve[1]+\ve[\iota],0).
	\end{align}
Combining these two bounds yields \refeq{Differencebound-4} and \refeq{Differencebound-5}.
\qed

\section{Proof of main results: Critical exponents $\nu, \gamma$ and $\eta$}
\label{sec-proof-main}
In this section, we prove our main results. We start by noting that Theorem \ref{thm-IRB} immediately follows from the fact that $f_2(z)\leq \gamma_2$ {\em uniformly} in $z<z_c$. This implies Theorem \ref{thm-IRB}. Theorem \ref{thm-bds-crit} follows from the numerical estimates and the facts that $f_i(z)\gamma_i$ uniformly in $z<z_c$ for $i=1,2$. We prove Corollary \ref{cor-TC-crit-exp} in Section \ref{sec-gamma-nu}, Theorem \ref{thm-k-space} in Section \ref{sec-k-space}, and Theorem \ref{thm-x-space} in Section \ref{sec-x-space}.

\subsection{Proof of $\gamma=1/2$ and $\nu=1/4$ in Corollary \ref{cor-TC-crit-exp}}
\label{sec-gamma-nu}
The finiteness of the square diagram follows from the infrared bound in Theorem \ref{thm-IRB}, together with the Fourier inversion theorem, which together imply that
	\eqn{
	\square(z_c)=\lim_{z\nearrow z_c}\square(z)=
	\lim_{z\nearrow z_c}\bar{G}^{\star 4}_z(0)=\lim_{z\nearrow z_c} \int_{-[\pi,\pi]^d} \hat{\bar{G}}_z(k)^4\frac{d^dk}{(2\pi)^d}<\infty.
	}
The finiteness of the integral being valid for $d>8$, under the infrared bound which holds for $d\geq \dmintree$ and $d\geq \dminanimal$ for lattice trees and lattice animals, respectively, by Theorem \ref{thm-IRB}. The fact that $\gamma=1/2$ follows directly from the finiteness of the square diagram, as proved in \cite{BovFroGla86b, Tasa86, HarTas87}.

In the following, we will show that $\nu=1/4$. The critical exponent $\nu$ governs the explosion for $z\nearrow z_c$ of $\xi_2(z)$ defined in \refeq{genTreestatement-xi}.
For this we rely on notation and techniques used in the accompanying paper \cite{FitHof13g-ext}. Indeed, we write, using \refeq{lace-exp-eq},
	\begin{align}
	\hat G_z(k)=\frac {\hat \Phi_z(k)}{1-\hat F_z(k)},
	\end{align}
with
	\begin{align}
      \lbeq{DefPhi}
	\hat \Phi_z(k) :=& 1+\hat \Xi_z(k)- \aaz(\v1+\vPsiz[k])^T\left[ \mD[k] + \aaz\mJ +\mPi[k]\right]^{-1} \vXiz[k],\\
	\lbeq{DefFFunction}
	\hat F_z(k) :=&\aaz(\v1 + \vPsiz[k])^T\left[ \mD[k] + \aaz\mJ +\mPi[k]\right]^{-1} \v1.
	\end{align}
By the completion of the bootstrap argument, we know that all the above NoBLE coefficients are well defined for all $z\leq z_c$.

In \cite[Section 3.3.1]{FitHof13g-ext}, we define the continuous Laplace operator $\bigtriangleup$: For a differentiable function $g$ and $s\in\{1,2,\dots,d\}$, let $\partial_s g(k)=\frac \partial {\partial k_s} g(k)$  and $\bigtriangleup g(k)=\sum_{s=1}^d\partial_s^2 g(k)$.
For the Laplace operator $\bigtriangleup$, we note that
	\begin{align}
	\lbeq{Laplace-Observation}
 	\sum_{x}\sum_{A\ni 0,x}|x|^2z^{|A|}=-\bigtriangleup \hat G_z(k)\big|_{k=0},
	\end{align}
and compute, now restricting to $z<z_c$,
	\begin{align}
 	\bigtriangleup \hat G_z(k)\big|_{k=0}=&
  	\frac {\bigtriangleup\hat \Phi_z(k)|_{k=0}}{1-\hat F_z(0)} +\frac {\hat \Phi_z(0)\bigtriangleup \hat F_z(k)|_{k=0}}{(1-\hat F_z(0))^2}\\
	& +2 \frac {\sum_s (\partial_s  \hat \Phi_z(k)\big|_{k=0})^2}{(1-\hat F_z(0))^3}
	+2 \frac {\sum_s \partial_s  \hat \Phi_z(k)\big|_{k=0}\partial_s  \hat F_z(k)\big|_{k=0}  }{(1-\hat F_z(0))^2}.\nn
	\end{align}
By the symmetry inherited from the symmetry of lattice trees and lattice animals, we know that all first order derivatives are zero at $k=0$, so that all elements in the second line are zero. Thus, we conclude that
	\begin{align}
 	\bigtriangleup \hat G_z(k)\big|_{k=0}=&
  	\frac {\bigtriangleup\hat \Phi_z(k)|_{k=0}}{\hat \Phi_z(0)} \chi(z)
  	+\frac {\bigtriangleup \hat F_z(k)|_{k=0}}{\hat \Phi_z(0)}\chi(z)^2.
	\end{align}
Starting from $\xi_2(z)$, defined in \refeq{genTreestatement-xi}, and \refeq{Laplace-Observation}, we conclude that
	\begin{align}
	\xi_2(z)^2  =&  -\frac {\bigtriangleup \hat F_z(k)|_{k=0}}{\hat \Phi_z(0)}\chi(z)
  	-\frac {\bigtriangleup\hat \Phi_z(k)|_{k=0}}{\hat \Phi_z(0)}\\
  	=& \chi(z) \frac {\sum_{x} |x|^2F_z(x)} {\hat \Phi_z(0)} + \frac {\sum_{x} |x|^2\Phi_z(x)} {\hat \Phi_z(0)}.\nn
	\end{align}
Our numerical estimates imply numerical lower and upper bounds of these factors in dimensions above $\dmintree$ and $\dminanimal$ for lattice trees and lattice animals, respectively, uniformly in $z<z_c$. In \cite[Appendix D]{FitHof13g-ext} we give a detailed description on how to bound such factors using the bounds on the coefficients summarized in
Section \ref{secBoundsSummary}. In particular, we obtain that, since $\sum_{x} |x|^2F_z(x)$ and $\sum_{x} |x|^2\Phi_z(x)$ are power-series with radius of convergence $z_c$, by Abel's theorem,
	\eqn{
	\lim_{z\nearrow z_c}\sum_{x} |x|^2F_z(x)=\sum_{x} |x|^2F_{z_c}(x),
	\qquad
	\text{and}
	\qquad
	\lim_{z\nearrow z_c}\sum_{x} |x|^2\Phi_z(x)=\sum_{x} |x|^2\Phi_{z_c}(x),
	}
and these sums are absolutely convergent. We thus conclude that
	\begin{align}
	\lim_{z\nearrow z_c}\frac{\xi_2(z)^2}{\chi(z)}  &=\frac {\sum_{x} |x|^2F_{z_c}(x)} {\hat \Phi_{z_c}(0)}.
	\end{align}
The fact that $\nu=1/4$ now follows from the fact that $\gamma=1/2$.
\qed

\subsection{Proof of $\eta=0$ in $k$-space in Theorem \ref{thm-k-space}}
\label{sec-k-space}
By Abel's theorem, and the fact that $1-\hat F_{z_c}(k)>0$ for $k\neq 0$ (which follows from Theorem \ref{thm-IRB}),
	\begin{align}
	\hat G_{z_c}(k)=\frac {\hat \Phi_{z_c}(k)}{1-\hat F_{z_c}(k)}.
	\end{align}
Note that $1-\hat F_{z_c}(0)=0$ (which follows from the finiteness of $\hat \Phi_{z_c}(0)$ and the fact that $\chi(z)$ blows up as $z\nearrow z_c$). Since $\bigtriangleup \hat F_z(k)|_{k=0}$ is well-defined, we thus immediately conclude that
	\begin{align}
	\hat G_{z_c}(k)=-\frac {\hat \Phi_{z_c}(0)}{|k|^2\bigtriangleup \hat F_{z_c}(k)|_{k=0}}(1+o(1)),
	\end{align}
so that Theorem \ref{thm-k-space} follows with
	\eqn{
	\bar{A}(d)=-\frac {\hat \Phi_{z_c}(0)}{\bigtriangleup \hat F_{z_c}(k)|_{k=0}}.
	}
\qed

\subsection{Proof of $\eta=0$ in $x$-space in Theorem \ref{thm-x-space}}
\label{sec-x-space}
The proof of Theorem \ref{thm-x-space} follows by using the $x$-space asymptotics proved by Takashi Hara in \cite{Hara08}. See in particular \cite[Proposition 1.3]{Hara08}.
%
For ease of reference, we also mention some useful bounds that follow from our analysis of LAs in dimension $d=17$. We only give the bounds for LAs, as these bounds also bound the corresponding lace-expansion coefficients for LTs. While we only prove Theorem \ref{thm-x-space} for $d\geq 27$, these bounds suggest that the classical lace expansion converges for $d\geq 17$ by our numerical estimates.

Using the bound on $f_2(z_c)$ directly, with $\bar{A}_2(17)=3.78,$ we obtain the rough bounds
	\begin{align}
	\bar{G}_{z_c}^{\star 4}(0)-1&\leq 292,
	\qquad \sup_{x\neq 0} \bar{G}_{z_c}^{\star 4}(x)\leq 45.
	\end{align}	
In the following, we improve these rough bounds using the ideas explained in detail in Section \ref{sec-BoundIdea-Heuristic}.

For this, we explain how to modify the bounds on the lace-expansion coefficients for the classical expansion, by allocating one-point functions efficiently, and repulsiveness of the corresponding diagrams. We emphasize that we apply all these bounds to the coefficients of the classical lace expansion, not to the NoBLE, to make them compatible with the classical lace expansion that \cite{Hara08} is based upon. Allocating the one-point functions properly, and avoiding overcounting them as described in Section \ref{secBoundsOnePointF}, we obtain
	\begin{align}
	\tilde{G}_{z_c}^{\star 4}(0)-1&\leq 2.777,
	\qquad
  	\sup_{x\neq 0} \tilde{G}_{z_c}^{\star 4}(x)\leq 0.5924.
	\end{align}	
Next, we use repulsiveness as defined in Section \ref{secBoundsRepdia}. Recall that repulsiveness implies that all paths of the square that actually connect the five (for open squares) or four (for closed squares) boundary points are mutually bond-disjoint. We can bound the squares arising as bounds on intermediate diagrams in the coefficients of the classical lace expansion by
	\begin{align}
	\label{un-weighted-LAs}
	\sum_{w,u,y}\diagRepulsiveLetter{S}_{1,0,0,0}(w,u,y,0)&\leq 0.2247,
	\qquad \sup_{x\neq 0}\sum_{w,u,y}\diagRepulsiveLetter{S}_{1,0,0,0}(w,u,y,x)\leq 0.0798561.
	\end{align}
We further obtain bounds on the (open and closed) weighted triangles $\bar{T}^{\sss(2,0)}=\sup_{y}T^{\sss(2,0)}(y)$, as required for \cite{Hara08}, given by
	\begin{align}
	\label{weighted-LAs}
	T^{\sss(2,0)}(0)&=\sum_{x}|x|^2G_z(x)(\tilde{G}_z\star \tilde{G}_z)(x)\leq 2.09182,\\
	\sup_{y\neq 0}T^{\sss(2,0)}(y)&=\sup_{y\neq 0}\sum_{x}|x|^2G_z(x)(\tilde{G}_z\star \tilde{G}_z)(x-y)\leq 0.456256.\nn
	\end{align}
In communication with Takashi Hara early 2019, we learned that finiteness of the diagrams in \eqref{weighted-LAs}, together with the maximum of the diagrams in \eqref{un-weighted-LAs} being strictly smaller than $\tfrac{1}{2}$, suffices for the $x$-space bounds in Theorem \ref{thm-x-space} to hold. The reason for the latter is that the base of the geometric convergence of the classical lace-expansion coefficients is at most twice the maximum of the squares in \eqref{un-weighted-LAs}. When this number if strictly smaller than one, the classical lace expansion converges, which is needed to apply \cite{Hara08}. These bounds also hold in any $d\geq 17$, using a numerical check in dimensions $17-29$, and the monotonicity arguments in Appendix \ref{sec-monoton-inD} for $d\geq 30$. The fact that we only obtain Theorem \ref{thm-x-space} for $d>27$ is due to a technical restriction in the proof in \cite{Hara08}.

The fact that the constant $A(d)$ in Theorem \ref{thm-x-space} equals the one in Theorem \ref{thm-k-space} follows directly from the proof in \cite{Hara08}.
\qed


\appendix
\section{Lower bound on the initial point}
\label{sec-lemmaAnalysisLABound}
In this section we prove that $z_I\geq (2d-1)^{-1}\e^{-1}$ by proving upper bounds on $g_{z_I}$ and $g^\iota_{z_I}$:
\begin{lemma}[Lower bound on $z_I$]
\label{lem-LBzI}
Let $z_I$ satisfy
	\eqn{
	z_I g_{z_I}^\iota =\frac{1}{2d-1}.
	}
Then, $z_I\geq (2d-1)^{-1}\e^{-1}$.
\end{lemma}

To prove Lemma \ref{lem-LBzI}, we use a similar argument as in \cite[Proof of Lemma 3.1]{HarSla90b}. Unfortunately, one step in \cite[Proof of Lemma 3.1]{HarSla90b} is not correct as claimed. We correct the argument by using Lemma \ref{lemmaAnalysisLABound} below. We explain the mistake of \cite[Proof of Lemma 3.1]{HarSla90b} in the proof of Lemma \ref{lemmaAnalysisLABound}.

The proof of Lemma \ref{lem-LBzI} relies on the following lemma, that we prove thereafter, and which is of independent interest:
\begin{lemma}[Upper bound on the number of $n$-bond LAs]
\label{lemmaAnalysisLABound}
The number of $n$-bond LAs that contain the origin is bounded above by
\begin{align}
\lbeq{lemmaAnalysisLABound-statement}
2d (2d-1)^{n-1}\frac {(n+1)^{n}}{(n+1)!}.
\end{align}
\end{lemma}

Let us first use Lemma \ref{lemmaAnalysisLABound} to prove Lemma \ref{lem-LBzI}:

\begin{proof}[Proof of Lemma \ref{lem-LBzI}]
We use Lemma \ref{lemmaAnalysisLABound} to show that $z_I\geq (2d-1)^{-1}\e^{-1}$.
For this, we bound $g_z$ using \refeq{lemmaAnalysisLABound-statement}, as
	\begin{align}
	\lbeq{LTAbstractBound}
	g_z&\leq 1 + \sum_{n=1}^{\infty}z^n (2d)(2d-1)^{n-1} \frac {(n+1)^{n}}{(n+1)!}.
	\end{align}
We note that
	\begin{eqnarray}
	\lbeq{LTAbstractBound-tool}
	\sum_{n=1}^\infty\frac {n^{n-1}}{n!}\e^{-n}=1,
	\end{eqnarray}
since the summand in \refeq{LTAbstractBound-tool} equals the probability that the total progeny of a branching process with Poisson offspring distribution with parameter one equals $n$, and the sum of these probabilities equals 1, since the tree is a.s.\ finite.
We obtain for $z \leq((2d-1)\e)^{-1}$ that
	\begin{align}
	g_z\leq&  1 + \sum_{n=2}^{\infty}z^{n-1} (2d)(2d-1)^{n-2} \frac {n^{n-1}}{n!}\leq 1 + \sum_{n=2}^{\infty}\frac{(2d)(2d-1)^{n-2}}{(2d-1)^{n-1}\e^{n-1}} \frac {n^{n-1}}{n!}\nnb
	=& 1 + \frac {2d \e}{(2d-1)}\sum_{n=2}^{\infty} \frac {n^{n-1}}{n!} \e^{-n}=1 + \frac {2d \e}{(2d-1)} \left( 1 - \frac {1}{\e}\right)= \e+ \frac {\e-1}{2d-1}.
	\end{align}
Since $\gj\leq g_z$ this is also a bound for $\gj$. We can improve this bound on $\gj$ by using the condition that $\ve[\iota]$ is not part of the tree/animal. A non-trivial LT/LA contains at least one edge connecting the origin to a neighboring point.
Therefore, at least at $1$ out of $2d$ LTs/LAs counted in $g_z$ does not contribute to $\gj$. We conclude that
	\begin{eqnarray}
	\lbeq{gjgzrelation-Initial}
	\gj\leq 1+\frac {2d-1}{2d}(g_z-1)=1+\frac {(2d-1)(\e-1)}{2d} \left(1+\frac 1 {2d-1}\right)=\e,
	\end{eqnarray}
and thus, for all $z\leq z_I$,
	\begin{align}
    \lbeq{Bound-Initial-f1}
	(2d-1)\gj z&\leq 1,\ &(2d-1)\bar \alpha_z =(2d-1)g_z z\leq 1+ \frac {1-\e^{-1}}{2d-1},
	\end{align}
as well as, by the definition of $z_I$ in \refeq{DefinitionOfzI},
	\begin{align}
	1=(2d-1)g^\iota_{z_I} {z_I}\leq(2d-1) \e z_I.
	\end{align}
This implies that $z_I\geq (2d-1)^{-1}\e^{-1}$.
\end{proof}

We complete this section by proving Lemma \ref{lemmaAnalysisLABound}:

\begin{proof}[Proof of Lemma \ref{lemmaAnalysisLABound}]
The proof of \cite[Lemma 3.1]{HarSla90b} uses a similar bound as stated in \refeq{lemmaAnalysisLABound-statement} in Lemma \ref{lemmaAnalysisLABound}, but contains an error.
Namely it is used that there are $(n+1)^{n}/(n+1)!$ abstract unlabeled rooted trees with $n$ edges. However, while Cayley's theorem states that there are $(n+1)^{n}$ labeled rooted trees, the removal of the labels does not create a factor $1/(n+1)!$ as stated. Indeed, there are less than $(n+1)!$ ways to label an abstract tree, e.g., for a tree with two edges there are only $3$ different ways to label the tree as the only difference is the label of the vertex that is part of two edges.

To prove the bound stated in \refeq{lemmaAnalysisLABound-statement} we first use the techniques of \cite[Sections 2 and 5]{BorChaHofSla99} to show the lemma only for LTs. Then, we adapt the arguments to LAs.

We begin by defining a non-backtracking branching random walk with Poisson offspring distribution. Abstract trees are the family trees of a critical branching process with Poisson offspring distribution. In more detail, we begin with a single individual having $\xi$ offspring, where $\xi$ is a Poisson random variable of mean $1$, i.e., $\prob(\xi=m)=(\e m!)^{-1}$. Each of the offspring then independently has offspring of its own with the same critical Poisson distribution.
We denote by $|\absT|$ the number of bonds of the tree $\absT$. For an abstract rooted tree $\absT$,
with the $i$th individual having $\xi_i$ offspring, this associates to $\absT$ the weight
	\begin{align}
	\prob(\absT)=\e^{-(|\absT|+1)}\prod_{i\in \absT}\frac 1 {\xi_i!}.
	\end{align}
We define an embedding of $\varphi$ of $\absT$ into $\Zd$ to be a mapping of the vertices of $\absT$ into $\Zd$, such that the root is mapped to the origin and adjacent vertices in the tree are mapped to nearest-neighbors in $\Zd$. Further, we restrict to embeddings in which the children are not mapped to the location of the grandparents. We define the pair $(\absT,\varphi)$ to be a non-backtracking branching random walk with Poisson offspring distribution (Poisson-NBBRW). Given an abstract tree $\absT$ with $n$ bonds, there are $2d (2d-1)^{n-1}$ possible embeddings $\varphi$ of $\absT$. The Poisson-NBBRW measure is given by
	\begin{align}
	\lbeq{Analysis-P-measure}
		\prob(\absT,\varphi)=\frac {1}{2d(2d-1)^{|\absT|-1}} \prob(\absT) =\frac {1}{2d(2d-1)^{|\absT|-1}}\e^{-(|\absT|+1)}\prod_{i\in \absT}\frac 1 {\xi_i!}.
	\end{align}
for all $\absT$ and $\varphi$.	We define $\Ibold^{\sss (t)}(\absT,\varphi)$ to be the indicator that the embedding $\varphi$ of $\absT$ is a LT, i.e., $\varphi\colon \absT\mapsto \Zd$ is injective. For a LT $T$, we write $T=\varphi(\absT)$ if the embedding $\varphi$ of $\absT$ equals $T$.
For abstract trees with $n$ bonds, we define the measure
	\begin{align}
	\Qbold^{\sss (t)}_n(\absT,\varphi)=\frac  1 {Z^{\sss (t)}_n} \prob(\absT,\varphi)\Ibold^{\sss (t)}(\absT,\varphi),
	\end{align}
with normalization
	\begin{align}
	Z^{\sss (t)}_n=\sum_{(\absT,\varphi)\colon |\absT|=n}\prob(\absT,\varphi)\Ibold^{\sss (t)}(\absT,\varphi).
	\end{align}
We prove that the number of $n$-bond LTs that contain the origin is given by
	\begin{align}
	\lbeq{Analysis-LT-Normalizer-Solved}
	t^{\sss (t)}_n(0) = Z^{\sss (t)}_n \e^{n+1}2d(2d-1)^{n-1},
	\end{align}
which is equivalent to
	\eqn{
	\lbeq{Analysis-LT-Normalizer-Solved-2}
	Z^{\sss (t)}_n=t^{\sss (t)}_n(0) \e^{-(n+1)}(2d)^{-1}(2d-1)^{-(n-1)}.
	}
To do that, we show that, for each $n$-bond LT $T$,
	\begin{align}
	\lbeq{Analysis-Q-is-Uniform-LT-measure}
	\sum_{(\absT,\varphi)\colon \varphi(\absT)=T} \Qbold^{\sss (t)}_n(\absT,\varphi)=\frac {\e^{-(n+1)}}{2d(2d-1)^{n-1}},
	\end{align}
which implies that the measure $\Qbold^{\sss (t)}_n$ corresponds to the uniform measure on all $n$-bond LTs and thus implies \refeq{Analysis-LT-Normalizer-Solved-2}, and thus \refeq{Analysis-LT-Normalizer-Solved}.

From \refeq{Analysis-P-measure} we conclude that, for all $T$ with $|T|=n$,
	\begin{align}
	\sum_{(\absT,\varphi)\colon \varphi(\absT)=T} \prob(\absT,\varphi) = \frac {\e^{-(n+1)}}{2d(2d-1)^{n-1}}
	\sum_{(\absT,\varphi)\colon \varphi(\absT)=T} \prod_{i\in \absT}\frac 1 {\xi_i!}.
	\end{align}
	
We continue by proving that, for any LT $T$,
	\begin{align}
	\lbeq{Analysis-LT-Poisson-isOne}
	\sum_{(\absT,\varphi)\colon \varphi(\absT)=T} \prod_{i\in \absT}\frac 1 {\xi_i!}=1.
	\end{align}
Let $d_0$ be the degree of $0$ in $T$, and given non-zero $x\in T$, let $d_x$ be the degree of $x$ in $T$
minus $1$. The set $\{d_x\colon x\in T\}$ must be equal to the set of $\xi_i$ for any $\absT$ that can be mapped to $T$.
Defining $\nu(T)=\#\{(\absT,\varphi)\colon \varphi(\absT)=T\}$, \refeq{Analysis-LT-Poisson-isOne} is therefore equivalent to
	\begin{align}
	\lbeq{Analysis-LT-Poisson-isOne-Tool}
	\nu(T)=\prod_{x\in T}d_x!.
	\end{align}
We prove \refeq{Analysis-LT-Poisson-isOne-Tool} by induction on the number $N$ of generations of $T$. By this, we mean the length of the longest self-avoiding path in $T$, starting from the origin. The identity \refeq{Analysis-LT-Poisson-isOne-Tool} clearly holds for $N=0$. Our induction hypothesis is that \refeq{Analysis-LT-Poisson-isOne-Tool} holds if there are $N-1$ or fewer generations. Suppose $T$ has $N$ generations, let $T_1,\dots,T_{d_0}$ denote the LTs resulting from deleting from $T$ all bonds incident on the origin. We regard each $T_a$ as rooted at the neighbor of the origin in the corresponding deleted bond. It suffices to show that $\nu(T)=b_0!\prod_{a=1}^{d_0}\nu(T_a)$, since each $T_a$ has fewer than $N$ generations. To prove this, we note that each $(\absT,\varphi)$ with $\varphi(\absT)=T$ induces a Poisson-NBBRW $(\absT_a,\varphi_a)$ such that $\varphi_a(\absT_a)=T_a$. This correspondence is $d_0!$ to $1$, since $(\absT,\varphi)$ is determined by the set of $(\absT_a,\varphi_a)$, up to permutations of the branches of $\absT$ at its root. This proves $\nu(T)=d_0!\prod_{a=1}^{d_0}\nu(T_a)$.

We finally complete the proof of the lemma for LTs. We rearrange \refeq{Analysis-LT-Normalizer-Solved} and use \refeq{Analysis-P-measure} to obtain
	\begin{align}
	t^{\sss (t)}_n(0) =& 2d(2d-1)^{n-1}\e^{n+1} \sum_{T\colon |T|=n}\sum_{(\absT,\varphi)\colon \varphi(\absT)=T} \prob(\absT,\varphi)\nnb
	=& \e^{n+1} \sum_{T\colon |T|=n} \sum_{(\absT,\varphi)\colon \varphi(\absT)=T} \prob(\absT)
	\leq  \e^{n+1}   \sum_{(\absT,\varphi)\colon |\absT|=n} \prob(\absT).
	\lbeq{Analysis-LT-Lemma-nextto-laststep}
	\end{align}
Since there are $2d (2d-1)^{n-1}$ embeddings $\varphi$, we then know that
	\begin{align}
	t^{\sss (t)}_n(0) \leq & 2d (2d-1)^{n-1} \e^{n+1} \sum_{\absT\colon |\absT|=n} \prob(\absT)\nnb
	= &2d(2d-1)^{n-1}\e^{n+1}       \prob(|\absT|=n).
	\lbeq{Analysis-LT-Lemma-laststep}
	\end{align}
The probability distribution of the total number of bonds of a Poisson branching process is given by $\prob(|\absT|=n)=\e^{-(n+1)} (n+1)^{n}/(n+1)!$
(as this is the number of vertices, also called the total progeny, minus 1). We insert this into \refeq{Analysis-LT-Lemma-laststep} and obtain the claimed bound for LTs.
\medskip

The claim for LAs is obtained using similar ideas, but requires some adaptations as LAs can contain loops. The main difficulty is to obtain a relation similar to \refeq{Analysis-LT-Poisson-isOne-Tool}.

We define $\Ibold^{\sss (a)}(\absT,\varphi)$ to be the indicator for the event that
\begin{enumerate}[(1)]
\item no two bonds of the abstract tree $\absT$ are mapped to the same bond in $\Zd$ by $\varphi$,
\item for all $i\in\absT$ either $\xi_i=0$ or the following two conditions hold:
  \begin{enumerate}
\item  there exist no $j\in\absT\setminus\{i\}$, such that $\varphi(i)=\varphi(j), \xi_j>0$ and ${\rm height}(j)={\rm height}(i)$,
\item  there exist no $j\in\absT\setminus\{i\}$, such that $\varphi(i)=\varphi(j)$ and ${\rm height}(j)<{\rm height}(i)$,
\end{enumerate}
where the height of a individual point $a\in\absT$ is intrinsic distance of $a$ to the root in $\absT$.
\end{enumerate}
For a Poisson-NBBRW $(\absT_a,\varphi_a)$ (1) guarantees that each bond is only used once by the process and (2) is the condition that whenever a point is visited by multiple individuals then only the first of them can have offspring. Thus, (1) and (2) together imply that the image $\varphi(\absT)$ is a LA. Here, for a LA $A$, we write $A=\varphi(\absT)$ if the embedding $\varphi$ of $\absT$ equals $A$.

Now we proceed as for the LT by defining
	\begin{align}
	\Qbold^{\sss (a)}_n(\absT,\varphi)=&\frac  1 {Z^{\sss (a)}_n} \prob(\absT,\varphi)\Ibold^{\sss (a)}(\absT,\varphi),
	\end{align}
where
	\eqn{
	Z^{\sss (a)}_n=\sum_{(\absT,\varphi)\colon |\absT|=n}\prob(\absT,\varphi)\Ibold^{\sss (a)}(\absT,\varphi).
	}
We next prove that the number of $n$-bond LAs containing the origin is given by
	\begin{align}
	\lbeq{Analysis-LA-Normalizer-Solved}
	t^{\sss (a)}_n(0) = Z^{\sss (a)}_n \e^{n+1}2d(2d-1)^{n-1}.
	\end{align}
This is analogous to the proof for LT in \refeq{Analysis-LT-Normalizer-Solved}, which the exception of the proof that
	\begin{align}
	\lbeq{Analysis-LA-Poisson-isOne}
	\sum_{(\absT,\varphi)\colon \varphi(\absT)=A} \Ibold^{\sss (a)}(\absT,\varphi) \prod_{i\in \absT}\frac 1 {\xi_i!}=1.
	\end{align}
We now explain how to adapt the argument so as to obtain \refeq{Analysis-LA-Poisson-isOne}. To prove this we define $d_x$ such that it corresponds to the number of offspring of the Poisson branching process. In contrast to the LT case, for LAs it can happen that $\varphi$ maps multiple individuals of $\absT$ to $x$. Therefore, it is not obvious how to choose $d_x$ such that $(\xi_i)_i$ and $(d_x)_x$ are equivalent.

For a LA $A$, and two vertices $x,y\in A$, we define $d_{A}(x,y)$ to be the intrinsic distance between $x$ and $y$, and call the distance $d_A(0,x)$ the {\em age} of $x$.  
For $N\in\Nbold$, we call the set of all vertices with age $N$ the {\em generation} $N$. We define an exploration process, that we use to define $(d_x)_x$, as follows:

\begin{enumerate}[1)]
\item We define $d_0$ to be the degree of $0$ in $A$;
\item We define $S$ to be the set of all vertices directly connected to the origin in $A$;
\item We define $B\subset A$ to be the subset of all bonds in $A$ that do not contain the origin;
\item Let $N$ be the age of the youngest vertex in $S$;
\item For the unique $x$ in the $N$th generation with the lowest lexicographic order:
\begin{enumerate}[i)]
\item we define $d_x$ to be the degree of $x$ in $B$;
\item we update $S$ by adding the endpoints of the bonds in $B$ adjacent to $x$, and give these endpoints an age that is equal to the age of $x$ plus one;
\item we remove $x$ from $S$ and the bonds that contain $x$ from $B$.
\end{enumerate}
\item If $S$ is non-empty, then we repeat the procedure starting in 4) with the updated sets $S$ and $B$.
\end{enumerate}
This procedure terminates when $d_x$ has been defined for all vertices of $A$.

Similarly to the LT case, we define $\nu(A)=\#\{(\absT,\varphi):\varphi(\absT)=A, \Ibold^{\sss (a)}(\absT,\varphi)=1\}$ and now show that
	\begin{align}
	\lbeq{Analysis-LA-Poisson-isOne-Tool}
	\nu(A)=\prod_{x\in A}d_x!,
	\end{align}
which implies \refeq{Analysis-LA-Poisson-isOne}.

We again prove \refeq{Analysis-LA-Poisson-isOne-Tool} by induction on the number $N$ of generations of $A$.
The number of generations corresponds to the age of the oldest particle in $A$. The identity \refeq{Analysis-LT-Poisson-isOne-Tool}
clearly holds for $N=0$. Our induction hypothesis is that \refeq{Analysis-LT-Poisson-isOne-Tool} holds if there are $N-1$ or fewer generations.

For a LT, the removal of the $d_0$ bonds adjacent to the origin splits the tree into $d_0$ unique, non-intersecting subtrees $T_1,\dots,T_{d_0}$.
Due to possible double connections in a LA, the removal of the bonds adjacent to the origin does not automatically split the animal into $d_0$ uniquely defined sub-animals $A_a$. We will now create a unique split of the animal $A$, that is consistent with the exploration process that we have used to define $(d_x)_x$.

For this, we first remove a certain set of bonds to create a LT $T\subseteq A$. Then we split the $T$ into $d_0$ subtrees as done above. In the last step we add each of the removed bonds to the subtrees to create a unique decomposition $A_1,\dots,A_{d_0}$. This procedure is visualized in Figure \ref{BoundLA-Figure-deconstructing-Animal}. In what follows, we explain how the set of bonds is being chosen. We fix a LA $A$.

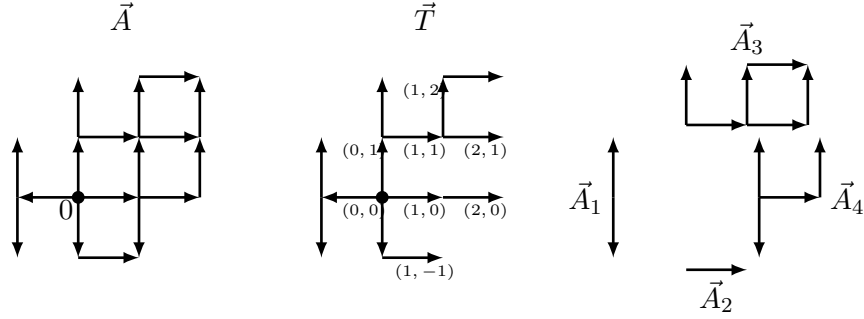
\begin{figure}[ht!]
\begin{center}
\begin{tikzpicture}[auto,scale=0.8]

\begin{scope}[shift={(0,0)},rotate=0]
\draw[->,line width=1 pt] (0,0) to (-1,0);
\draw[<-,line width=1 pt] (-1,1) to (-1,0);
\draw[<-,line width=1 pt] (-1,-1) to (-1,0);
\draw[->,line width=1 pt] (0,0) to (0,1);
\draw[->,line width=1 pt] (0,1) to (0,2);
\draw[->,line width=1 pt] (0,1) to (1,1);
\draw[->,line width=1 pt] (0,0) to (1,0);
\draw[->,line width=1 pt] (0,0) to (0,-1);
\draw[->,line width=1 pt] (0,-1) to (1,-1);
\draw[<-,line width=1 pt] (1,-1) to (1,0);
\draw[->,line width=1 pt] (1,0) to (2,0);
\draw[->,line width=1 pt] (1,1) to (2,1);
\draw[->,line width=1 pt] (1,2) to (2,2);
\draw[->,line width=1 pt] (1,0) to (1,1);
\draw[->,line width=1 pt] (2,0) to (2,1);
\draw[->,line width=1 pt] (1,1) to (1,2);
\draw[->,line width=1 pt] (2,1) to (2,2);

\fill (0,0) circle (3pt);

\node at(-0.2,-0.2)   {$0$};
\node at(0.7,3)   {$\vec A$};
\end{scope}

\begin{scope}[shift={(5,0)},rotate=0]
\draw[->,line width=1 pt] (0,0) to (-1,0);
\draw[<-,line width=1 pt] (-1,1) to (-1,0);
\draw[<-,line width=1 pt] (-1,-1) to (-1,0);
\draw[->,line width=1 pt] (0,0) to (0,1);
\draw[->,line width=1 pt] (0,1) to (0,2);
\draw[->,line width=1 pt] (0,1) to (1,1);
\draw[->,line width=1 pt] (0,0) to (1,0);
\draw[->,line width=1 pt] (0,0) to (0,-1);
\draw[->,line width=1 pt] (0,-1) to (1,-1);
\draw[->,line width=1 pt] (1,0) to (2,0);
\draw[->,line width=1 pt] (1,1) to (2,1);
\draw[->,line width=1 pt] (1,2) to (2,2);
\draw[->,line width=1 pt] (1,1) to (1,2);

\begin{scope}[shift={(-0.3,-0.23)},rotate=0]
{\tiny
\node at(0,1)   {$(0,1)$};
\node at(0,0)   {$(0,0)$};
\node at(1,-1)   {$(1,-1)$};
\node at(1,0)   {$(1,0)$};
\node at(1,1)   {$(1,1)$};
\node at(2,0)   {$(2,0)$};
\node at(1,2)   {$(1,2)$};
\node at(2,1)   {$(2,1)$};
}
\end{scope}

\fill (0,0) circle (3pt);
\node at(0.7,3)   {$\vec T$};
\end{scope}

\begin{scope}[shift={(10,0)},rotate=0]
\begin{scope}[shift={(-0.2,0)},rotate=0]
\draw[<-,line width=1 pt] (-1,1) to (-1,0);
\draw[<-,line width=1 pt] (-1,-1) to (-1,0);
\node[left] at(-1,0)   {$\vec A_1$};
\end{scope}

\begin{scope}[shift={(0,0.2)},rotate=0]
\draw[->,line width=1 pt] (0,1) to (0,2);
\draw[->,line width=1 pt] (0,1) to (1,1);
\draw[->,line width=1 pt] (1,1) to (2,1);
\draw[->,line width=1 pt] (1,1) to (1,2);
\draw[->,line width=1 pt] (2,1) to (2,2);
\draw[->,line width=1 pt] (1,2) to (2,2);
\node[above ] at(1,2)   {$\vec A_3$};
\end{scope}
\begin{scope}[shift={(0.2,0)},rotate=0]
\draw[->,line width=1 pt] (1,0) to (2,0);
\draw[<-,line width=1 pt] (1,-1) to (1,0);
\draw[->,line width=1 pt] (1,0) to (1,1);
\draw[->,line width=1 pt] (2,0) to (2,1);
\node[right] at(2,0)   {$\vec A_4$};
\end{scope}

\begin{scope}[shift={(0,-0.2)},rotate=0]
\draw[->,line width=1 pt] (0,-1) to (1,-1);
\node[below] at(0.5,-1)   {$\vec A_2$};
\end{scope}
\end{scope}


\end{tikzpicture}
\caption{The first picture shows how we orient the bonds of $A$ to create $\vec A$.
In the second picture we see the LT that is created by the removal of the bonds. The coordinates of the vertices that we compare in step (2) of the algorithm below are indicated. The last picture shows the result of the partition of $\vec A\setminus \{0\}$, which is the oriented graph obtained by removing the vertex 0 and all edges incident to it from $\vec A$.}
\label{BoundLA-Figure-deconstructing-Animal}
\end{center}
\end{figure}

\begin{enumerate}[(1)]
\item We orient all bonds of $A$ away from the origin in the intrinsic distance, see Figure \ref{BoundLA-Figure-deconstructing-Animal}, and denote by $\vec A$ the created LA of oriented bonds. Thus, the oriented bonds $b=(x,y)$ are such that the age of $y$ equals that of $x$ plus one (equal ages cannot happen due to the bipartite structure of the hypercubic lattice $\Z^d$);

\item For all $x\in \Zd$ which have multiple ingoing bonds we remove all ingoing bonds from $\vec A$ with the exception of the ingoing bond $(b_i=(\bb_i,\tb_i=x))_i$ that contains the starting point $\bb_i$ with the lowest lexicographic order;

\item We denote the resulting LA by $\vec T$ and the set of all removed oriented bonds by $\vec S$;

\item We know that $\vec A=\vec T\cup \vec S$ and that in fact $\vec T$ is a LT since step (3) leaves at most one ingoing bond for every $x$ (and precisely one for all vertices except for 0). Further, we see that $\vec T$ contains the same set of vertices as $A$, so that $\vec{T}$ can also be considered a spanning tree of the graph $A$;

\item We delete from $\vec T$ all bonds incident to the origin and denote the created LTs by $\vec T_1,\dots,\vec T_{d_0}$;
\item For $i=1,\dots, d_0$, we define
	\begin{align}
	\vec A_i=\vec T_i\cup\{ b\in \vec S \mid \bb\in \vec T_i\},
	\end{align}
and $A_i$ to be the LA obtained by removing the orientation of the bonds of $\vec A_i$.
\end{enumerate}
We regard $A_{i}$ as rooted at the unique vertex in it that was directly connected to the origin in $A$.
We see that for all $i=1,\dots,d_0$ and $x\in A_i$, the number $d_x$ corresponds to the number of outgoing edges of $x$ in $A_i$.
We note that for all bonds $b\in \vec S\cup \vec A_i$, the end-vertex $\tb$ has no outgoing bonds. This is related to the condition $(2)$ of the definition of $\Ibold^{\sss (a)}(\absT,\varphi)$.

It suffices to show that $\nu(A)=d_0!\prod_{a=1}^{d_0}\nu(A_{a})$, since each $A_{i}$ has at most $N-1$ generations. We note that each $(\absT,\varphi)$ with $\varphi(\absT)=A$ induces a Poisson-NBBRW $(\absT_{a},\varphi_{a})$ such that $\varphi_a(\absT_{a})=A_{a}$. This correspondence is $d_0!$ to $1$, since $(\absT,\varphi)$ is determined by the set of $(\absT_{a},\varphi_{a})$, up to permutations of the branches of $\absT$ at its root. This proves $\nu(A)=d_0!\prod_{a=1}^{d_0}\nu(A_{a})$. This completes the proof of \refeq{Analysis-LA-Normalizer-Solved}.

Knowing that \refeq{Analysis-LA-Normalizer-Solved} holds, we repeat the steps between \refeq{Analysis-LT-Lemma-laststep} and \refeq{Analysis-LT-Lemma-nextto-laststep} for LAs, and obtain the claimed bound on $t^{\sss (a)}_n(0)$, which completes the proof of Lemma \ref{lemmaAnalysisLABound}.
\end{proof}


\section{Extension to all dimensions $d\geq 30$}
\label{sec-monoton-inD}

In this section we prove that when a modified version of our analysis is applied successfully in some dimension $d'$, which we choose to be $d'=30$, then this analysis will also succeed in {\em any} dimension $d\geq d'$. For this, we need to modify the analysis somewhat, so that the bounds that apply in dimension $d'$ will be guaranteed to apply to any dimension $d\geq d'$. The basic ideas, that also guide the structure of this appendix, are the following:
\begin{itemize}
\item[(1)] In Appendix \ref{sec-monoton-inD-1} we prove that all SRW integrals that are used in our bounds for $d\geq d'$ are monotone decreasing in $d$.
\item[(2)] In Appendix \ref{sec-monoton-inD-2} we prove that the contributions of short explicit paths within our bounds are all monotone decreasing in $d$.
\item[(3)] In Appendix \ref{sec-monoton-inD-3} we argue that the first two observations imply that the bounds on all simple diagrams are monotone decreasing in $d$.
This will imply that all bounds on the coefficients $\beta_{\bullet}$, as summarized in Section \ref{secBoundsSummary}, are monotone decreasing in $d$.
\item[(4)] In Appendix \ref{sec-monoton-inD-4} we show that, since our bounds on the coefficients are all monotone in the dimension $d$, all bounds on the coefficients arising in the rewrite used in the general analysis \cite{FitHof13b}, are also monotone.
\item[(5)] In Appendix \ref{sec-monoton-inD-5} we describe how the above steps imply that whenever the bootstrap succeeds with a given set of parameters $(\gamma_i,\Gamma_i)_{i=1}^3$
in $d=d'$, then it also succeeds in {\em all} dimension $d\geq d'$. This in turn implies our results in all $d\geq d'$.
\end{itemize}
We emphasize that the bounds derived here are slightly worse than the ones used for a specific dimension. This is due to the fact that we do not know that {\em all} SRW integrals that we rely upon in the bounds for explicit dimensions are monotone decreasing. Thus, we need to rely on the subset of SRW integrals that we can prove to be monotone decreasing.

We implement these bounds, which are provably monotone in $d$, using the notebook \verb|LAmonotone|.
Before starting this, let us clarify the notion of monotonicity in $d$:

\begin{definition}[Monotonicity in dimension]
\label{def-mon-dim}
Let $g^{\sss(d)}\colon \Z^d\to \R$ be a family of functions. Recall that we say that $g^{\sss(d)}(x^{\sss(d)})$ is monotone in $d$ when $g^{\sss(d+1)}((x^{\sss(d)},0))\leq g^{\sss(d)}(x^{\sss(d)})$, where $(x^{\sss(d)},0)=(x_1, \ldots, x_d,0)$ is $x^{\sss(d)}$ appended with an extra coordinate that takes the value zero.
\end{definition}
Prominent examples of such families are the critical SRW Green's function in $d$ dimensions $C_{1/(2d)}^{\sss(d)}$, for which we investigate its monotonicity for $x^{\sss(d)}$ equal to the origin, as well as some points close to the origin.

\subsection{SRW Integrals}
\label{sec-monoton-inD-1}
We start by showing that all the bounds on the SRW integrals $(I_{n,m},K_{n,m},T_{n,m},U_{n})$ and $\Isupx_{n,l}(x)$ that we use in our analysis can be uniformly bounded in $d$, using bounds that are monotone decreasing in $d$. We emphasize that we do not claim that all integrals, such as $U_n$, are monotone decreasing in $d$, only that the bounds that we rely upon are.
The reason for this distinction is that the only monotonicity that we can actually prove is given by the two lemmas below:

\begin{lemma}[Monotonicity of SRW Green's function in $d$]
\label{Lemma-Mono-SRW}
Let $|x|_{\sss \infty}$ be the supremum norm of $x\in \Z^d$. Then, for every $n\geq 1$ and $x^{\sss(d)}\in \Z^d$ with $|x^{\sss(d)}|_{\sss \infty}\leq 2$,
	\eqn{
	I_{n,0}^{\sss(d)}(x^{\sss(d)})\geq I_{n,0}^{\sss(d+1)}(x^{\sss(d+1)}),
	}
where $x^{\sss(d+1)}_i=x^{\sss(d)}_i$ for $i\in\{1,\dots,d\}$, while $x^{\sss(d+1)}_{d+1}=0$.
\end{lemma}

\begin{lemma}[Monotonicity of SRW Green's function in $x$]
\label{Lemma-Mono-SRW-x}
For any positive integer $n\geq 1$, the function $x\mapsto I_{n,0}(x)$ is monotone decreasing in each $|x|_\mu$ with $\mu=1,2,\dots, d$.
\end{lemma}

Lemma \ref{Lemma-Mono-SRW-x} is \cite[Lemma B.3]{HarSla92b} and follows from the monotonicity of Bessel functions,
see \eqref{fN-def}. Lemma \ref{Lemma-Mono-SRW} is already proven for all $x$ with $|x^{\sss(d)}|_{\sss \infty}\leq 1$, in the proof of \cite[Lemma C.1]{HarSla92b}. For this reason we copy parts of the proof almost verbatim from \cite[Lemma C.1]{HarSla92b}. For $|x^{\sss(d)}|_{\sss \infty}=2$, we only need to add some additional ideas.

\proof Equation \cite[(B.2)]{HarSla92b}, states, after a rescaling of the $t$ variable, that
	\eqn{
	I_{n,0}^{\sss(d)}(x^{\sss(d)})=\frac{1}{(n-1)!} \int_0^\infty dt\ t^{n-1} \prod_{\mu=1}^d f_{|x^{\sss(d)}_i|}(t/d)dt,
	}
where
	\eqn{
	\label{fN-def}
	f_{N}(z) = \int_{-\pi}^{\pi} \frac{d\theta}{2\pi} \e^{-z[1-\cos(\theta)]} \cos(N\theta).
	}
Let $\ell_1$ be the number of coordinates equal to $\pm 1$ in $x^{\sss(d)}$, and $\ell_2$ the number of coordinates equal to $\pm2$ in $x^{\sss(d)}$. Then,
	\eqn{
	I_{n,0}^{\sss(d)}(x^{\sss(d)})=\frac{1}{(n-1)!} \int_0^\infty t^{n-1} f_{0}(t/d)^d \Big(\frac{f_{1}(t/d)}{f_{0}(t/d)}\Big)^{\ell_1} \Big(\frac{f_{2}(t/d)}{f_{0}(t/d)}\Big)^{\ell_2} dt.
	}
Denoting the $L_p(-\pi,\pi)$ norm of $g_s(\theta):=\exp(-s(1-\cos\theta))$ by $\|g_s\|_p$, we obtain by \eqref{fN-def} that
$f_0(s/d)^d=\|g_s\|_{1/d}$, which is monotone decreasing in $d$ for fixed $s$. This already yields the statement for $x^{\sss(d)}=0$.

As shown by Hara and Slade \cite[Proof of Lemma C.1]{HarSla92b}, the factor $f_1(s/d)/f_0(s/d)$ is monotone decreasing in $d$, for fixed $s$, because $z\mapsto f_1(z)/f_0(z)$ is monotone increasing in $z>0$.
This can be seen from a direct calculation of its derivative, which gives
	\eqn{
	\frac d {dz} \Big(\frac{f_{1}(z)}{f_{0}(z)}\Big)=\langle\cos^2(\theta)\rangle-\langle\cos(\theta)\rangle^2\geq 0,
	}
where
	\eqn{
	\langle f(\theta)\rangle=\frac {\int_{0}^{2\pi} d\theta\ f(\theta) \e^{z \cos\theta} }{\int_{0}^{2\pi} d\theta \e^{z \cos\theta}}.
	}
This proves the statement for all $x^{\sss(d)}$ with $|x^{\sss(d)}|_{\sss \infty}=1$, and starts the point where we extend \cite[Lemma C.1]{HarSla92b}.

For even $N$, we use the periodicity of $\cos(\theta)$ and $\cos(N\theta)$, as well as $\cos(N\theta)=\cos(N(\theta+\pi))$ and \eqref{fN-def}, to see that
	\begin{align*}
	f_{2}(z) &=\e^{-z}\int_{-\pi/2}^{3\pi/2} \frac{d\theta}{2\pi} \e^{z\cos(\theta)} \cos(N\theta)\\
	&=\e^{-z}\int_{-\pi/2}^{\pi/2} \frac{d\theta}{2\pi} \big[\cos(N\theta)\e^{z\cos(\theta)}+\cos(N(\theta+\pi))\e^{z\cos(\theta+\pi)})\big]\\
	&=\e^{-z}\int_{-\pi/2}^{\pi/2} \frac{d\theta}{2\pi} \cos(N\theta) (\e^{z\cos(\theta)}+\e^{-z\cos(\theta)})
	=\e^{-z} \int_{-\pi/2}^{\pi/2} \frac{d\theta}{\pi} \cosh(z\cos(\theta)) \cos(Nz).
	\end{align*}
For odd $N$, the same can be derived where $\cosh$ is replaced by $\sinh$, as $\cos(N\theta)=-\cos(N(\theta+\pi))$ for odd $N$, but we will not rely on this fact.

Using a uniform random variable $U$ on $[-\pi/2,\pi/2]$, we can rewrite the above as
	\eqan{
	\frac{d}{d z} \frac{f_{2}(z)}{f_{0}(z)}
	&=\frac{\expec[\sinh(z\cos(U)) \cos(2U)\cos(U)]}{\expec[\cosh(z\cos(U))]}\\
	&\qquad-\frac{\expec[\cosh(z\cos(U)) \cos(2U)]\expec[\sinh(z\cos(U))\cos(U)]}{\expec[\cosh(z\cos(U))]^2}.\nn
	}
Using the random variable $X$ with density $f_{\sss X}(z)=\cosh(zx)/\int_0^\pi \cosh(zy)dy$, we can simplify
	\eqn{
	\frac{d}{d z} \frac{f_{2}(z)}{f_{0}(z)}=\mathrm{Cov}(\tanh(z\cos(X)) \cos(X),\cos(2X)).
	}
Then, we note that $\cos(2x)=2\cos(x)^2-1$, so
	\eqan{
	\frac{d}{d z} \frac{f_{2}(z)}{f_{0}(z)}&=\mathrm{Cov}(\tanh(z\cos(X)) \cos(X),2\cos(X)^2-1)\nn\\
	&=2\mathrm{Cov}(\tanh(z\cos(X)) \cos(X),\cos(X)^2).
	}
Finally, let $Y=\cos(X)$, which is non-negative as $X\in [-\pi/2,\pi/2]$, to conclude
	\eqn{
	\frac{d}{d z} \frac{f_{2}(z)}{f_{0}(z)}=2\mathrm{Cov}(\tanh(zY) Y,Y^2).
	}
Since both $y\mapsto \tanh(zy) y$ as well as $y\mapsto y^2$ are increasing on $[0,1]$, the above covariance is non-negative. Here we use that $\mathrm{Cov}(f(Y),g(Y))\geq0$ when both $f$ and $g$ are increasing on the support of $Y$, with strict inequality when both are not constant.
\qed
\medskip

In what follows, {\em every} bound on SRW-integrals that we use is concluded from
the monotonicity in $d$ of $I_{n,0}(x)$ of Lemma \ref{Lemma-Mono-SRW}, as well as the monotonicity in $|x_\mu|$ of $I_{n,0}(x)$ in Lemma \ref{Lemma-Mono-SRW-x}.
In this way all these bounds will inherit their monotonicity from these $I_{n,0}(x)$.
To compute these other bounds, we use that
	 \eqn{
 	I_{n,m}(x)= (D \star I_{n,m-1}(x))=\sum_{\iota} I_{n,m-1}(x+\ve[\iota]).
 	}
and apply this bound iteratively starting from $m=1$ until $m=10$ to create our bounds.

As we know that $I_{n,m}(x)$ is monotone in $d$ only for $m=0$ and appropriate values of $x$,
it is a priori not clear that these bounds are also monotone decreasing in $d$.
For this reason, let us create a bound that is monotone decreasing by construction.
For $x\in\Zd$, let
	\begin{align*}
	\underline{x}:=(\underline{x}_i)_{i=1}^d,
	\qquad
	\text{where}
	\qquad
	\begin{cases}
	\underline{x}_i=x_i&\text{ when }|x_i|\leq 2,\\
	\underline{x}_i=2 &\text{ otherwise}.
	\end{cases}
	\end{align*}
Thus, $\underline{x}$ is the point that is closest to $x$ and is still covered by Lemma \ref{Lemma-Mono-SRW}.
By Lemma \ref{Lemma-Mono-SRW-x}. we know that $I_{n,0}(x)\leq I_{n,0}(\underline{x})$, which allows us to use monotone bounds.

Let us now explain how we use the above ideas to derive monotone bounds on $I_{n,m}(x)$ for every $n,m,x$. We define $\bar I_{n,0}(x):= I_{n,0}(\underline{x})$ and, for $m\geq 1$, define $\bar I_{n,m}(x)$ recursively by
	\begin{align*}
	\bar I_{n,m}(x):=\sum_{\iota} \bar I_{n,m-1}(x+\ve[\iota]).
	\end{align*}
To finalize the bounds we need to show that
	\begin{align}
	\lbeq{equation-Inmbar}
	I_{n,m}(x)\leq \bar I_{n,m}(x),
	\end{align}
which we will do by induction on $m$.

For $m=0$, the inequality \refeq{equation-Inmbar} holds by definition and Lemma \ref{Lemma-Mono-SRW-x}.
For $m\geq 1$, we note that
	\begin{align*}
	I_{n,m}(x)= (D \star I_{n,m-1}(x))=\sum_{\iota} I_{n,m-1}(x+\ve[\iota])\leq \sum_{\iota} \bar I_{n,m-1}(x+\ve[\iota])=\bar I_{n,m}(x),
	\end{align*}
where the inequality follows from \refeq{equation-Inmbar} for $m-1$. This advances the induction, and thus proves \refeq{equation-Inmbar}. It also concludes our construction of bounds that are monotone in $d$, as $\bar I_{n,m}$ is monotone in the sense of Definition \ref{def-mon-dim}.

Using these monotone bounds on $I_{n,m}(x)$ we conclude bounds on $L_n(x)$, as given in \cite[(5.17)]{FitHof13b}, e.g.
	\begin{align*}
	L_n(\ve[1]+\ve[2])=& \frac {(d-2)(d-3)}{d(d-1)} I_{n,0}(\ve[1]+\ve[2]+\ve[3]+\ve[4])\\
	&+\frac {d-2}{2d(d-1)}\left(I_{n,0}(\ve[1]+\ve[2])+I_{n,0}(2\ve[1]+\ve[2]+\ve[3])\right)\\
	&+\frac {1}{4d(d-1)}\left(I_{n,0}(0)+I_{n,0}(2\ve[1]+2\ve[2])+2I_{n,0}(2\ve[1])\right)\\
    	\leq& \bar I_{n,0}(\ve[1]+\ve[2]+\ve[3]+\ve[4])+\frac {1}{2d}\left(\bar I_{n,0}(\ve[1]+\ve[2])+\bar  I_{n,0}(2\ve[1]+\ve[2])\right)\\
	&+\frac {1}{4d(d-1)}\left(\bar  I_{n,0}(0)+I_{n,0}(2\ve[1]+\ve[2])+2\bar I_{n,0}(2\ve[1])\right),
	\end{align*}
which again, by the monotonicity in the coordinates of $I_{n,0}(x)$, creates a bound on $L_n(x)$ that is monotone decreasing in $d$ as well.
Using the same idea and the methods in \cite[Section 5.2]{FitHof13b}, we can construct bounds on $K_{n,m}(x)$, $T_{n,m}(x)$, $V_{n,m}$, $T_{n,m}(x)$, $ \Isupx_{n,l}(x)$.
These bounds inherit their monotonicity in $d$ from the $\bar I_{n,m}(x)$.
Note that we also remove all terms that are subtracted, e.g. $2I_{n,l}(2\ve[1])$ in \cite[(5.25)]{FitHof13b}, and
use only points $x$ covered by Lemma \ref{Lemma-Mono-SRW}, by exploiting the monotonicity of $\bar{I}_{n,m}(x)$ in $x$.
All this can be found in the second section of the (very readable) \verb|LAMonotone| Mathematica notebook,
which can be downloaded from \cite{FitNoblePage}, in its original form, as well as in pdf-format.

\subsection{Monotonicity of small-step paths}
\label{sec-monoton-inD-2}
A corner stone of our numerical bounding procedure is to extract explicit short contributions from our diagrams,
for which we can use the repulsiveness of connections involved. For this step we require that
	$\tfrac {c_n(x)}{(2d-1)^n}$
is monotone decreasing in $d$, where $c_n$ is the size of a collection of paths, such as simple random walk, self-avoiding walk, or bond-self-avoiding walk paths, of $n$ steps, starting at the origin and ending at $x\in \Z^d$. Our argument applies to each of these collections.

\begin{lemma}[Monotonicity of small-step paths]
\label{Lemma-SmallSteps}
Fix $n$ and $x$. Let $c_n(x)$ be the number of $n$-step paths from 0 to $x$ of any type and $\mathrm{dim}(x)$ be the dimensionality of $x$, i.e.,
	\begin{align*}
	\mathrm{dim}(x)&:=\#\{\text{\rm{non-zero coordinates of $x$}}\}.
	\end{align*}
Then, $d\mapsto c_n(x)/(2d-1)^n$ is monotonically decreasing in $d$ for
$2d\geq \max\{2\mathrm{dim}(x)+n-|x|_1, 2n-2\}$.
\end{lemma}

\proof
We define the dimension of a path be the maximal dimension of the points it traverses, and let
	\begin{align*}
	c_n^{\sss(d)}(x)&:=\#\{\text{number of paths from $0$ to $x$ which have exactly $d$ dimension}\}.
	\end{align*}
This $c_n^{\sss(d)}$ has the following properties, which hold for any type of nearest-neighbor paths:
	\begin{enumerate}[(i)]
  	\item $c_n^{\sss(d)}(x)\equiv 0$ when $n<|x|_1$, as there are simply not enough steps $n$ to reach $x$
  	\item $c_n^{\sss(d)}(x)\equiv 0$ when $d< \mathrm{dim}(x)$, as the point $x$ lives in larger dimension than $d$.
  	\item $c_n^{\sss(d)}(x)\equiv 0$ when $n-|x|_1$ is odd, by parity of the lattice.
  	\item $c_n^{\sss(d)}(x)\equiv 0$ when $d>\mathrm{dim}(x)+\tfrac {n-|x|_1}2$, as the path can simply not use all the required dimensions $d$ and still reach $x$ in $n$ steps.
	\end{enumerate}
By grouping the paths from $0$ to $x$ according to the number of dimensions $d'$ used, regarding the properties above,  we obtain
	\begin{align*}
	c_n(x)=&\sum_{d'=\mathrm{dim}(x)}^{\mathrm{dim}(x)+\tfrac {n-|x|_1} 2} {d-\mathrm{dim}(x) \choose  d'-\mathrm{dim}(x)}c^{\sss(d')}_n(x)
	=\sum_{d'=0}^{\tfrac {n-|x|_1}2} {d-\mathrm{dim}(x) \choose  d'}c^{\sss(d'+\mathrm{dim}(x))}_n(x).
	\end{align*}
The binomial coefficient represents the free choices for the additional dimensions $d'$, not imposed by $x$ itself, that the path uses.
We then write
	\eqn{
	\frac {c_n(x)}{(2d-1)^n}=\sum_{d'=0}^{\tfrac {n-|x|_1}2 } f_n(d',d,x)c^{\sss(d'+\mathrm{dim}(x))}_n(x),
	}
where
	\eqn{
	f_n(d',d,x):=  \frac { {d-\mathrm{dim}(x) \choose  d'} }{(2d-1)^n}.
	}
Our restriction to $2d\geq 2\mathrm{dim}(x)+n-|x|_1$ ensures that each summand contributes, which allows us to conclude the
required monotonicity of $d\mapsto c_n(x)/(2d-1)^n$ by showing that $d\mapsto f_n(d',d+1)$ is monotonically decreasing for every $d'$ fixed.
 For this, we need to show that
	\eqn{
	\frac {f_n(d',d+1,x)}{f_n(d',d,x)}=\frac {d+1-\mathrm{dim}(x)}{d+1-\mathrm{dim}(x)-d'} \Big(\frac {2d-1}{2d+1}\Big)^n,
	}
is smaller than one. The first fraction is monotone increasing in $\mathrm{dim}(x)$ and $d'$.
We use that $d'\leq \tfrac {n-|x|_1}2$, $\mathrm{dim}(x)\leq |x|_1$ and then $n\geq |x|_1\geq 0$ to obtain
	\eqan{
	\label{ratio-bd-1}
	\frac {d+1-\mathrm{dim}(x)}{d+1-\mathrm{dim}(x)-d'} &\leq
	\frac {d+1-|x|_1}{d+1-|x|_1-\tfrac {n-|x|_1}2}=1+\frac {n-|x|_1}{2d+2-|x|_1-n}\\
	&\leq 1+\frac {n}{2d+2-n}=\frac{2d+2}{2d+2-n},
	}
where the final inequality follows since $n\leq d+1$.
For the second fraction we note that $(1+y)^n\leq 1+yn+y^2n(n-1)/2$
for all $n\leq y$ and $y>-1$, so that
	\begin{align*}
 	\Big(\frac {2d-1}{2d+1}\Big)^n=\Big(1-\frac {2}{2d-1}\Big)^n\leq 1-\frac {2n}{2d-1}+\frac{2n(n-1)}{(2d-1)^2}
	=\frac {(2d-1)^2-2n(2d-1)+2n(n-1)}{(2d-1)^2}.
	\end{align*}
Combining both bounds, we obtain
	\eqan{
	\frac {f_n(d',d+1,x)}{f_n(d',d,x)}&\leq \frac {(2d-1)^2-2n(2d-1)+2n(n-1)}{(2d-1)^2}\frac{2d+2}{2d+2-n}\\
	&=1+\frac {n(2d-1)^2-2n(2d-1)(2d+2)+n(n-1)(2d+2)}{(2d-1)^2(2d+2-n)}\nn\\
	&=1+n\frac{(2d-1)^2-2(2d-1)(2d+2)+(n-1)(2d+2)}{(2d-1)^2(2d+2-n)}.\nn
	}
We finally note that, since $n\leq d$,
	\eqan{
	&(2d-1)^2-2(2d-1)(2d+2)+(n-1)(2d+2)\\
	&\qquad\leq (2d-1)^2-2(2d-1)^2+(d-1)(2d+2)
	=2(d^2-1)-(2d-1)^2<0,\nn
	}
when $d\geq 1$.
	
\qed

\subsection{Monotonicity of bounds on the NoBLE coefficients}
\label{sec-monoton-inD-3}
We bound simple and repulsive diagrams using the bootstrap assumption,
SRW-integrals and by extracting short explicit contributions. This is explained in detail in \cite[Section 5.3.1]{FitHof13b}.
For example,
	\begin{align*}
 	\diagRepulsiveLetter{T}_{1,1,0}(x)
	\leq& \sum_{i=2}^{M-1} \frac {(i-1)i}{2}p_{i}(x) (z\gj)^{i}+\frac{(M-2)(M-3)}{2}(2dz\gj)^{M} (D^{\star M}\star G_{z})(x) \\
	&+(M-2)(2dz\gj)^{M} (D^{\star M}\star G^{\star 2}_{z})(x)+(2dz\gj)^{M} (D^{\star M}\star G^{\star 3}_{z})(x),
	\end{align*}
where $p_n$ is the number of $n$-step bond-self-avoiding paths from $0$ to $x$, i.e., paths that do not use the same bond twice. For dimension $d=16,\dots, 29$ we use $M=10$.  For our analysis in $d\geq 30$,  we extract only $7$ steps, so that $M=8$.

To obtain numerical values for these bounds, we use
	\begin{align*}
	z\gj \leq& \frac {1} {2d-1} \Gamma_1\quad \text{ and }(D^{\star M}\star G^{\star n}_{z})(x)\leq \big(\frac {2d-2}{2d-1} \Gamma_2 \big)^n K_{n,M}(x)\leq  \Gamma_2^n K_{n,M}(x),
	\end{align*}
and use the following bound that can be numerically computed:
	\begin{align*}
 	\diagRepulsiveLetter{T}_{1,1,0}(x)
	\leq& \sum_{i=2}^{7} \frac {(i-1)i}{2} \frac {p_{i}(x)} {(2d-1)^i} \Gamma^i_1
        +15 \left(\frac {2d}{2d-1}\Gamma_1\right)^8 \Gamma_2 K_{1,8}(x)\\
	&  +6 \left(\frac {2d}{2d-1}\Gamma_1\right)^8 \Gamma_2^2 K_{2,8}(x)
        +\left(\frac {2d}{2d-1}\Gamma_1\right)^8 \Gamma_2^3 K_{3,8}(x).
	\end{align*}
By Lemma \ref{Lemma-SmallSteps} the first term is monotone decreasing in $d$. Further, we bound $K_{n,8}(x)$ using bounds that are monotone in $d$, see
Section \ref{sec-monoton-inD-1}, which means that our bound on $\diagRepulsiveLetter{T}_{1,1,0}(x)$ is monotone decreasing in $d$.

We use the same kind of bounds for {\em all } simple and repulsive diagrams. As we bound all NoBLE coefficients using these diagrams (recall Section \ref{secStatingtheBounds}),
the bounds that we rely upon inherit this monotonicity. Having made this point we should add that in three bounds in Lemma \ref{BoundNZero-XiIota} we see a factor of $(2d-2)/2d$,
which is not monotone decreasing. We resolve this non-monotonicity issue by simply replacing $(2d-2)/2d$ by $1$ in our bounds/implementation for $d\geq 30$.

Unfortunately our verification of the technical condition $\Pi^{\ssc[0],\iota,\kappa}(x)\leq zg_z \Xi^{\ssc[0],\iota}(x)$ is numerical. We have already proved this for all $x\neq \ve[\iota]+\ve[\kappa]$. Even though the contribution for $x=\ve[\iota]+\ve[\kappa]$ is quite small (recall \refeq{tech-cond-prob-iota-kappa}), we did not manage to find a proof that is uniform in the dimension for $x=\ve[\iota]+\ve[\kappa]$. This causes problems in the application of the general analysis in \cite{FitHof13b}. To bypass this, we use a bound that applies to $\Xi^{\ssc[0],\iota}$ and $\Pi^{\ssc[0],\iota,\kappa}$ at the same time, and replace $\beta_{\sss\Xi^\iota}^\ssc[0], \beta_{{\sss\Delta\Xi^\iota},0}^\ssc[0], \beta_{{\sss\Delta\Xi^\iota},\iota}^\ssc[0]$ that are currently used as to bound $\Xi^{\ssc[0],\iota}$ only, by a bound that applies also to $\Pi^{\ssc[0],\iota,\kappa}$. Numerically, this makes no difference, as our crude numerical bound already applies to both these terms at the same time.

This problem has repercussions in the proof of \refeq{Bound-Pi0RemI} and \refeq{Bound-Pi0RemIW}, where this bound was used. 
However, the problem does not arise in \refeq{Bound-Pi0AlphaI}, since it does not involve $x=\ve[\iota]+\ve[\kappa]$.
This problem does not arise in \refeq{Bound-Pi0RemIW} either, since it is zero when $x=\ve[\iota]+\ve[\kappa]$ (recall also \refeq{Pi(0)-bd}).

In conclusion, all bounds that are used in the general analysis of \cite{FitHof13b}, as summarized in Section \ref{secBoundsSummary}, are monotone decreasing in $d$.


\subsection{Monotonicity of bounds on the general analysis coefficients}
\label{sec-monoton-inD-4}
In this section we discuss the analysis of the generalized NoBLE analysis as explained in \cite{FitHof13b}.
Our aim is to show that the monotonicity of the bounds on the coefficients implies that
{\em if} a bootstrap, as described in Proposition \ref{prop-analysis-is-success}, is
successful in the given dimension $d'=30$, {\em then} the bootstrap will be successful in every larger dimension $d\geq d'$.
For this, we fix the assumed bounds $\Gamma_i, c_{n,l,S}$ and use them uniformly for all $d\geq 30$. We next explain how we can see that the success of
the bootstrap for $d=30$ implies that for $d\geq 30$, splitting between the various bootstrap function $f_i(z)$ for $i=1,2,3$:

\paragraph{Monotonicity of success for the initialization and improvement of $f_1$.}
We start with the bootstrap function
	\begin{align*}
	f_1(z)=\max\left\{(2d-1)zg_z,c_\mu (2d-1)zg_z^\iota\right\}.
	\end{align*}
In \refeq{Bound-Initial-f1} we have already seen the bound on $f_1(z_I)$, which is clearly monotone decreasing in $d$. This deals with the initialization.

For $z\in(z_I,z_c)$, the bounds on $f_1(z)$ as given in \cite[(3.5)]{FitHof13b} can be bounded uniformly in $d$ as follows. In our argument, we write
$\beta_\bullet(d)$ for the parameter $\beta_\bullet$ in dimension $d$. Note that the analysis in the previous section implies the monotonicity of $d\mapsto \beta_\bullet (d)$ for {\em all}
$\beta_\bullet (d)$ discussed in Section \ref{secBoundsSummary}. This leads to
	\begin{align*}
	f_1(z)\leq& \betaaa(d) \frac {1+\overline{\beta}_{\sss \Pi^{\iota}}(d) }{ 1 - \frac {2d}{2d-1}\underline {\beta}_{\sss \Psi^{\kappa}}(d)}
	\leq \betaaa(30) \frac {1+\overline{\beta}_{\sss \Pi^{\iota}}(30)}{ 1 - \frac {60}{59}\underline {\beta}_{\sss \Psi^{\kappa}}(30) }.
	\end{align*}
For $\betaaa(d)$, which bounds $g_z/g_z^{\iota}$, we note that
	\begin{align*}
	g_z-g_z^{\iota}=\bar G_z(\ve[\iota]).
	\end{align*}
This implies
	\begin{align*}
	\frac {\aabz}{\aaz}=\frac {g_z}{g_z^{\iota}}= 1+\frac {\bar G_z(\ve[\iota])}{g_z^{\iota}}=1+\frac {\aabz}{\aaz} G_z(\ve[\iota])
	\Rightarrow \frac {\aabz}{\aaz}=\frac 1 {1-G_z(\ve[\iota])}.
	\end{align*}
We bound $G_z(\ve[\iota])$ in terms of simple diagrams, see Appendix \ref{sec-monoton-inD-3}. This
creates a bound that is monotone decreasing. Thus, our bound on $\betaaa(d)$ is decreasing in $d$ as well.

\paragraph{Bounds on the rewrite.}
We bound $f_2$ and $f_3$ using a rewrite as this made the analysis much clearer.
To proceed we next show that all coefficients of the rewrite are monotone decreasing in $d$.
This is done by checking each bound on the rewrite, as given in \cite[Appendix D]{FitHof13b}, one bound at a time.

For most bounds it is obvious that they are monotone, as they only use the bounds on the coefficients $\beta_{\bullet}$ (which we showed are monotone already in the previous section).
For this reason we only comment on three issues, for which it is not obvious that the bounds are monotone decreasing.

Regarding the bounds on \cite[(D.4)]{FitHof13b} and \cite[(D.5)]{FitHof13b}, we note that we have
defined $\underline{\beta}_{\sss \Psi}^\ssc[0]   =  \underline{\beta}_{\sss \sum \Pi}^\ssc[1]=0$,
so that the last terms in these lines are not a concern regarding monotonicity.

In the bounds on the remainder terms in \cite[ Steps 2-5 in Appendix D]{FitHof13b}, the following terms appear:
	\begin{align*}
	\frac {\aa} {1-\aa^2} \quad \text{and}\quad \frac {\aab} {1-\aa^2}\leq \betaaa \frac {\aa} {1-\aa^2}.
	\end{align*}
We next derive bounds on these terms that are uniform in $d$.

The function $\tfrac {x}{1-x^2}$ is monotone increasing on the interval $[0,1]$, and
	$\aa\leq \tfrac {\Gamma_1}{2d-1}\leq \tfrac {\Gamma_1}{59}$ for $d\geq 30$, so that
	\begin{align*}
	\frac {\aa} {1-\aa^2} \leq \frac {\Gamma_1} {(2d-1)(1-\frac {\Gamma_1^2}{(2d-1)^2})} \leq \frac  {\tfrac {\Gamma_1}{59} }{1- \tfrac {\Gamma_1^2}{59^2}}.
	\end{align*}
The term $\afz\approx 1$, which plays a central role in the NoBLE rewrite and is defined in  \cite[(4.18)]{FitHof13b}, is bounded from above and below.
We require an upper bound that is decreasing and a lower bound that is increasing in $d$.
The upper bound, stated in \cite[(D.2)]{FitHof13b},
	\begin{align*}
	\afz\leq\frac {2d \aa} {1-\aa^2}	\left[1+\beta_{{\sss\sum \Psi^{\iota}_\alpha,I}}^\ssc[0-1]
	+\aa\beta_{{\sss\sum \Psi^{\iota}_\alpha,II}}^\ssc[1-0]
	-\frac {1} {1-\aa^2} \underline{\beta}_{ \sss  \sum \Pi_{\alpha}}^\ssc[0]\right]
	\end{align*}
is clearly monotone, as we have chosen $\underline{\beta}_{ \sss  \sum \Pi_{\alpha}}^\ssc[0]=0$
and bound the other two terms using repulsive diagrams, see \refeq{Differencebound-4}, \refeq{Differencebound-7}.
As $\afz$ arises in the dominator of $\hat G_z(k)$, see  \cite[(1.37)]{FitHof13b},
we require a monotone increasing lower bound on $\afz$,
to obtain then a bound on $\hat G_z(k)$ that is monotone decreasing. We use the lower bound in \cite[(D.2)]{FitHof13b} that reads
	\begin{align*}
	\afz\geq \frac {2d \aa} {1-\aa^2}
	&\left[1-\beta_{{\sss\sum \Psi^{\iota}_\alpha,I}}^\ssc[1-0]
	-\aa\beta_{{\sss\sum \Psi^{\iota}_\alpha,II}}^\ssc[0-1]
	-\frac {1} {1-\aa^2}  \bar{\beta}_{\sum \sss \Pi_{\alpha}}^\ssc[0]\right].
	\end{align*}
The term in the brackets is bounded by
	\begin{align*}
	1-\beta_{{\sss\sum \Psi^{\iota}_\alpha,I}}^\ssc[1-0](d)
	-\frac {\Gamma_1}{(2d-1)}\beta_{{\sss\sum \Psi^{\iota}_\alpha,II}}^\ssc[0-1](d)
	-\frac {1} {1-\frac {\Gamma_1^2}{(2d-1)^2}}  \bar{\beta}_{\sum \sss \Pi_{\alpha}}^\ssc[0](d):=1-\zeta(d).
	\end{align*}
This sum of bounds in $\zeta(d)$ is clearly decreasing in $d$, so that $1-\zeta(d)\geq 1-\zeta(30)$.
Let us remark that we have implicitly assume that $\zeta(d)<1$, which a relatively weak condition that we have verified numerically.
The initial factor $\tfrac {2d \aa} {1-\aa^2}$ is monotone increasing in $\aa$ and we know that
	\begin{align*}
	\aa=z\gj \geq z_Ig_{z_I}^\iota = \frac 1 {2d-1},
	\end{align*}
for all $z\geq z_I$, by definition of $z_I$. This means that
	\begin{align*}
	\frac {2d \aa} {1-\aa^2}\geq   \frac {2d-1}{2d-2}=1 +\frac {1}{2d-2}\geq 1,
	\end{align*}
so that
	\begin{align*}
	\afz\geq \frac {2d-1}{2d-2}\left[1-\zeta(30)\right].
	\end{align*}
Due to the factor $\tfrac {2d-1}{2d-2}$ this is decreasing in $d$ and not increasing as required for the lower bounds on
$\afz$. We can circumvent by simply using a uniform bound $\tfrac {2d-1}{2d-2}\leq \tfrac {59}{58}$, which we use in the implementation.

\paragraph{Monotonicity of success for the initialization and improvement of $f_2$.}
The bootstrap function $f_2$, defined in \refeq{defFunc2}, bounds $\hat G_z(k)$. The bounds for the initialization follow the same bounds as for the improvement, so we only discuss the improvement here.

From \cite[(3.9)]{FitHof13b} and the bounds on $\afz$ that we derived above, we obtain
	\begin{align*}
	|\hat G_z(k)| [1-\hat D(k)]&
	\leq \frac {\upcp(d)+\betaap(d)+\betaRp(d)}{\lowaf(d)-\betadeltaRfzlow(d)}
    	\leq \frac {\upcp(30)+\betaap(30)+\betaRp(30)}{\tfrac {2d-1}{2d-2}\left[1-\zeta(30)\right]-\betadeltaRfzlow(30)}.
	\end{align*}
The definition \refeq{defFunc2} states
	\begin{align*}
	f_2(z)=& \frac{2d-1}{2d-2}\sup_{k\in(-\pi,\pi)^d} [1-\hat D(k)]\ \hat G_z(k)\\
	\leq&  \frac {\upcp(30)+\betaap(30)+\betaRp(30)}
	{\tfrac{2d-2}{2d-1}\left(\tfrac {2d-1}{2d-2}\left[1-\zeta(30)\right]-\betadeltaRfzlow(30)\right)}
	\leq\frac {\upcp(30)+\betaap(30)+\betaRp(30)}{\left[1-\zeta(30)\right]- \betadeltaRfzlow(30)},
	\end{align*}
which creates a bound that holds uniformly for all $d\geq 30$.

\paragraph{Monotonicity of success for the initialization and improvement of $f_3$.}
The bound on $f_3$ is derived in multiple steps, so that we have to check the monotonicity of multiple expressions.
All of these are quite similar and the required monotonicity always follows from the
monotonicity of the NoBLE coefficients of the rewrite and the SRW integrals.

The bounds for the initialization of the bootstrap for $f_3$ at $z_I$ is given in \cite[(3.30)-(3.31)]{FitHof13b}.
To guarantee that we use a monotone bound we omit from $\Isupx_{n,l}(x)$, in \cite[(3.30)]{FitHof13b}
the term $-\tfrac 1 d I_{N+3,l}(x)$ and use the spatial monotonicity of $I_{n+3,l}(x)$ in $|x_\mu|$
to ensure that we only use values for $I_{n+3,l}(x+2\ve[\iota])$ for which $|x+2\ve[\iota]|_{\infty}\leq 2$.
This results in a uniform bound on $f_3(z_I)$.

The bound on $f_3(z)$ for $z\in(z_I,z_c)$  is stated in \cite[(3.87)]{FitHof13b}
and consists of many individual terms derived in the lines \cite[(3.60)-(3.85)]{FitHof13b}.
All these terms consists of simple products of bounds on the terms of the rewrite and SRW integrals,
which, as argued before, are monotone decreasing in $d$.

To remove any doubt about monotonicity, we redefine some bounding coefficients whose definition included a factor
$\tfrac {2d-2}{2d-1}$:
	\begin{align*}
	\frac 1 {\afz}&\leq \frac 1 {1-\zeta(30)}:=\lowaf^{-1},\\
	\lowK&:=\frac 1 {1-\zeta(30)-\betadeltaRfzlow(30)},\\
	\Gamma_2'&:=\Gamma_2.
	\end{align*}
This results in a bound in which each individual piece is monotone decreasing in $d$,
so that the bound on $f_3$ holds uniformly in $d\geq 30$.

\subsection{Conclusion for all $d\geq 30$}
\label{sec-monoton-inD-5}
Our aim was that to show that our results, as stated in Theorem \ref{thm-IRB}, Corollary \ref{cor-TC-crit-exp}, Theorem  \ref{thm-k-space} and Theorem \ref{thm-x-space},
hold for all $d\geq d'=30$.

The proof of these results is described in Section \ref{sec-overview}. Regarding the proof we note that Propositions \ref{prop-LE} and \ref{prop-bds-LEC} hold
regardless of the dimension, even for $d\leq 16$. Our restriction to $d\geq \dmintree$ for LT and $d\geq \dminanimal$ for LA, respectively,
is only necessary as Proposition \ref{prop-analysis-is-success} does not hold in smaller dimensions.

In Appendix \ref{sec-monoton-inD-4} we conclude a bound on the bootstrap functions that holds uniformly for $d\geq d'$.
Using our numerical verification in \verb|LAmonotone|, we check that Proposition \ref{prop-analysis-is-success} holds using (only) our monotone bounds. In other words in $d'=30$ our bootstrap was successful.
As the bounds on the bootstrap function are uniform in $d$ we know that Proposition \ref{prop-analysis-is-success} therefore also holds for $d'\geq d$. For this reason, our analysis succeeds in all $d\geq d'=30$.
\qed


\def\picAZeroZeroA[#1]{
\begin{tikzpicture}[line width=1pt,auto,scale=#1]
  \draw [-](0,0) to [out=90,in=90]    (2,0);
  \draw [-](0,0) to [out=270,in=270]    (2,0);
  \fill (0,0) circle (3pt);
  \node[left]   at(0,0)      {$0$};
  \node[right]  at (2,0)   {$\geq 4$};
\end{tikzpicture}}

\def\picAZeroZeroB[#1]{
\begin{tikzpicture}[line width=1pt,auto,scale=#1]
  \draw [-](0,0) to [out=90,in=90]    (3,0);
  \draw [-](0,0) to [out=270,in=270]    (3,0);
  \fill (0,0) circle (3pt);
  \fill (3,0) circle (3pt);
  \node[left]   at(0,0)      {$0$};
  \node[right]  at (3,0)     {$x$};
  \node[above]  at (1,0.7)     {$\geq 1$};
  \node[below]  at (2,-0.7)     {$\geq 2$};
\end{tikzpicture}}

\def\picAZeroZeroC[#1]{
\begin{tikzpicture}[line width=1pt,auto,scale=#1]
  \draw [-] (0,0)--(1,1.5);
  \draw [-] (1,1.5)-- (2,0);
  \draw [-] (2,0)--(0,0);
  \fill (0,0) circle (3pt);
  \fill (1,1.5) circle (3pt);
  \fill (2,0) circle (3pt);
  \node[left]   at(0,0)       {$0$};
  \node[right]   at (2,0)       {$x$};
  \node[left]   at (1,1.5)      {$w$};
  \node[below]  at (1,0)        {$\geq 1$};
  \node[left]   at (0.6,.9)     {$\geq 1$};
  \node[right]  at (1.4,0.9)    {$\geq 1$};
\end{tikzpicture}}

\def\picAZeroOneA[#1]{
\begin{tikzpicture}[line width=1pt,auto,scale=#1]
  \draw [-](0,0) to [out=90,in=90]    (3,0);
  \draw [-](0,0) to [out=270,in=270]    (3,0);
  \fill (0,0) circle (3pt);
  \fill (3,0) circle (3pt);
  \node[left]   at(0,0)      {$0$};
  \node[right]  at (3,0)     {$y$};
  \node[above]  at (1,0.7)     {$\geq 3$};
  \node[below]  at (2,-0.7)     {$= 1$};
\end{tikzpicture}}

\def\picAZeroOneB[#1]{
\begin{tikzpicture}[line width=1pt,auto,scale=#1]
  \draw [-](0,0) to [out=90,in=180]    (3,2);
  \draw [-](3,0) to    (3,2);
  \draw [-](0,0) to    (3,0);
  \fill (0,0) circle (3pt);
  \fill (3,0) circle (3pt);
  \fill (3,2) circle (3pt);
  \node[left]   at(0,0)      {$0$};
  \node[right]  at (3,0)     {$y$};
  \node[right]  at (3,2)     {$x$};
  \node[below]  at (1,0)  {$\geq 1$};
  \node[above]  at (1.2,1.6)   {$\geq 1$};
  \node[right]  at (3,1)     {$=1$};
\end{tikzpicture}}

\def\picAZeroOneC[#1]{
\begin{tikzpicture}[line width=1pt,auto,scale=#1]
 \draw [-](0,0) to    (1,2);
 \draw [-](0,0) to    (3,0);
 \draw [-](1,2) to    (3,1.5);
 \draw [-](3,0) to    (3,1.5);
  \fill (1,2) circle (3pt);
  \fill (3,0) circle (3pt);
  \fill (3,1.5) circle (3pt);
  \fill (0,0) circle (3pt);
  \node[left]   at(0,0)      {$0$};
  \node[left]   at(1,2)      {$w$};
  \node[right]  at (3,0)     {$y$};
  \node[right]  at (3,2)     {$x$};
  \node[left]  at (0.5,1)      {$\geq 1$};
  \node[below]  at (1.5,0)     {$\geq 1$};
  \node[above]  at (2.2,1.7)   {$\geq 1$};
  \node[right]  at (3,1)     {$=1$};
\end{tikzpicture}}

\def\picAZeroTwoA[#1]{
\begin{tikzpicture}[line width=1pt,auto,scale=#1]
  \draw [-](0,0) to [out=90,in=90]    (3,0);
  \draw [-](0,0) to [out=270,in=270]    (3,0);
  \fill (0,0) circle (3pt);
  \fill (3,0) circle (3pt);
  \node[left]   at(0,0)      {$0$};
  \node[right]  at (3,0)     {$y$};
  \node[above]  at (1,0.7)     {$\geq 1$};
  \node[below]  at (2,-0.7)     {$\geq 2$};
\end{tikzpicture}}

\def\picAZeroTwoB[#1]{
\begin{tikzpicture}[line width=1pt,auto,scale=#1]
  \draw [-](0,0) to [out=90,in=180]    (3,2);
  \draw [-](3,0) to    (3,2);
  \draw [-](0,0) to    (3,0);
  \fill (0,0) circle (3pt);
  \fill (3,0) circle (3pt);
  \fill (3,2) circle (3pt);
  \node[left]   at(0,0)      {$0$};
  \node[right]  at (3,0)     {$y$};
  \node[right]  at (3,2)     {$x$};
  \node[below]  at (1,0)  {$\geq 1$};
  \node[above]  at (1.2,1.6)   {$\geq 1$};
  \node[right]  at (3,1)     {$\geq 1$};
\end{tikzpicture}}

\def\picAZeroTwoC[#1]{
\begin{tikzpicture}[line width=1pt,auto,scale=#1]
 \draw [-](0,0) to    (1,2);
 \draw [-](0,0) to    (3,0);
 \draw [-](1,2) to    (3,1.5);
 \draw [-](3,0) to    (3,1.5);
  \fill (1,2) circle (3pt);
  \fill (3,0) circle (3pt);
  \fill (3,1.5) circle (3pt);
  \fill (0,0) circle (3pt);
  \node[left]   at(0,0)      {$0$};
  \node[left]   at(1,2)      {$w$};
  \node[right]  at (3,0)     {$y$};
  \node[right]  at (3,2)     {$x$};
  \node[left]  at (0.5,1)      {$\geq 1$};
  \node[below]  at (1.5,0)     {$\geq 1$};
  \node[above]  at (2.2,1.7)   {$\geq 1$};
  \node[right]  at (3,1)     {$\geq 2$};
\end{tikzpicture}}

\def\picAZeroMinusOneA[#1]{
\begin{tikzpicture}[line width=1pt,auto,scale=#1]
  \draw [-](0,0) to [out=90,in=180]    (3,2);
  \draw [-](3,0) to    (3,2);
  \draw [-](0,0) to    (3,0);
  \fill (0,0) circle (3pt);
  \fill (3,0) circle (3pt);
  \fill (3,2) circle (3pt);
  \node[left]   at(0,0)      {$0$};
  \node[right]  at (3,0)     {$x$};
  \node[right]  at (3,2)     {$y$};
  \node[below]  at (1.5,0)  {$= 1$};
  \node[above]  at (1.2,1.6)   {$\geq 2$};
  \node[right]  at (3,1)     {$= 1$};
\end{tikzpicture}}

\def\picAZeroMinusOneB[#1]{
\begin{tikzpicture}[line width=1pt,auto,scale=#1]
 \draw [-](0,0) to    (1,2);
 \draw [-](0,0) to    (3,0);
 \draw [-](1,2) to    (3,1.5);
 \draw [-](3,0) to    (3,1.5);
  \fill (1,2) circle (3pt);
  \fill (3,0) circle (3pt);
  \fill (3,1.5) circle (3pt);
  \fill (0,0) circle (3pt);
  \node[left]   at(0,0)      {$0$};
  \node[left]   at(1,2)      {$w$};
  \node[right]  at (3,0)     {$x$};
  \node[right]  at (3,2)     {$y$};
  \node[left]  at (0.5,1)      {$\geq 1$};
  \node[below]  at (1.5,0)     {$=1$};
  \node[above]  at (2.2,1.7)   {$\geq 1$};
  \node[right]  at (3,1)     {$=1$};
\end{tikzpicture}}

\def\picAZeroMinusOneC[#1]{
\begin{tikzpicture}[line width=1pt,auto,scale=#1]
 \draw [-](0,0) to    (1,2);
 \draw [-](0,0) to    (3,0);
 \draw [-](1,2) to    (3,1.5);
 \draw [-](3,0) to    (3,1.5);
  \fill (1,2) circle (3pt);
  \fill (3,0) circle (3pt);
  \fill (3,1.5) circle (3pt);
  \fill (0,0) circle (3pt);
  \node[left]   at(0,0)      {$0$};
  \node[left]   at(1,2)      {$w$};
  \node[right]  at (3,0)     {$x$};
  \node[right]  at (3,2)     {$y$};
  \node[left]  at (0.5,1)      {$\geq 0$};
  \node[below]  at (1.5,0)     {$\geq 2$};
  \node[above]  at (2.2,1.7)   {$\geq 0$};
  \node[right]  at (3,1)     {$=1$};
\end{tikzpicture}}

\def\picAZeroMinusTwoA[#1]{
\begin{tikzpicture}[line width=1pt,auto,scale=#1]
 \draw [-](0,0) to    (1,2);
 \draw [-](0,0) to    (3,0);
 \draw [-](1,2) to    (3,1.5);
 \draw [-](3,0) to    (3,1.5);
  \fill (1,2) circle (3pt);
  \fill (3,0) circle (3pt);
  \fill (3,1.5) circle (3pt);
  \fill (0,0) circle (3pt);
  \node[left]   at(0,0)      {$0$};
  \node[left]   at(1,2)      {$w$};
  \node[right]  at (3,0)     {$x$};
  \node[right]  at (3,2)     {$y$};
  \node[left]  at (0.5,1)      {$\geq 1$};
  \node[below]  at (1.5,0)     {$=1$};
  \node[above]  at (2.2,1.7)   {$\geq 0$};
  \node[right]  at (3,1)     {$\geq 2$};
\end{tikzpicture}}

\def\picAZeroMinusTwoB[#1]{
\begin{tikzpicture}[line width=1pt,auto,scale=#1]
  \draw [-](0,0) to [out=90,in=180]    (3,2);
  \draw [-](3,0) to    (3,2);
  \draw [-](0,0) to    (3,0);
  \fill (0,0) circle (3pt);
  \fill (3,0) circle (3pt);
  \fill (3,2) circle (3pt);
  \node[left]   at(0,0)      {$0$};
  \node[right]  at (3,0)     {$x$};
  \node[right]  at (3,2)     {$y$};
  \node[below]  at (1.5,0)  {$\geq 2$};
  \node[above]  at (1.2,1.6)   {$\geq 0$};
  \node[right]  at (3,1)     {$\geq 2$};
\end{tikzpicture}}

\def\picAZeroMinusTwoC[#1]{
\begin{tikzpicture}[line width=1pt,auto,scale=#1]
 \draw [-](0,0) to    (1,2);
 \draw [-](0,0) to    (3,0);
 \draw [-](1,2) to    (3,1.5);
 \draw [-](3,0) to    (3,1.5);
  \fill (1,2) circle (3pt);
  \fill (3,0) circle (3pt);
  \fill (3,1.5) circle (3pt);
  \fill (0,0) circle (3pt);
  \node[left]   at(0,0)      {$0$};
  \node[left]   at(1,2)      {$w$};
  \node[right]  at (3,0)     {$x$};
  \node[right]  at (3,2)     {$y$};
  \node[left]  at (0.5,1)      {$\geq 1$};
  \node[below]  at (1.5,0)     {$\geq 2$};
  \node[above]  at (2.2,1.7)   {$\geq 0$};
  \node[right]  at (3,1)     {$\geq 2$};
\end{tikzpicture}}

\def\picAMinusOneZeroA[#1]{
\begin{tikzpicture}[line width=1pt,auto,scale=#1]
  \draw [-](0,0) to [out=0,in=0]    (0,2);
  \fill (0,0) circle (3pt);
  \fill (0,2) circle (3pt);
  \node[left]   at(0,0)      {$0$};
  \node[left]  at (0,2)     {$v=x$};
  \node[right]  at (0.6,1)   {$\geq 2$};
\end{tikzpicture}}

\def\picAMinusOneZeroB[#1]{
\begin{tikzpicture}[line width=1pt,auto,scale=#1]
 \draw [dotted](0,0) to    (0,1.5);
 \draw [-](0,0) to    (3,0);
 \draw [-](0,1.5) to    (2,2);
 \draw [-](3,0) to    (2,2);
  \fill (0,1.5) circle (3pt);
  \fill (3,0) circle (3pt);
  \fill (2,2) circle (3pt);
  \fill (0,0) circle (3pt);
  \node[left]   at(0,0)      {$0$};
  \node[right]   at(2,2)      {$w$};
  \node[left]   at(0,1.5)      {$v$};
  \node[right]  at (3,0)     {$x=y$};
  \node[above]  at (1,1.75)    {$\geq 0$};
  \node[below]  at (1.5,0)   {$=1$};
  \node[left]  at (0,0.75)   {$=1$};
  \node[above]  at (3,0.7) {$\geq 1$};
\end{tikzpicture}}

\def\picAMinusOneZeroC[#1]{
\begin{tikzpicture}[line width=1pt,auto,scale=#1]
 \draw [dotted](0,0) to    (0,1.5);
 \draw [-](0,0) to    (3,0);
 \draw [-](0,1.5) to    (2,2);
 \draw [-](3,0) to    (2,2);
  \fill (0,1.5) circle (3pt);
  \fill (3,0) circle (3pt);
  \fill (2,2) circle (3pt);
  \fill (0,0) circle (3pt);
  \node[left]   at(0,0)      {$0$};
  \node[right]   at(2,2)      {$w$};
  \node[left]   at(0,1.5)      {$v$};
  \node[right]  at (3,0)     {$x=y$};
  \node[above]  at (1,1.75)    {$\geq 0$};
  \node[below]  at (1.5,0)   {$\geq 2$};
  \node[left]  at (0,0.75)   {$=1$};
  \node[above]  at (3,0.7) {$\geq 0$};
\end{tikzpicture}}

\def\picXilotaInitialStructure[#1]{
\begin{tikzpicture}[line width=1pt,auto,scale=#1]
 \tikzstyle{every node}=[font=\small]

 \begin{scope}[shift={(0,0)},rotate=0]
 \draw [-](0+0.03,-1-0.2) to (0+0.03,-1+0.2);
 \draw [-](0+0.5-0.03,-1-0.2) to   (0+0.5-0.03,-1+0.2);
  \fill (0,1) circle (3pt);
  \fill (0,-1) circle (3pt);
  \fill (-2,-1) circle (3pt);
  \fill (2,1) circle (3pt);
  \fill (2,-1) circle (3pt);
  \draw [-](0.5,-1) to  (2,-1);
  \draw [-](0,1) to  (0,-1);
  \draw [-](2,1) to  (2,-1);
  \draw [-](0,1) to  (2,1);
  \draw [-](0,-1) to  (-2,-1);
  \draw [-](-1-0.15,-1-0.25) to   (-1+0.15,-1+0.25);

  \node[below]   at(0.6,-1)      {$\bar b_1$};
  \node[below]   at(0,-1)      {$0$};
  \node[above]   at(0,1)      {$w$};
  \node[above]   at(-2,-1)      {$\ve[\iota]$};
  \node[right]   at(2,1)      {$y$};
  \node[right]   at(2,-1)      {$x$};
  \node[right]   at(2,0)      {$m$};

  \node   at(1,-2.25)      {\Large $B^{\ssc[1],m}$};
  \end{scope}

 \begin{scope}[shift={(5,0)},rotate=0]
 \draw [-](0+0.03,-1-0.2) to (0+0.03,-1+0.2);
 \draw [-](0+0.5-0.03,-1-0.2) to   (0+0.5-0.03,-1+0.2);
  \fill (0,1) circle (3pt);
  \fill (0,-1) circle (3pt);
  \fill (0,-0.5) circle (3pt);
  \fill (2,1) circle (3pt);
  \fill (2,-1) circle (3pt);
  \draw [-](0.5,-1) to  (2,-1);
  \draw [-](0,1) to  (0,-1);
  \draw [-](2,1) to  (2,-1);
  \draw [-](0,1) to  (2,1);

  \node[below]   at(0.6,-1)      {$\bar b_1$};
  \node[below]   at(0,-1)      {$0$};
  \node[above]   at(0,1)      {$w$};
  \node[left]   at(0,-0.5)      {$\ve[\iota]$};
  \node[right]   at(2,1)      {$y$};
  \node[right]   at(2,-1)      {$x$};
  \node[right]   at(2,0)      {$m$};
  \node   at(1,-2.25)      {\Large $B^{\ssc[2],m}$};
  \end{scope}

 \begin{scope}[shift={(10.5,0)},rotate=0]
 \draw [-](-1,-1) to  (0,0.5);
 \draw [-](0,0.5) to  (0,-1);

 \draw [-](-0.2,-0.35-0.15) to   (+0.2,-0.35+0.15);

 \draw [-](-0.6+0.15,-0.35-0.15) to  (-0.6-0.15,-0.35+0.15);

  \fill (-1,-1) circle (3pt);
  \fill (0,-1) circle (3pt);
  \node[below]   at(0,-1)      {$0$};
  \node[below]   at(-1,-1)      {$e_\iota$};
 \draw [-](0+0.03,-1-0.2) to (0+0.03,-1+0.2);
 \draw [-](0+0.5-0.03,-1-0.2) to   (0+0.5-0.03,-1+0.2);
  \node[above]   at(0-0.3,0.5)      {$u$};
  \fill (0,0.5) circle (3pt);
  \fill (1,1.2) circle (3pt);
  \fill (2,1) circle (3pt);
  \fill (2,-1) circle (3pt);

 \draw [-](0,0.5) to  (1,1.2);
 \draw [-](1,1.2) to  (2,1);
 \draw [-](2,-1) to  (2,1);
 \draw [-](0.5,-1) to  (2,-1);
  \node[right]   at(2,1)      {$y$};
  \node[right]   at(2,-1)      {$x$};
  \node[above]   at(1,1.2)      {$w$};
  \node[right]   at(2,0)      {$m$};
  \node[below]   at(0.6,-1)      {$\bar b_1$};
  \node   at(1,-2.25)      {\Large $B^{\ssc[3],m}$};
  \end{scope}

 \begin{scope}[shift={(15.5,0)},rotate=0]
 \draw [-](0,1) to  (0,-1);
  \node[left]   at(0,-1)      {$0$};
 \draw [-](0+0.03,-1-0.2) to (0+0.03,-1+0.2);
 \draw [-](0+0.5-0.03,-1-0.2) to   (0+0.5-0.03,-1+0.2);
 \node[below]   at(0.5,-1)      {$\tb_1=e_\iota$};
  \node[above]   at(0,1)      {$w$};
  \fill (0,1) circle (3pt);
   \fill (2,1) circle (3pt);
  \fill (2,-1) circle (3pt);
  \draw [-](0,1) to  (2,1);
  \draw [-](2,-1) to  (2,1);
  \draw [-](0.5,-1) to  (2,-1);
  \node[right]   at(2,1)      {$y$};
  \node[right]   at(2,-1)      {$x$};
  \node[right]   at(2,0)      {$m$};
  \node   at(1,-2.25)      {\Large $B^{\ssc[4],m}$};
  \end{scope}

 \begin{scope}[shift={(-5,-4)},rotate=0,scale=0.6]
 \draw [-](4,0) to  (5,1);
 \draw [-](4,0) to  (5,-1);
 \draw [-](5,1) to  (5,-1);
 \draw [-](4.5-0.2,-0.5-0.2) to (4.5+0.2,-0.5+0.2);

  \fill (4,0) circle (3pt);
  \node[left]   at(4,0)      {$0$};
 \draw [-](5+0.03,-1-0.2) to (5+0.03,-1+0.2);
 \draw [-](5+0.5-0.03,-1-0.2) to   (5+0.5-0.03,-1+0.2);
 \node[below]   at(5.5,-1)      {\small $e_\iota$};
 \node[left]   at(5.1,-1.2)      {\small  $\bb_1$};

  \node[above]   at(5,1)      {\small $v$};
  \fill (5,1) circle (3pt);
  \fill (6,1.2) circle (3pt);
  \fill (7,1) circle (3pt);
  \fill (7,-1) circle (3pt);

 \draw [-](5,1) to  (6,1.2);
 \draw [-](6,1.2) to  (7,1);
 \draw [-](5.5,-1) to  (7,-1);
 \draw [-](7,1) to  (7,-1);

  \node[right]   at(7,1)      {\small $y$};
  \node[right]   at(7,0)      {\small $m$};
  \node[right]   at(7,-1)      {\small $x$};
  \node[above]   at(6,1.2)      {\small $w$};
  \node   at(6,-2.5)      {\Large $B^{\ssc[5],m}$};
  \end{scope}

 \begin{scope}[shift={(-1,-4)},rotate=0,scale=0.6]
 \draw [-](4,0) to  (5,1);
 \draw [-](4,0) to  (5,-1);
 \draw [-](3,0) to  (4,0);
 \draw [-](5,1) to  (5,-1);
 \draw [-](4.5-0.2,-0.5-0.2) to (4.5+0.2,-0.5+0.2);

  \fill (3,0) circle (3pt);
  \fill (4,0) circle (3pt);
  \node[above]   at(4,0)      {$0$};
  \node[below]   at(3,0)      {$e_\iota$};
 \draw [-](5+0.03,-1-0.2) to (5+0.03,-1+0.2);
 \draw [-](5+0.5-0.03,-1-0.2) to   (5+0.5-0.03,-1+0.2);
  \node[above]   at(5,1)      {$v$};
  \fill (5,1) circle (3pt);
  \fill (6,1.2) circle (3pt);
  \fill (7,1) circle (3pt);
  \fill (7,-1) circle (3pt);

 \draw [-](5,1) to  (6,1.2);
 \draw [-](6,1.2) to  (7,1);
 \draw [-](5.5,-1) to  (7,-1);
 \draw [-](7,1) to  (7,-1);

  \node[left]   at(5.1,-1.2)      {$\underline b_1$};
  \node[right]   at(7,1)      {$y$};
  \node[right]   at(7,-1)      {$x$};
  \node[above]   at(6,1.2)      {$w$};

  \node[right]   at(7,0)      {$m$};
  \node   at(6,-2.5)      {\Large $B^{\ssc[6],m}$};
  \end{scope}

 \begin{scope}[shift={(2.75,-4)},rotate=0,scale=0.6]

  \node[above]   at(4,0)      {$0$};
  \node[above]   at(3,0)      {$e_\iota$};
 \draw [-](5+0.03,-0.2) to (5+0.03,+0.2);
 \draw [-](5+0.5-0.03,-0.2) to   (5+0.5-0.03,+0.2);

 \draw [-](4.5-0.15,-0.3-0.2) to (4.5+0.15,-0.3+0.2);
 \draw [-](4.5-0.15,+0.3-0.2) to (4.5+0.15,+0.3+0.2);

  \fill (3,0) circle (3pt);
  \fill (4,0) circle (3pt);
  \fill (5,0) circle (3pt);
  \fill (6,1.2) circle (3pt);
  \fill (7,1) circle (3pt);
  \fill (7,-1) circle (3pt);

 \draw [-](5,0) to [out=90,in=180] (6,1.2);
 \draw [-](4,0) to [out=90,in=90] (5,0);
 \draw [-](4,0) to [out=270,in=270] (5,0);
 \draw [-](3,0) to [out=270,in=270] (5,0);
 \draw [-](6,1.2) to  (7,1);
 \draw [-](5.5,0) to  (7,-1);
 \draw [-](7,1) to  (7,-1);

  \node[below]   at(5.1,-0.3)      {$u=\underline b_1=v$};
  \node[right]   at(7,1)      {$y$};
  \node[right]   at(7,-1)      {$x$};
  \node[above]   at(6,1.2)      {$w$};

  \node[right]   at(7,0)      {$m$};
  \node   at(6,-2.5)      {\Large $B^{\ssc[7],m}$};
  \end{scope}

 \begin{scope}[shift={(6.25,-4)},rotate=0,scale=0.6]
 \draw [-](4,0) to  (5,1);
 \draw [-](4,0) to   (5,-1);
 \draw [-](3,0) to  [out=270,in=210] (5,-1);
 \draw [-](5,1) to  (5,-1);
  \fill (3,0) circle (3pt);
  \fill (4,0) circle (3pt);
  \node[above]   at(4,0)      {$0$};
  \node[above]   at(3,0)      {$e_\iota$};
 \draw [-](5+0.03,-1-0.2) to (5+0.03,-1+0.2);
 \draw [-](5+0.5-0.03,-1-0.2) to   (5+0.5-0.03,-1+0.2);
  \node[above]   at(5,1)      {$v$};
  \fill (5,1) circle (3pt);
  \fill (6,1.2) circle (3pt);
  \fill (7,1) circle (3pt);
  \fill (7,-1) circle (3pt);

 \draw [-](5,1) to  (6,1.2);
 \draw [-](6,1.2) to  (7,1);
 \draw [-](5.5,-1) to  (7,-1);
 \draw [-](7,1) to  (7,-1);
 \draw [-](4.5-0.2,-0.5-0.2) to (4.5+0.2,-0.5+0.2);
 \draw [-](5-0.2,-0-0.2) to (5+0.2,-0+0.2);

  \node[below]   at(5,-1)      {$\bb_1=u$};
  \node[right]   at(7,1)      {$y$};
  \node[right]   at(7,-1)      {$x$};
  \node[above]   at(6,1.2)      {$w$};

  \node[right]   at(7,0)      {$m$};
  \node   at(6,-2.5)      {\Large $B^{\ssc[8],m}$};
  \end{scope}

 \begin{scope}[shift={(10,-4)},rotate=0,scale=0.6]
 \draw [-](4,0) to  (5,1);
 \draw [-](4,0) to  (5,-1);
 \draw [-](3,0) to  [out=90,in=150] (5,1);
 \draw [-](5,1) to  (5,-1);
  \fill (3,0) circle (3pt);
  \fill (4,0) circle (3pt);
  \node[below]   at(4,0)      {$0$};
  \node[below]   at(3,0)      {$e_\iota$};
 \draw [-](5+0.03,-1-0.2) to (5+0.03,-1+0.2);
 \draw [-](5+0.5-0.03,-1-0.2) to   (5+0.5-0.03,-1+0.2);
  \node[above]   at(4,1)      {$v=u$};
  \fill (5,1) circle (3pt);
  \fill (6,1.2) circle (3pt);
  \fill (7,1) circle (3pt);
  \fill (7,-1) circle (3pt);

 \draw [-](5,1) to  (6,1.2);
 \draw [-](6,1.2) to  (7,1);
 \draw [-](5.5,-1) to  (7,-1);
 \draw [-](7,1) to  (7,-1);
 \draw [-](4.5-0.2,-0.5-0.2) to (4.5+0.2,-0.5+0.2);
 \draw [-](5-0.2,-0-0.2) to (5+0.2,-0+0.2);

  \node[below]   at(4.8,-0.8)      {$\bb_1$};
  \node[right]   at(7,1)      {$y$};
  \node[right]   at(7,-1)      {$x$};
  \node[above]   at(6,1.2)      {$w$};

  \node[right]   at(7,0)      {$m$};
  \node   at(6,-2.5)      {\Large $B^{\ssc[9],m}$};
  \end{scope}

 \begin{scope}[shift={(13.5,-4)},rotate=0,scale=0.6]
  \node[above]   at(4,0.2)      {$0$};
  \node[above]   at(3,0)      {$e_\iota$};
 \draw [-](5+0.03,-0.2) to (5+0.03,+0.2);
 \draw [-](5+0.5-0.03,-0.2) to   (5+0.5-0.03,+0.2);

  \fill (3,0) circle (3pt);
  \fill (4,0) circle (3pt);
  \fill (5,0) circle (3pt);
  \fill (6,1.2) circle (3pt);
  \fill (7,1) circle (3pt);
  \fill (7,-1) circle (3pt);
  \fill (4.5,-0.3) circle (3pt);

 \draw [-](4,0) to [out=90,in=180] (6,1.2);
 \draw [-](4,0) to [out=90,in=90] (5,0);
 \draw [-](4,0) to [out=270,in=270] (5,0);
 \draw [-](3,0) to [out=270,in=240] (4.5,-0.3);
 \draw [-](6,1.2) to  (7,1);
 \draw [-](5.5,0) to  (7,-1);
 \draw [-](7,1) to  (7,-1);
 \draw [-](4.5-0.15,+0.3-0.2) to (4.5+0.15,+0.3+0.2);
 \draw [-](4.75-0.1,-0.2+0.15) to (4.75+0.15,-0.2-0.2);
 \draw [-](4.25+0.1,-0.2+0.15) to (4.25-0.15,-0.2-0.2);

  \node[below]   at(5.2,-0.1) {$\underline b_1$};
  \node[below]   at(4.5,-0.4) {$u$};
  \node[right]   at(7,1)      {$y$};
  \node[right]   at(7,-1)     {$x$};
  \node[above]   at(6,1.2)    {$w$};

  \node[right]   at(7,0)      {$m$};
  \node   at(6,-2.5)      {\Large $B^{\ssc[10],m}$};
  \end{scope}

 \begin{scope}[shift={(-5,-7.5)},rotate=0,scale=0.6]
 \draw [-](4,0) to  (5,1);
 \draw [-](4,0) to  (5,-1);
 \draw [-](4.5,0.5) to  [out=135,in=90]   (3,0);
  \draw [-](5,1) to  (5,-1);
  \fill (4,0) circle (3pt);
  \fill (4.5,0.5) circle (3pt);
  \node[below]   at(4,0)      {$0$};
  \node[above]   at(4.3,0.5)      {$u$};
  \node[below]   at(3,0)      {$e_\iota$};
 \draw [-](5+0.03,-1-0.2) to (5+0.03,-1+0.2);
 \draw [-](5+0.5-0.03,-1-0.2) to   (5+0.5-0.03,-1+0.2);
  \node[above]   at(5,1)      {$v$};
  \fill (5,1) circle (3pt);
  \fill (6,1.2) circle (3pt);
  \fill (7,1) circle (3pt);
  \fill (7,-1) circle (3pt);
  \fill (3,0) circle (3pt);

 \draw [-](4.5-0.2,-0.5-0.2) to (4.5+0.2,-0.5+0.2);
 \draw [-](5-0.2,-0-0.2) to (5+0.2,-0+0.2);

 \draw [-](4.75+0.15,0.75-0.15) to (4.75-0.15,0.75+0.15);
 \draw [-](4.25+0.15,0.25-0.15) to (4.25-0.15,0.25+0.15);

 \draw [-](5,1) to  (6,1.2);
 \draw [-](6,1.2) to  (7,1);
 \draw [-](5.5,-1) to  (7,-1);
 \draw [-](7,-1) to  (7,1);
  \node[right]   at(7,0)      {$m$};
  \node[right]   at(7,1)      {$y$};
  \node[right]   at(7,-1)      {$x$};
  \node[above]   at(6,1.2)      {$w$};
  \node   at(6,-2.5)      {\Large $B^{\ssc[11],m}$};
  \end{scope}

 \begin{scope}[shift={(-1,-7.5)},rotate=0,scale=0.6]
 \draw [-](4,0) to  (5,1);
 \draw [-](4,0) to  (5,-1);
 \draw [-](4.5,-0.5) to  [out=225,in=270]   (3,0);
  \draw [-](5,1) to  (5,-1);
  \fill (4,0) circle (3pt);
  \fill (4.5,-0.5) circle (3pt);
  \node[above]   at(4,0)      {$0$};
  \node[below]   at(4.5,-0.5)      {$u$};
  \node[above]   at(3,0)      {$e_\iota$};
 \draw [-](5+0.03,-1-0.2) to (5+0.03,-1+0.2);
 \draw [-](5+0.5-0.03,-1-0.2) to   (5+0.5-0.03,-1+0.2);
  \node[above]   at(5,1)      {$v$};
  \fill (5,1) circle (3pt);
  \fill (6,1.2) circle (3pt);
  \fill (7,1) circle (3pt);
  \fill (7,-1) circle (3pt);
  \fill (3,0) circle (3pt);

 \draw [-](5,1) to  (6,1.2);
 \draw [-](6,1.2) to  (7,1);
 \draw [-](7,-1) to  (7,1);
 \draw [-](5.5,-1) to  (7,-1);

 \draw [-](4.25-0.15,-0.25-0.15) to (4.25+0.15,-0.25+0.15);
 \draw [-](4.75-0.15,-0.75-0.15) to (4.75+0.15,-0.75+0.15);

 \draw [-](5-0.2,-0-0.2) to (5+0.2,-0+0.2);

 \draw [-](4.5+0.2,0.5-0.2) to (4.5-0.2,0.5+0.2);

  \node[right]   at(7,1)      {$y$};
  \node[right]   at(7,0)      {$m$};

  \node[right]   at(7,-1)      {$x$};
  \node[above]   at(6,1.2)      {$w$};
  \node   at(6,-2.5)      {\Large $B^{\ssc[12],m}$};
  \end{scope}

 \begin{scope}[shift={(2.75,-7.5)},rotate=0,scale=0.6]
 \draw [-](4,0) to  (5,1);
 \draw [-](4,0) to  (5,-1);
 \draw [-](3,0) to  [out=7  0,in=150] (5,0);
 \draw [-](5,1) to  (5,-1);
  \fill (3,0) circle (3pt);
  \fill (4,0) circle (3pt);
  \node[below]   at(4,0)      {$0$};
  \node[below]   at(3,0)      {$e_\iota$};
 \draw [-](5+0.03,-1-0.2) to (5+0.03,-1+0.2);
 \draw [-](5+0.5-0.03,-1-0.2) to   (5+0.5-0.03,-1+0.2);
  \node[above]   at(3.6,1)      {$v=u$};
  \fill (5,1) circle (3pt);
  \fill (5,0) circle (3pt);
  \fill (6,1.2) circle (3pt);
  \fill (7,1) circle (3pt);
  \fill (7,-1) circle (3pt);

 \draw [-](5,1) to  (6,1.2);
 \draw [-](6,1.2) to  (7,1);
 \draw [-](5.5,-1) to  (7,-1);
 \draw [-](7,1) to  (7,-1);
 \draw [-](4.5-0.2,-0.5-0.2) to (4.5+0.2,-0.5+0.2);
 \draw [-](5-0.15,-0.5-0.15) to (5+0.15,-0.5+0.15);
 \draw [-](5-0.15,0.5-0.15) to (5+0.15,0.5+0.15);

  \node[below]   at(4.8,-0.7)      {$\bb_1$};
  \node[right]   at(7,1)      {$y$};
  \node[right]   at(7,-1)      {$x$};
  \node[above]   at(6,1.2)      {$w$};

  \node[right]   at(7,0)      {$m$};
  \node   at(6,-2.5)      {\Large $B^{\ssc[13],m}$};
  \end{scope}

 \begin{scope}[shift={(6.25,-7.5)},rotate=0,scale=0.6]
  \node[below]   at(4,0)      {$0$};
  \node[below]   at(3,0)      {$e_\iota$};
 \draw [-](5+0.03,-0.2) to (5+0.03,+0.2);
 \draw [-](5+0.5-0.03,-0.2) to   (5+0.5-0.03,+0.2);

  \fill (3,0) circle (3pt);
  \fill (4,0) circle (3pt);
  \fill (5,0) circle (3pt);
  \fill (6,1.2) circle (3pt);
  \fill (5.4,1) circle (3pt);
  \fill (7,1) circle (3pt);
  \fill (7,-1) circle (3pt);
 \draw [-](4.5-0.15,-0.3-0.15) to (4.5+0.15,-0.3+0.15);
 \draw [-](4.5-0.15,+0.3-0.15) to (4.5+0.15,+0.3+0.15);

 \draw [-](5.2+0.2,+0.6-0.2) to (5.2-0.2,+0.6+0.15);

 \draw [-](5,0) to (5.4,1);
 \draw [-](5.4,1) to (6,1.2);
 \draw [-](4,0) to [out=90,in=90] (5,0);
 \draw [-](4,0) to [out=270,in=270] (5,0);
 \draw [-](3,0) to [out=90,in=180] (5.4,1);
 \draw [-](6,1.2) to  (7,1);
 \draw [-](5.5,0) to  (7,-1);
 \draw [-](7,1) to  (7,-1);

  \node[below]   at(5.2,-0.2)      {$\underline b_1=v$};
  \node[right]   at(7,1)      {$y$};
  \node[right]   at(7,-1)      {$x$};
  \node[above]   at(6,1.2)      {$w$};
  \node[above]   at(5.4,1)      {$u$};

  \node[right]   at(7,0)      {$m$};
  \node   at(6,-2.5)      {\Large $B^{\ssc[14],m}$};
  \end{scope}

 \begin{scope}[shift={(10,-7.5)},rotate=0,scale=0.6]
  \node[below]   at(4,-0.2)      {$0$};
  \node[below]   at(3,0)      {$e_\iota$};
 \draw [-](5+0.03,-0.2) to (5+0.03,+0.2);
 \draw [-](5+0.5-0.03,-0.2) to   (5+0.5-0.03,+0.2);

  \fill (3,0) circle (3pt);
  \fill (4,0) circle (3pt);
  \fill (5,0) circle (3pt);
  \fill (6,1.2) circle (3pt);
  \fill (7,1) circle (3pt);
  \fill (7,-1) circle (3pt);
  \fill (5,1) circle (3pt);
 \draw [-](4.5-0.15,-0.3-0.15) to (4.5+0.15,-0.3+0.15);
 \draw [-](4.5-0.15,+0.3-0.15) to (4.5+0.15,+0.3+0.15);

 \draw [-](4.5+0.2,+0.7-0.1) to (4.5-0.3,+0.7+0.1);

 \draw [-](5,1) to (6,1.2);
 \draw [-](4,0) to [out=90,in=210] (5,1);
 \draw [-](4,0) to [out=90,in=90] (5,0);
 \draw [-](4,0) to [out=270,in=270] (5,0);
 \draw [-](3,0) to [out=90,in=180] (5,1);
 \draw [-](6,1.2) to  (7,1);
 \draw [-](5.5,0) to  (7,-1);
 \draw [-](7,1) to  (7,-1);

  \node[below]   at(5.2,-0.1) {$\underline b_1$};
  \node[above]   at(5,1) {$u$};
  \node[right]   at(7,1)      {$y$};
  \node[right]   at(7,-1)     {$x$};
  \node[above]   at(6,1.2)    {$w$};

  \node[right]   at(7,0)      {$m$};
  \node   at(6,-2.5)      {\Large $B^{\ssc[15],m}$};
  \end{scope}

 \begin{scope}[shift={(13.5,-7.5)},rotate=0,scale=0.6]
 \draw [-](4,0) to  (5,1);
 \draw [-](4,0) to  (5,-1);
 \draw [-](3,0) to  [out=90,in=150] (5.5,1.1);
 \draw [-](5,1) to  (5,-1);
  \fill (3,0) circle (3pt);
  \fill (4,0) circle (3pt);
  \node[below]   at(4,0)      {$0$};
  \node[below]   at(3,0)      {$e_\iota$};
 \draw [-](5+0.03,-1-0.2) to (5+0.03,-1+0.2);
 \draw [-](5+0.5-0.03,-1-0.2) to   (5+0.5-0.03,-1+0.2);
  \node[left]   at(5,1)      {$v$};
  \node[below]   at(5.5,1.1)      {$u$};
  \fill (5,1) circle (3pt);
  \fill (5.5,1.1) circle (3pt);
  \fill (6,1.2) circle (3pt);
  \fill (7,1) circle (3pt);
  \fill (7,-1) circle (3pt);

 \draw [-](4.5-0.15,-0.5-0.15) to (4.5+0.15,-0.5+0.15);
 \draw [-](5-0.2,-0-0.2) to (5+0.2,-0+0.2);
 \draw [-](4.5+0.2,0.5-0.2) to (4.5-0.2,0.5+0.2);

 \draw [-](5.25+0.1,+1.05-0.3) to (5.25-0.1,+1.05+0.2);

 \draw [-](5,1) to  (6,1.2);
 \draw [-](6,1.2) to  (7,1);
 \draw [-](5.5,-1) to  (7,-1);
 \draw [-](7,1) to  (7,-1);

  \node[below]   at(4.8,-0.7)      {$\bb_1$};
  \node[right]   at(7,1)      {$y$};
  \node[right]   at(7,-1)      {$x$};
  \node[above]   at(6,1.2)      {$w$};

  \node[right]   at(7,0)      {$m$};
  \node   at(6,-2.5)      {\Large $B^{\ssc[16],m}$};
  \end{scope}

\end{tikzpicture}}

\iflongversion
\section{Formal definition of building blocks}
\label{sec-formaldefinition}
In this section we provide the formal definition of the building blocks described in Section \ref{secBuildingBlocks}.
We begin with the blocks used in the bound without spatial weights as shown in Figure \ref{TreeXiFourDeomposedStructure}
This means, we define intermediates blocks $A^{m,0}(0,v,x,x)$ and  as initial piece of the diagram $P^{\ssc[1],m}, P^{\ssc[1],\iota,m}$ for $\Xi^{\ssc[N]}_z$ and $\Xi^{\ssc[N],\iota}_z$, respectively.
Then, we define the weighted analogue $C^{m,l}$ and $\Delta^{\text {start},m}$
$\Delta^{\text {iota},{\sss I},m}$, $\Delta^{\text {iota},{\sss I},m}$.

The blocks are designed to bound the NoBLE coefficients. We define them using as many avoidance/repulsiveness properties of the coefficients as possible, to create the sharpest numerical bounds. To give a visual example, we define $A^{0,-2}$, $A^{0,-1}$, $A^{0,0}$, $A^{0,1}$, $A^{0,-2}$ using Tables \ref{BoundTableAZeroZero}-\ref{BoundTableAZeroMinusTwo}. In these tables, we consider various different cases, we give a visual representation of the bounding diagram for the respective case, and give the corresponding bounding diagrams in formulas in the right column. Let us illustrate this by the example $A^{0,-1}$, which we define using Table \ref{BoundTableAZeroMinusOne} as the sum of the elements in the right-most column:
	\begin{align*}
	A^{0,-1}(0,v,x,y)=\delta_{0,v}\Big(\diagRepulsiveLetter{T}_{\underline 1,\underline 1,2}(x,y,0)+
	\sum_w \big(\diagRepulsiveLetter{S}_{\underline 1,\underline 1,1,1}(x,y,w,0)+\diagRepulsiveLetter{S}_{2,\underline 1,0,0}(x,y,w,0)\big) \Big),
	\end{align*}
where $\delta_{0,v}$ denotes the Kronecker delta.
For our definitions we use translation variance to simplify the presentation, e.g., we use
	\begin{align*}
  	A^{m,l}(u,v,x,y)=A^{m,l}(0,v-u,x-u,y-u),
	\end{align*}
and we restrict to defining $A^{m,l}(0,v,x,y)$ for $m,l\in \{-2,-1,0,1,2\}$.

\subsection{Intermediate diagrams $A^{m,l}$ and $\bar{A}^{m,l}$}
Here we formally define the bounding diagrams $A^{m,l}$ and $\bar{A}^{m,l}$ as informally described in Section \ref{secBuildingBlocks}.
\paragraph{Definition of $A^{m,l}(0,v,x,y)$ for $m,l\in \{-2,-1,0,1,2\}$.}
Tables \ref{BoundTableAZeroZero}--\ref{BoundTableAZeroMinusTwo} define $A^{0,l}$ for all values of $l\in \{-2,-1,0,1,2\}$.

\threecolomntable
{Definition of $A^{0,0}(0,v,x,y)=\delta_{0,v}\delta_{x,y}A^{0,0}(0,0,x,x)$}{BoundTableAZeroZero}{
 $\begin{array}{c} x=0 \\ \Rightarrow w=0,\\ d_\omega(0,x)\geq 4 \end{array}$ &     \picAZeroZeroA[0.5]& $\sum_\iota \diagRepulsiveLetter{B}_{3,\underline 1}(\ve[\iota],0)$  \\ \cline{1-3}
 $\begin{array}{c} x\neq 0 \\w\in\{0,x\} \\ \Rightarrow d_\omega(0,x)\geq 2 \end{array}$ &     \picAZeroZeroB[0.5]& $2\diagRepulsiveLetter{B}_{2,1}(x,0)$  \\ \cline{1-3}
$\begin{array}{c} x\neq 0,\\ w\nin\{0,x\} \end{array}$&     \picAZeroZeroC[0.6]& $\sum_w \diagRepulsiveLetter{T}_{1,1,1}(x,w,0)$  \\ \cline{1-3}
}

\threecolomntable
{Definition of $A^{0,1}(0,v,x,y)=\delta_{0,v}\indic{y\neq 0}A^{0,1}(0,0,x,y)$}{BoundTableAZeroOne}{
 $\begin{array}{c} x=0 \\ \Rightarrow w=0,\\ d_\omega(0,x)\geq 4 \end{array}$ &     \picAZeroOneA[0.5]& $\diagRepulsiveLetter{B}_{3,\underline 1}(y,0)$  \\ \cline{1-3}
 $\begin{array}{c} x\neq 0 \\w\in\{0,x\} \\ \Rightarrow d_\omega(0,x)\geq 2 \end{array}$ &     \picAZeroOneB[0.5]& $2\diagRepulsiveLetter{T}_{1,\underline 1,1}(y,x,0)$  \\ \cline{1-3}
 $\begin{array}{c} x\neq 0,\\ w\nin\{0,x\} \end{array}$&  \picAZeroOneC[0.5]& $\sum_w \diagRepulsiveLetter{S}_{1,\underline 1,1,1}(y,x,w,0)$  \\ \cline{1-3}
}

\threecolomntable
{Definition of $A^{0,2}(0,v,x,y)=\delta_{0,v}\indic{y\neq 0}A^{0,2}(0,0,x,y)$}{BoundTableAZeroTwo}{
 $\begin{array}{c} x=0 \\ \Rightarrow w=0,\\ d_\omega(0,x)\geq 4 \end{array}$ &     \picAZeroTwoA[0.5]& $\diagRepulsiveLetter{B}_{1,2}(y,0)$  \\ \cline{1-3}
 $\begin{array}{c} x\neq 0 \\w\in\{0,x\} \\ \Rightarrow d_\omega(0,x)\geq 2 \end{array}$ &     \picAZeroTwoB[0.5]& $2\diagRepulsiveLetter{T}_{1,2,1}(y,x,0)$  \\ \cline{1-3}
 $\begin{array}{c} x\neq 0,\\ w\nin\{0,x\} \end{array}$&  \picAZeroTwoC[0.5]& $\sum_w \diagRepulsiveLetter{S}_{1,2,1,1}(y,x,w,0)$  \\ \cline{1-3}
}

\threecolomntable
{Definition of $A^{0,-1}(0,v,x,y)=\delta_{0,v}A^{0,-1}(0,0,x,y)$}{BoundTableAZeroMinusOne}{
 $\begin{array}{c} y=w,d_\omega(0,x)=1 \\ \Rightarrow w\neq 0\end{array}$ &     \picAZeroMinusOneA[0.5]& $\diagRepulsiveLetter{T}_{\underline 1,\underline 1,2}(x,y,0)$  \\ \cline{1-3}
 $\begin{array}{c} y\neq w,d_\omega(0,x)=1 \\ \Rightarrow w\neq 0 \end{array}$ &     \picAZeroMinusOneB[0.5]& $ \sum_w \diagRepulsiveLetter{S}_{\underline 1,\underline 1,1,1}(x,y,w,0)$  \\ \cline{1-3}
 $\begin{array}{c} d_\omega(0,x)\geq 2 \\ w\nin\{0,y\} \end{array}$&  \picAZeroMinusOneC[0.5]& $\sum_w \diagRepulsiveLetter{S}_{2,\underline 1,0,0}(x,y,w,0)$  \\ \cline{1-3}
}

\threecolomntable
{Definition of $A^{0,-2}(0,v,x,y)=\delta_{0,v}A^{0,-2}(0,0,x,y)$}{BoundTableAZeroMinusTwo}{
 $\begin{array}{c} d_\omega(0,x)=1 \\ \Rightarrow w\neq 0\end{array}$ &     \picAZeroMinusTwoA[0.5]& $\sum_w\diagRepulsiveLetter{S}_{\underline 1,2,0,1}(x,y,w,0)$  \\ \cline{1-3}
 $\begin{array}{c} 0=w,d_\omega(0,x)\geq 2 \end{array}$ &     \picAZeroMinusTwoB[0.5]& $ \diagRepulsiveLetter{T}_{2,2,0}(x,y,0)$  \\ \cline{1-3}
 $\begin{array}{c} y\neq 0,d_\omega(0,x)\geq 2 \end{array}$&  \picAZeroMinusTwoC[0.5]& $\sum_w \diagRepulsiveLetter{S}_{2,2,0,1}(x,y,w,0)$  \\ \cline{1-3}
}

We next define $A^{-1,0}, A^{-1,-1}$ and $A^{-1,-2}$ as
	\begin{align*}
	A^{-1,0}(0,v,x,y)&=2dD(v)\delta_{x,y}\left(\sum_w \diagRepulsiveLetter{T}_{\underline 1,1,0}(x,w,0)+\diagRepulsiveLetter{T}_{2,0,0}(x,w,v)\right),\\
	A^{-1,-1}(0,v,x,y)&=2dD(v)\sum_w \diagRepulsiveLetter{S}_{1,\underline 1,0,0}(x,y,w,v),\\
	A^{-1,-2}(0,v,x,y)&=2dD(v)\sum_w \diagRepulsiveLetter{S}_{1,2,0,0}(x,y,w,v),
	\end{align*}
and $A^{-2,0}, A^{-2,-1}$ and $A^{-2,-2}$ as
	\begin{align*}
	A^{-2,0}(0,v,x,y)&=\delta_{x,y}\left(\sum_w \diagRepulsiveLetter{T}_{\underline 1,1,0}(x,w,v)+\diagRepulsiveLetter{T}_{2,0,0}(x,w,v)\right),\\
	A^{-2,-1}(0,v,x,y)&=\sum_w \diagRepulsiveLetter{S}_{1,\underline 1,0,0}(x,y,w,v),\\
	A^{-2,-2}(0,v,x,y)&=\sum_w \diagRepulsiveLetter{S}_{1,2,0,0}(x,y,w,v).
	\end{align*}
We define the remaining $A^{m,l}$ by the symmetries, for $m\in \{1,2\}$,
	\begin{align*}
	A^{m,0}(0,v,x,y)=A^{-m,0}(0,v,x,y),
	\end{align*}
and, for $m, l\in\{1,2\}$,
	\begin{align*}
	 A^{m,l}(0,v,x,y)=A^{m,-l}(0,v,y,x).
	\end{align*}
	
\paragraph{Definition of $\bar{A}^{m,l}(0,v,x,y)$ for $m,l\in \{-2,-1,0,1,2\}$.}
Next, we define the double open diagram $\bar{A}^{m,l}(0,v,x,y)$, to which the  connection $x$ to $y$, as present in $A^{m,l}(0,v,x,y)$, does not contribute.
For $l=0$, this connection is not present in $A^{m,l}(0,v,x,y)$ either, as $x=y$, so we simply define, for $m\in\{-2,-1,0, 1,2\}$,
	\begin{align*}
	\bar A^{m,0}(0,v,x,y)= A^{m,0}(0,v,x,y),
	\end{align*}
and use symmetry to define
	\begin{align*}
	\bar A^{0,m}(0,v,x,y)= A^{m,0}(x,y,0,v).
	\end{align*}
Further, we define $\bar A^{-1,-1}, \bar A^{-1,-2}$ and $\bar A^{-2,-2}$ as
	\begin{align*}
	\bar A^{-1,-1}(0,v,x,y)=&(2d)^2D(v)D(x-y) \tilde G_{1,z}(x)\sum_w\diagRepulsiveLetter{B}_{0,0}(w-v,x-v),\\
	\bar A^{-1,-2}(0,v,x,y)=& 2dD(v) \tilde G_{1,z}(x)\sum_w\diagRepulsiveLetter{B}_{0,0}(w-v,x-v),\\
	\bar A^{-2,-2}(0,v,x,y)=&\tilde G_{1,z}(x)\sum_w\diagRepulsiveLetter{B}_{0,0}(w-v,x-v),
	\end{align*}
and define the remaining diagrams by
	\begin{align*}
	\bar A^{m,l}(0,v,x,y)&=\bar A^{m,-l}(0,v,y,x), \qquad  \bar A^{m,l}(0,v,x,y)=\bar A^{-m,l}(0,v,y,x),\\
    \bar A^{m,l}(0,v,x,y)&=\bar A^{l,m}(x,y,0,v).
	\end{align*}
for $m,l\in\{-2,-1,1,2\}$ for which $\bar A^{m,l}$ was not yet defined.

\subsection{Initial diagrams of $\Xi$ and $\Xi^{\iota}$}
\label{Appendix-def-Initial-Xi-XiIota}
In this section, we define the bounds on the initial diagrams of $\Xi$ and $\Xi^{\iota}$.

\subsubsection{Initial diagrams of $\Xi$}
\label{Appendix-def-Initial-Xi}
Here we define the left-most part of the bounding diagram, $P^{\ssc[1],m}$, for $\Xi^{\ssc[N]}$. This diagram represent the first piece of the rib-/sausage-walk, until the first self-intersection. For LTs we simply defined $P^{\ssc[1],m}(x,y)$ to be equal to $ A^{0,m}(0,0,x,y)$. The LAs can contain an additional triangle, created by the case $\bb_1\neq 0 $, as discussed in Section \ref{secBoundsOnePointF}
(recall also Definition \ref{defLASausagewalks}). To include this case, we define, for LAs, recalling also Definition \ref{def-repulsive-diamgrams},
	\begin{align*}
	S^{0}(x,y) =&\delta_{x,0}\delta_{y,0} + \delta_{x,y}\diagRepulsiveLetter{D}_{1}(x),\\
	S^{-1}(x,y) =& \delta_{y,0} \diagRepulsiveLetter{B}_{\underline 1,3}(x,0)+(1-\delta_{y,0})\diagRepulsiveLetter{T}_{1,\underline 1,1}(x,y,0),\\
	S^{-2}(x,y) =& \delta_{y,0} \diagRepulsiveLetter{D}_{2}(0)+(1-\delta_{y,0})\diagRepulsiveLetter{T}_{1,2,1}(x,y,0).
	\end{align*}
Using this notation we define
	\begin{align*}
	P^{\ssc[1],m}(x,y)=\sum_{l=-2}^0\sum_{u,v\in\Zd} S^{l}(u,v)A^{l,m}(u,v,x,y).
	\end{align*}
	
\subsubsection{Initial diagrams of $\Xi^{\iota}$}
\label{Appendix-def-Initial-XiIota}
Now we define $P^{\ssc[1],\iota,m}$ that bounds the left-most diagram of $\Xi^{\ssc[N],\iota}$. These bounds are also used for $\Pi^{\ssc[N],\iota,\kappa}$.
Due to the constraint that either $\ve[\iota]\in A_0$ or $\ve[\iota]=\tb_1$, this diagram can have many different forms.
For a clear definition we group them into $16$ groups that we denote by $B^{\ssc[M],m}$ for $M=1,\dots,16$, depicted in Figure \ref{BoundLA-Figure-initial},
and define
	\begin{align*}
	P^{\ssc[1],\iota,m}(x,y)=\sum_{M=1}^{16} B^{\ssc[M],m}(x,y).
	\end{align*}
For $M\geq 5$ the diagrams for LTs are trivial, i.e. $B^{\ssc[M],m}=0$ for $M\geq 5$, as they require a double connection from $0$ to $\bb_1$ in the first rib/sausage, which is not possible for LTs.

\begin{figure}[h!]
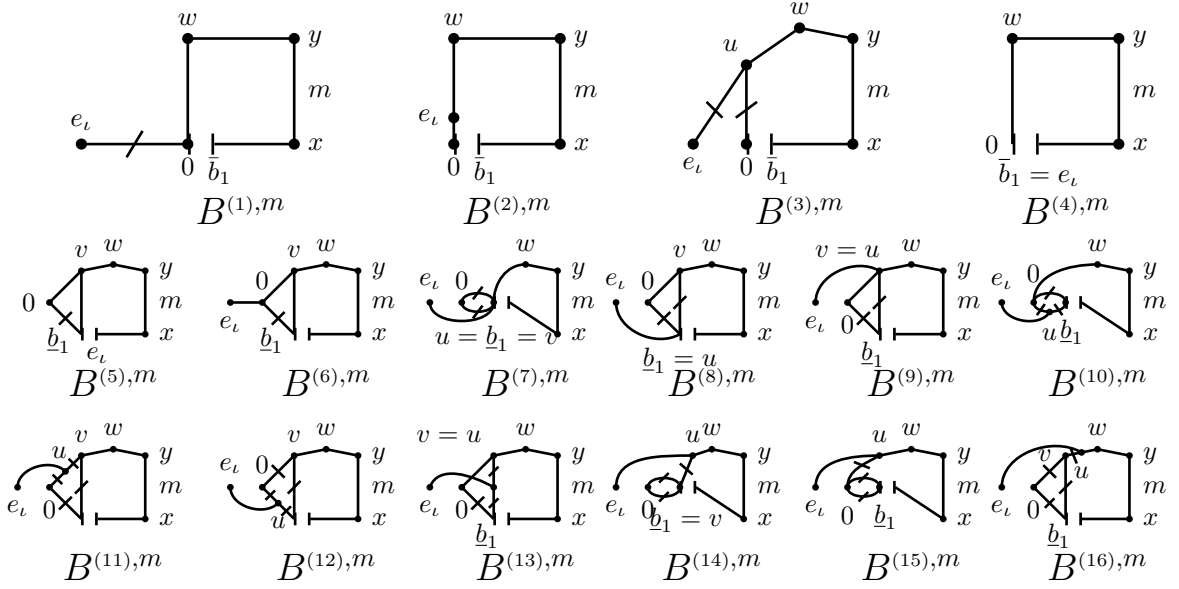

\begin{center}
{\Large
\picXilotaInitialStructure[0.7]
\\ }
\caption{
The shapes of the initial diagrams $\Xi^{\ssc[N],\iota}_z$ and $\Pi^{\ssc[N],\iota,\kappa }_z$ for $N\geq 2$. For $N=1$, these diagrams, for $x=y$ (or, equivalently, $m=0$) act as bounding diagrams for $\Xi^{\ssc[N],\iota}_z(x)$. Only the labeling for $m\leq 0$ is shown. In the diagram we sum over $u,v,w,\tb_1$, if they are present.
The lines without dash can collaps to a point, while the dashed lines include at least one step.
If $\bb_1$ is shown, then it is supposed to be unequal to zero.}
\label{BoundLA-Figure-initial}
\end{center}
\end{figure}

\noindent
In the following, we show that the sixteen $B^{\ssc[M],m}$ in Figure \ref{BoundLA-Figure-initial} capture all possible shapes. A sausage-/rib-walk contributing to $\Xi^{\ssc[N],\iota}_z$ and $\Pi^{\ssc[N],\iota,\kappa }_z$ needs to satisfy that either $\tb_1=\ve[\iota]$ or $\ve[\iota]\in A_0$. We bound the contribution to $\tb_1=\ve[\iota]$ using $B^{\ssc[4],m}$ and $B^{\ssc[5],m}$ by considering the case $\bb_1=0$ and $\bb_1\neq 0$, respectively.

The other fourteen diagrams bound contributions with $\ve[\iota]\in A_0$. For these diagrams we denote by $w$ the first intersection point and by $u$ the last point that all paths from $0 \conn w$ and $0 \conn \ve[\iota]$ have in common. The diagrams $B^{\ssc[1],m}-B^{\ssc[3],m}$ represent the case $\bb_1=0$, where they individually describe the cases where $u=0$, $u=\ve[\iota]$ and $u\nin\{0,\ve[\iota]\}$.

Thus, only the case $\ve[\iota]\in A_0$ with $\bb_1\neq 0$ remains to be considered. This contribution will be bounded by $B^{\ssc[6],m}-B^{\ssc[16],m}$ by splitting according to the location of $u$ and various other vertices involved.
In these diagrams, we consider six distinct vertices: $0$, $\ve[\iota]$, $\bb_1$, the first intersection point $w$, the last common vertex $u$, and the
last point $v$ that the paths $0\conn w$ and $0\conn \bb_1$ have in common.

We order the remaining eleven $B^{\ssc[M],m}$  by the locations of $u$.
$B^{\ssc[6],m}$-$B^{\ssc[9],m}$ bound the three cases $u\in\{0,v,\bb_1\}$, while
$B^{\ssc[7],m}$ bounds the special case where $u=\bb_1=v$.

In the other diagrams, we put $u$ on one of the four paths $0\conn v$, $0\conn \bb_1$, $\bb_1\conn v$ and $v\conn w$, but not the endpoints of these paths (as these have already been considered in $B^{\ssc[6],m}$-$B^{\ssc[9],m}$). To be precise

	\renewcommand{\labelitemi}{$\rhd$}
	\begin{itemize}
  	\item $B^{\ssc[10],m}$: $v=0$ with $u$ on the double connection between $0$ and $\bb_1$ with $u\nin\{0,\bb_1\}$;
  	\item $B^{\ssc[11],m}$: $u$ on connection $0\conn v$, for which $v\neq 0$ and $u\nin \{0,v\}$;
 	 \item $B^{\ssc[12],m}$: $u$ on connection $0\conn \bb_1$, for which $v\neq 0$ and $u\nin \{0,\bb_1\}$;
 	 \item $B^{\ssc[13],m}$: $u$ on connection $\bb_1\conn v$, for which $v\neq 0$ and $u\nin \{\bb_1,v\}$;
  	\item $B^{\ssc[14],m}$: $u$ on connection $v\conn w$ and $v=\bb_1$, for which $u\neq v$;
  	\item $B^{\ssc[15],m}$: $u$ on connection $v\conn w$ and $v=0$, for which $u\neq v$;
  	\item $B^{\ssc[16],m}$: $u$ on connection $v\conn w$ and $v\nin\{0,\bb_1\}$, for which $u\neq v$.
	\end{itemize}
The restrictions on the vertices involved (such as $u\nin \{0,v\}$ for $B^{\ssc[11],m}$) follow from the fact that $u\in \{0,\bb_1,v\}$ were already considered in $B^{\ssc[6],m}$-$B^{\ssc[9],m}$. We also split according to $v=0$ (bounded by $B^{\ssc[10],m}$ and $B^{\ssc[15],m}$) and the case where $v\neq 0$ (the other ones). This demonstrates that
$B^{\ssc[1],m}-B^{\ssc[16],m}$ really bound all possible contributions.

After this overview and validation that we have indeed considered all possible contributions, we define these diagrams formally. For $u=0$, we define, for $m=-2,\dots,2$,
	\begin{align*}
	B^{\ssc[1],m}(x,y)= \tilde G_z(\ve[\iota])A^{0,m}(0,0,x,y).
	\end{align*}
The diagram $B^{\ssc[2],m}$ represents the case that $u=\ve[\iota]$ and is very similar to the diagrams $A^{0,m}$, defined in
Tables \ref{BoundTableAZeroOne}-\ref{BoundTableAZeroMinusTwo}, and are given by
	\begin{align*}
	B^{\ssc[2],0}(x,y)=&\delta_{x,y}\sum_{w}\diagRepulsiveLetter{S}_{\underline 1,1,0,1}(x,w,\ve[\iota],0)+\diagRepulsiveLetter{S}_{2,0,0,1}(x,w,\ve[\iota],0),\\
	B^{\ssc[2],-1}(x,y)=&\sum_{w}\diagRepulsiveLetter{P}_{1,\underline 1,0,0,1}(x,y,w,\ve[\iota],0),\qquad
	\qquad B^{\ssc[2],1}(x,y)=B^{\ssc[2],-1}(y,x),\\
	B^{\ssc[2],-2}(x,y)=&\sum_{w}\diagRepulsiveLetter{P}_{1,2,0,0,1}(x,y,w,\ve[\iota],0),\qquad \qquad
	B^{\ssc[2],2}(x,y)=B^{\ssc[2],-2}(y,x).
	\end{align*}
The diagrams $B^{\ssc[3],m}$ represents the remaining possibilities for $u\nin \{0,\ve[\iota]\}$, and are defined by
	\begin{align*}
	B^{\ssc[3],0}(x,y)=&\delta_{x,y}\sum_u \diagRepulsiveLetter{B}_{1,1}(u,\ve[\iota])
	\sum_w\left(\diagRepulsiveLetter{T}_{\underline 1,1,0}(x,w,u)+\diagRepulsiveLetter{T}_{2,0,0}(x,w,u)\right),\\
	B^{\ssc[3],-1}(x,y)=&\sum_u \diagRepulsiveLetter{B}_{1,1}(u,\ve[\iota])\sum_w\diagRepulsiveLetter{S}_{1,\underline 1,0,0}(x,y,w,u),\qquad
	B^{\ssc[3],1}(x,y)=B^{\ssc[3],-1}(y,x),\\
	B^{\ssc[3],-2}(x,y)=&\sum_u \diagRepulsiveLetter{B}_{1,1}(u,\ve[\iota])\sum_w\diagRepulsiveLetter{S}_{1,2,0,0}(x,y,w,u),\qquad
	B^{\ssc[3],2}(x,y)=B^{\ssc[3],-2}(y,x).
	\end{align*}
Next we define $B^{\ssc[4],m}$, that represents the contribution of $b_1=(0,\ve[\iota])$, as
	\begin{align*}
	B^{\ssc[4],0}(x,y)=&\delta_{x,y}\left(
	\delta_{x,0}\diagRepulsiveLetter{B}_{\underline 1,3}(\ve[\iota],0)
	+2\diagRepulsiveLetter{T}_{\underline 1,1,1}(\ve[\iota],x,0)
	+\sum_{w}\diagRepulsiveLetter{S}_{\underline 1,0,1,1}(\ve[\iota],x,w,0)\right),\nnb
	B^{\ssc[4],1}(x,y)=&
	\delta_{x,0}\diagRepulsiveLetter{T}_{\underline 1,2,\underline 1}(\ve[\iota],y,0)
	+2\diagRepulsiveLetter{S}_{\underline 1,0,\underline 1,1}(\ve[\iota],y,x,0)
	+\sum_w \diagRepulsiveLetter{P}_{\underline 1,0,\underline 1,1,1}(\ve[\iota],y,x,w,0),\nnb
	B^{\ssc[4],2}(x,y)=&
	\delta_{x,0}\diagRepulsiveLetter{T}_{\underline 1,0,2}(\ve[\iota],y,0)
	+2\diagRepulsiveLetter{S}_{\underline 1,0,2,1}(\ve[\iota],y,x,0)
	+\sum_w \diagRepulsiveLetter{P}_{\underline 1,0,2,1,1}(\ve[\iota],y,x,w,0),\nnb
	B^{\ssc[4],-1}(x,y)=&
	\delta_{x,\ve[\iota]}\diagRepulsiveLetter{T}_{\underline 1,\underline 1,2}(\ve[\iota],y,0)
	+\delta_{x,\ve[\iota]}\sum_w\diagRepulsiveLetter{S}_{\underline 1,\underline 1,1,1}(\ve[\iota],y,w,0)
	+\sum_w \diagRepulsiveLetter{P}_{\underline 1,1,\underline 1,0,0}(\ve[\iota],x,y,w,0),\nnb
	B^{\ssc[4],-2}(x,y)=&
	\delta_{x,\ve[\iota]}\sum_w\diagRepulsiveLetter{S}_{\underline 1,2,0,1}(\ve[\iota],y,w,0)
	+\diagRepulsiveLetter{S}_{\underline 1,1,2,0}(\ve[\iota],x,y,0)
	+\sum_w \diagRepulsiveLetter{P}_{\underline 1,1,2,0,1}(\ve[\iota],x,y,w,0).
	\end{align*}
Further, the diagram $B^{\ssc[5],m}$ bounds the case where $\tb_1=\ve[\iota]$ and $\bb_1\neq 0$, and is defined as
	\begin{align*}
	B^{\ssc[5],0}(x,y)=&z\gj\delta_{x,y}
	\sum_{\kappa\neq-\iota}\diagRepulsiveLetter{D}_{2}(\ve[\iota]+\ve[\kappa])\nnb
	&\times \sum_{w}\big(\delta_{x,\ve[\iota]}\diagRepulsiveLetter{B}_{1,1}(w-\ve[\iota],\ve[\kappa])+\diagRepulsiveLetter{T}_{1,0,0}(x-\ve[\iota],w-\ve[\iota],\ve[\kappa])
 	+\diagRepulsiveLetter{T}_{0,0,1}(w,x,\ve[\iota])\big),\nnb &+z\gj\delta_{x,y}\sum_{\kappa\neq-\iota}\sum_{v,w}\diagRepulsiveLetter{T}_{2,1,1}(\ve[\iota]+\ve[\kappa],v,0)\diagRepulsiveLetter{T}_{0,0,0}(w-v,x-v,\ve[\iota]-v),\nnb
	B^{\ssc[5],-1}(x,y)=&z\gj \sum_{\kappa\neq-\iota}\diagRepulsiveLetter{D}_{2}(\ve[\iota]+\ve[\kappa])\sum_{w}\big(\diagRepulsiveLetter{S}_{0,0,\underline 1,0}(w,y,x,\ve[\iota])
	+\diagRepulsiveLetter{S}_{0,\underline 1,0,0}(x-\ve[\iota],y-\ve[\iota],w-\ve[\iota],\ve[\kappa])\big)\nnb
	&+z\gj\sum_{\kappa\neq-\iota}\sum_{v,w}\diagRepulsiveLetter{T}_{2,1,1}(\ve[\iota]+\ve[\kappa],v,0)
	\diagRepulsiveLetter{S}_{0,0,\underline 1,0}(w-v,y-v,x-v,\ve[\iota]-v),\nnb
	B^{\ssc[5],-2}(x,y)=&
	z\gj\sum_{\kappa\neq-\iota}\diagRepulsiveLetter{D}_{2}(\ve[\iota]+\ve[\kappa])\sum_{w}
	\big(\diagRepulsiveLetter{S}_{0,0,2,0}(w,y,x,\ve[\iota])+\diagRepulsiveLetter{S}_{0,2,0,0}(x-\ve[\iota],y-\ve[\iota],w-\ve[\iota],\ve[\kappa])\big)\nnb
	&+z\gj\sum_{\kappa\neq-\iota}\sum_{v,w}\diagRepulsiveLetter{T}_{2,1,1}(\ve[\iota]+\ve[\kappa],v,0)
	\diagRepulsiveLetter{S}_{0,0,2,0}(w-v,y-v,x-v,\ve[\iota]-v),\nnb
	B^{\ssc[5],1}(x,y)=&B^{\ssc[5],-1}(y,x),\qquad \qquad
	B^{\ssc[5],2}(x,y)=B^{\ssc[5],-2}(y,x).
	\end{align*}
The diagram $B^{\ssc[6],m}$ captures the contribution of $\bb_1\neq 0=u=0$. For this, we simply re-use the diagram $P^{\ssc[1],m}$, and define, for $m=-2,\dots,2$,
	\begin{align*}
	B^{\ssc[6],m}(x,y)=&\tilde G_z(\ve[\iota])(P^{\ssc[1],m}(x,y)-A^{0,m}(0,0,x,y)).\phantom{========================================================}
	\end{align*}
Then, we define $B^{\ssc[7],m},B^{\ssc[8],m},B^{\ssc[9],m}$ to bound the contributions due to $\bb_1=v=u$, $\bb_1=v\neq u$ and $\bb_1\neq v=u$, respectively, for $m=-2,\dots,2$, as
	\begin{align*}
	B^{\ssc[7],m}(x,y)=& \sum_{\bb_1}\diagRepulsiveLetter{D}_{1}(\bb_1)\tilde G_z(\ve[\iota]-\bb_1)A^{0,m}(\bb_1,\bb_1,x,y),	
	\phantom{===========================================================}\\
	B^{\ssc[8],m}(x,y)=& \diagRepulsiveLetter{B}_{3, 1}(\ve[\iota],0)A^{-1,m}(\ve[\iota],0,x,y)
	+\sum_v\diagRepulsiveLetter{T}_{1,1,1}(\ve[\iota],v,0)A^{-2,m}(\ve[\iota],v,x,y)\\
	&+\sum_{v,\bb_1\neq \ve[\iota]} \diagRepulsiveLetter{B}_{1,1}(\bb_1,\ve[\iota])\diagRepulsiveLetter{B}_{0,1}(v,\bb_1)A^{-2,m}(\bb_1,v,x,y),\\
	B^{\ssc[9],m}(x,y)=&
	\sum_{\bb_1} \diagRepulsiveLetter{T}_{1,\underline 1,1}(\ve[\iota],\bb_1,0) A^{-1,m}(\bb_1,\ve[\iota],x,y)
	+\diagRepulsiveLetter{T}_{1,2,1}(\ve[\iota],\bb_1,0) A^{-2,m}(\ve[\iota],\bb_1,x,y)\\
	&+\sum_{\bb_1,v\neq \ve[\iota]} \diagRepulsiveLetter{B}_{1,1}(v,\ve[\iota])\diagRepulsiveLetter{B}_{1,1}(\bb_1,v)A^{-2,m}(\bb_1,v,x,y).
	\end{align*}
In the diagrams $B^{\ssc[10],m}-B^{\ssc[13],m}$ the point $u$ is somewhere on the triangle $0,v,\bb_1$, excluding the corner points, which were already considered using $B^{\ssc[6],m}-B^{\ssc[9],m}$. Thus, we define, for $m=-2,\dots,2$,
	\begin{align*}
	B^{\ssc[10],m}(x,y)=&\sum_{v,\bb_1} \diagRepulsiveLetter{B}_{1,1}(u,\ve[\iota])
	\left(\diagRepulsiveLetter{B}_{\underline 1,1}(\bb_1,u) A^{-1,m}(\bb_1,v,x,y)+\diagRepulsiveLetter{B}_{2,1}(\bb_1,u)A^{-2,m}(\bb_1,v,x,y)\right),\\
	B^{\ssc[11],m}(x,y)=&
	\sum_{v,\bb_1}
	\left(\diagRepulsiveLetter{S}_{1,\underline 1,1,1}(\bb_1,v,\ve[\iota],0) A^{-1,m}(\bb_1,v,x,y)+\diagRepulsiveLetter{S}_{1,2,1,1}(\bb_1,v,\ve[\iota],0)A^{-2,m}(\bb_1,v,x,y)\right)\\ 
	&+\sum_{u,v,\bb_1} \diagRepulsiveLetter{B}_{1,1}(u,\ve[\iota])
	\left(\diagRepulsiveLetter{T}_{1,\underline 1,1}(\bb_1,v,u) A^{-1,m}(\bb_1,v,x,y)+\diagRepulsiveLetter{T}_{1,2,1}(\bb_1,v,u)A^{-2,m}(\bb_1,v,x,y)\right),\\
	B^{\ssc[12],m}(x,y)=&
	\sum_{v,\bb_1}
	\left(\diagRepulsiveLetter{S}_{1, 1,\underline 1,1}(\ve[\iota],\bb_1,v,0) A^{-1,m}(\bb_1,v,x,y)+\diagRepulsiveLetter{S}_{1,1,2,1}(\ve[\iota],\bb_1,v,0) A^{-2,m}(\bb_1,v,x,y)\right)\\ 
	&+\sum_{u,v,\bb_1} \diagRepulsiveLetter{B}_{1,1}(u,\ve[\iota])
	\left(\diagRepulsiveLetter{T}_{1,\underline 1,1}(v,\bb_1,u) A^{-1,m}(\bb_1,v,x,y)+\diagRepulsiveLetter{T}_{1,2,1}(v,\bb_1,u)A^{-2,m}(\bb_1,v,x,y)\right),\\
	B^{\ssc[13],m}(x,y)=&
	\sum_{v,\bb_1}\diagRepulsiveLetter{S}_{1, 1 ,1,1}(\bb_1,\ve[\iota],v,0) A^{-2,m}(\bb_1,v,x,y)\\ 
	&\qquad+\sum_{u,v,\bb_1} \diagRepulsiveLetter{B}_{1,1}(\bb_1,u)\diagRepulsiveLetter{T}_{1,1,1}(v,u,\ve[\iota])A^{-2,m}(\bb_1,v,x,y).
	\end{align*}
In the diagrams $B^{\ssc[14],m}$ ,$B^{\ssc[15],m}$,$B^{\ssc[16],m}$, the vertex $u$ is on the connection $v$ to $w$,
where $u\neq v$, while $u=w$ is possible. The three diagrams represent the cases $v=0$, $v=\bb_1$ and $v\neq \{0,v\}$.
Below, the first part corresponds to $u=\ve[\iota]$ and the second to $u\neq \ve[\iota]$. We thus define
	\begin{align*}
	B^{\ssc[14],0}(x,y)=&\delta_{x,y} \sum_{\bb_1,w,u}
	\left(\diagRepulsiveLetter{T}_{\underline 1,1,0}(\bb_1,u,\ve[\iota])\tilde G_{3,z}(\bb_1)
	+\diagRepulsiveLetter{T}_{2,1,0}(\bb_1,u,\ve[\iota])\tilde G_{2,z}(\bb_1)\right)\\
	&\qquad \times\diagRepulsiveLetter{T}_{1,0,0}(x-\bb_1,w-\bb_1,u-\bb_1),\\
	B^{\ssc[14],-1}(x,y)=& \sum_{\bb_1,w,u}
	\left(\diagRepulsiveLetter{T}_{\underline 1,1,0}(\bb_1,u,\ve[\iota])\tilde G_{3,z}(\bb_1)
	+\diagRepulsiveLetter{T}_{2,1,0}(\bb_1,u,\ve[\iota])\tilde G_{2,z}(\bb_1)\right)\\
	&\qquad \times\diagRepulsiveLetter{S}_{1,\underline 1, 0,0}(x-\bb_1,y-\bb_1,w-\bb_1,u-\bb_1),\\
	B^{\ssc[14],-2}(x,y)=& \sum_{\bb_1,w,u}
	\left(\diagRepulsiveLetter{T}_{\underline 1,1,0}(\bb_1,u,\ve[\iota])\tilde G_{3,z}(\bb_1)
	+\diagRepulsiveLetter{T}_{2,1,0}(\bb_1,u,\ve[\iota])\tilde G_{2,z}(\bb_1)\right)\\
	&\qquad \times \diagRepulsiveLetter{S}_{1,2, 0,0}(x-\bb_1,y-\bb_1,w-\bb_1,u-\bb_1),\\
	B^{\ssc[14],1}(x,y)=&B^{\ssc[14],-1}(y,x),\qquad \quad B^{\ssc[14],2}(x,y)=B^{\ssc[14],-2}(y,x),
	\end{align*}
and
	\begin{align*}
	B^{\ssc[15],0}(x,y)=&\delta_{x,y}
	\sum_{\bb_1,w}
	\left(\diagRepulsiveLetter{P}_{\underline 1,1,0,0,1}(\bb_1,x,w,\ve[\iota],0)\tilde G_{3,z}(\bb_1)
	+\diagRepulsiveLetter{P}_{2,1,0,0,1}(\bb_1,x,w,\ve[\iota],0)\tilde G_{2,z}(\bb_1)\right)\phantom{===========================================================}\\
	&+\delta_{x,y}\sum_{\bb_1,w,u\neq\ve[\iota]}\diagRepulsiveLetter{D}_{1}(\bb_1)\diagRepulsiveLetter{B}_{1,1}(u,\ve[\iota])
	\diagRepulsiveLetter{T}_{1,0,0}(x-\bb_1,w-\bb_1,u-\bb_1),\\
	B^{\ssc[15],-1}(x,y)=& \tilde G_{1,z}(\ve[\iota])\sum_{\bb_1,w}
	\left(\diagRepulsiveLetter{P}_{\underline 1,1,\underline 1,0,0}(\bb_1,x,y,w,\ve[\iota])\tilde G_{3,z}(\bb_1)
	+\diagRepulsiveLetter{P}_{2,1,\underline 1, 0,0}(\bb_1,x,y,w,\ve[\iota])\tilde G_{2,z}(\bb_1)\right)
	\phantom{===========================================================}\\
	&+\sum_{\bb_1,w,u\neq\ve[\iota]}\diagRepulsiveLetter{D}_{1}(\bb_1)\diagRepulsiveLetter{B}_{1,1}(u,\ve[\iota])
	\diagRepulsiveLetter{S}_{1,\underline 1,0,0}(x-\bb_1,y-\bb_1,w-\bb_1,u-\bb_1),
	\end{align*}
	\begin{align*}
	B^{\ssc[15],-2}(x,y)=& \sum_{\bb_1,w}
	\tilde G_{1,z}(\ve[\iota])
	\diagRepulsiveLetter{P}_{\underline 1,1,2,0,0}(\bb_1,x,y,w,\ve[\iota])\tilde G_{3,z}(\bb_1)\phantom{===========================================================}\\
	&+\sum_{\bb_1,w} \tilde G_{1,z}(\ve[\iota]) \diagRepulsiveLetter{D}_{2}(\bb_1)
	\diagRepulsiveLetter{S}_{1,2,0,0}(x-\bb_1,y-\bb_1,w-\bb_1,\ve[\iota]-\bb_1)\\
	&+\sum_{\bb_1,w,u\neq\ve[\iota]}\diagRepulsiveLetter{D}_{1}(\bb_1)\diagRepulsiveLetter{B}_{1,1}(u,\ve[\iota])
	\diagRepulsiveLetter{S}_{1,2,0,0}(x-\bb_1,y-\bb_1,w-\bb_1,u-\bb_1),\\
	B^{\ssc[15],1}(x,y)=&B^{\ssc[15],-1}(y,x),\qquad \quad B^{\ssc[15],2}(x,y)=B^{\ssc[15],-2}(y,x),
	\end{align*}
Finally, we define
	\begin{align*}
	B^{\ssc[16],0}(x,y)=&\delta_{x,y}\sum_{w,v,\bb_1} \diagRepulsiveLetter{B}_{1,1}(v,\ve[\iota])\diagRepulsiveLetter{B}_{1,1}(\bb_1,v)
	\diagRepulsiveLetter{T}_{1,0,0}(x-\bb_1,w-\bb_1,\ve[\iota]-\bb_1)\phantom{===========================================================}\\
	&+\delta_{x,y}\sum_{u,w,v,\bb_1} \diagRepulsiveLetter{T}_{1,1,1}(\bb_1,v,0) \diagRepulsiveLetter{B}_{1,1}(u-v,\ve[\iota]-v)\diagRepulsiveLetter{T}_{1,0,0}(x-\bb_1,w-\bb_1,u-\bb_1),\\
	B^{\ssc[16],-1}(x,y)=&\sum_{w,v,\bb_1} \diagRepulsiveLetter{B}_{1,1}(v,\ve[\iota])\diagRepulsiveLetter{B}_{1,1}(\bb_1,v)
	\diagRepulsiveLetter{S}_{1,\underline 1, 0,0}(x-\bb_1,y-\bb_1,w-\bb_1,\ve[\iota]-\bb_1)\\
	&+\sum_{u,w,v,\bb_1} \diagRepulsiveLetter{T}_{1,1,1}(\bb_1,v,0) \diagRepulsiveLetter{B}_{1,1}(u-v,\ve[\iota]-v)\diagRepulsiveLetter{S}_{1,\underline 1,0,0}(x-\bb_1,y-\bb_1,w-\bb_1,u-\bb_1),\\
	B^{\ssc[16],-2}(x,y)=&\sum_{w,v,\bb_1} \diagRepulsiveLetter{B}_{1,1}(v,\ve[\iota])\diagRepulsiveLetter{B}_{1,1}(\bb_1,v)\diagRepulsiveLetter{S}_{1,2,0,0}(x-\bb_1,y-\bb_1,w-\bb_1,\ve[\iota]-\bb_1)\\
	&+\sum_{u,w,v,\bb_1} \diagRepulsiveLetter{T}_{1,1,1}(\bb_1,v,0) \diagRepulsiveLetter{B}_{1,1}(u-v,\ve[\iota]-v)\diagRepulsiveLetter{S}_{1,2,0,0}(x-\bb_1,y-\bb_1,w-\bb_1,u-\bb_1),\\
	B^{\ssc[16],1}(x,y)=&B^{\ssc[16],-1}(y,x),\qquad \quad B^{\ssc[16],2}(x,y)=B^{\ssc[16],-2}(y,x).
	\end{align*}
This defines the bounding diagrams $B^{\ssc[M],m}$ for $M=1, \ldots, 16$ and $m=-2, \ldots, 2$, that bound the initial part of the diagrams for $\Xi^{\iota}$ and $\Pi^{\iota,\kappa}$.

\subsection{Weighted intermediate diagrams $C^{m,l}$}
\label{sec-formal-C}
In this section we define the weighted diagrams that we use to bound diagrams such as $\sum_{x}|x|^2 \Xi^{\ssc[N]}(x)$.
The shape of such diagrams is shown in the {\it informal definition} part of Section \ref{secBuildingBlocks}.
There are $25$ different cases for this diagram. Using some symmetry arguments, these reduce to $16$ diagrams, that we now define one by one. We describe the bounds for LAs only, as LTs are special cases of LAs. For the bounds on LTs, we simply restrict the definitions below to LTs.

Again we use translation variance in the form that, for all $m,l\in \{-2,-1,0,1,2\}$,
	\begin{align*}
  	C^{m,l}(u,v,x,y)=C^{m,l}(0,v-u,x-u,y-u).\phantom{===========================================================}
	\end{align*}
We first define $C^{0,2}, C^{-2,0}, C^{-2,2}, C^{-2,1}$ and $C^{-1,2}$ as
	\begin{align*}
	C^{0,2}(0,v,x,y)&=\delta_{0,v}|y|^2G_z(y)  \sum_{w}\diagRepulsiveLetter{B}_{0,0}(w,x),\phantom{===========================================================}\\
	C^{-2,0}(0,v,x,y)&=\delta_{x,y}|y|^2G_z(y) \sum_{w}\diagRepulsiveLetter{B}_{0,0}(w-v,y-v),\\
	C^{-2,2}(0,v,x,y)&= |y|^2 G_z(y) \sum_{w}\diagRepulsiveLetter{B}_{0,0}(w-v,x-v),\\
	C^{-2,1}(0,v,x,y)&=  2d D(x-y) |y|^2 G_z(y) \sum_{w}\diagRepulsiveLetter{B}_{0,0}(w-v,x-v),\\
	C^{-1,2}(0,v,x,y)&=  2d D(v) |y|^2 G_z(y) \sum_{w}\diagRepulsiveLetter{B}_{0,0}(w-v,x-v).
	\end{align*}
In the following definitions of $C^{2,0}, C^{2,2}, C^{2,1}$ and $C^{1,2}$, the weight is distributed along the line $v\to w\to y$:
	\begin{align*}
	C^{2,0}(0,v,x,y)&=\delta_{x,y}\sum_{w}\tilde G_{1,z}(v) 2^{\indic{w\nin \{0,y\}}}( |w|^2 G_z(w)\tilde G_z(y-w)+ |y-w|^2\tilde G_z(w) G_z(y-w) ),	
	\phantom{===========================================================}\\
	C^{2,2}(0,v,x,y)&= \tilde G_z(x-v) \sum_{w}2^{\indic{w\nin \{0,y\}}}( |w|^2 G_z(w)\tilde G_z(y-w)+ |y-w|^2\tilde G_z(w) G_z(y-w) ),\\
	C^{2,1}(0,v,x,y)&=   2d D(x-y)\tilde G_z(x-v)\\
	&\qquad \qquad \times \sum_{w}2^{\indic{w\nin \{0,y\}}}( |w|^2 G_z(w)\tilde G_z(y-w)+ |y-w|^2\tilde G_z(w) G_z(y-w)),	
	\phantom{===========================================================}\\
	C^{1,2}(0,v,x,y)&=   2d D(v)\tilde G_z(x-v)\\
	&\qquad \qquad \times \sum_{w}2^{\indic{w\nin \{0,y\}}}( |w|^2 G_z(w)\tilde G_z(y-w)+ |y-w|^2\tilde G_z(w) G_z(y-w)).
	\end{align*}
We define $C^{-1,1}$ using sums over LAs, to retain as many repulsive properties as possible:
	\begin{align*}
	C^{-1,1}(0,v,x,y)=&\frac{1}{g_z} (2d)^2 D(v)D(x-y) |x|^2 \sum_{A_1\ni v, A_2\ni y} \sum_{A_3\colon 0,x\in A_3} z^{|A_1|+|A_2|+|A_3|},
	\end{align*}
where $A_1, A_2$ are planted animals while $A_3$ is a LA, and $A_1,A_2, A_3$ are such that (a) $A_1,A_2$ share at least one vertex, while (b) $A_1,A_3$ and $A_2,A_3$ do not share a vertex. Similarly, we define $C^{1,1}$ as
	\begin{align*}
	C^{1,1}(0,v,x,y)=& \frac{1}{g_z}(2d)^2D(v)D(x-y) \sum_{A_1, A_2} \sum_{A_3\colon 0,x\in A_3} z^{|A_1|+|A_2|+|A_3|}\\
	&\qquad \times\sum_{w\in A_1, w\in A_2} 2^{\indic{w\nin \{v,y\}}} \Big(\indic{v\in A_1, y\in A_2} |v-w|^2+\indic{v\in A_2, y\in A_1} |y-w|^2\Big),
	\end{align*}
where now $A_2, A_3$ are planted animals while $A_1$ is a LA, and $A_1,A_2, A_3$ are such that (a) $w$ denote the first intersection point between $A_1$ and $A_2$, while (b) $A_1,A_3$ and $A_2,A_3$ do not share a vertex. These definitions allow us to make maximal use of repulsiveness between the lines in the bounding diagram. If we discard that repulsiveness, then we arrive instead at the bounds
	\begin{align*}
	C^{-1,-1}(0,v,x,y)&\leq (2d)^2 D(v)D(x-y)|x|^2 G_z(x) \sum_{w}\diagRepulsiveLetter{B}_{0,0}(w-v,y-v),\\
	C^{1,1}(0,v,x,y)&\leq(2d)^2D(v)D(x-y)\tilde G_{1,z}(x) \\
	&\qquad \times\sum_{w}2^{\indic{w\nin \{v,y\}}}( |w|^2 G_z(w)\tilde G_z(y-w)+ |y-w|^2\tilde G_z(w) G_z(y-w)).
	\end{align*}
Including the repulsiveness allows us to reduce the bounds on these building blocks by about 45\%, which is significant.

Each connection in our repulsive diagrams $\diagRepulsiveLetter{B},\diagRepulsiveLetter{T},\diagRepulsiveLetter{S}$ characterizes a planted animal.
For the following bounds, the line that carries the spatial weight $|x|^2$ should correspond to a regular (i.e., a non-planted) animal.
For this we modify Definitions \ref{defSkeletonLA} and \ref{def-repulsive-diamgrams}, so that the initial and last sausage of one line might be non-trivial, see part (4) of Definition \ref{defSkeletonLA}. When we ignore repulsiveness, this line will correspond to a regular two-point function $\bar{G}_z$.

Instead of repeating Definition \ref{defSkeletonLA}, let us just mention the differences that arise in it:
We extend the index range $j_i\in\{l_i,\underline l_i\}$ corresponding to lines of at least $l_i$ and exactly $\underline{l}_i$ bonds, respectively, to $j_i\in\{l_i,\underline l_i,\bar l_i,\bar {\underline l}_i\}$.
Here, the indices $\bar l_i,\bar {\underline l}_i$ correspond to the line being a regular LA, and the number of bonds being at least $\bar{l}_i$, respectively, exactly equal to $\bar {\underline l}_i$.
In more detail, in Definition \ref{defSkeletonLA}(2), $l_i$ and $\bar l_i$ both enforce that $|p^i|\leq l_i$, while
$\underline l_i$ and $\underline {\bar l}_i$ both enforce that $|p^i|=l_i$. Further, we drop the constraint in Definition \ref{defSkeletonLA}(4) when $j_i\in\{\bar l_i,\bar {\underline l}_i\}$, so that the sausage-walk
$\omega_i$ corresponds to a regular LA. Then, we extend Definition  \ref{def-repulsive-diamgrams} using these skeletons to create diagrams alike $\diagRepulsiveLetter{T}_{\bar 1,1,1}$ and $\diagRepulsiveLetter{S}_{1,\bar 0,1,\underline 1}$. To summarize, the line with a bar on top of its length corresponds to a regular animal.

We use this to define $C^{0,0}, C^{0,1}, C^{-1,0}$ and $C^{1,0}$ as
	\begin{align*}
	C^{0,0}(0,v,x,y)=&\ \frac {\delta_{0,v}\delta_{x,y}}{g_z}|x|^2\left(2\diagRepulsiveLetter{B}_{\bar 1,2}(x,0)+\sum_w \diagRepulsiveLetter{T}_{\bar 1,1,1}(x,w,0)\right),\\
	C^{0,1}(0,v,x,y)=&\ 2d D(x-y)\frac{\delta_{0,v}}{g_z}|y|^2\left(2\diagRepulsiveLetter{B}_{\bar 1,1}(-y,y-x)+\sum_w \diagRepulsiveLetter{T}_{\bar 1,1,1}(-y,w-y,x-y)\right),\\
	C^{-1,0}(0,v,x,y)=&\ 2d D(v)\frac{\delta_{x,y}}{g_z}|x|^2\left(2\diagRepulsiveLetter{B}_{\bar 1,1}(x,0)+\sum_w \diagRepulsiveLetter{T}_{\bar 1,1,1}(x,w,0)\right),\\
	C^{1,0}(0,v,x,y)=&\ 2dD(v)\delta_{x,y}|x-v|^2\frac {\gj}{g_z}\left(\diagRepulsiveLetter{B}_{\underline 1, 1}(v,x)+\sum_w \diagRepulsiveLetter{T}_{ \underline 1,1, 0}(x,w  ,v)\right)\\
&+2dD(v)\frac{\delta_{x,y}}{g_z}\sum_w2^{\indic{w\nin \{v,x\}}}
\left(|w-v|^2 \diagRepulsiveLetter{T}_{2,1,\bar 1}(x,w,v)+|x-w|^2 \diagRepulsiveLetter{T}_{2,\bar 1,1}(x,w,v)\right).
	\end{align*}
The factor $1/g_z$ is created by the normalisation of the coefficients, as introduced in \refeq{defXiNanimalPrime}.
We define the remaining cases of $m\in \{-2,-1,0,1,2\}$ and $l\in \{-1,-2\}$ using symmetry by
	\begin{align*} C^{m,-l}(0,v,x,y)=C^{m,l}(0,v,y,x).\phantom{===========================================================}.
	\end{align*}

\subsection{Weighted initial diagrams of $\Xi$ and $\Xi^{\iota}$}
\label{Appendix-def-Initial-Xi-XiIota-Delta}
In this section, we define all the weighted versions of the initial diagrams of $\Xi$ and $\Xi^{\iota}$. These initial diagrams, without any weights, were already defined in Appendix \ref{Appendix-def-Initial-Xi-XiIota}.

\subsubsection{Weighted initial diagrams of $\Xi$}
\label{Appendix-def-Initial-Xi-Delta}
Recall the discussion of $b_1$ in Section \ref{Appendix-def-Initial-Xi}. For the weighted diagram bounding $\Xi$ we first define the weighted version of the initial triangle for the case that $\bb_1\neq 0$ as
	\begin{align*}
	S^{w,0}(x,y) =&\delta_{x,y}\diagRepulsiveLetter{D}_{1}(x)|x|^2,\\
	S^{w,-1}(x,y) =& 2d D(x-y) \big( \delta_{y,0} \tilde G_{3,z}(x)+(1-\delta_{y,0})|x|^2  G_{z}(x) \tilde G_{z}(y)\big),\\
	S^{w,-2}(x,y) =& \delta_{y,0} |x|^2 G_{2,z}(x)+(1-\delta_{y,0})|x|^2 G_{z}(x) \tilde G_{z}(y).
	\end{align*}
These are indeed weighted versions of $S^{m}(x,y)$ defined in Appendix \ref{Appendix-def-Initial-Xi}.
Then, we define the weighted diagram by
	\begin{align*}
	\Delta^{{\rm start},m}(x,y)=& C^{0,m}(0,0,x,y)\\
	&+2\sum_{l=-2}^0\sum_{u,v\in\Zd} \Big( S^{w,l}(u,v) A^{l,m}(u,v,x,y)+
	\big(S^{m}(u,v)-\delta_{m,0}\delta_{u,0}\big) C^{l,m}(u,v,x,y)\Big).\nn
	\end{align*}
This serves as a bound for the weighted initial diagram $\bar P^{\ssc[1],m}$.
The first term captures the case $\bb_1=0$. For the case $\bb_1=u\neq 0$, we split the weight using $|x|^2\leq 2(|u|^2+|x-u|^2)$. For LTs the second term is not required.

\subsubsection{Weighted initial diagrams of $\Xi^{\iota}$}
\label{Appendix-def-Initial-XiIota-Delta}
In this section we define $\Delta^{\text {iota},{\sss I},m}$ and $\Delta^{\text {iota},{\sss II},m}$, which are used to bound the initial parts of the bounding diagram of the weighted NoBLE coefficients
	\begin{align}
	\label{two-weighted-Xis}
	\sum_x |x|^2 \Xi^{\ssc[N],\iota}_z(x) \qquad \text{and} \qquad \sum_x|x-\ve[\iota]|^2 \Xi^{\ssc[N],\iota}_z(x),
	\end{align}
respectively.  We defined it as the weighted versions of the diagrams shown in Figure \ref{BoundLA-Figure-initial} and indicated as $B^{\ssc[M],m}(x,y)$ with $M=1, \ldots, 16$, in which the path $x$ to $y$ is removed (so that there is no two-point function $G_z(y-x)$ or any of its close friends). We always put the weight on the bottom part of the diagram.

Recall that unweighted versions of such diagrams were defined in Appendix \ref{Appendix-def-Initial-XiIota}. There, we have explained that we need $16\times 5$ of such diagrams. When weighing them, this number increases to $2\times 16\times 5$, since we need diagrams that are weighted with $|x|^2$ as well as with $|x-\ve[\iota]|^2$ (see \eqref{two-weighted-Xis}). In more detail, we denote the $2\times 16$ diagrams by $D^{\ssc[M],{\sss I},m}(x,y)$ and $D^{\ssc[M],{\sss II},m}(x,y)$, where $I$ corresponds to diagrams with the $|x|^2$ weight as in the first term in \eqref{two-weighted-Xis}, while $II$ corresponds to the diagrams with the $|x-\ve[\iota]|^2$ weight as in the second term in \eqref{two-weighted-Xis}. We then write
	\eqn{
	\Delta^{\text {iota},{\sss I},m}(x,y)=\sum_{M=1}^{16} D^{\ssc[M],{\sss I},m}(x,y),
	\qquad
	\Delta^{\text {iota},{\sss II},m}(x,y)=\sum_{M=1}^{16} D^{\ssc[M],{\sss II},m}(x,y).
	}
We are left to define $D^{\ssc[M],{\sss I},m}(x,y), D^{\ssc[M],{\sss II},m}(x,y)$ for $M\in\{1, \ldots, 16\}$ and $m\in \{-2,\ldots, 2\}$.

For $M=1$, we define
	\begin{align*}
	D^{\ssc[1],{\sss I},m}(x,y)&= \tilde G_z(\ve[\iota])C^{0,m}(0,0,x,y),\\
	D^{\ssc[1],{\sss II},-m}(x,y)=& 2\tilde G_z(\ve[\iota]) \left(C^{0,m}(0,0,x,y)+A^{m,0}(x,y,0,0)\right).
	\end{align*}
We use the spatial symmetry, as in \refeq{BoundLT-tmp-3}, to improve the bound for $m=0$ to
	\begin{align*}
	D^{\ssc[1],{\sss II},0}(x,y)=& \delta_{x,y} \left(\tilde G_z(\ve[\iota])C^{0,0}(0,0,x,x)+G_z(\ve[\iota]) A^{0,0}(0,0,x,x)\right).
	\end{align*}
	
For $M=2$, we define
	\begin{align*}
	D^{\ssc[2],{\sss I},0}(x,y)=&\frac {\delta_{x,y}|x|^2}{g_z} \diagRepulsiveLetter{S}_{1,0,0,\bar 0}(\ve[\iota],w,x,0),\\
	D^{\ssc[2],{\sss I},-1}(x,y)=&2d D(x-y) \frac {|x|^2}{g_z}   \diagRepulsiveLetter{S}_{0,1,0,\bar 0}(w-y,\ve[\iota]-y,-y,x-y),\\
	D^{\ssc[2],{\sss I},-2}(x,y)=& \frac {|x|^2}{g_z} \diagRepulsiveLetter{S}_{0,1,0,\bar 0}(w-y,\ve[\iota]-y,-y,x-y),\\
	D^{\ssc[2],{\sss I},1}(x,y)=&D^{\ssc[2],{\sss I},-1}(y,x),\qquad \quad D^{\ssc[2],{\sss I},2}(x,y)=D^{\ssc[2],{\sss I},-2}(y,x),
	\end{align*}
and $D^{\ssc[2],{\sss II},m}(x,y)$ in the same manner where $|x|^2$ is replaced by $|x-\ve[\iota]|^2$.
For $D^{\ssc[2],{\sss II},m}(x,y)$, we note a convenient estimate, which relies on the bound
	\begin{align*}
	\sum_{\footnotesize\begin{array}{c}
                            A \text{ animal} \\
                            0,\ve[\iota],x,w_1\in A
                          \end{array}}z ^{|A|}\leq
	\sum_{\footnotesize\begin{array}{c}
                            A \text{ animal} \\
                            \ve[\iota],x,w_1\in A
                          \end{array}}z ^{|A|}\leq \tilde G_z(w-\ve[\iota]) \bar G(x-\ve[\iota]).
	\end{align*}
For example, this allows us to bound
	\begin{align*}
	D^{\ssc[2],{\sss II},0}(x,y)\leq |x-\ve[\iota]|^2 \bar G(x-\ve[\iota])  \tilde G_z(w-\ve[\iota])  \tilde G_z(w-x).
	\end{align*}
This has the major advantage that we can directly use $f_3$ to bound such diagrams. We will do so when the path connecting $\ve[\iota]$ and $x$ has length at least $6$. For shorter connections we use explicit computation to improve our bounds further.

For $M=3$, we define
	\begin{align*}
	D^{\ssc[3],{\sss I},m}(x,y)=&\diagRepulsiveLetter{B}_{\underline 1,1}(u,\ve[\iota])C^{-1,m}(0,u,x,y)+\diagRepulsiveLetter{B}_{2,1}(u,\ve[\iota])C^{-2,m}(0,u,x,y),\\
	D^{\ssc[3],{\sss II},m}(x,y)=&2D^{\ssc[3],{\sss I},m}(x,y)+2\tfrac{\gj}{g_z} \big(\tilde G_{2,z}(u-\ve[\iota]) A^{m,-1}(x,y,0,u)+\tilde G_{1,z}(u-\ve[\iota])A^{m,-2}(x,y,0,u)\big).
	\end{align*}
	
For $M=4$, we define
	\begin{align*}
	D^{\ssc[4],{\sss I},0}(x,y)=&\frac {|x|^2}{g_z}\delta_{x,y}\diagRepulsiveLetter{S}_{\bar 0,\underline 1,0,0,0}(\ve[\iota]-x,-x,w-x,y-x),\\
	D^{\ssc[4],{\sss I},-1}(x,y)=&\frac {|x|^2}{g_z} 2d D(x-y)\diagRepulsiveLetter{S}_{\bar 0,\underline 1,0,0,0}(\ve[\iota]-x,-x,w-x,y-x),\\
	D^{\ssc[4],{\sss I},-2}(x,y)=&\frac {|x|^2}{g_z} \diagRepulsiveLetter{S}_{\bar 0,\underline 1,0,0}(\ve[\iota]-x,-x,w-x,y-x),\\
	D^{\ssc[4],{\sss I},1}(x,y)=&D^{\ssc[4],{\sss I},-1}(y,x),\qquad \quad D^{\ssc[4],{\sss I},2}(x,y)=D^{\ssc[4],{\sss I},-2}(y,x),
	\end{align*}
and we define $D^{\ssc[4],{\sss II},m}(x,y)$ in the same manner were $|x|^2$ is replaced by $|x-\ve[\iota]|^2$.

For $M=5$, we define
	\begin{align*}
	D^{\ssc[5],{\sss II},0}(x,y)=&\sum_{\bb_1}\diagRepulsiveLetter{D}_{2}(\bb_1)2dz D(\bb_1-\ve[\iota]) |x-\ve[\iota]|^2 \\ &\times\left(\diagRepulsiveLetter{T}_{\bar 0,0,0}(x-\ve[\iota],w-\ve[\iota],\bb_1-\ve[\iota])
+\diagRepulsiveLetter{T}_{\bar 0,0,0}(x-\ve[\iota],w-\ve[\iota],0-\ve[\iota])\right)\nnb
	&+\sum_{\bb_1,v}\diagRepulsiveLetter{T}_{2,1,1}(\bb_1,v,0)  |x-\ve[\iota]|^2 2dz D(\bb_1-\ve[\iota]) \diagRepulsiveLetter{S}_{\bar 0,0,0,0}(x-\ve[\iota],w-\ve[\iota],u-\ve[\iota]),\nnb
	D^{\ssc[5],{\sss II},-1}(x,y)=&2d D(x-y)\sum_{\bb_1}\diagRepulsiveLetter{D}_{2}(\bb_1)2dz D(\bb_1-\ve[\iota]) |x-\ve[\iota]|^2 G(x-\ve[\iota])\\
	&\qquad \times\left(\diagRepulsiveLetter{B}_{0,0}(w-y,\bb_1-y)+\diagRepulsiveLetter{B}_{0,0}(w-y,-y)\right)\nnb
	&+2d D(x-y)\sum_{\bb_1,v}\diagRepulsiveLetter{T}_{2,1,1}(\bb_1,v,0) 2dz D(\bb_1-\ve[\iota]) |x-\ve[\iota]|^2 G(x-\ve[\iota])\diagRepulsiveLetter{B}_{0,0}(w-y,v-y),\nnb
	D^{\ssc[5],{\sss II},-2}(x,y)=&\sum_{\bb_1}\diagRepulsiveLetter{D}_{2}(\bb_1)2dz D(\bb_1-\ve[\iota]) |x-\ve[\iota]|^2
	 G(x-\ve[\iota])\left(\diagRepulsiveLetter{B}_{0,0}(w-y,\bb_1-y)+\diagRepulsiveLetter{B}_{0,0}(w-y,-y)\right)\\
	&+\sum_{\bb_1,v}\diagRepulsiveLetter{T}_{2,1,1}(\bb_1,v,0) 2dz D(\bb_1-\ve[\iota]) |x-\ve[\iota]|^2 G(x-\ve[\iota])\diagRepulsiveLetter{B}_{0,0}(w-y,v-y),\nnb
	D^{\ssc[5],{\sss II},1}(x,y)=&D^{\ssc[5],{\sss II},-1}(y,x),\qquad \quad D^{\ssc[5],{\sss II},2}(x,y)=D^{\ssc[5],{\sss II},-2}(y,x),
	\end{align*}
and we define $D^{\ssc[5],{\sss I},m}(x,y)$ in the same manner where $|x-\ve[\iota]|^2$ is replaced by $|x|^2$.

For $M=6$, we define
	\begin{align*}
	D^{\ssc[6],{\sss I},m}(x,y)=& \tilde G_z(\ve[\iota]) \big(\Delta^{\text {start},m}(x,y)-C^{0,m}(0,0,x,y)\big),\\
	D^{\ssc[6],{\sss II},m}(x,y)=&\frac 3 2 D^{\ssc[6],{\sss I},m}(x,y)
	+3 \frac{\tilde G_z(\ve[\iota])}{g_z}\sum_{u}\diagRepulsiveLetter{D}_{1}(u)A^{m,0}(u,u,x,y)\nnb
	&+3   \frac {\tilde G_z(\ve[\iota])}{g_z}\sum_{u,v}\big(\delta_{0,v}\tilde G_{2,z}(u) +\diagRepulsiveLetter{B}_{1,1}(-v,u-v)\big)\sum_{l=-2}^{-1}A^{m,l}(u,v,x,y).\nn
	\end{align*}
The factor $\tfrac 3 2$ is created by the split of weights.
For $D^{\ssc[6],{\sss I},m}(x,y)$ we split the weight using $|x|^2\leq 2 (|\bb_1|^2+|x-\bb_1|^2)$.
For $D^{\ssc[6],{\sss I},m}(x,y)$ this is replaced in by $|x-\ve[\iota]|^2\leq 3 (|\ve[\iota]|^2+|\bb_1|^2+|x-\bb_1|^2)$,
so we $\tfrac 3 2$ simply replaces the factor $2$ by $3$ for these contribution.

For $M=7$, we define
	\begin{align*}
	D^{\ssc[7],{\sss I},m}(x,y)=&
	2\sum_{u}\diagRepulsiveLetter{D}_{1}(u)\tilde G_z(u-\ve[\iota])
\big(|u|^2\frac {\gj}{g_z}A^{m,0}(x,y,u,u)+ C^{0,m}(u,u,x,y)\big),\\
	D^{\ssc[7],{\sss II},m}(x,y)=&\diagRepulsiveLetter{D}_{1}(\ve[\iota])C^{0,m}(\ve[\iota],\ve[\iota],x,y)+2\sum_{u\neq \ve[\iota]}\diagRepulsiveLetter{D}_{1}(u)\tilde G_{1,z}(u-\ve[\iota])C^{0,m}(u,u,x,y)\nnb
	&+\frac{2}{g_z}\sum_{u}|u-\ve[\iota]|^2\diagRepulsiveLetter{B}_{1,\bar 0}(u,\ve[\iota]) \tilde G_{2,z}(u) A^{m,0}(x,y,u,u).\nn
	\end{align*}
	
For $M=8$, we define
	\begin{align*}
	D^{\ssc[8],{\sss I},m}(x,y)=
	&2\diagRepulsiveLetter{D}_{1}(\ve[\iota]) C^{-1,m}(\ve[\iota],0,x,y)+2  G_{3,z}(\ve[\iota])A^{-1,m}(\ve[\iota],0,x,y)\\
	&+2\sum_{v} \left(\diagRepulsiveLetter{T}_{1,\underline 1,1}(\ve[\iota],v,0)C^{-1,m}(\ve[\iota],v,x,y)+\diagRepulsiveLetter{T}_{1,2,1}(\ve[\iota],v,0)C^{-2,m}(\ve[\iota],v,x,y)\right)\nnb
	&+\frac {2}{g_z}\sum_{v} \diagRepulsiveLetter{B}_{ 1,1} (-\ve[\iota],v-\ve[\iota])\left( A^{m,-1}(x,y,\ve[\iota],v)+A^{m,-2}(x,y,\ve[\iota],v)\right)\nnb
	&+2\sum_{u,v} \diagRepulsiveLetter{B}_{1,1}(u,\ve[\iota])\left(\diagRepulsiveLetter{B}_{0,\underline 1}(v,u)C^{-1,m}(u,v,x,y)
	+\diagRepulsiveLetter{B}_{0,2}(v,u)C^{-2,m}(u,v,x,y)\right)\nnb
	&+\frac {2}{g_z}\sum_{u,v\neq\ve[\iota]}  \diagRepulsiveLetter{B}_{\bar 1,1}(u,\ve[\iota]) \tilde G_z(v) \sum_{b=-2}^{-1} A^{a,-1}(x,y,u,v),\nnb
	D^{\ssc[8],{\sss II},m}(x,y)=&\diagRepulsiveLetter{D}_{1}(\ve[\iota]) C^{-1,m}(\ve[\iota],0,x,y)\\
	&+\sum_{v} \left(\diagRepulsiveLetter{T}_{1,\underline 1,1}(\ve[\iota],v,0)C^{-1,m}(\ve[\iota],v,x,y)+\diagRepulsiveLetter{T}_{1,2,1}(\ve[\iota],v,0)C^{-2,m}(\ve[\iota],v,x,y)\right)\nnb
	&+2\sum_{u,v} \diagRepulsiveLetter{B}_{1,1}(u,\ve[\iota])\left(\diagRepulsiveLetter{B}_{0,\underline 1}(v,u)C^{-1,m}(u,v,x,y)
	+\diagRepulsiveLetter{B}_{0,2}(v,u)C^{-2,m}(u,v,x,y)\right)\nnb
	&+\frac {2}{g_z}\sum_{u,v\neq\ve[\iota]} \diagRepulsiveLetter{B}_{1,\bar 1}(u,\ve[\iota]) \tilde G_z(v) \sum_{b=-2}^{-1} A^{a,-b}(x,y,u,v).\nnb
	\end{align*}
	
For $M=9$, we define
	\begin{align*}
	D^{\ssc[9],{\sss I},m}(x,y)=
	&2\sum_{u}\diagRepulsiveLetter{T}_{1,\underline 1,1}(\ve[\iota],u,0)C^{-1,m}(u,\ve[\iota],x,y)
+\diagRepulsiveLetter{T}_{1,2,1}(\ve[\iota],u,0)C^{-2,m}(u,\ve[\iota],x,y)\\
	&+2\sum_{u, v\neq\ve[\iota]} \diagRepulsiveLetter{B}_{1,1}(v,\ve[\iota])\diagRepulsiveLetter{B}_{1,1}(u,v)C^{-2,m}(u,v,x,y)\nnb
	&+\frac {2}{g_z}\sum_{u}|u|^2\diagRepulsiveLetter{T}_{\bar 1, 1,0}(-u,v-u,\ve[\iota]-u)
\big(A^{-1,m}(u,\ve[\iota],x,y)+ A^{-2,m}(u,\ve[\iota],x,y)\big)\nnb
	\end{align*}
and
	\begin{align*}
	D^{\ssc[9],{\sss II},m}(x,y)=
	&2\sum_{u}\diagRepulsiveLetter{T}_{1,\underline 1,1}(\ve[\iota],u,0) C^{-1,m}(u,\ve[\iota],x,y)
+\diagRepulsiveLetter{T}_{1,2,1}(\ve[\iota],u,0)C^{-2,m}(u,\ve[\iota],x,y)\\
	&\frac {2}{g_z}\sum_{u}|u-\ve[\iota]|^2\diagRepulsiveLetter{T}_{1,\bar 1,1}(\ve[\iota],u,0)\bar A^{-2,m}(u,\ve[\iota],x,y)\\
	&+3\sum_{v,u}\diagRepulsiveLetter{B}_{1,1}(v,\ve[\iota])\diagRepulsiveLetter{B}_{1,1}(u,v)C^{-2,m}(u,v,x,y)\nnb
	&+\frac {3}{g_z}\sum_{u}(1+|u|^2)\diagRepulsiveLetter{T}_{\bar 1, 1,0}(-u,v-u,\ve[\iota]-u)
\big(A^{-1,m}(u,\ve[\iota],x,y)+ A^{-2,m}(u,\ve[\iota],x,y)\big)
	\end{align*}
	
For $M=10$, we define
	\begin{align*}
	D^{\ssc[10],{\sss I},m}(x,y)=&
	2\sum_{u,\bb_1} \diagRepulsiveLetter{B}_{1,1}(u,\ve[\iota])\tilde G_{1,z}(u-\bb_1)
\left( z C^{-1,m}(\bb_1,0,x,y)+|\underline b_1|^2A^{m,-1}(x,y,\bb_1,0)\right)\\
	&+2\sum_{u,\bb_1} \diagRepulsiveLetter{B}_{1,1}(u,\ve[\iota])\left( \diagRepulsiveLetter{B}_{2,1}(\bb_1, u)C^{-2,m}(\bb_1,0,x,y)+\frac{|\bb_1|^2}{g_z} \diagRepulsiveLetter{B}_{\bar 2,1}(\bb_1, u) \bar A^{-2,m}(\bb_1,0,x,y)\right),\\
	D^{\ssc[10],{\sss II},m} (x,y)=&
	2\sum_{u,\bb_1} \diagRepulsiveLetter{B}_{1,1}(u,\ve[\iota])\tilde G_{1,z}(u-\bb_1)
\left( z C^{-1,m}(\bb_1,0,x,y)+|\underline b_1-\ve[\iota]|^2 A^{m,-1}(x,y,\bb_1,0)\right)\\
	&+3\sum_{u,\bb_1} \diagRepulsiveLetter{B}_{1,1}(u,\ve[\iota])\diagRepulsiveLetter{B}_{2,1}(\bb_1, u)C^{-2,m}(\bb_1,0,x,y)\\
	&+\frac{3}{g_z}\sum_{u,\bb_1} \diagRepulsiveLetter{B}_{ 2,1}(-\bb_1, u-\bb_1) \bar A^{-2,m}(\bb_1,0,x,y)\\
&\qquad \times \big( |u-\bb_1|^2\diagRepulsiveLetter{B}_{\bar 1,1}(u-\bb_1,\ve[\iota]-\bb_1)
+|\ve[\iota]-u|^2\diagRepulsiveLetter{B}_{1,\bar 1}(u-\bb_1,\ve[\iota]-\bb_1)
 \big).
	\end{align*}
	
For $M=11$, we define
	\begin{align*}
	D^{\ssc[11],{\sss I},m}(x,y)=&
	2\sum_{v,\bb_1} \left( \diagRepulsiveLetter{S}_{1,1,\underline 1,1}(\ve[\iota],v,\bb_1,0)C^{-1,m}(\bb_1,v,x,y)
+\frac{|\bb_1|^2}{g_z}\diagRepulsiveLetter{S}_{1,1,\underline 1,1}(\ve[\iota],v,\bb_1,0) \bar A^{-1,m}(\bb_1,v,x,y)\right)\\
	&+2\sum_{v,\bb_1} \diagRepulsiveLetter{S}_{1,1,2,1}(\ve[\iota],v,\bb_1,0)C^{-2,m}(\bb_1,v,x,y)\\
	&+2\sum_{v,\bb_1}\frac {|\bb_1|^2}{g_z}   \diagRepulsiveLetter{T}_{\bar 1,1,1}(-\bb_1,\ve[\iota]-\bb_1,v-\bb_1)  A^{m,-2}(x,y,\bb_1,v)\\
	&+2\sum_{u,v,\bb_1} \diagRepulsiveLetter{B}_{1,1}(u,\ve[\iota])
	\left(\diagRepulsiveLetter{T}_{1,\underline 1,1}(\bb_1,v,u)C^{-1,m}(\bb_1,v,x,y)+\frac {|\bb_1|^2}{g_z}\diagRepulsiveLetter{T}_{\bar 1,\underline 1,1}(\bb_1,v,u)\bar A^{-1,m}(\bb_1,v,x,y)\right)\\
	&+2\sum_{u,v,\bb_1} \diagRepulsiveLetter{B}_{1,1}(u,\ve[\iota]) \diagRepulsiveLetter{T}_{1,2,1}(\bb_1,v,u) C^{-2,m}(\bb_1,v,x,y)\\
	&+2\sum_{u,v,\bb_1} \diagRepulsiveLetter{B}_{1,1}(u,\ve[\iota]) |\bb_1|^2 G_z(\bb_1)\tilde G_{1,z}(u-v)A^{m,-2}(x,y,\bb_1,v),
	\end{align*}
and
	\begin{align*}
	D^{\ssc[11],{\sss I},m}(x,y)=&
	2\sum_{v,\bb_1} \left( \diagRepulsiveLetter{S}_{1,1,\underline 1,1}(\ve[\iota],v,\bb_1,0)C^{-1,m}(\bb_1,v,x,y)
+\frac{|\bb_1-\ve[\iota]|^2}{g_z}\diagRepulsiveLetter{S}_{1,1,\underline 1,1}(\ve[\iota],v,\bb_1,0) \bar A^{-1,m}(\bb_1,v,x,y)\right)\\
	&+2\sum_{v,\bb_1} \diagRepulsiveLetter{S}_{1,1,2,1}(\ve[\iota],v,\bb_1,0)C^{-2,m}(\bb_1,v,x,y)\\
	&+2\sum_{v,\bb_1}\frac {|\bb_1-\ve[\iota]|^2}{g_z}   \diagRepulsiveLetter{T}_{\bar 1,1,1}(-\bb_1,\ve[\iota]-\bb_1,v-\bb_1)  A^{m,-2}(x,y,\bb_1,v)\\
&+3\sum_{u,v,\bb_1} \diagRepulsiveLetter{B}_{1,1}(u,\ve[\iota])
	\diagRepulsiveLetter{T}_{1,\underline 1,1}(\bb_1,v,u)\big( C^{-1,m}(\bb_1,v,x,y)+\frac {1}{g_z}\bar A^{-1,m}(\bb_1,v,x,y)\big)\\
	&+3\sum_{u,v,\bb_1} \diagRepulsiveLetter{B}_{1,1}(u,\ve[\iota])
{|\bb_1|^2}{g_z}\diagRepulsiveLetter{T}_{\bar 1,\underline 1,1}(\bb_1,v,u)\bar A^{-1,m}(\bb_1,v,x,y)\\
	&+3\sum_{u,v,\bb_1} \diagRepulsiveLetter{B}_{1,1}(u,\ve[\iota]) \diagRepulsiveLetter{T}_{1,2,1}(\bb_1,v,u) \big(C^{-2,m}(\bb_1,v,x,y)+\frac 1 {g_z} \bar A^{-2,m}(\bb_1,v,x,y)\big)\\
 	&+3\sum_{u,v,\bb_1} \diagRepulsiveLetter{B}_{1,1}(u,\ve[\iota]) |\bb_1|^2 G_z(\bb_1)\tilde G_{1,z}(u-v)A^{m,-2}(x,y,\bb_1,v).
	\end{align*}
	
For $M=12$, we define
	\begin{align*}
	D^{\ssc[12],{\sss I},m}(x,y)=&
	2\sum_{v,\bb_1}\big(  \diagRepulsiveLetter{S}_{1,1,\underline 1,1}(\ve[\iota],\bb_1,v,0)C^{-1,m}(\bb_1,v,x,y)+\diagRepulsiveLetter{S}_{1,1,2,1}(\ve[\iota],\bb_1,v,0)C^{-2,m}(\bb_1,v,x,y)\big)\\
&+2\sum_{v,\bb_1}\frac {|\bb_1|^2}{g_z}\big( \diagRepulsiveLetter{S}_{1,\bar 1,\underline 1,1}(\ve[\iota],\bb_1,v,0)\bar A^{-1,m}(\bb_1,v,x,y)+\diagRepulsiveLetter{S}_{1,\bar 1,2,1}(\ve[\iota],\bb_1,v,0)\bar A^{-2,m} (\bb_1,v,x,y)\big)\\
	&+3\sum_{u,v,\bb_1} \diagRepulsiveLetter{B}_{1,1}(u,\ve[\iota]) \big(\diagRepulsiveLetter{T}_{1,\underline 1,1}(\bb_1,v,u)C^{-1,m}(\bb_1,v,x,y)+\diagRepulsiveLetter{T}_{1,2,1}(\bb_1,v,u)C^{-2,m}(\bb_1,v,x,y)\big)\\
	&+\frac {3}{g_z}\sum_{u,v,\bb_1}|u|^2 \diagRepulsiveLetter{B}_{\bar 1,1}(u,\ve[\iota]) \\
&\qquad \times \big(
\diagRepulsiveLetter{T}_{1,\underline 1,1}(\bb_1,v,u)\bar A^{-1,m}(\bb_1,v,x,y)
+\diagRepulsiveLetter{T}_{1,2,1}(\bb_1,v,u)\bar A^{-2,m}(\bb_1,v,x,y)\big)\\
&+\frac {3}{g_z}\sum_{u,v,\bb_1} \diagRepulsiveLetter{B}_{1,1}(u,\ve[\iota]) |u-\bb_1|^2\\
&\qquad \times \big(\diagRepulsiveLetter{T}_{\bar 1,\underline 1,1}(\bb_1,v,u)\bar A^{-1,m}(\bb_1,v,x,y)
+\diagRepulsiveLetter{T}_{\bar 1,2,1}(\bb_1,v,u)\bar A^{-2,m}(\bb_1,v,x,y)\big),
	\end{align*}
and
	\begin{align*}
D^{\ssc[12],{\sss I},m}(x,y)=&
	2\sum_{v,\bb_1}\big(  \diagRepulsiveLetter{S}_{1,1,\underline 1,1}(\ve[\iota],\bb_1,v,0)C^{-1,m}(\bb_1,v,x,y)+\diagRepulsiveLetter{S}_{1,1,2,1}(\ve[\iota],\bb_1,v,0)C^{-2,m}(\bb_1,v,x,y)\big)\\
&+2\sum_{v,\bb_1}\frac {|\bb_1-\ve[\iota]|^2}{g_z}\big( \diagRepulsiveLetter{S}_{1,\bar 1,\underline 1,1}(\ve[\iota],\bb_1,v,0)\bar A^{-1,m}(\bb_1,v,x,y)+\diagRepulsiveLetter{S}_{1,\bar 1,2,1}(\ve[\iota],\bb_1,v,0)\bar A^{-2,m} (\bb_1,v,x,y)\big)\\
	&+3\sum_{u,v,\bb_1} \diagRepulsiveLetter{B}_{1,1}(u,\ve[\iota]) \big(\diagRepulsiveLetter{T}_{1,\underline 1,1}(\bb_1,v,u)C^{-1,m}(\bb_1,v,x,y)+\diagRepulsiveLetter{T}_{1,2,1}(\bb_1,v,u)C^{-2,m}(\bb_1,v,x,y)\big)\\
	&+\frac {3}{g_z}\sum_{u,v,\bb_1}|u-\ve[\iota]|^2\diagRepulsiveLetter{B}_{ 1,\bar 1}(u,\ve[\iota]) \\
&\qquad \times \big(
\diagRepulsiveLetter{T}_{1,\underline 1,1}(\bb_1,v,u)\bar A^{-1,m}(\bb_1,v,x,y)
+\diagRepulsiveLetter{T}_{1,2,1}(\bb_1,v,u)\bar A^{-2,m}(\bb_1,v,x,y)\big)\\
&+\frac {3}{g_z}\sum_{u,v,\bb_1} \diagRepulsiveLetter{B}_{1,1}(u,\ve[\iota]) |u-\bb_1|^2\\
&\qquad \times \big(\diagRepulsiveLetter{T}_{\bar 1,\underline 1,1}(\bb_1,v,u)\bar A^{-1,m}(\bb_1,v,x,y)
+\diagRepulsiveLetter{T}_{\bar 1,2,1}(\bb_1,v,u)\bar A^{-2,m}(\bb_1,v,x,y)\big).
\end{align*}

For $M=13$, we define
	\begin{align}
	D^{\ssc[13],{\sss I},m}(x,y)=&
2\sum_{u,v,\bb_1} \diagRepulsiveLetter{T}_{1,1,0}(v,u,\ve[\iota])  \left( \diagRepulsiveLetter{B}_{1,1}(\bb_1,u)C^{-2,m}(\bb_1,v,x,y)+\frac {|\bb_1|^2}{g_z}
\diagRepulsiveLetter{B}_{\bar 1,1}(\bb_1,u)  \bar A^{-2,m}(\bb_1,v,x,y)\right),\nnb
	D^{\ssc[13],{\sss II},m}(x,y)=&
2\sum_{v,\bb_1} \diagRepulsiveLetter{B}_{1,1}(v,\ve[\iota])
\left( \diagRepulsiveLetter{B}_{1,1}(\bb_1,u) C^{-2,m}(\bb_1,v,x,y)+\frac {|\bb_1-\ve[\iota]|^2}{g_z}
\diagRepulsiveLetter{B}_{1,\bar 1}(\bb_1,\ve[\iota])  \bar A^{-2,m}(\bb_1,v,x,y)\right)\nnb
&+3\sum_{u,v,\bb_1} \diagRepulsiveLetter{T}_{1,1,1}(v,u,\ve[\iota])  \left( \diagRepulsiveLetter{B}_{1,1}(\bb_1,u)C^{-2,m}(\bb_1,v,x,y)+\frac {|\bb_1|^2}{g_z}
\diagRepulsiveLetter{B}_{\bar 1,1}(\bb_1,u)  \bar A^{-2,m}(\bb_1,v,x,y)\right)\nnb
&+3\sum_{u,v,\bb_1} \frac {|u-\ve[\iota]|^2}{g_z} \diagRepulsiveLetter{T}_{1,1,\bar 1}(v,u,\ve[\iota])
\diagRepulsiveLetter{B}_{1,1}(\bb_1,u) \bar A^{-2,m}(\bb_1,v,x,y).\nn
	\end{align}
	
For $M=14$, we define
	\begin{align}
	D^{\ssc[14],{\sss I},m}(x,y)=&
	2\sum_{u,\bb_1} \diagRepulsiveLetter{T}_{\underline 1,1,0}(\bb_1, u,\ve[\iota]) \tilde G_{3,z}(\bb_1)(1+ |x-\bb_1|^2) G_z(x-\bb_1)\diagRepulsiveLetter{B}_{0,0}(w-u,y-u)\nn\\
	&+2\sum_{u,\bb_1} \diagRepulsiveLetter{T}_{2,1,0}(\bb_1, u,\ve[\iota]) \tilde G_{2,z}(\bb_1) |x-\bb_1|^2 G_z(x-\bb_1)\diagRepulsiveLetter{B}_{0,0}(w-u,y-u)\nnb
	&+2\sum_{u,\bb_1} \frac{|\bb_1|^2}{g_z} \diagRepulsiveLetter{D}_{1}(\bb_1)  \diagRepulsiveLetter{B}_{1,0}(u-\bb_1,\ve[\iota]-\bb_1)  \bar A^{-2,m}(\bb_1,u,x,y),\nnb
	D^{\ssc[14],{\sss II},m}(x,y)=&
	2\sum_{u,\bb_1} \diagRepulsiveLetter{T}_{\underline 1,1,0}(\bb_1, u,\ve[\iota]) \tilde G_{3,z}(\bb_1)(|\bb_1-\ve[\iota]|^2+ |x-\bb_1|^2) G_z(x-\bb_1)\diagRepulsiveLetter{B}_{0,0}(w-u,y-u)\nn\\
	&+3\sum_{u,\bb_1} \diagRepulsiveLetter{T}_{2,1,0}(\bb_1, u,\ve[\iota]) \tilde G_{2,z}(\bb_1) (1+|x-\bb_1|^2) G_z(x-\bb_1)\diagRepulsiveLetter{B}_{0,0}(w-u,y-u)\nnb
	&+3\sum_{u,\bb_1}\frac{|\bb_1|^2}{g_z}|\bb_1|^2 \diagRepulsiveLetter{D}_{1}(\bb_1) \diagRepulsiveLetter{B}_{1,0}(u-\bb_1,\ve[\iota]-\bb_1)   \bar A^{-2,m}(\bb_1,u,x,y).\nn
	\end{align}
	
For $M=15$, we define
	\begin{align*}
	D^{\ssc[15],{\sss I},0}(x,y)=&\frac 2 {g_z}\delta_{x,y}\sum_{w,\bb_1}\diagRepulsiveLetter{D}_{1}(\bb_1)
\diagRepulsiveLetter{S}_{1,0,0,\bar 1}(\ve[\iota],w,x,\bb_1) (|\bb_1|^2+|\bb_1-x|^2)\nn\\ 
	&+2\delta_{x,y}\sum_{w,u,\bb_1}\diagRepulsiveLetter{D}_{1}(\bb_1) \diagRepulsiveLetter{B}_{1,1}(u,\ve[\iota]) \diagRepulsiveLetter{T}_{0,0,\bar 1}(w-u,x-u,\bb_1-u) (|\bb_1|^2+|\bb_1-x|^2),\nnb 	
	D^{\ssc[15],{\sss I},-1}(x,y)=&4d D(x-y)\sum_{w,\bb_1}\diagRepulsiveLetter{D}_{1}(\bb_1) \diagRepulsiveLetter{T}_{1,0,0}(\ve[\iota],w,y)  G_z(x-\bb_1)(|\bb_1|^2+|\bb_1-x|^2)\nn\\
	&+4d D(x-y)\sum_{w,u,\bb_1}\diagRepulsiveLetter{D}_{1}(\bb_1) \diagRepulsiveLetter{B}_{1,1}(u,\ve[\iota]) \diagRepulsiveLetter{B}_{1,0}(w-u,y-u)  G_z(x-\bb_1)\nn \\[-5mm]
	&\qquad\qquad\qquad\qquad\qquad\qquad\times (|\bb_1|^2+|\bb_1-x|^2),\nnb
	D^{\ssc[15],{\sss I},-2}(x,y)=&2\sum_{w,\bb_1}\diagRepulsiveLetter{D}_{1}(\bb_1) \diagRepulsiveLetter{T}_{1,0,0}(\ve[\iota],w,y)  G_z(x-\bb_1)(|\bb_1|^2+|\bb_1-x|^2)\nn\\
	&+2\sum_{w,u,\bb_1}\diagRepulsiveLetter{D}_{1}(\bb_1) \diagRepulsiveLetter{B}_{1,1}(u,\ve[\iota]) \diagRepulsiveLetter{B}_{1,0}(w-u,y-u)  G_z(x-\bb_1)2(|\bb_1|^2+|\bb_1-x|^2),\nnb
	D^{\ssc[15],{\sss I},1}(x,y)=&D^{\ssc[15],{\sss I},-1}(y,x),\qquad \quad D^{\ssc[15],{\sss I},2}(x,y)=D^{\ssc[15],{\sss I},-2}(y,x),\nn
	\end{align*}
and
	\begin{align}
	D^{\ssc[15],{\sss II},0}(x,y)=&\frac 3 2 D^{\ssc[15],{\sss I},0}(x,y)+3 B^{\ssc[15],0}(x,y),\nn\\
	D^{\ssc[15],{\sss II},-1}(x,y)=&\frac 3 2 D^{\ssc[15],{\sss I},-1}(x,y)+6 d D(x-y)\sum_{w,\bb_1}\diagRepulsiveLetter{D}_{1}(\bb_1) \diagRepulsiveLetter{T}_{1,0,0}(\ve[\iota],w,y)  G_z(x-\bb_1)\nn\\
	&+6d D(x-y)\sum_{w,u,\bb_1} \diagRepulsiveLetter{D}_{1}(\bb_1) \diagRepulsiveLetter{B}_{1,1}(\bb_1,\ve[\iota]) \diagRepulsiveLetter{B}_{0,0}(w-u,y-u)  G_z(x-\bb_1),\nnb
	D^{\ssc[15],{\sss II},-2}(x,y)=&\frac 3 2 D^{\ssc[15],{\sss I},-2}(x,y)+3\sum_{w,\bb_1}\diagRepulsiveLetter{D}_{2}(\bb_1) \diagRepulsiveLetter{T}_{1,0,0}(\ve[\iota],w,y) \tilde G_{1,z}(x-\bb_1)\nn\\
	&+3\sum_{w,\bb_1} \diagRepulsiveLetter{P}_{1,\underline 1, 1,0,0}(\bb_1-x,-x,\ve[\iota]-x,w-x,y-x)  G_{3,z}(\bb_1)\nnb 
	&+3\sum_{w,u,\bb_1} \diagRepulsiveLetter{D}_{1}(\bb_1) \diagRepulsiveLetter{B}_{1,1}(\bb_1,\ve[\iota]) \diagRepulsiveLetter{B}_{0,0}(w-u,y-u)  G_z(x-\bb_1),\nnb
	D^{\ssc[15],{\sss II},1}(x,y)=&D^{\ssc[15],{\sss II},-1}(y,x),\qquad \quad D^{\ssc[15],{\sss II},2}(x,y)=D^{\ssc[15],{\sss II},-2}(y,x).\nn
	\end{align}
	
For $M=16$, we define
	\begin{align*}
	D^{\ssc[16],{\sss I},0}(x,y)=&2\delta_{x,y}\sum_{u,v,w,\bb_1}
	\diagRepulsiveLetter{B}_{\bar 1,1}(\bb_1,v)\diagRepulsiveLetter{T}_{1,1,0}(v,u,\ve[\iota])\nn\\[-3mm]
	&\qquad \qquad \qquad \times\diagRepulsiveLetter{T}_{0,0,1}(w-u,x-u,\bb_1-u) |\bb_1|^2 \nnb
    &+2\delta_{x,y}\sum_{u,v,w,\bb_1}\diagRepulsiveLetter{B}_{1,1}(\bb_1,v)\diagRepulsiveLetter{T}_{1,1,0}(v,u,\ve[\iota])\nn\\[-3mm]
	&\qquad \qquad \qquad \times\diagRepulsiveLetter{T}_{0,0,\bar 1}(w-u,x-u,\bb_1-u) |\bb_1-x|^2, \nnb
	D^{\ssc[16],{\sss I},-1}(x,y)=&4dD(x-y)\sum_{u,v,w,\bb_1}
	\diagRepulsiveLetter{B}_{\bar 1,1}(\bb_1,v)\diagRepulsiveLetter{T}_{1,1,0}(v,u,\ve[\iota])\diagRepulsiveLetter{B}_{0,0}(w-u,x-u)\nn\\[-3mm]
	&\qquad \qquad \qquad \qquad\times G_z(x-\bb_1) |\bb_1|^2, \nnb
&+4dD(x-y)\sum_{u,v,w,\bb_1} \diagRepulsiveLetter{B}_{1,1}(\bb_1,v)\diagRepulsiveLetter{T}_{1,1,0}(v,u,\ve[\iota])\diagRepulsiveLetter{B}_{0,0}(w-u,x-u)\\[-3mm]
	&\qquad \qquad \qquad \qquad\times G_z(x-\bb_1)|\bb_1-x|^2, \nnb
	D^{\ssc[16],{\sss I},-1}(x,y)=&2\sum_{u,v,w,\bb_1}
	\diagRepulsiveLetter{B}_{\bar 1,1}(\bb_1,v)\diagRepulsiveLetter{T}_{1,1,0}(v,u,\ve[\iota])\diagRepulsiveLetter{B}_{0,0}(w-u,x-u)\nn\\[-3mm]
	&\qquad \qquad \qquad \qquad\times\bar G_z(x-\bb_1) |\bb_1|^2 \nnb
    &+\sum_{u,v,w,\bb_1} \diagRepulsiveLetter{B}_{1,1}(\bb_1,v)\diagRepulsiveLetter{T}_{1,1,0}(v,u,\ve[\iota])\diagRepulsiveLetter{B}_{0,0}(w-u,x-u)\nn\\[-3mm]
	&\qquad \qquad \qquad \qquad\times G_z(x-\bb_1)|\bb_1-x|^2, \nnb
	D^{\ssc[16],{\sss I},1}(x,y)=&D^{\ssc[16],{\sss I},-1}(y,x),\qquad \quad D^{\ssc[16],{\sss I},2}(x,y)=D^{\ssc[16],{\sss I},-2}(y,x),\nn
	\end{align*}
and
	\begin{align*}
	D^{\ssc[16],{\sss II},0}(x,y)=&\frac {3}{g_z}\delta_{x,y}\sum_{u,v,w,\bb_1}
	\diagRepulsiveLetter{B}_{\bar 1,1}(\bb_1,v)
\diagRepulsiveLetter{T}_{1,1,0}(v,u,\ve[\iota])\nn\\[-3mm]
	&\qquad \qquad \qquad \times\diagRepulsiveLetter{T}_{0,0,1}(w-u,x-u,\bb_1-u) (|\bb_1|^2+1), \nnb
&+\frac {3}{g_z}\delta_{x,y}\sum_{u,v,w,\bb_1}
	\diagRepulsiveLetter{B}_{ 1,1}(\bb_1,v)
\diagRepulsiveLetter{T}_{1,1,0}(v,u,\ve[\iota])\nn\\[-3mm]
	&\qquad \qquad \qquad \times\diagRepulsiveLetter{T}_{0,0,\bar 1}(w-u,x-u,\bb_1-u) |\bb_1-x|^2, \nnb
	D^{\ssc[16],{\sss II},-1}(x,y)=&3dD(x-y)\sum_{u,v,w,\bb_1}
	\diagRepulsiveLetter{B}_{\bar 1,1}(\bb_1,v)\diagRepulsiveLetter{T}_{1,1,0}(v,u,\ve[\iota])\diagRepulsiveLetter{B}_{0,0}(w-u,x-u)\nn\\[-3mm]
	&\qquad \qquad \qquad \qquad\times G_z(x-\bb_1) (|\bb_1|^2+|\bb_1-x|^2+1), \nnb
	D^{\ssc[16],{\sss II},-2}(x,y)=&3\sum_{u,v,w,\bb_1}
	\diagRepulsiveLetter{B}_{\bar 1,1}(\bb_1,v)\diagRepulsiveLetter{T}_{1,1,0}(v,u,\ve[\iota])\diagRepulsiveLetter{B}_{0,0}(w-u,x-u)\nn\\[-3mm]
	&\qquad \qquad \qquad \qquad\times G_z(x-\bb_1) (|\bb_1|^2+|\bb_1-x|^2+1), \nnb
	D^{\ssc[16],{\sss II},1}(x,y)=&D^{\ssc[16],{\sss II},-1}(y,x),\qquad \quad D^{\ssc[16],{\sss II},2}(x,y)=D^{\ssc[16],{\sss II},-2}(y,x).\nn
	\end{align*}
This defines all the bounding diagrams to the contributions to the weighted diagrams in \eqref{two-weighted-Xis}, where the weight falls on the initial piece of the diagram.
\fi

\paragraph{Acknowledgements.}
This work was supported in part by the Netherlands Organisation for Scientific Research (NWO) through VICI grant 639.033.806 and the Gravitation {\sc Networks} grant 024.002.003. We thank
David Brydges, Takashi Hara and Gordon Slade for their constant encouragement, as well as for several stimulating discussions. This work builds upon the work by Takashi Hara and Gordon Slade, originally used for self-avoiding walk and percolation. We are indebted to Takashi for his help in the proof of Theorem \ref{thm-x-space}, which relies on an improved version of this analysis in \cite{Hara08} that Takashi shared with us in 2015. Early 2019, Takashi significantly helped us once more by clarifying the requirements for  \cite{Hara08} to apply for LTs and LAs (see Section \ref{sec-x-space}). The work of RF was partially performed while being employed by the Institute for Complex Molecular Systems at Eindhoven University of Technology.

\bibliographystyle{plain}
\bibliography{robert}

\begin{thebibliography}{10}

\bibitem{BauBrySla15b}
R.~Bauerschmidt, D.~Brydges, and G.~Slade.
\newblock Critical two-point function of the 4-dimensional weakly self-avoiding
  walk.
\newblock {\em Comm. Math. Phys.}, {\bf 338}(1):169--193, (2015).

\bibitem{BauBrySla15a}
R.~Bauerschmidt, D.~Brydges, and G~Slade.
\newblock Logarithmic correction for the susceptibility of the 4-dimensional
  weakly self-avoiding walk: a renormalisation group analysis.
\newblock {\em Comm. Math. Phys.}, {\bf 337}(2):817--877, (2015).

\bibitem{BorChaHofSla99}
C.~Borgs, J.~Chayes, R.~van~der Hofstad, and G.~Slade.
\newblock Mean-field lattice trees.
\newblock {\em Ann. Comb.}, {\bf 3}(2-4):205--221, (1999).
\newblock On combinatorics and statistical mechanics.

\bibitem{BovFroGla86b}
A.~Bovier, J.~Fr{\"o}hlich, and U.~Glaus.
\newblock Branched polymers and dimensional reduction.
\newblock In {\em Ph\'enom\`enes critiques, syst\`emes al\'eatoires, th\'eories
  de jauge, Part I, II (Les Houches, 1984)}, pages 725--893. North-Holland,
  Amsterdam, (1986).

\bibitem{BryImb03}
D.~Brydges and J.~Imbrie.
\newblock Branched polymers and dimensional reduction.
\newblock {\em Ann. of Math. (2)}, {\bf 158}(3):1019--1039, (2003).

\bibitem{BrySpe85}
D.C. Brydges and T.~Spencer.
\newblock Self-avoiding walk in 5 or more dimensions.
\newblock {\em Commun. Math. Phys.}, {\bf 97}:125--148, (1985).

\bibitem{DerSla97}
E.~Derbez and G.~Slade.
\newblock Lattice trees and super-{Brownian} motion.
\newblock {\em Canad.\ Math.\ Bull.}, {\bf 40}:19--38, (1997).

\bibitem{DerSla98}
E.~Derbez and G.~Slade.
\newblock The scaling limit of lattice trees in high dimensions.
\newblock {\em Commun.\ Math.\ Phys.}, {\bf 193}:69--104, (1998).

\bibitem{FitNoblePage}
R.~Fitzner.
\newblock www.fitzner.nl/noble/.

\bibitem{Fit13}
R.~Fitzner.
\newblock Non-backtracking lace expansion.
\newblock {\em PhD. thesis, TU Eindhoven}, (2013).

\bibitem{FitHof13a}
R.~Fitzner and R.~van~der Hofstad.
\newblock Non-backtracking random walk.
\newblock {\em J. Statist. Phys.}, {\bf 150}(2):264--284, (2013).

\bibitem{FitHof13b}
R.~Fitzner and R.~van~der Hofstad.
\newblock Generalized approach to the non-backtracking lace expansion.
\newblock {\em Probab. Theory Related Fields}, {\bf 169}(3-4):1041--1119,
  (2017).

\bibitem{FitHof13d}
R.~Fitzner and R.~van~der Hofstad.
\newblock Mean-field behavior for nearest-neighbor percolation in {$d>10$}.
\newblock {\em Electron. J. Probab.}, {\bf 22}:Paper No. 43, 65, (2017).

\bibitem{FitHof13g-ext}
R.~Fitzner and R.~van~der Hofstad.
\newblock {NoBLE} for lattice trees and lattice animals: Extended version.
\newblock (2019).

\bibitem{Hara08}
T.~Hara.
\newblock Decay of correlations in nearest-neighbor self-avoiding walk,
  percolation, lattice trees and animals.
\newblock {\em Ann. Probab.}, {\bf 36}(2):530--593, (2008).

\bibitem{Hara08ext}
T.~Hara.
\newblock Private communication.
\newblock (2015).

\bibitem{HarSla90a}
T.~Hara and G.~Slade.
\newblock Mean-field critical behaviour for percolation in high dimensions.
\newblock {\em Commun. Math. Phys.}, {\bf 128}:333--391, (1990).

\bibitem{HarSla90b}
T.~Hara and G.~Slade.
\newblock On the upper critical dimension of lattice trees and lattice animals.
\newblock {\em J. Stat. Phys.}, {\bf 59}:1469--1510, (1990).

\bibitem{HarSla92b}
T.~Hara and G.~Slade.
\newblock The lace expansion for self-avoiding walk in five or more dimensions.
\newblock {\em Reviews in Math.\ Phys.}, {\bf 4}:235--327, (1992).

\bibitem{HarSla92c}
T.~Hara and G.~Slade.
\newblock The number and size of branched polymers in high dimensions.
\newblock {\em J. Stat. Phys.}, {\bf 67}:1009--1038, (1992).

\bibitem{HarSla92a}
T.~Hara and G.~Slade.
\newblock Self-avoiding walk in five or more dimensions. {I.} {The} critical
  behaviour.
\newblock {\em Commun.\ Math.\ Phys.}, {\bf 147}:101--136, (1992).

\bibitem{HarTas87}
T.~Hara and H.~Tasaki.
\newblock Critical behavior in a system of branched polymers.
\newblock {\em Progress of Theoretical Physics Supplement}, {\bf 92}:14--25,
  (1987).

\bibitem{HofHol13}
R.~van~der Hofstad and M.~Holmes.
\newblock The survival probability and {$r$}-point functions in high
  dimensions.
\newblock {\em Ann. of Math. (2)}, {\bf 178}(2):665--685, (2013).

\bibitem{HofHolPer15}
R.~van~der Hofstad, M.~Holmes, and E.~Perkins.
\newblock A criterion for convergence to super-{B}rownian motion on path space.
\newblock {\em Ann. Probab.}, {\bf 45}(1):278--376, (2017).

\bibitem{HofSak04}
R.~van~der Hofstad and A.~Sakai.
\newblock Gaussian scaling for the critical spread-out contact process above
  the upper critical dimension.
\newblock {\em Electron. J. Probab.}, {\bf 9}:710--769 (electronic), (2004).

\bibitem{HofSla02}
R.~van~der Hofstad and G.~Slade.
\newblock A generalised inductive approach to the lace expansion.
\newblock {\em Probab. Theory Related Fields}, {\bf 122}(3):389--430, (2002).

\bibitem{Holm08b}
M.~Holmes.
\newblock Convergence of lattice trees to super-{B}rownian motion above the
  critical dimension.
\newblock {\em Electron. J. Probab.}, {\bf 13}:no. 23, 671--755, (2008).

\bibitem{HolPer07}
M.~Holmes and E.~Perkins.
\newblock Weak convergence of measure-valued processes and {$r$}-point
  functions.
\newblock {\em Ann. Probab.}, {\bf 35}(5):1769--1782, (2007).

\bibitem{HsuNadGra05}
H.~Hsu, W.~Nadler, and P.~Grassberger.
\newblock Simulations of lattice animals and trees.
\newblock {\em J. Phys. A}, {\bf 38}(4):775--806, (2005).

\bibitem{Jens08}
I.~Jensen.
\newblock Enumerations of lattice animals and trees.
\newblock {\em Journal of Statistical Physics}, {\bf 102}(3-4):865--881,
  (2001).

\bibitem{KenWin09}
R.~Kenyon and P.~Winkler.
\newblock {Branched Polymers}.
\newblock {\em American Mathematical Monthly}, {\bf 116}(7):612--628, (2009).

\bibitem{LubIsa79}
T.~C. Lubensky and Joel Isaacson.
\newblock Statistics of lattice animals and dilute branched polymers.
\newblock {\em Phys. Rev. A}, {\bf 20}(5):2130--2146, Nov (1979).

\bibitem{MadSla93}
N.~Madras and G.~Slade.
\newblock {\em The Self-Avoiding Walk}.
\newblock Birkh{\"a}user, Boston, (1993).

\bibitem{MirSla12}
Y.~Mej\'{i}a~Miranda and G.~Slade.
\newblock Expansion in high dimension for the growth constants of lattice trees
  and lattice animals.
\newblock {\em Combin. Probab. Comput.}, {\bf 22}(4):527--565, (2013).

\bibitem{NguYan93}
B.G. Nguyen and W-S. Yang.
\newblock Triangle condition for oriented percolation in high dimensions.
\newblock {\em Ann.\ Probab.}, {\bf 21}:1809--1844, (1993).

\bibitem{NguYan95}
B.G. Nguyen and W-S. Yang.
\newblock Gaussian limit for critical oriented percolation in high dimensions.
\newblock {\em J. Stat. Phys.}, {\bf 78}(3):841--876, (1995).

\bibitem{ParSou81}
G.~Parisi and N.~Sourlas.
\newblock Critical behavior of branched polymers and the {L}ee-{Y}ang edge
  singularity.
\newblock {\em Phys. Rev. Lett.}, {\bf 46}(14):871--874, (1981).

\bibitem{Perk02}
E.~Perkins.
\newblock Dawson-{W}atanabe superprocesses and measure-valued diffusions.
\newblock In {\em Lectures on probability theory and statistics (Saint-Flour,
  1999)}, volume~{\bf 1781} of {\em Lecture Notes in Math.}, pages 125--324.
  Springer, Berlin, (2002).

\bibitem{Saka01}
A.~Sakai.
\newblock Mean-field critical behavior for the contact process.
\newblock {\em J. Statist. Phys.}, {\bf 104}(1-2):111--143, (2001).

\bibitem{Saka07}
A.~Sakai.
\newblock Lace expansion for the {I}sing model.
\newblock {\em Comm. Math. Phys.}, {\bf 272}(2):283--344, (2007).

\bibitem{Slad87}
G.~Slade.
\newblock The diffusion of self-avoiding random walk in high dimensions.
\newblock {\em Commun. Math. Phys.}, {\bf 110}:661--683, (1987).

\bibitem{Slad06}
G.~Slade.
\newblock {\em The lace expansion and its applications}, volume~{\bf 1879} of
  {\em Lecture Notes in Mathematics}.
\newblock Springer-Verlag, Berlin, (2006).

\bibitem{Tasa86}
H.~Tasaki.
\newblock Stochastic geometric methods in statistical physics and field
  theories.
\newblock {\em PhD. thesis, University of Tokio}, (1986).

\end{thebibliography}


\end{document}